\documentclass[11pt]{book}
\usepackage{fullpage}

\usepackage{algorithm}
\usepackage{algpseudocode}
\usepackage{amsmath}
\usepackage{amssymb}
\usepackage{amsthm}
\usepackage{bm}
\usepackage{bbm}
\usepackage{changepage}
\usepackage{color}
\usepackage{graphicx}
\usepackage{ifthen}
\usepackage{longtable}
\usepackage{mathtools}
\usepackage{multirow}
\usepackage{nicefrac}
\usepackage{psfrag}
\usepackage{ragged2e}
\usepackage{soul}
\usepackage{xcolor}

\usepackage[export]{adjustbox}

\usepackage{acro}[=v2]
\ExplSyntaxOn
\prop_new:N \l__acro_foreign_format_prop
\ExplSyntaxOff

\usepackage[bookmarks=false]{hyperref}
\newtheorem{define}{Definition}[chapter]
\newtheorem{prop}{Proposition}
\newtheorem{statement}{Statement}[chapter]

\newcommand{\Kt}{K_{\text{tot}}}
\newcommand{\Ka}{K_a}
\newcommand{\br}[1]{\left(#1\right)}

\newcommand{\Fig}[1]{Fig.~\textup{\ref{#1}}}

\DeclareMathOperator*{\argmax}{arg\,max}
\DeclareMathOperator*{\argmin}{arg\,min}

\newcommand{\bY}{\mathbf{Y}}

\newcommand{\norm}[1]{\left\| #1 \right\|}

\newcommand{\Pb}[1]{\Pr\left[#1 \right]}

\newcommand{\Exp}[1]{\exp\left(#1 \right)}
\newcommand{\determ}[1]{\left|#1\right|}

\newcommand\cn[1]{\mathcal{CN}\left(#1\right)}
\newcommand\prob[1]{\Pr\left[#1\right]}
\newcommand\Span[1]{\text{span}\left\{#1\right\}}
\newcommand\rscalar[1]{\mathsf{#1}}
\newcommand\rvec[1]{\mathsf{\mathbf{#1}}}
\newcommand\rmat[1]{\bm{\mathsf{#1}}}
\newcommand\dvec[1]{\bm{#1}}
\newcommand\dset[1]{\mathcal{#1}}
\newcommand\vdset[1]{\bm{\mathcal{#1}}}
\newcommand\dmat[1]{\bm{#1}}

\newcommand\logdet[1]{\log\left|#1\right|}

\DeclareRobustCommand{\stirling}{\genfrac\{\}{0pt}{}}

\newcommand{\norms}[1]{\norm{#1}^2}
\newcommand{\sbr}[1]{\left[#1\right]}
\newcommand{\brs}[1]{\left[#1\right]}
\newcommand{\brc}[1]{\left\{#1\right\}}

\newcommand\ind[3]{ {#1}^{(#2)}_{#3} }
\newcommand\beg[2]{ {#1}_{[#2]}}

\newcommand{\Gens}{\mathcal{G}\br{M,n,P}}
\newcommand{\Gensp}{\mathcal{G}\br{M,n,P'}}
\newcommand{\Gensr}{\mathcal{G}\br{2^k,n',P}}
\newcommand{\Bens}{\mathcal{B}\br{M,n,P}}

\newcommand{\pM}{p_{m}} % Inner code missed detection probability
\newcommand{\pF}{p_{f}} % Inner code false detection probability

\newcommand{\bx}{\mathbf{x}}

\newcommand{\code}{\mathcal{C}}

\newcommand{\rbz}{\rvec{z}}
\newcommand{\dbz}{\dvec{z}}

\newcommand{\rbcm}{\rvec{s}_{\dset{M}}}
\newcommand{\dbcm}{\dvec{s}_{\dset{M}}}
\newcommand{\rbcmp}{\rvec{s}_{\dset{M}'}}
\newcommand{\dbcmp}{\dvec{s}_{\dset{M}'}}

\newcommand{\rbcf}{\rvec{s}_{\dset{F}}}

\newcommand{\rbcfp}{\rvec{s}_{\dset{F}'}}

\newcommand{\rbcc}{\rvec{s}_{\dset{C}}}

\newcommand{\rbct}{\rvec{s}_{\dset{T}}}
\newcommand{\dbct}{\dvec{s}_{\dset{T}}}

\newcommand{\cov}{\mathrm{Cov}}

\newcommand{\bzero}{\dvec{0}}
\newcommand{\bI}{\dmat{I}}
\newcommand{\be}{\pmb{\eta}}
\newcommand{\bA}{\dmat{A}}
\newcommand{\bB}{\dmat{B}}
\newcommand{\rank}{\mathrm{rank}}
\newcommand{\supp}{\mathrm{supp}}
\newcommand{\wt}[1]{\mathrm{wt}_H\left\{#1\right\}}

\newcommand{\1}{\mathbbm{1}}

\algnewcommand\algorithmicinput{\textbf{Input:}}
\algnewcommand\INPUT{\item[\algorithmicinput]}
\algnewcommand\algorithmicoutput{\textbf{Output:}}
\algnewcommand\OUTPUT{\item[\algorithmicoutput]}
\algnewcommand{\LineComment}[1]{\State \(\triangleright\) #1}

\def\PgfColorA{0, 0, 0}
\def\PgfColorB{27,158,119}
\def\PgfColorC{217,95,2}
\def\PgfColorD{117,112,179}
\def\PgfColorE{231,41,138}
\def\PgfColorF{102,166,30}
\def\PgfColorG{230,171,2}
\def\PgfColorH{166,118,29}

\definecolor{MainColorA}{RGB}{\PgfColorA}
\definecolor{MainColorB}{RGB}{\PgfColorB}
\definecolor{MainColorC}{RGB}{\PgfColorC}
\definecolor{MainColorD}{RGB}{\PgfColorD}
\definecolor{MainColorE}{RGB}{\PgfColorE}
\definecolor{MainColorF}{RGB}{\PgfColorF}
\definecolor{MainColorG}{RGB}{\PgfColorG}
\definecolor{MainColorH}{RGB}{\PgfColorH}

\newtheorem{corollary}{Corollary}[chapter]
\newtheorem{definition}{Definition}[chapter]
\newtheorem{example}{Example}[chapter]
\newtheorem{lemma}{Lemma}[chapter]
\newtheorem{proposition}{Proposition}[chapter]
\newtheorem{remark}{Remark}[chapter]
\newtheorem{theorem}{Theorem}[chapter]

\allowdisplaybreaks

\title{\Huge \sf Unsourced Random Access \vspace{1in}}

\author{
\begin{tabular}{ll}
{\sf Kirill Andreev} & \href{mailto:k.andreev@skoltech.ru}{k.andreev@skoltech.ru} \\[2mm]
{\sf Pavel Rybin},   & \href{mailto:p.rybin@skoltech.ru}{p.rybin@skoltech.ru}     \\[2mm]
{\sf Alexey Frolov}, & \href{mailto:al.frolov@skoltech.ru}{al.frolov@skoltech.ru}
\end{tabular}
}

\date{\sf Skolkovo Institute of Science and Technology}

\DeclareAcronym{3gpp}{
short=3GPP,
long=Third-Generation Partnership Project
}

\DeclareAcronym{5gnr}{
short=5G NR,
long=Fifth-generation New Radio
}

\DeclareAcronym{ack}{
short=ACK,
long=acknowledgment
}

\DeclareAcronym{amp}{
short=AMP,
long=approximate message passing
}

\DeclareAcronym{aoa}{
short=AoA,
long=angle-of-arrival
}

\DeclareAcronym{arq}{
short=ARQ,
long=Automatic Repeat reQuest
}

\DeclareAcronym{asr}{
short=ASR,
long=approximate support recovery
}

\DeclareAcronym{awgn}{
short=AWGN,
long=additive white Gaussian noise
}

\DeclareAcronym{bac}{
short=BAC,
long=binary adder channel
}

\DeclareAcronym{bch}{
short=BCH,
long=Bose-Chaudhuri-Hocquenghem
}

\DeclareAcronym{bp}{
short=BP,
long=belief propagation
}

\DeclareAcronym{bpsk}{
short=BPSK,
long=binary phase-shift keying
}

\DeclareAcronym{bs}{
short=BS,
long=base station
}

\DeclareAcronym{ccs}{
short=CCS,
long=coded compressed sensing
}

\DeclareAcronym{cdma}{
short=CDMA,
long=code division multiple-access
}

\DeclareAcronym{cnop}{
short=CNOP,
long=check-node operation
}

\DeclareAcronym{crc}{
short=CRC,
long=cyclic redundancy check
}

\DeclareAcronym{crdsa}{
short=CRDSA,
long=contention resolution diversity slotted ALOHA
}

\DeclareAcronym{cs}{
short=CS,
long=compressed sensing
}

\DeclareAcronym{csi}{
short=CSI,
long=channel state information
}

\DeclareAcronym{csma}{
short=CSMA,
long=carrier-sense multiple-access
}
\DeclareAcronym{csmaca}{
short=CSMA/CA,
long=carrier-sense multiple-access with collision avoidance
}

\DeclareAcronym{css}{
short=CSS,
long=chirp spread spectrum
}

\DeclareAcronym{dbpsk}{
short=DBPSK,
long=differential binary phase-shift keying
}

\DeclareAcronym{de}{
short=DE,
long=density evolution
}

\DeclareAcronym{dft}{
short=DFT,
long=discrete Fourier transform
}

\DeclareAcronym{dsss}{
short=DSSS,
long=direct-sequence spread spectrum
}

\DeclareAcronym{far}{
short=FAR,
long=false alarm rate
}

\DeclareAcronym{fbl}{
short=FBL,
long=finite blocklength
}

\DeclareAcronym{fcfs}{
short=FCFS,
long={first-come, first-serve}
}

\DeclareAcronym{fdma}{
short=FDMA,
long=frequency division multiple-access
}

\DeclareAcronym{fft}{
short=FFT,
long=fast Fourier transform
}

\DeclareAcronym{ft}{
short=FT,
long=Fourier transform
}

\DeclareAcronym{gmac}{
short=GMAC,
long=Gaussian MAC
}

\DeclareAcronym{gura}{
short=G-URA,
long=URA over Gaussian MAC
}

\DeclareAcronym{idma}{
short=IDMA,
long=interleave division multiple-access
}

\DeclareAcronym{iot}{
short=IoT,
long=Internet of Things
}

\DeclareAcronym{irsa}{
short=IRSA,
long=irregular repetition slotted ALOHA
}

\DeclareAcronym{irsab}{
short=IRSA-B,
long={basic IRSA protocol}
}

\DeclareAcronym{irsaf}{
short=IRSA-F,
long={IRSA with per-frame preambles}
}

\DeclareAcronym{irsas}{
short=IRSA-S,
long={IRSA with per-slot preambles}
}

\DeclareAcronym{jsc}{
short=JSC,
long=joint successive cancellation
}

\DeclareAcronym{ldpc}{
short=LDPC,
long=low-density parity check
}

\DeclareAcronym{lds}{
short=LDS,
long=low-density signature
}

\DeclareAcronym{llr}{
short=LLR,
long=log-likelihood ratio
}

\DeclareAcronym{lmmse}{
short=LMMSE,
long=linear MMSE
}
\DeclareAcronym{lpwan}{
short=LPWAN,
long=low-power wide area network
}

\DeclareAcronym{mac}{
short=MAC,
long=multiple-access channel
}

\DeclareAcronym{mf}{
short=MF,
long=matched filter
}

\DeclareAcronym{mimo}{
short=MIMO,
long={multiple input, multiple output}
}

\DeclareAcronym{ml}{
short=ML,
long=maximum likelihood
}

\DeclareAcronym{mmse}{
short=MMSE,
long=minimum mean squared error
}

\DeclareAcronym{mmv}{
short=MMV,
long=multiple measurement vector
}

\DeclareAcronym{mpa}{
short=MPA,
long=message passing algorithm
}

\DeclareAcronym{mse}{
short=MSE,
long=mean squared error
}

\DeclareAcronym{mtc}{
short=MTC,
long=machine-type communications
}

\DeclareAcronym{mud}{
short=MUD,
long=multi-user detector
}

\DeclareAcronym{musa}{
short=MUSA,
long=multi-user shared access
}

\DeclareAcronym{nnls}{
short=NNLS,
long=non-negative least squares
}

\DeclareAcronym{noma}{
short=NOMA,
long=non-orthogonal multiple-access
}

\DeclareAcronym{odma}{
short=ODMA,
long=on-off division multiple access
}

\DeclareAcronym{ofdm}{
short=OFDM,
long=orthogonal frequency-division multiplexing
}

\DeclareAcronym{ofdma}{
short=OFDMA,
long=orthogonal frequency-division multiple-access
}

\DeclareAcronym{omp}{
short=OMP,
long=orthogonal matching pursuit
}

\DeclareAcronym{pan}{
short=PAN,
long=personal-area network
}

\DeclareAcronym{pexit}{
short=PEXIT,
long=protograph extrinsic information transfer
}

\DeclareAcronym{po}{
short=PO,
long=physical uplink shared channel occasion
}

\DeclareAcronym{pupe}{
short=PUPE,
long=per-user probability of error
}

\DeclareAcronym{pusch}{
short=PUSCH,
long=physical uplink shared channel
}

\DeclareAcronym{qos}{
short=QoS,
long=quality of service
}

\DeclareAcronym{qpsk}{
short=QPSK,
long=quadrature phase-shift keying
}

\DeclareAcronym{ra}{
short=RA,
long=random access
}

\DeclareAcronym{rach}{
short=RACH,
long=random access channel
}

\DeclareAcronym{rcb}{
short=RCB,
long=random coding bound
}

\DeclareAcronym{rfid}{
short=RFID,
long=radio-frequency identification
}

\DeclareAcronym{rip}{
short=RIP,
long=restricted isometry property
}

\DeclareAcronym{rm}{
short=RM,
long=Reed-Muller
}

\DeclareAcronym{rs}{
short=RS,
long=Reed-Solomon
}

\DeclareAcronym{rsma}{
short=RSMA,
long=rate-splitting multiple-access
}

\DeclareAcronym{sa}{
short=SA,
long=slotted ALOHA
}

\DeclareAcronym{sc}{
short=SC,
long=successive cancellation
}

\DeclareAcronym{scl}{
short=SCL,
long=successive cancellation list
}

\DeclareAcronym{scma}{
short=SCMA,
long=sparse-coded multiple-access
}

\DeclareAcronym{sf}{
short=SF,
long=spreading factor
}

\DeclareAcronym{sic}{
short=SIC,
long=successive interference cancellation
}

\DeclareAcronym{snr}{
short=SNR,
long=signal-to-noise ratio
}

\DeclareAcronym{sinr}{
short=SINR,
long=signal-to-interference-plus-noise ratio
}

\DeclareAcronym{soic}{
short=SoIC,
long=soft SIC
}

\DeclareAcronym{sparcs}{
short=SPARCs,
long=sparse regression codes
}

\DeclareAcronym{spc}{
short=SPC,
long=single-parity-check
}

\DeclareAcronym{tdma}{
short=TDMA,
long=time-division multiple-access,
}

\DeclareAcronym{tin}{
short=TIN,
long=treat interference as noise
}

\DeclareAcronym{tinsic}{
short=TIN-SIC,
long=TIN followed by SIC
}

\DeclareAcronym{ue}{
short=UE,
long=user equipment
}

\DeclareAcronym{ura}{
short=URA,
long=unsourced random access
}

\DeclareAcronym{vnop}{
short=VNOP,
long=variable-node operation
}

\begin{document}

\maketitle
\tableofcontents

\newpage

\printacronyms
\chapter*{Abstract}
Current wireless networks are designed to optimize spectral efficiency for human users, who typically require sustained connections for high-data-rate applications like file transfers and video streaming. However, these networks are increasingly inadequate for the emerging era of \ac{mtc}. With a vast number of devices exhibiting sporadic traffic patterns consisting of short packets, the grant-based multiple access procedures utilized by existing networks lead to significant delays and inefficiencies. To address this issue the \ac{ura} paradigm has been proposed. This paradigm assumes the devices to share a common encoder thus simplifying the reception process by eliminating the identification procedure. The \ac{ura} paradigm not only addresses the computational challenges but it also considers the \ac{ra} as a coding problem, i.e., takes into account both medium access protocols and physical layer effects. In this monograph we provide a comprehensive overview of the \ac{ura} problem in noisy channels, with the main task being to explain the major ideas rather than to list all existing solutions.

\chapter*{Notation}
Throughout the monograph we use the following notations and abbreviations:

\begin{longtable}{p{.15\textwidth}  p{.75\textwidth} } 
$x$, $X$ & deterministic scalar value \\
$\dvec{x}$ & deterministic column-vector \\
$\dmat{X}$ & deterministic matrix \\
$\dmat{I}_n$ & $n \times n$ identity matrix \\
$\mathrm{diag}$ & diagonal matrix \\
$\norm{\dvec{x}}_2$ & Euclidean norm of vector $\dvec{x}$ \\
$\mathrm{supp}\br{\dvec{x}}$ & support of vector $\dvec{x}$ \\
$\mathrm{wt}  \br{\dvec{x}}$ & the number of nonzero components in vector $\dvec{x}$ \\
$\mathbb{N}$ & set of natural numbers  \\
$\mathbb{C}$ & set of complex numbers  \\
$\mathbb{F}_q$ & finite field with $q$ elements \\
$\dset{X}$ & set \\
$\bm{\mathcal{X}}$ & sequence of sets \\
$\bigsqcup$ & disjoint union \\
$[n]$ & $[n] = \{1,\dots,n\}$, where $n \in \mathbb{N}$ \\
$\dvec{a}_\dset{I}$ & $\dvec{a}_\dset{I} = (a_{i_1},\ldots,a_{i_s})$, where $\dvec{a} = \br{a_1, \ldots, a_n}$ and $\dset{I} = \{i_1, \ldots, i_s\} \subseteq [n]$ with $i_1 < \ldots < i_s$ \\
$\dmat{A}_\dset{I}$ & $\dmat{A}_\dset{I} = \br{\dvec{a}_{i_1},\ldots,\dvec{a}_{i_s}}$, where $\dmat{A} = \br{\dvec{a}_1, \ldots, \dvec{a}_n}$ and $\dset{I} = \{i_1, \ldots, i_s\} \subseteq [n]$ with $i_1 < \ldots < i_s$ \\
$\cn{\bm{0},\dmat{I}_n}$ & circularly symmetric complex standard normal distribution \\
$\mathrm{Bern}(p)$ & Bernoulli distribution with parameter $p$ \\
$\mathrm{Unif}([Q])$ & uniform distribution on $[Q]$\\
$\rscalar{x}$, $\rscalar{X}$ & random scalar value \\
$\bm{\mathsf{x}}$ & random column-vector \\
$\rmat{X}$ & random matrix \\
$E$ & event \\
$E^c$ & complementary event to $E$ \\
$\1_{E}$ & indicator of the event $E$ \\
$\prob{E}$ & probability of the event $E$  \\
$\mathbb{E}$ & expectation operator  \\
$H(\rscalar{x})$ & entropy of discrete random variable $\rscalar{x}$  \\
$I(\rscalar{x}, \rscalar{y})$ & mutual information of random variables $\rscalar{x}$ and $\rscalar{y}$ \\
% $h(p)$\footnote{We assume that $0\log_{2}(0)$ is defined to be $0$, thus $h(1)=h(0)=0$} & $h(p)=-p\log_2(p)-(1-p)\log_2(1-p)$, $0\leq p\leq 1$ \\
$h(p)$ & $h(p)=-p\log_2(p)-(1-p)\log_2(1-p)$, $0\leq p\leq 1$, where $h(1)=h(0)=0$ \\
% $\tau(x)$ & \ac{bpsk} modulation $\tau: \{0,1\} \rightarrow \{ +\sqrt{P}, -\sqrt{P}\}$ with power $P$ \\
$\tau(x)$ & \ac{bpsk} modulation $\tau\br{x} = \br{1-2x}\sqrt{P}$ \\
w.l.o.g. & ``without loss of generality'' \\
r.v. & ``random variable'' \\
i.i.d. & ``independent identically distributed'' \\
p.m.f. & ``probability mass function'' \\
\end{longtable}
\chapter{Massive \acl{mtc}}\label{label:chap1}

\Acf{mtc} dramatically change traffic patterns. Instead of focusing on peak data rates and low latencies, massive connectivity becomes a key requirement. The \ac{mtc} concept involves a massive number of autonomous devices and sensors being connected to a gateway: a node (or a set of nodes) responsible for data collection. \Ac{mtc} is a crucial component of the \ac{iot} paradigm, which defines the infrastructure and scenarios for interconnecting devices rather than humans. \Ac{iot} encompasses various tasks, including monitoring, remote and automated control, data collection, and data-related services. \ac{mtc} is a communication technology specifically designed to support this paradigm, enabling connectivity for a vast number of devices. In this monograph, we focus on the communication aspects and use the terms \ac{iot} and \ac{mtc} interchangeably.

\Ac{iot} applications encompass a wide range of use cases, including:
\begin{itemize}
\item Environmental and health monitoring.
\item Smart homes, cities, and industries.
\item Road traffic monitoring and tracking to improve efficiency and safety.
\end{itemize}

Typical \ac{mtc} transmissions involve short measurement reports generated either regularly or sporadically, resulting in additional requirements.

\begin{figure}[t]
\centering
\includegraphics{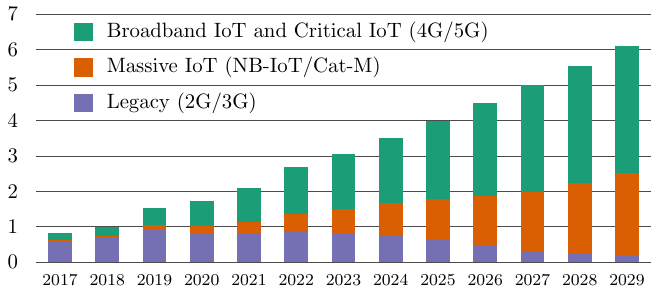}
\caption{Predicted number of worldwide \ac{iot} cellular subscriptions (billions) in accordance with Ericsson mobility report~\cite{Ericsson202311} (Nov. 2023).\label{fig:iot_growth}}
\end{figure}

\begin{enumerate}
\item Improved \emph{battery life} is essential. Wiring a large number of devices to the electricity grid would require expensive cabling, making it preferable for these devices to be autonomous. Moreover, monitoring the battery status of thousands (or even millions) of devices may also be prohibitively costly. Therefore, the battery lifetime must match the device lifetime (approximately 10 years). As a result, energy-efficient solutions must utilize simple radio-frequency devices with low-complexity signal processing algorithms.
\item There is a need for \emph{improved coverage}. Many sensors may be located in so-called deep indoor environments (e.g., building basements), leading to significant signal loss between the transmitter and receiver.
\item \emph{Low cost} is a critical factor. The challenge of low cost is twofold:
\begin{itemize}
\item since \ac{iot} devices generate only small amounts of data, subscription fees should be much lower compared to those for ordinary smartphones, minimizing operational expenses;
\item the massive deployment of \ac{iot} devices necessitates a low cost per device, reducing capital expenses.
\end{itemize}
\end{enumerate}

According to Ericsson's forecast~\cite{Ericsson202311}, the number of cellular \ac{iot} devices (see~Fig.~\ref{fig:iot_growth}) is expected to reach 6.1 billion by 2029, while the total number of connected devices across all \ac{iot} technologies is projected to reach approximately 39 billion. Currently, the \ac{iot} industry exhibits a compound annual growth rate of approximately $16\%$. Given this rapid growth, standardization plays a crucial role in ensuring sustainable \ac{iot} development.

\section{\acl{iot} standardization}

There are various standards for \ac{iot}, each fulfilling different requirements specified above. We distinguish three main branches of \ac{iot} technologies: short-range, wide-area, and cellular. 

Short-range \ac{iot} encompasses a variety of technologies, including \ac{rfid} and \acp{pan} such as Bluetooth and Zigbee. These technologies typically operate in unlicensed spectrum and within a very short range. Massive connectivity in this case is limited by the small coverage area (on the order of several meters) of the corresponding radio devices~\cite{Gratton2013}.

In contrast, wide-area technologies have the potential to connect millions of autonomous devices to a single \ac{bs}. In this short overview, we focus on Sigfox and LoRa technologies, which are described in Section~\ref{sec:iot_wan}.

\subsection{Short-range \acs{iot}}

Short-range, or \ac{pan}, provides a communication environment for various \ac{iot} applications, such as communication between wearable devices, short-range location tags, home controllers, and more. A typical network consists of at most a few dozen devices, making communication between them relatively simple in terms of channel access. The low communication range is beneficial for energy efficiency and security.

An extreme example of energy efficiency (particularly regarding remote device battery life) is \ac{rfid} technology, where the energy required to transmit data is induced by an interrogation pulse from a nearby reader device. The short communication range also prevents signals from being detected over large distances between the transmitter and receiver, which significantly simplifies security protocols.

Due to the absence of massive connectivity issues and the wide variety of \ac{pan} technologies, we will
not consider them in the remainder of this monograph.

\subsection{Wide-area \acs{iot}}\label{sec:iot_wan}

To fulfill the requirements described above, several solutions are commonly employed in wide-area \ac{iot} systems -- \acp{lpwan}. To improve system range, narrowband signals are typically utilized. Since thermal noise power is proportional to the processed bandwidth, narrowband signals can tolerate a greater link loss for the same transmitter power and hence enable coverage of a wider area. Additionally, due to the extended communication distances, narrowband signals benefit from the reduced frequency selectivity of the wireless channel. Consequently, there is no need for complex multi-carrier modulations or high-complexity signal processing algorithms at the transmitter.

To simplify remote devices, wide-area \ac{iot} communication systems often have limited or even completely absent downlink functionality. This limitation reduces the ability to coordinate transmitting devices and typically eliminates the possibility of employing \ac{arq} mechanisms. Additionally, the use of \ac{arq} in scenarios involving sporadic data transmission by a massive number of devices can overwhelm the control channel.

To address this issue, transmissions can be repeated multiple times, potentially at different frequencies, to exploit both frequency and time diversity, assuming the retransmission delay exceeds the channel coherence time. Further transmitter simplifications often include the use of constant-envelope modulation formats, which do not require expensive or power-inefficient signal amplifiers.

At the receiver, additional solutions are applied to enhance performance. The simplest way to increase diversity is through \emph{multiple receptions}. When multiple nodes collect the transmitted data, many can detect, demodulate, and decode the message, thereby enhancing \emph{spatial diversity}.

These typical solutions are implemented in Sigfox and LoRa, the most popular wide-area \ac{iot} technologies, which are described below. Both approaches rely on distributed resource coordination mechanisms\footnote{These mechanisms are often referred to as medium access control (MAC). In this monograph, however, we interpret the abbreviation \acs{mac} as \acl{mac}.} based on ALOHA and \ac{csmaca}, rather than the centralized coordination employed in cellular systems.

\subsubsection{Sigfox -- an ultra-narrowband \acs{iot} technology}

The main feature of Sigfox is its extremely narrow transmission bandwidth of just $100$ Hz. This narrow bandwidth significantly reduces in-band thermal noise to very low levels (approximately $-154$ dBm), which is a key enabler of its long-range communication. Downlink functionality is extremely limited, with a complete absence of synchronization, coordination, and \ac{arq} mechanisms.

To send a packet, a device selects a transmission frequency within a $192$ kHz band and transmits the packet, followed by two replicas at different randomly selected frequencies to enhance diversity. Multiple reception is achieved by deploying multiple \ac{bs}s, which continuously scan the entire $192$ kHz bandwidth to detect uplink messages.

The transmitted packets are exceptionally short. Each packet consists of:
\begin{itemize}
\item a $4$-byte preamble,
\item a $2$-byte frame-synchronization sequence,
\item a $4$-byte device identifier,
\item a payload of up to $12$ bytes,
\item a variable-length hash code for packet authentication within the Sigfox network, and
\item a $2$-byte \ac{crc}.
\end{itemize}
Uplink packets are typically modulated using \ac{dbpsk}. Differential modulation is chosen to allow for non-coherent detection. The combination of simple modulation and coding, along with a very low sampling rate, results in a highly cost-effective solution. Additionally, the high link budget enables a reduction in transmit power, thereby fulfilling battery efficiency requirements with ease.

Downlink messages (if configured) are triggered by a transmitting device in the form of a callback. The \ac{bs}'s response has a fixed delay and is transmitted at the reception frequency of the request plus a predefined frequency shift. The payload size for a downlink message is fixed at $8$ bytes.

\subsubsection{LoRa}

The LoRa (Long Range) protocol is another \ac{iot} alternative. Similar to the previously described ultra-narrowband solutions, this communication standard assumes a \ac{bs} and end devices connected to it, managed by a simple medium access protocol. LoRa supports wider communication bandwidths ($125$ or $500$ kHz) based on the \ac{css} technique, which utilizes linear frequency modulation.

Let $B$ denote the total transmission bandwidth at a carrier frequency $f_c$, and let $T$ represent the symbol duration. The instantaneous frequency $f(t), \quad t \in [0, T]$, of the \ac{css} signal changes linearly with a rate of $\nicefrac{B}{T}$, wrapping around when it reaches the edges of the transmitted bandwidth:
$$
f(t) = f_c - \frac{B}{2} + \left\{ \frac{B}{T}t + \frac{B}{M}i\right\} \, \text{mod} \, B,
$$
where $i \in [M]$ is the transmitted message index, and the transmission rate is $\log_2 M$ bits per symbol. Different values of $M$ correspond to different \ac{sf} values~\cite{LoraTutorial2024}.

Different \ac{sf}s are supported by LoRa, and signals corresponding to different \ac{sf}s are almost orthogonal~\cite{LoraTutorial2024}. This property reduces interference between transmissions with different \ac{sf}s, thereby improving the communication range. The medium access mechanism also allows for downlink transmissions. The LoRa standard supports three different classes of devices.

Devices of the first class (A) behave similarly to Sigfox: they send data as it becomes available, and downlink messages can be sent during \emph{receive windows}, whose timings are configurable. After an uplink transmission, the device waits for the start of a downlink message within two receive windows. If no downlink message is detected, the device enters sleep mode. For this type of device, the sender spends most of its time in sleep mode, enabling long battery life. Class-A functionality is basic and must be supported by all devices.

Devices of the second class (B) extend the downlink functionality by opening periodic receive windows (or ping slots). To manage this functionality, periodic beacons are sent by the network to maintain synchronization.

Finally, devices of the third class (C) enhance the capabilities of Class-A devices by keeping the receive windows open unless transmitting an uplink message. As a result, Class-C devices can receive downlink messages almost any time, offering very low latency for downlink transmissions.

Typical use cases for Class-A devices include sensors that periodically report measurements or send data triggered by an alarm event. Class-B devices are useful for applications requiring measurements on request, while Class-C devices are ideal for remote control mechanisms powered by a continuous power source.

\subsection{Cellular \acs{iot}}\label{sec:cellular_iot_overview}

Cellular systems are highly attractive for massive deployments due to their wide coverage, straightforward subscription procedures, operation over licensed spectrum with effective interference management, and robust security protocols.

However, cellular systems employ centralized coordination algorithms (as opposed to distributed \ac{csmaca}), which are better suited for high data rates among a fixed and relatively small number of active users within the coverage area of a single \ac{bs}.

Cellular networks did not support machine-type devices prior to \ac{3gpp} Release 12. Earlier releases assumed that any device connecting to the \ac{bs} would support the full bandwidth of $20$ MHz. This large bandwidth requirement was not suitable for achieving the long battery life needed by \ac{iot} devices.

\subsubsection{Different device types}

The \ac{3gpp} standards define various categories of devices based on their capabilities. These categories are represented by numbers, where higher numbers indicate devices that support higher peak uplink and downlink rates, a greater number of supported antennas, and so on. However, these categories primarily pertain to devices operated by humans, while the requirements for \ac{iot} devices can differ significantly.

\ac{3gpp} Release 13 introduced the Cat-M category (with ``M'' denoting \ac{mtc}). This user equipment category was the first narrowband device type, supporting a bandwidth of $6$ resource blocks ($1.08$ MHz). This new device type required novel approach to control channel design.

Further optimization for \ac{mtc} devices was achieved in Release 13 with the introduction of the NB (narrowband) device category. Devices in this category reduced the total supported bandwidth to $200$ kHz. Additionally, optimizations enabled narrowband transmissions with bandwidths reduced to a single subcarrier, referred to as single-tone transmission. The subcarrier bandwidth for these transmissions is either $15$ kHz or $3.75$ kHz.

The \ac{5gnr} standard introduced broadband \ac{iot}, allowing sensing devices to transmit larger amounts of data. In Release 17, the RedCap (Reduced Capabilities) network type was introduced. The RedCap standard aims to support all industrial applications by enabling broadband communication services for machine-type devices. This standard assumes the utilization of up to $20$ MHz bandwidth in the frequency range below $6$ GHz.

\subsubsection{\Acl{ra} procedures in cellular systems}

To initiate a connection with a cellular network, each device must proceed with \acf{ra} procedure. The introduction of \ac{iot} devices and their massive deployments could overload the control channel. To address this issue, a new RA procedure should be developed, aiming to reduce the overall communication overhead.

As specified in \ac{3gpp} TS 138.321, the original RA procedure follows a four-way handshake, as illustrated in~\Fig{fig:3pp_ra_timing}\acuse{bs}\acuse{ue}. This handshake consists of the following phases:
\begin{enumerate}
\item \ac{ra} through preamble transmission to identify users.
\item Resource allocation provided by the \ac{bs}.
\item Data transmission using orthogonal resources assigned to the identified users.
\item Final \acf{ack}. 
\end{enumerate}

The described above procedure separates the \ac{ra} phase from data transmission. Starting with \ac{3gpp} Release 16, this procedure was simplified, requiring only a two-way handshake. In this updated procedure, the preamble transmission also announces the resources to be used for data transmission, which follows immediately. Users select preambles from a predefined orthogonal set. Data is then transmitted during specified positions of the \ac{po}s. If the \ac{bs} successfully receives the data, it sends an \acl{ack}. Otherwise, the traditional four-way handshake is performed. This transmission scheme is depicted \acuse{rach}\acuse{pusch} in~\Fig{fig:2step_ra_tx} (see the detailed description in~\cite{agostini2024evolution5gnewradio}).

\section{Challenges for the next-generation cellular systems}

\begin{figure}
\centering
\includegraphics{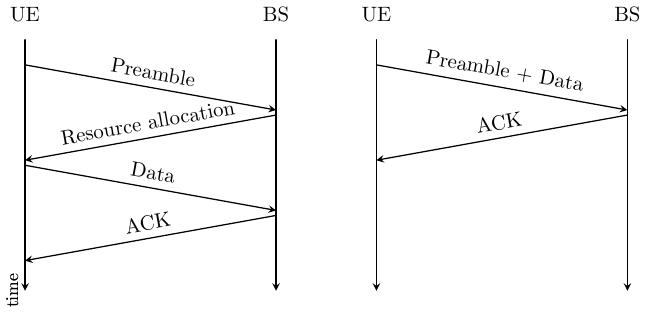}
\caption{Four-step \ac{ra} (left) and two-step \ac{ra} procedures specified in \ac{3gpp} TS 138.321. A time diagram corresponds to message exchange between \acf{ue} and \acf{bs}.\label{fig:3pp_ra_timing}}
\end{figure}
\begin{figure}
\begin{center}
\includegraphics{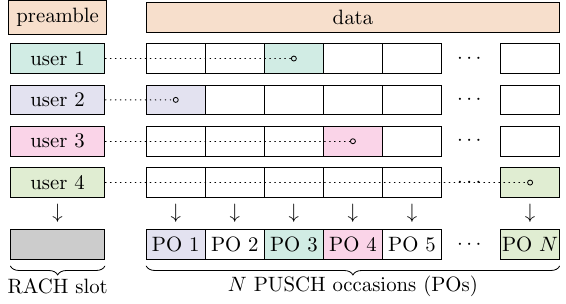}
\caption{Two-step \ac{ra} procedure with data transmission.\label{fig:2step_ra_tx}}
\end{center}
\end{figure}
\begin{figure}
\centering
\includegraphics{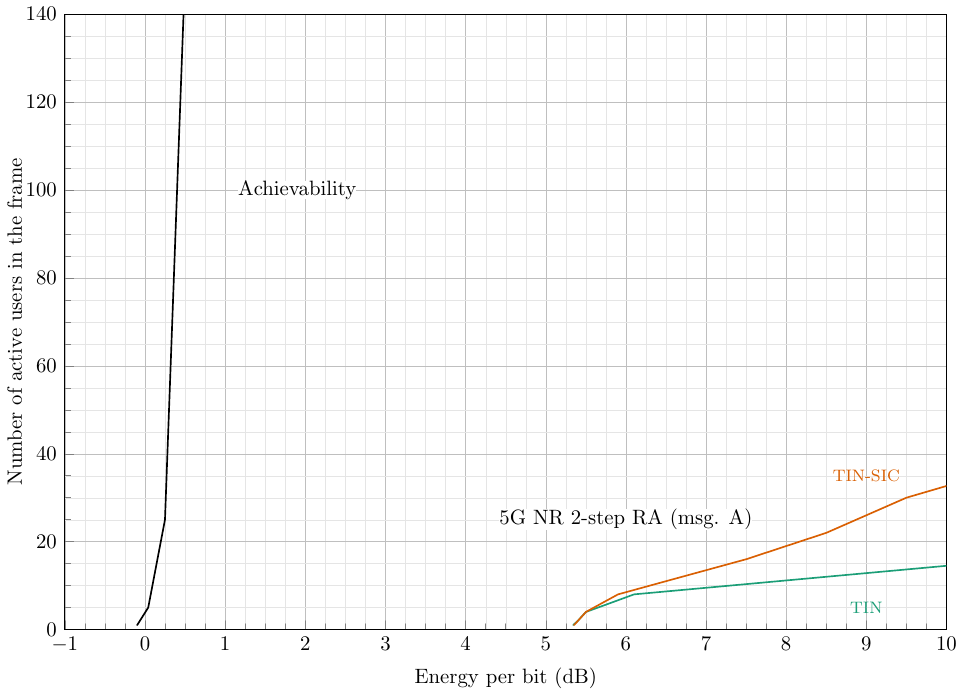}
\caption{Energy efficiency of the proposed 2-step \ac{ra} procedure in \ac{5gnr} versus achievability bound on energy efficiency~\cite{polyanskiy2017perspective}. \acs{awgn} channel model, reference signals are not considered~\cite{agostini2024evolution5gnewradio}.}
\label{fig:chap1:num_results}
\end{figure}
\begin{table}
\caption{Simulation parameters for numerical comparison presented in~\Fig{fig:chap1:num_results}}
\label{tab:chap_1:sim_params}
\centering
\begin{tabular}{c|c}
Parameter & Value \\
\hline
Preamble length & $2\times 139$ (A1 configuration) \\
Error-correcting code & $(500, 100)$ \acs{ldpc} (\ac{5gnr} base graph 2) \\
Modulation & \Acf{qpsk} \\
Decoding algorithm & \acs{tin} / \acs{tinsic} \\
Pilot configuration & Pilot-free \\
Number of \acp{po} & $64$ \\
Overall frame length & $16278$ channel uses
\end{tabular}
\end{table}

The key challenge for the next generation of radio-access networks is managing the massive number of infrequently communicating sensors. Current solutions are inadequate due to their reliance on centralized resource allocation, which orthogonalizes access for different users. For \ac{mtc}, this approach is unacceptable as it results in significant control-layer overhead and latency. Therefore, a new communication solution is required.

As evidence, we will first demonstrate the performance of the two-step \ac{ra} procedure presented above. Following the exposition in~\cite{agostini2024evolution5gnewradio}, we reproduce several numerical results from the referenced manuscript. The objective of the numerical setup is to highlight the significant gap between the energy efficiency of the two-step \ac{ra} procedure and the achievable energy efficiency in a massive \ac{ra} scenario (see Theorem~\ref{th:polyanskiy_bound}). Energy efficiency is defined as the minimum energy required to transmit a single information bit (or energy per bit) under certain \ac{qos} constraints. The exact definition of energy efficiency will be provided in Chapter~\ref{chap3}.

The technical details of this numerical experiment are as follows. The preamble dictionary consists of $64$ Zadoff-Chu sequences with varying lengths. For short preambles, the length is $139$, while for long preambles, it is $839$. Each preamble corresponds to a \ac{po}. Within a \ac{po}, data is transmitted using \acf{ldpc} codes as specified in \ac{5gnr} standards. The \acf{awgn} channel model is considered, and reference signals are not required in this setup. The system parameters outlined in Table~\ref{tab:chap_1:sim_params}\acuse{qpsk}, and the energy efficiency as a function of the number of simultaneously active users is presented in~\Fig{fig:chap1:num_results}.

Current schemes that are part of existing standards exhibit significantly lower energy efficiency compared to theoretical bounds. Moreover, the two-step \ac{ra} procedure proposed by \ac{3gpp} remains optional.  The lack of energy-efficient schemes has motivated many researchers to extensively study this new massive \ac{mtc} scenario.

In this monograph, we outline the core ideas behind these theoretical bounds and provide a brief overview of the challenges in designing low-complexity schemes. We demonstrate that some of these schemes closely approach the achievability bounds.

In our introductory example in~\Fig{fig:chap1:num_results}, we considered the simple case of a Gaussian channel. However, real wireless channels are affected by multipath propagation, which introduces additional random effects during signal transmission. In the subsequent chapters, we address various challenges posed by real propagation environments.

\section{Monograph organization}

The monograph is organized as follows:
\begin{itemize}
\item Chapter~\ref{chap2} explores \ac{mac} problems and their formulations, emphasizing the differences between classical \ac{mac} scenarios and the \acf{ura} setup.
\item Chapter~\ref{chap3} defines the \ac{ura} problem and introduces \ac{pupe}, the primary measure of the system's operational quality. It also revisits the definition of energy efficiency (see~\eqref{eq:ch3_energy_efficiency}) previously discussed in~\Fig{fig:chap1:num_results}. Additionally, this chapter establishes a connection between the \ac{ura} problem and the well-known \ac{cs} problem.
\item Chapter~\ref{chap4} investigates the fundamental limits of energy efficiency under \ac{pupe} constraints in the Gaussian channel.
\item Chapter~\ref{chap5} focuses on low-complexity schemes for the Gaussian channel.
\item Chapter~\ref{chap6} examines more realistic channels with fading effects, specifically the quasi-static Rayleigh fading channel. It covers scenarios where the \ac{bs} is equipped with either a single antenna or multiple antennas.
\item Chapter~\ref{chap7} concludes the monograph by highlighting the remaining challenges and open problems.
\end{itemize}

\chapter{Classical multiple-access problem statements}\label{chap2}

The study of \acp{mac} originated in Shannon's paper~\cite{Shannon1961}. Today, various variants of the \ac{mac} problem exist. These variants differ in terms of noise models (noiseless, noisy, or worst-case/adversarial), probability of error (zero-error, vanishing, per-user, or joint) the presence of feedback, user activity, and other factors. Following~\cite{gallager1985_perspectiveMAC, polyanskiy2017perspective}, in this chapter, we classify \ac{mac} problems based on user activity.

We focus on the vanishing probability of error case and briefly describe the zero-error results in Section~\ref{sec:ch2_zero}. We consider two main setups from the literature: all-active users (information-theoretic approach) and partial user activity. It is important to distinguish between cases where users utilize different encoders or codebooks (e.g.,\cite{ahlswede1973multi, liao1972coding, SCMA}) and cases where users share the same encoder. The scenario with all-active users utilizing different encoders is typical for \textit{coordinated multiple access}, i.e., a form of transmission where channel access is managed by a centralized scheduler that assigns resources to the users. The scenario of partial user activity and same encoders is typical for \textit{uncoordinated multiple access} or \textit{\acl{ra}} \cite{Roberts1975, Capetanakis79, tsybakov1985, Casini2007, liva2011graph, liva2015}. In the latter case, the total number of users is irrelevant, and it is common to assume it to be infinite, which aligns well with the massive access scenario. Note that for the case of partial user activity, we distinguish between the total number of users and the number of active users.

The classification and, consequently, the organization of this chapter is depicted in Fig.~\ref{fig:chap2_structure}.

We start with zero-error decoding, as presented in Section~\ref{sec:ch2_zero}. Then, we consider two different cases depending on the user activity. We present the case of all-active users in Section~\ref{ch2:sec_all_active}. Another important branch is partial user activity, as described in Section~\ref{ch2:sec_par_active}. Next, we consider the so-called many-access channel, as presented in Section~\ref{sec:ch2_manymac}. This is an extreme asymptotic case, where the number of active users grows simultaneously with the block length. This setup is the closest to the \ac{ura} model. Finally, we discuss the main differences and justify the need for a new information-theoretic framework.

\begin{figure}[t]
\centering
\includegraphics{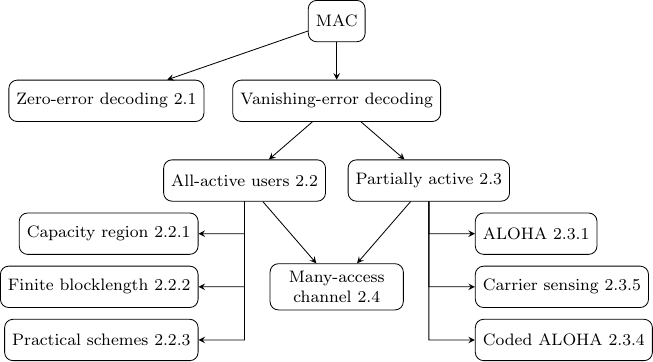}
\caption{Classification of \ac{mac} problems based on user activity. References to corresponding section of this chapter are added.\label{fig:chap2_structure}}
\end{figure}

\section{Zero-error decoding}\label{sec:ch2_zero}

A significant amount of effort has been devoted to evaluating the zero-error capacity for various natural and relatively simple \ac{mac} models. One of the most popular models is the noiseless adder \ac{mac}, where users send $0$ and $1$, and the channel output equals their ordinary sum. This problem was first studied in the context of group testing~\cite{Renyi1963, Soderberg1063, Lindstrom1964}, and was rediscovered nearly twenty years later within \ac{mac} theory~\cite{Weldon1979}. Notably, one of the most intriguing topics in \ac{mac} theory is determining the zero-error capacity for the simplest case of the two-user \ac{bac}, with the current best upper bound being $0.5753$, as reported in~\cite{COHEN2001152}.

The majority of zero-error coding schemes focus on the case of partially active users ($T$-of-$N$ schemes for activity detection). These solutions can be viewed as Sidon sets~\cite{Erdos1941}, specifically for scenarios where the input alphabet consists of binary or $q$-ary vectors ($\{0,1\}^n$ or $\{0,1,\ldots,q-1\}^n$). Relevant work includes $B_2$-sequences~\cite{LINDSTROM1969402, LINDSTROM1972261} and subsequent studies~\cite{BarDavid1993, Ericson88, Furedi99, DyaRyk81}. The construction proposed in~\cite{BarDavid1993} has been utilized as a component of several low-complexity \ac{ura} schemes (see Chapter~\ref{chap5}).

For an overview of the latest results, we refer the reader to~\cite{OrdentlichZeroOverview}.

\section{\acs{mac} with all-active users}\label{ch2:sec_all_active}

Classical information theory provides an exact solution for the case of a constant number of users, an infinite blocklength $n\rightarrow \infty$, and a vanishing probability of error, as described in~\cite{liao1972coding, ahlswede1973multi}.

Consider a fixed number of transmitters, $K$, sending their messages over a shared channel. The receiver observes a noisy linear combination of the transmitted signals. Assuming a discrete-time memoryless channel:
\begin{equation} \label{eq:MAC_statement}
\mathbf{y} = \sum\limits_{k = 1}^{K}\mathbf{x}_k + \mathbf{z}, \quad \mathbf{z}\sim\mathcal{N}\br{0, \mathbf{I}_n}, \quad \left\|\mathbf{x}_k\right\|^2 \leq nP_k
\end{equation}
where $P_k$ is the transmit power of the $k$-th user, and the terms $\mathbf{x}_k$ and $\mathbf{z}$\footnote{For simplicity, we consider real degrees of freedom.} represent the transmitted signal of the $k$-th user and the noise, respectively. This setup is known as the \ac{gmac}.

\subsection{Capacity region}\label{sec:ch2_cap_region}

The most straightforward way to manage $K$ simultaneously active users in a shared channel is to allocate them orthogonal resources. Orthogonality can be achieved in the time, frequency, code, or spatial domains. The corresponding strategies are known as \ac{tdma}, \ac{fdma}, \ac{cdma}, and spatial multiplexing, respectively.

As we see in Section~\ref{sec:ch2_low_complexity}, some of these strategies may, in certain cases, achieve optimal performance. Nevertheless, in many scenarios, an orthogonalization strategy is suboptimal.

In 1973, Ahlswede~\cite{ahlswede1973multi} and Liao~\cite{liao1972coding} derived a coding theorem for the \ac{mac}. These theorems state that given a $K$-user \ac{mac} channel with transmit powers $P_1, \ldots, P_K$, any tuple of rates $R_1, \ldots, R_K$ is achievable if
\begin{equation}\label{eq:cap_region}
\sum\limits_{k \in \dset{K}}R_k \leq C\br{\sum\limits_{k \in \dset{K}}P_k} \quad \forall \dset{K} \subseteq [K],
\end{equation}
where $C\br{P}$ is Shannon's capacity, given by
\begin{equation}\label{eq:cap_shannon}
 C\br{P} = \frac{1}{2}\log_2\br{1 + P}.
\end{equation}

\begin{figure}[t]
\centering
\includegraphics{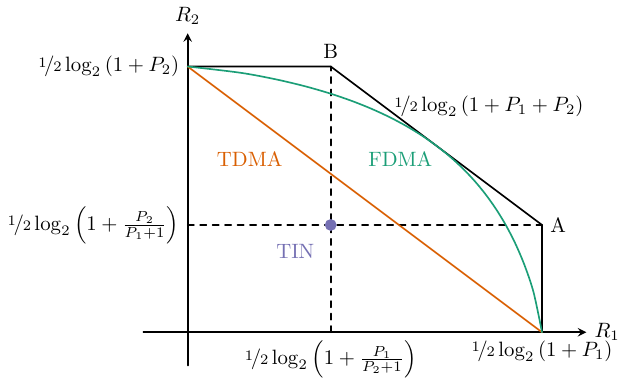}
\caption{A $2$-user \ac{mac} for an \ac{awgn} channel with real degrees of freedom. The information-theoretic capacity region is represented by a black polygon passing through points A and B. The capacity region of the \ac{tdma} scheme without power control is shown by an orange line, while the \ac{fdma} capacity region is depicted by a green line. The capacity region of \acs{tin} is a single point, denoted by a blue dot.\label{fig:MAC2user}}
\end{figure}

An illustration for the case $K=2$ is provided in Fig.~\ref{fig:MAC2user}. It is worth noting that user $1$ can increase
their data rate by increasing $P_1$ without affecting user $2$'s data rate until point B is reached. The line A--B of the resulting capacity region corresponds to the condition
\begin{equation}
\label{eq:MAC_sum_rate}
R_1 + R_2 = C\br{P_1 + P_2},
\end{equation}
which represents a point-to-point communication channel with a total transmit power of ${P_1 + P_2}$ and is referred to as the \emph{dominant facet}.

\subsection{\Acl{fbl}}\label{sec:ch2_fbl}

In this section, we examine a normal approximation introduced in~\cite{ppv2010, ppv2011} and applied to the capacity region of the $K$-user \ac{mac}, as given in~\eqref{eq:cap_region}. The core concept of the normal approximation is \emph{channel dispersion}. For an \ac{awgn} channel with transmit power $P$, the dispersion is given by~\cite[Theorem 54]{ppv2010}:

\begin{equation}\label{eq:channel_dispersion_su}
V\br{P} = \frac{\log^2 e}{2} \cdot \frac{P\br{P + 2}}{\br{1 + P}^2}.
\end{equation}

The following theorem~\cite[Theorem 3]{molavianjazi2014second} considers the joint probability of error $\varepsilon$ for a $K$-user \ac{mac} in the \acf{fbl} regime. The authors consider the codebook generated from the uniform distribution on the $n$-dimensional sphere, which is typically referred to as a \emph{power shell} (see Appendix~\ref{a3:cwbook:sphere}).

\begin{theorem}{}\label{th:chap2:fbl_cap_region}
A second-order achievable rate region with power shell inputs for the $K$-user \ac{gmac} with a power constraint $P_u$ for each user $u \in [K]$ is the set of all rate tuples $\br{R_1,\ldots,R_K}$ satisfying
\begin{equation}\label{eq:cap_region_fbl}
R_{\dset{K}} < C\br{P_{\dset{K}}} - \sqrt{\frac{V\br{P_{\dset{K}}} + V_c\br{P_{\dset{K}}}}{n}}Q^{-1}\br{\lambda_{\dset{K}}\varepsilon} + \mathcal{O}\br{\frac{1}{n}},
\end{equation}
for all subsets $\dset{K} \subseteq [K]$ and any choice of non-negative coefficients $\lambda_{\dset{K}}$ that sum to one:
$$
\sum\limits_{\dset{K}\subseteq [K]}\lambda_{\dset{K}} = 1.
$$
Here, $P_{\dset{K}}$ is the total transmit power of users in subset $\dset{K}$:
$$
P_{\dset{K}} = \sum\limits_{u\in\dset{K}}P_u,
$$
$V\br{\cdot}$ is the channel dispersion defined in~\eqref{eq:channel_dispersion_su}, and $V_c\br{\cdot}$ is the channel cross-dispersion given by:
$$
V_c\br{P_{\dset{K}}} = \frac{\log^2e}{2} \cdot \frac{\br{P_{\dset{K}}}^2-\sum_{u\in \dset{K}}P_u^2}{\br{1 + P_{\dset{K}}}^2}.
$$
\end{theorem}

To gain some intuition about this theorem, let us consider the sum-rate case, which corresponds to the dominant facet, $\dset{K} = [K]$, with equal power allocation $P_u = P$ for all $u$. In this case, the sum-rate is given by:
$$
R^\star = C\br{KP} - \sqrt{\frac{V\br{KP} + V_c\br{K, P}}{n}}Q^{-1}\br{\varepsilon} + \mathcal{O}\br{\frac{1}{n}},
$$
where
$$
V_c\br{K, P} = \frac{\log_2 e}{2} \cdot \frac{K\br{K - 1} P^2}{\br{1 + KP}^2}.
$$
One can observe an increased dispersion value over the dominant facet.

\subsection{Practical schemes}\label{sec:ch2_low_complexity}

In this section, we consider practical coding schemes that achieve points on (or the entirety of) the dominant facet of the capacity region and may be implemented with low computational complexity through appropriate single-user code selection.

\subsubsection{Orthogonal schemes}

Let us begin with orthogonal schemes by considering \ac{tdma}, \ac{fdma}, and the orthogonal case of \ac{cdma}. To proceed with our discussion, recall the capacity region presented in~\Fig{fig:MAC2user}.

\paragraph{\Ac{tdma}.}
Given $P_i$ as the transmission power of the $i$-th user ($i = 1, 2$), the corresponding capacities of the point-to-point channels are $C\br{P_1}$ and $C\br{P_2}$. In \ac{tdma}, orthogonalization is performed in the time domain. Let $\alpha$ be the fraction of resources allocated to user 1, meaning that $1 - \alpha$ is the fraction allocated to user 2. For a \ac{mac} with \ac{tdma}, the set of achievable sum rates is given by:
\begin{equation}\label{eq:cap_region_tdma}
\alpha C\br{P_1} + \br{1 - \alpha}C\br{P_2},    
\end{equation}
where $\alpha C\br{P_1}$ and $\br{1 - \alpha}C\br{P_2}$ form a pair of achievable individual transmission rates. Equation~\eqref{eq:cap_region_tdma} represents a straight line, as illustrated in~\Fig{fig:MAC2user}. In terms of the capacity region, this scheme is strictly suboptimal for all $\alpha \in (0, 1)$. Suboptimality means that the achievable rate pairs lie strictly inside the capacity region.

\paragraph{\Ac{fdma}.}
In the \ac{tdma} strategy, each transmission occupies the entire bandwidth allocated for the shared channel. In contrast, \ac{fdma} allows for better utilization of transmit power spectral density by occupying only a fraction of the allocated bandwidth. Given a bandwidth fraction $\alpha$, the channel capacity for a fraction $\alpha$ of the total frequency resources is:
$$
C_f\br{P, \alpha} = \frac{\alpha}{2}\log_2\br{1 + \frac{P}{\alpha}}.
$$
For the previously discussed two-user \ac{mac} with the \ac{fdma} strategy, the achievable individual transmission rate pairs are:
$$
R_1 = C_f\br{P_1, \alpha}, \quad R_2 =  C_f\br{P_2, 1-\alpha},
$$
which corresponds to the green curve shown in~\Fig{fig:MAC2user}.

This transmission scheme becomes optimal (touching the dominant facet of the capacity region) when:
\begin{equation}\label{eq:fdma_optimal_point}
\alpha^\star = \frac{P_1}{P_1 + P_2},
\end{equation}
due to the increased transmit power spectral density, compared to \ac{tdma} without power control as considered previously.

\begin{remark}[\ac{tdma} with power control]
In the previous discussion, we considered \ac{tdma} without power control. According to~\eqref{eq:MAC_statement}, the energy transmitted by the first user is limited to $\alpha n P_1$, and by the second user to $\br{1 - \alpha} n P_2$. Hence, to fulfill the energy constraint from~\eqref{eq:MAC_statement}, and given a fraction of allocated resources equal to $\alpha$, the first user can transmit with an instantaneous power of $\nicefrac{P_1}{\alpha}$, and the second user with $\nicefrac{P_2}{1-\alpha}$. As a result, the \ac{tdma} scheme with power control (the term described in~\cite{elgamal2011book}) has the same capacity region as \ac{fdma}.
\end{remark}

\paragraph{\Ac{cdma}.}
To achieve orthogonality in the \emph{code} domain, the transmitted signal can be \emph{spread} using a spreading sequence $\mathbf{a}^{\br{i}}$. We refer the reader to \eqref{eq:ch5_spread_rkon} and Section~\ref{sec:encoder} for details. To detect a user's signal from the received sequence, the receiver employs a \ac{mf} matched to the spreading sequence $\mathbf{a}^{\br{i}}$ assigned to the $i$-th user. Given orthogonal spreading sequences, i.e., $\left(\mathbf{a}^{\br{i}}\right)^T \cdot \mathbf{a}^{\br{j}} = 0$, applying a \ac{mf} converts the two-user \ac{mac} channel into a \ac{tdma} channel, as described above, with $\alpha = \nicefrac{1}{2}$\footnote{This can be easily verified by expressing the spread signal in the form~\eqref{eq:ch5:spread_reception} and~\eqref{eq:ch5:spread_matrix}, keeping in mind that the $K$ spreading sequences form a $K\times K$ orthogonal matrix. Then, a simple basis change transforms~\eqref{eq:ch5:spread_reception} into a \ac{tdma} channel (identity spreading matrix) without altering the noise properties.}. Thus, the capacity region of \ac{cdma} with orthogonal spreading sequences reduces to a single point on the \ac{tdma} capacity region. In the case of $P_1=P_2$, and following~\eqref{eq:fdma_optimal_point}, this point lies on the dominant facet of the capacity region. Note that the maximum number of orthogonal sequences is limited by their length, meaning that the maximum number of users that can be orthogonalized equals the length of the spreading sequence.

Additionally, the spreading technique also helps to perform time-domain channel equalization in the presence of multipath propagation with delay spread, using a RAKE receiver~\cite{Tse_Viswanath_2005}.

\subsubsection{\Acl{tin}}

To recover all $K$ transmitted codewords from the received signal~\eqref{eq:MAC_statement} in a non-orthogonal regime, a joint decoder is required. However, the complexity of joint decoding can be prohibitively high for practical applications. In practice, particularly in a non-symmetrical \ac{mac} -- where the received powers $P_i$ are unequal -- the user with the strongest received power is more likely to be successfully decoded by a \emph{single-user} decoder, even when treating other users' signals as noise. This approach is known as \acf{tin}.

Applying an information-theoretic analysis to \ac{tin} in the two-user \ac{mac}, the noise and interference\footnote{Under the assumption of random coding with i.i.d. Gaussian codebooks, the sum of noise and interference will also follow a Gaussian distribution.} power experienced by the first user's \ac{tin} decoder is $1 + P_2$, resulting in a maximum achievable communication rate of $R_1$:
$$
R_1 = \frac{1}{2}\log_2\br{1 + \frac{P_1}{1 + P_2}}, \quad R_2 = \frac{1}{2}\log_2\br{1 + \frac{P_2}{1 + P_1}},
$$
where $R_2$ is obtained similarly. Note that in this case, the relative received power of each user does not affect the result. This rate pair is represented by the blue dot in Fig.~\ref{fig:MAC2user}. This scenario was previously discussed in the context of \ac{cdma} with pseudo-orthogonal spreading sequences.

Notably, depending on the values of $P_1$ and $P_2$, the \ac{tin} point can lie either inside or outside the capacity region of the \ac{tdma} scheme. The reference case in~\Fig{fig:MAC2user} corresponds to equal power allocation.

Now, let us analyze a $K$-user symmetric \ac{mac} and consider the achievable sum rate as $K\rightarrow \infty$ in the case where a \ac{tin} strategy is applied to the codewords of each user:
$$
\lim_{K\rightarrow \infty}\frac{K}{2}\log_2\br{1 + \frac{P}{\br{K - 1}P + 1}}\approx \frac{\log_2 e}{2}.
$$
This sum rate is limited and falls far below the dominant facet of the capacity region, where the maximum sum rate is equal to $C\br{KP}$, which is unbounded.

\subsubsection{\Acl{sic}}

Can the previously discussed \ac{tin} strategy be improved without significantly increasing computational complexity?

After performing the \ac{tin} step and applying capacity-achieving codes, the receiver successfully decodes a codeword. This codeword can then be subtracted from the received sequence. For instance, if user $1$ is decoded using the \ac{tin} approach, the residual signal -- after subtracting the first user's codeword -- will consist only of user $2$'s signal and noise. This effectively reduces the problem to a single-user decoding scenario in the \ac{awgn} channel, where user $2$ can achieve the rate $C(P_2)$. The corresponding rate pair is given by:
$$
R_1 = \frac{1}{2}\log_2\br{1 + \frac{P_1}{1 + P_2}}, \quad R_2 = C\br{P_2}.
$$
This rate pair corresponds to point $B$ in the capacity region shown in \Fig{fig:MAC2user}. Similarly, point $A$ can be achieved by reversing the order in which the \ac{tin} step and subtraction are applied.

For the $K$-user \ac{mac}, the \ac{tin} step followed by subtraction can be repeated multiple times. This method is known as \acf{sic}. After each \ac{sic} step, the dimensionality of the problem is reduced.

In the two-user \ac{mac} example, applying \ac{sic} allows us to achieve points $A$ and $B$ on the dominant facet of the capacity region. To reach any point along the line segment A--B, a \ac{tdma} strategy can be employed, alternating between different decoding orders of the \ac{tin} step followed by \ac{sic}. As a result, we obtain a low-complexity scheme capable of achieving any point on the dominant facet of the capacity region.

\subsubsection{Non-orthogonal \ac{cdma}}

Orthogonal sequences allow us to use a single-user decoder to solve a multiple-access problem. However, the spreading sequences may not always be orthogonal. In this case, the system load may be higher compared to orthogonal spreading, at the cost of additional interference. In the presence of interference, a \ac{mud} may outperform a single-user one. The problem of \ac{mud} in \ac{cdma} has been studied in~\cite{verdu1986minimum, lupas1989linear}. In~\cite{honig1995blind}, the idea of blind detection was proposed -- a scenario in which the receiver knows the target spreading sequence but does not know the interfering spreading sequences.

Regarding the relationship between \ac{cdma} and the capacity region, the authors of~\cite{Massey1994} proposed the design of spreading sequences for synchronous overloaded \ac{cdma} (where the number of active users exceeds the length of the spreading sequence) in an equal-power setting and estimated the achievable sum rate. The unequal-power case was considered in~\cite{Viswanath1999}. Finally, \ac{mud} combined with \ac{sic}, as proposed in~\cite{varanasi1997optimum}, achieved full capacity region.

For randomly generated sequences, the authors of~\cite{TseHanly99, VerduShamai99} analyzed capacity for random binary and spherical sequences under joint decoding and linear \ac{mud}s. Additionally, an iterative decoding scheme for turbo codes with spreading was proposed in~\cite{softSIC_wang_poor}. Lastly, let us highlight the work of Montanari and Tse in~\cite{Montanari2006}, where they proposed message-passing algorithms to design near-optimal \ac{mud}s. This approach later evolved into \ac{lds} \ac{cdma}~\cite{LDS} and \ac{scma}, as proposed in~\cite{SCMA}.

\subsubsection{\Acl{rsma}}

In the previous discussion, we found that the dominant facet of the capacity region can be achieved using the \ac{tinsic} strategy with time-sharing. However, is there a simpler strategy that can be achieved using only single-user codes? As proposed in~\cite{rimoldi1996rate}, this can be accomplished through \ac{rsma}.

To explain the main idea, let us redefine Shannon's capacity~\eqref{eq:cap_shannon} by explicitly considering the noise power $\sigma^2$:
$$
C_\sigma\br{P, \sigma^2} = \frac{1}{2}\log_2\br{1 + \frac{P}{\sigma^2}}
$$
As we observed previously, the maximum sum-rate $R_\Sigma$ of the two-user \ac{mac} is achieved on the dominant facet, resulting in
\begin{equation}\label{eq:cap_superposition}
R_\Sigma = C_\sigma\br{P_1 + P_2, \sigma^2} = C_\sigma\br{P_2, \sigma^2} + C_\sigma\br{P_1, P_2 + \sigma^2},
\end{equation}
where the last equation is referred to as the \emph{chain rule}. The physical meaning of this chain rule is as follows: In the first step, we apply a \ac{tin} step, resulting in $R_1 = C_\sigma\br{P_1, P_2 + \sigma^2}$, followed by a \ac{sic} step, yielding $R_2 = C_\sigma\br{P_2, \sigma^2}$. This strategy allows us to achieve one corner point of the capacity region pentagon.

An interesting observation is that the chain rule~\eqref{eq:cap_superposition} allows us to split a high-rate code into two codes with lower rates. Then, the \ac{tinsic} strategy can be used to decode two superimposed codes. As proposed in~\cite{rimoldi1996rate}, any point in the two-user \ac{mac} capacity region can be achieved by splitting the higher-rate code of the user with the larger received power into two superimposed codes. Then, a \ac{tinsic} strategy can be applied to decode all three codes.

To illustrate this, let us consider the case where $P_1 > P_2$, and let us choose some $0 < \delta < P_1$. Next, we define
\begin{align*}
p_1= & \delta, \\
p_2= & P_2, \\
p_3= & P_1 - \delta.
\end{align*}
Then, let $r_i = C_\sigma\br{p_i, \sigma^2}$ be the rate achievable by the $i$-th code. The sum rate of the three codes equals
$$
r_1 + r_2 + r_3 = C_\sigma\br{p_1 + p_2 + p_3, \sigma^2} = C_\sigma\br{P_1 + P_2, \sigma^2} = R_\Sigma.
$$
Note also that for any $\delta$, there is a unique $r_1$ value. Hence, any point on the dominant facet can be achieved.

Furthermore, the result above was generalized to the case of a $K$-user \ac{mac} in~\cite{rimoldi1996rate}. It was shown that a total of $2K - 1$ superimposed codes are required to achieve an arbitrary point on the dominant facet.

\subsubsection{Schemes based on polar codes}

\begin{figure}
\begin{center}
\includegraphics{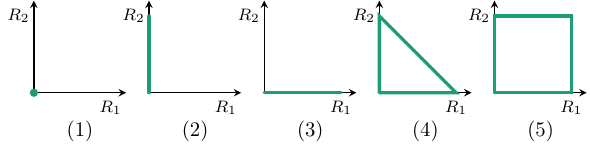}
\caption{Five extreme channels for two-user \ac{mac} polarization~\cite{Telatar2user}. The figure illustrates: (1) a completely useless channel, where the capacity region is a single point at $\br{0, 0}$; (2, 3) two channels that are useless for one of the two users; (4) a \ac{mac} with a \ac{tdma}-like capacity region; and (5) a channel where both users can transmit in parallel.\label{fig:MAC2user_polarization}}
\end{center}
\end{figure}

Several works present theoretical results on the achievability of the capacity region using polar coding. In~\cite{Telatar2user}, it was shown that two users with a uniform data source, using Ar{\i}kan's construction, may achieve a point on the dominant facet of the capacity region. The resulting \ac{mac} polarizes into five extreme channels, with the corresponding capacity regions illustrated in~\Fig{fig:MAC2user_polarization}. Next, it was shown that a polar coding technique allows the achievement of any point in the capacity region of a two-user \ac{mac} \cite{Arikan2012}.\footnote{A Slepian-Wolf source coding problem was considered, which is a problem dual to the two-user \ac{mac}. The solution is based on monotone chain rule expansions.} Finally, a method for improving the performance of two-user \ac{mac} polar coding through list decoding was described in \cite{Onay2013}.

Theoretical results for two-user \ac{mac} polarization were further extended to the $K$-user \ac{mac} case in~\cite{TelatarMuser}. Moreover, the achievability of the entire uniform rate region\footnote{For a \ac{mac}, the uniform rate region is the achievable region when the input distributions are uniform.} for the $K$-user \ac{mac} was proposed in~\cite{Mahdavifar2016}.

At the same time, there are no efficient decoding or optimization methods for the \ac{fbl} case.

\subsubsection{Schemes based on \ac{ldpc} codes}

Finally, let us consider the capacity-approaching scheme based on \ac{ldpc} codes presented in~\cite{LDPC_2user}. The authors employed a \ac{mpa} to iteratively decode both codes, and each \ac{ldpc} code was optimized for degree distributions using \ac{de}. Their results showed that the resulting decoding threshold is only $0.18$ dB away from Shannon's limit.

Moreover, another capacity-approaching scheme is based on spatially coupled \ac{ldpc} codes~\cite{Kudekar2011}, whose authors validated their claim using \ac{de}.

\section{\acs{mac} with partial user activity}\label{ch2:sec_par_active}

Now, let us consider \ac{ra} protocols. In this scenario, users generate messages at random time instances and attempt to transmit them over a shared channel.

In 1970, Abramson proposed the first random access protocol, which is also the simplest, known as ALOHA~\cite{Abramson70}. The core idea is straightforward: whenever a new packet arrives, the transmitter immediately sends it. If two or more transmissions occur simultaneously, the signals interfere, making them undecodable and resulting in a collision. This original protocol is known as pure ALOHA.

In this non-synchronized system, a collision occurs even when transmitted messages partially overlap. To mitigate this, time synchronization was later introduced~\cite{Roberts1975}, giving rise to what is now known as \ac{sa}. An example illustrating how synchronization reduces the number of collisions is shown in \Fig{fig:ch2_aloha}. Thus, in the subsequent analysis, we consider only the \ac{sa} system.

Below are some key assumptions of the \ac{sa} model. Based on these assumptions, various throughput optimization techniques have been developed.

\begin{figure}
\centering
\includegraphics{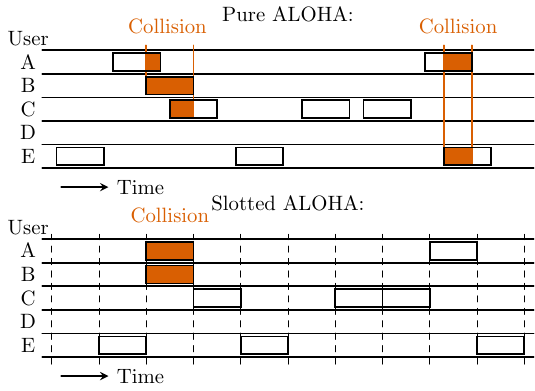}
\caption{Pure (top) and slotted (bottom) ALOHA. Packets are generated at the same time instances, resulting in five packet collisions during transmission in the case of pure ALOHA. In \ac{sa}, transmission is allowed only in the nearest slot after a packet appears. As a result, the number of collided packets is reduced to two. Slot boundaries are marked by dashed lines.\label{fig:ch2_aloha}}
\end{figure}

\subsection{\Acl{sa} model assumptions}\label{sec:ch2_saloha}

As we demonstrate later, the use of collision resolution algorithms can significantly enhance the performance of \ac{sa}. We now outline the assumptions for \ac{sa} used in the study of collision resolution algorithms.

\begin{enumerate}
\item Slotted transmission system: Each message has a fixed duration and is transmitted within a unit slot length.
\item Collision: Occurs when two or more simultaneous transmissions happen within a slot, and successful reception is impossible.
\item Immediate feedback: The system provides feedback indicating whether the slot was empty (no transmission), a success (a single packet was transmitted), or a collision (two or more packets were transmitted).
\end{enumerate}

To work with a mathematical model of ALOHA, we make the following additional assumptions:
\begin{enumerate}
\item An infinite number of transmitters.
\item The arrival rate of incoming messages follows a Poisson point process with intensity $\lambda$. This approximation is reasonable given a large number of independent sources.
\end{enumerate}

\begin{remark}[Worst-case performance estimation in terms of collisions]
The assumption of an infinite number of transmitters combined with a finite message rate allows for estimating worst-case performance in terms of collisions. Indeed, if the number of transmitters were finite, a message arriving at a transmitting station would be queued rather than causing a collision.
\end{remark}

\begin{remark}[Correlated transmissions]
We should also note that correlated transmissions (which are typical in sensor networks~\cite{Kalor2018}) may lead to significantly worse performance than a Poisson model assuming independent transmissions.
\end{remark}

\begin{remark}[Collision model considerations] The ALOHA model does not account for physical-layer coding techniques. The previously mentioned \ac{sic}, which is a physical-layer coding technique, could significantly alter system performance. Nevertheless, the above assumptions provide a simplified setup for studying collision resolution algorithms.
\end{remark}

\begin{remark}[Collisions and capture effect]
In the case of real signal propagation models and \acf{awgn}, a phenomenon known as the \emph{capture effect} occurs~\cite{Paolini2017capture}. When two messages collide and arrive with different received power levels, the \ac{tin} decoding process may allow for partial message recovery, where at least the user with the higher \ac{sinr} value is successfully decoded.
\end{remark}

\subsection{Throughput and stability of \acl{sa}}

Given a unit slot length, it is reasonable to define throughput as the average number of successfully delivered packets per slot. For an \ac{sa} system without any collision resolution algorithms or retransmissions, the throughput depends on the input data rate, $\lambda$, as $\lambda e^{-\lambda}$, which reaches its maximum at $\lambda = 1$, yielding $\nicefrac{1}{e}\approx 0.3679$ (following the Poisson traffic model proposed in Section~\ref{sec:ch2_saloha}). Thus, only about 36.8\% of slots contain a non-collided transmission, while the remaining slots are either empty or experience collisions.

To define stability, let us introduce a retransmission mechanism. When a transmitter has channel feedback, collisions can be detected immediately. In this case, the packet becomes \emph{backlogged} for future retransmission. Recall that the number of transmitters is infinite, and a backlogged packet is not affected by subsequent messages. Different retransmission algorithms may influence the system's throughput.

Let us first specify the simplest retransmission policy: any backlogged packet is transmitted with probability $p$ in the subsequent slot, independently of past slots and other packets. The value of $p$ can be adjusted based on collision statistics. Small values of $p$ lead to high delays, whereas large values result in excessive collisions.

This system can be represented as a Markov chain, where the state is the number of backlogged messages, $k$. Given $\lambda$ and $p$, the set of transition probabilities can be easily derived~\cite{gallager1985_perspectiveMAC}. It turns out that the resulting Markov chain is non-ergodic~\cite{Kaplan1979}.

The \emph{drift} of the system is defined as the expected increment $D_k\br{\lambda, p}$ in the number of backlogged packets for each state $k$ during a slot duration. When $D_k > 0$, the number of packets scheduled for retransmission is expected to increase. By analyzing $D_k$ for different values of $k$, we can identify regimes where the system becomes unstable, meaning the number of backlogged packets grows without bound. Consequently, the system's throughput approaches zero.

The problem of ALOHA stabilization was first studied in~\cite{Hajek1982}. The authors examined a class of algorithms that control $p$ based solely on collision feedback. The resulting control algorithm was stable for $\lambda < \nicefrac{1}{e}$. Note that for $\lambda > \nicefrac{1}{e}$, the drift value is always positive~\cite{gallager1985_perspectiveMAC}. Hence, the throughput of the \ac{sa} system is upper-limited by $\nicefrac{1}{e}$, and collisions remain a crucial throughput-limiting factor. In the next section, we describe advanced collision resolution algorithms that may increase the throughput of the \ac{sa} system.

\subsection{Lower and upper bounds on throughput in \acl{sa}}

\begin{figure}
\centering
\includegraphics{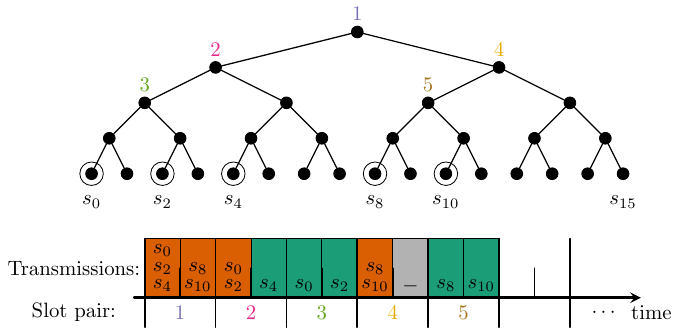}
\caption{Example of the tree-split algorithm, reproduced from~\cite{Capetanakis79}. (Top) Alignment of all stations into a tree. Active stations $s_0$, $s_2$, $s_4$, $s_8$, and $s_{10}$ are marked by circles. (Bottom) Sequence of transmissions during the collision resolution algorithm. Slots with collisions are marked by orange rectangles with active station indices, collision-free slots are represented by green rectangles, and empty slots are shown in grey. The slot pair index and the corresponding root node are indicated by colored numbers.}
\label{fig:sa_tree_collision}
\end{figure}

To achieve higher throughput and lower delay, the following modifications can be made to the retransmission strategy outlined above. First, the immediate transmission of a newly arrived packet can be delayed if previous collisions remain unresolved. Second, more advanced collision resolution policies can be employed. In this section, we briefly review various explored methods for throughput optimization.

We begin with the tree-split algorithm proposed in~\cite{Capetanakis79, Hayes1978, TsyMik78}. The idea of this algorithm is to arrange all stations in a binary tree, as shown in~\Fig{fig:sa_tree_collision}.\footnote{For simplicity, we consider the case where the number of stations is a power of two.} When a collision occurs, the system enters a collision resolution regime. In this regime, all subsequent slots are grouped into pairs, and a tree-traversing algorithm (e.g., depth-first search) is initiated. The goal of this algorithm is to mark all leaf nodes (stations) as visited (i.e., collisions resolved).

To achieve this, the algorithm starts at the root node. Stations from the left and right branches of the root transmit in the first and second slots, respectively. If a collision-free transmission or an empty slot occurs, all nodes connected to the corresponding branch are marked as visited (resolved). The tree-traversing algorithm continues splitting until all nodes are marked as visited.

Subsequent modifications of tree-based algorithms, such as more than two groups at the initial collision resolution step, resulted in a throughput of $0.46$~\cite{Massey1981}, which remains stable for smaller input rates.

The second family of algorithms is based on the \ac{fcfs} algorithms, where each arriving packet is marked with a timestamp, and packets whose timestamps are within a certain window are allowed to transmit~\cite{TsyMik80}. The window edges move in accordance with channel feedback.

The state-of-the-art algorithm for the classical \ac{sa} setup was proposed in~\cite{tsybakov1985}, achieving a throughput of up to $0.4877$. Finally, the state-of-the-art upper bound for classical \ac{sa} throughput is given in~\cite{tsybakov1987}.
\begin{equation}\label{eq:sa_throughput_bounds}
0.4877 \overset{\text{\cite{tsybakov1985}}}{\leq} R \overset{\text{\cite{tsybakov1987}}}{\leq} 0.5683.
\end{equation}
Note that by leveraging the \ac{sic} technique, a throughput of 0.693 can be achieved using the tree-based algorithm, as reported in~\cite{Giannakis2007}. However, the use of \emph{coded} \ac{sa}, as presented in the next section, may further improve the system's throughput.

\subsection{Coded \acl{sa}}\label{ch2:coded_aloha}

Let us introduce time diversity into the ALOHA system: a single transmission is followed by the transmission of a \emph{replica} of the original message~\cite{Casini2007}, which is sent in a randomly selected slot within a \emph{frame}. In this setup, the frame serves to separate different bursts and limit the distance between two replicas. This approach is known as \ac{crdsa}. If one of the two replicas does not experience a collision, the position of the second replica can be determined from the successfully decoded message, allowing the second replica to be subtracted from another slot, potentially resolving additional collisions -- a process known as the \ac{sic} step.

This idea was further improved in~\cite{liva2011graph}, where an arbitrary number of replicas was allowed and chosen randomly from a specific distribution. All replicas are transmitted in randomly selected slots within a frame. The parameters of this distribution were optimized, significantly improving the resulting throughput, which can reach approximately $0.97$ for large frames and close to $0.8$ in practical implementations -- a significant improvement over ordinary \ac{sa}, as shown in~\eqref{eq:sa_throughput_bounds}. This method is known as \ac{irsa}. An example of \ac{irsa} is shown in~\Fig{fig:chap5:aloha_tx} and~\Fig{fig:chap5:aloha_rx} and described in Section~\ref{ch5:sec:coded_sa}.

The transmission and decoding processes in the \ac{irsa} system can be described using a bipartite graph, known as a Tanner graph. The vertex set of the graph consists of user nodes and slot nodes. An edge exists between a user node and a slot node if a transmission occurs from the corresponding user in the corresponding slot. The random family of these bipartite graphs can be characterized by a corresponding degree distribution. Hence, density evolution~\cite{RichardsonUrbanke2008}, similar to \ac{ldpc} optimization, can be applied here.

\subsection{Carrier sensing}\label{sec:ch2_csma}

\Ac{csma} was proposed in~\cite{Kleinrock1975}. The main idea of carrier sensing is to quickly detect when the slot is idle and immediately proceed to the next slot. This idle-detection time depends on the practical implementation of the system and may include hardware switching time, propagation delays, etc. There are two differences compared to the \ac{sa} system described previously. The first is that an idle slot has a duration of $\alpha \ll 1$. The second is that a packet arriving during another transmission is immediately backlogged, and its transmission starts with probability $p$ after each subsequent idle slot. Packets arriving during an idle slot are transmitted in the next slot, as usual.

The analysis of the system can be performed via a Markov chain~\cite{gallager1985_perspectiveMAC}. The drift value $D_k$ can be expressed as follows:
$$
D_k = \lambda \alpha + \lambda \br{1 - e^{-\lambda \alpha}\br{1 - p}^k}- \br{\lambda \alpha+ \frac{pk}{1 - p}}e^{-\lambda \alpha}\br{1 - p}^k.
$$
For $\lambda\br{1 + \alpha} < 1$, the value of $D_k$ reaches its minimum at
$$
p^\star = \frac{1 - \lambda\br{1 + \alpha}}{k - \lambda\br{1 + \alpha}}.
$$
$D_k$ is negative for all values of $k$ (with $p^\star$ substituted) as soon as
$$
\lambda\br{1 + \alpha} <e^{-1 + \lambda}.
$$
Applying a power series expansion over $1 - \lambda$, the throughput approaches one for small $\alpha$:
$$
\lambda < 1 - \sqrt{2\alpha}.
$$
Moreover, within \ac{csma}, the difference between slotted and pure ALOHA disappears, with the maximum throughput of pure ALOHA with \ac{csma} being equal to $1 - 2\sqrt{\alpha}$~\cite{gallager1985_perspectiveMAC}.

\section{Many-access channels}\label{sec:ch2_manymac}

Finally, let us consider the information-theoretic approach, in which the number of users is allowed to grow with the blocklength. This line of research is the closest to the \ac{ura} direction. Henceforth, we refer to this paradigm as \emph{many-user information theory} and to the corresponding channels as \emph{many-access channels}.

We present the many-access channel results in three steps, following their evolution in the literature. In the first step (see Section~\ref{sec:many_bac}), the asymptotics over blocklength $n\rightarrow \infty$ were considered first, followed by the limit of the user count $K\rightarrow \infty$. Next, in Section~\ref{sec:many_mac1}, we consider the case where both $K$ and $n$ tend to infinity simultaneously, providing additional insights into the relationship between $K$ and $n$. However, in this regime, each user expends infinite energy to send a message with a fixed number of information bits, leading to an infinite $E_b/N_0$. Finally, this issue is addressed in Section~\ref{sec:many_mac2}, where the main conclusion is that achieving zero probability of joint decoding error becomes infeasible in this case.

\subsection{Noiseless and noisy many-access \ac{bac}}\label{sec:many_bac}

We begin with the paper \cite{Weldon1979}, assuming that all $K$ users in the system are active. The $K$ messages emanating from the $K$ sources are encoded independently using $K$ binary block codes\footnote{Here, we consider the case of different encoders or codebooks.}
$$
\mathcal{C}_1, \mathcal{C}_2, \ldots, \mathcal{C}_K, \quad \mathcal{C}_i \subseteq \{0,1\}^n, \quad |\mathcal{C}_i| = M_i
$$
of the same length $n$.

The goal is to maximize the sum rate
\[
R_\Sigma(K) = \sum\limits_{i=1}^K R_i, \quad R_i = \frac{\log_2 M_i}{n}, i \in [K].
\]

The authors consider a noiseless \ac{bac}, where the number of users increases with the blocklength. The $i$-th user sends an integer $X_i \in \{0, 1\}$, and the output symbol $Y$ is the real sum of the $K$ inputs, i.e., the channel output is
$$
Y = \sum\nolimits_{i=1}^K X_i \in \{0,1,\ldots,K\}.
$$
The first result is the capacity region for such a channel. Here, we provide only the sum rate inequality (dominant facet inequality):
\[
R_\Sigma(K) \leq C_\Sigma(K) = \max\limits_{P_{X_1} \otimes \ldots \otimes P_{X_K}} I(X_1, \ldots, X_K; Y) = \sum\limits_{i=0}^K \frac{\binom{K}{i}}{2^K} \log_2 \frac{2^K}{\binom{K}{i}},
\]
where the maximization is performed over independent distributions $P_{X_i}$, and the maximum is achieved for $P_{X_i} = \textrm{Bern}(1/2)$, $i \in [K]$.

The asymptotics can be summarized as follows.

\begin{theorem}
The maximal achievable rate of a $K$-user uniquely decodable code for the binary $K$-user noiseless \ac{bac} is asymptotically equal to $\frac{1}{2} \log_2(\pi e K/2)$ with increasing $K$.
\end{theorem}

Finally, the authors proposed uniquely decodable codes for the many-access \ac{bac}. The key ingredient of the construction is the difference matrix $\dmat{D}$ of size $K \times n$ over $\{0, 1,-1\}$, where the rows of $\dmat{D}$ are linearly independent over $\{0, 1,-1\}$. Given $\dmat{D}$, it is possible to construct $K$ user codes $\mathcal{C}_1, \mathcal{C}_2, \ldots, \mathcal{C}_K$ that are uniquely decodable\footnote{Note that in this case, we have zero-error decoding.}. Each user code consists of two codewords, i.e., $M_i = 2$, $i \in [K]$.

The most significant result is that the sum rate of such codes,
\[
R_\Sigma(K) = \frac{K}{n} = 1 + \frac{1}{2} \log_2 n 
\]
achieves $R_\Sigma(K)$ asymptotically, i.e.,
$$
\lim\limits_{K\to\infty} \frac{R_\Sigma(K)}{C_\Sigma(K)} = 1.
$$

The authors also generalized the construction for noisy channels. In this case, the sum rate is given by
\[
R_\Sigma(K) = \frac{K}{n} = 1 + \frac{1}{2} \log_2 \frac{n}{d_{\min}}, 
\]
where $d_{\min}$ is the minimal $L_1$ distance of the $K$-user code, defined as
\[
d_{\min} = \min \norm{\sum\limits_{i=1}^K \dvec{c}_i - \sum\limits_{i=1}^K \dvec{c}'_i}_1,
\]
where the minimum is taken over all pairs $(\dvec{c}_1, \ldots, \dvec{c}_K) \ne (\dvec{c}'_1, \ldots, \dvec{c}'_K)$ with $\dvec{c}_i, \dvec{c}'_i \in \mathcal{C}_i$, $i \in [K]$. 

\subsection{Many-access \ac{gmac}}\label{sec:many_mac1}

Now, let us proceed with the inspirational papers of X.~Chen, D.~Guo, and their co-authors on many-access channels~\cite{guoMAC, guoRandom, Guo2017}. The authors proposed a new paradigm, referred to as many-user information theory, in which the number of users is allowed to grow with the blocklength. Notably, this is one of the first works to address both the partial activity case and the issue of a growing number of users as the blocklength increases.

The paper~\cite{Guo2017} highlights an important observation: even in works that consider a growing number of users, the blocklength is typically sent to infinity before the number of users. Indeed,~\cite{Weldon1979} follows exactly this approach. We note that while the capacity computed in this manner can still be used in the converse bound, it does not necessarily apply to the achievability bound. The reason is as follows: the argument that joint typicality holds with probability 1 as the blocklength grows to infinity does not directly extend to models where the number of users also grows to infinity. Thus, the achievability bound derivation must be reworked.

Now, let us consider the model and the main result. Let $K$ be the total number of users in the system. The users share the same message set $[M]$ but use different encoders $f_i: [M] \bigcup \{0\} \to \mathbb{R}^n$. Each user is active with probability $\alpha^{\br{n}}$, and both this probability and the message selection strategy are identical for all users ($i \in [K]$):
\[
\Pr\left[ W_i = W\right] = \left\{ \begin{array}{cc}
     1 - \alpha^{(n)}, & W=0, \\
     \alpha^{(n)}/M, & W \in [M]. 
\end{array}\right.
\]

Suppose each user accesses the channel independently with identical probability $\alpha^{\br{n}}$ during any given block. The messages are encoded and sent via the Gaussian many-access channel:
\begin{equation}\label{eq:gmac_for_converse}
\rvec{y} = \sum\limits_{i=1}^K f_i(W_i) + \rvec{z},
\end{equation}
where $\rvec{z} \sim \mathcal{N}\br{\dvec{0}, \dmat{I}_n}$ is Gaussian noise. Note that $\norms{f_i(W_i)} \leq Pn$. If user  $i$ is inactive, it is assumed to transmit the all-zero codeword, i.e., $f_i(0) = \dvec{0}$.

The authors focus on the \emph{joint} probability of error:
\[
P_J^{(n)} = \Pr\left[ \mathrm{dec}(\rvec{y}) \ne (W_1, \ldots, W_K)\right], 
\]
where the decoder must return zeros for users who were inactive.

The main result is stated in~\cite[Theorem 1]{Guo2017}. This theorem considers different regimes; for the reader’s convenience, we do not reproduce it in full but instead provide the result relevant to the regime of interest. We focus on the case where the average number of active users ($\Ka$) grows linearly with the blocklength ($n$), i.e.,
\[
K_{a}^{(n)} = \alpha^{(n)} K^{(n)} = \mu n,
\]
where $\mu$ is some constant. At the same time, the total number of users can scale faster, e.g., as $n^2$.

For this case, the symmetric message-length capacity (in bits) is given by:
\[
k^{(n)} \sim \left( \frac{1}{2 \mu} \log_2(1 + \mu n P) - \frac{h(\alpha^{(n)})}{\alpha^{(n)}} \right)_+,
\]
which means there exists a sequence of codebooks with message lengths (in bits) arbitrarily close to $k_n$ such that as $n \to \infty$, the average probability of error $P_J^{(n)}$ vanishes. The capacity is achieved by i.i.d. Gaussian inputs.

Note that $f^{(n)} \sim g^{(n)}$ if and only if
\[
\lim\limits_{n \to \infty} \frac{f^{(n)}}{g^{(n)}} = 1.
\]
We use equivalence rather than equality, since if $k^{(n)}$ is the message-length capacity, then $k^{(n)} + o(k^{(n)})$ is also considered capacity.

\paragraph{Overview of the results.}
The expression for $k_n$ consists of two parts:
\begin{itemize}
\item The first term resembles the classical capacity expression and depends only on the average number of active users.
\item The second term represents the loss due to the need to identify active users and depends on the ratio of the average number of active users to the total number of users.
\end{itemize}

\begin{example}
 Let us consider the example from~\cite{Guo2017}. Assume $P = 10$ dB and $K_{a}^{(n)} = n/4$. We consider three cases: (a) $K^{(n)} = n$, (b) $K^{(n)} = n^2$ and (c) $K^{(n)} = n^3$. The corresponding message-length capacities are:
\begin{itemize}
\item[(a)] $k^{(n)} \sim 2 \log_2(1 + 5 n / 2) - 4 h (1/4)$
\item[(b)] $k^{(n)} \sim 2 \log_2(1 + 5 n / 2) - \log_2 (4 e n)$
\item[(c)] $k^{(n)} = 0$.
\end{itemize}
\end{example}

\begin{remark}[Infinite energy per bit]
The authors of~\cite{Guo2017} assume a fixed power $P$. However, the energy required to send one information bit tends to infinity as $n$ grows. Indeed, consider the case where all users are active:
\[
K_{a}^{(n)} = K^{(n)} = \mu n.
\]
The energy per information bit is given by:
\[
\frac{Pn}{\frac{1}{2 \mu} \log_2(1 + \mu n P)} \to \infty.
\]
\end{remark}

\begin{remark}
By \emph{energy per bit}, we consider the ratio $E_b/N_0$, where $E_b$ represents the energy required to transmit one information bit and $N_0$ denotes the noise power spectral density. Through appropriate scaling of both signal and noise, we may assume $N_0 = 1$ without loss of generality, making the quantities $E_b$ and $E_b/N_0$ equivalent. Therefore, these terms may be used interchangeably throughout this monograph. Furthermore, given $n$ channel uses and a transmit \emph{power} $P$, the total transmitted \emph{energy} equals $Pn$.
\end{remark}

\subsection{Finite energy and joint probability of error}\label{sec:many_mac2}

Recall that energy efficiency is of critical importance for massive \ac{mtc}. Namely we are interested in  \emph{finite} energy per information bit. We have the following result~\cite{Poly2018sk}.

\begin{proposition}\label{prop:p_joint}
Assume the case where all users are active, i.e., $K_{a} = K$, and each user aims to transmit $k = 1$ bit over the channel \eqref{eq:gmac_for_converse} with finite energy $\mathcal{E}$. Then we have
\[
1 - P_J \leq \frac{\mathcal{E} \log e + \log2}{\log K}.  
\]

Thus, if $K \to \infty$, then the success probability $1-P_J \to 0$, implying that the users cannot reliably transmit even a single bit.
\end{proposition}
\begin{proof}

First, let us show that the standard Fano-inequality-based converse does not work. We have (recalling that $M=2$)
\[
1 - P_J \leq \frac{\frac{n}{2} \log\left( 1 + \frac{K}{n} \mathcal{E} \right) + \log 2}{K \log M} = \frac{1}{2 \log2} \frac{n}{K} \log\left( 1 + \frac{K}{n} \mathcal{E} \right) + \frac{1}{K},
\]
and we may obtain different relationships depending on the ratio $K/n$ (e.g., in the case $K = \mu n$).

Now, let us introduce a \emph{genie} that provides us with the vector
$$
\dvec{W}' = (W'_1, \ldots, W'_K) \in \{0,1\}^K
$$
such that
\[
d(\dvec{W}, \dvec{W}') = 1,
\]
where $\dvec{W} = (W_1, \ldots, W_K)$ represents the transmitted messages, and $d(\dvec{a}, \dvec{b})$ is the Hamming distance between vectors $\dvec{a}$ and $\dvec{b}$.

We note that the genie’s help reduces the number of hypotheses. Initially, there are $2^K$ possible messages, but with the genie's assistance, this number is reduced to $K$, meaning that we only need to determine the error position. Thus,
\[
1 - P_J \leq \Pr\left[ U = \hat{U} \right],
\]
where $U$ is the index of the user whose message was received in error.

Next, we reformulate the problem in terms of the modified received signal:
\[
Y' = Y - \sum\limits_{i=1}^K f_i(W'_i) = f_U(W_U) - f_U(W'_U) + \dvec{z}.
\]
The goal is now to identify $U$.

Consider the new codebook:
\[
\dmat{C} = [\dvec{c}_1, \ldots, \dvec{c}_K],
\]
of size $n \times K$, where
$$
\dvec{c}_i = f_i(W'_i \oplus 1) - f_i(W'_i),
$$
and note that $\norms{\dvec{c}_i} \leq 2\mathcal{E}$.

Applying Fano's inequality to this new problem, we obtain
\[
\Pr\left[ U = \hat{U} \right] \leq \frac{\frac{n}{2} \log\left( 1 + \frac{2\mathcal{E}}{n}  \right) + \log 2}{\log{K}} \leq  \frac{\mathcal{E} \log e + \log 2}{\log{K}}.
\]
\end{proof}

Thus, the probability of joint decoding error cannot be made arbitrarily small under a finite energy constraint.
\begin{remark}
Given these results, in the subsequent discussion, we switch from joint decoding to the \ac{pupe} criterion. The asymptotic results for finite energy are studied in Appendix~\ref{app:asymptotics}.
\end{remark}

\section{Outcomes}

In this monograph, we focus on uncoordinated multiple access, as this form of channel access is both natural and advantageous for massive \ac{mtc} scenarios involving a large number of infrequently transmitting devices. Coordinated multiple access is a separate problem and falls beyond the scope of this monograph, although some tools from this line of research (e.g., decoding approaches) are employed within the \ac{ura} framework.

In this chapter, we examined different approaches to the \ac{mac} problem, discussing both information-theoretic and random-access methods. The former focuses on physical layer effects but does not account for the randomness of user activity. Moreover, it assumes a fixed number of users while the blocklength tends to infinity. In contrast, the latter provides deep insights into collision resolution algorithms by considering the randomness of user activity, but most works largely ignore the impact of noise at the physical layer. Some papers do take noise into consideration (see, e.g., \cite[Appendix]{liva2011graph}), but none of them aim to formulate the corresponding system requirements.

To bridge this gap, the \ac{ura} paradigm was introduced in~\cite{polyanskiy2017perspective}. The \ac{ura} framework not only addresses the computational challenges posed by a vast number of devices but also treats \ac{ra} as a coding problem---incorporating both medium access protocols and physical layer effects. The \ac{ura} paradigm extends previous work in several ways: it incorporates \ac{fbl} performance, allows for an unbounded total number of users, and prioritizes energy per bit over the traditional notion of rate under vanishing error probability.

At the same time, we emphasize that \ac{ura} is an information-theoretic framework that offers a new perspective on random access, providing a fundamental benchmark to assess the performance of various uncoordinated access schemes. \ac{ura} does not represent a new method of accessing a shared medium.

The following chapters provide a detailed overview of the \ac{ura} problem, covering fundamental limits, low-complexity schemes, and its connection to the \ac{cs} problem.

\chapter{\acs{ura} problem statement and \acl{cs} interpretation}\label{chap3}

Recall that the main problem is to provide multiple access to a massive number of uncoordinated and infrequently communicating devices. In this chapter, we highlight the most important aspects of this problem and provide the corresponding system model and information-theoretic formulation, as proposed by Y.~Polyanskiy in \cite{polyanskiy2017perspective}.

\section{System model}\label{chap3:sys_model}

We consider the scenario of partial user activity. Assume that there are $\Kt \gg 1$ users\footnote{In what follows, $\Kt$ is of no importance, and one may assume $\Kt = \infty$.} in the system, but only $K_a \ll K_\text{tot}$ of them are active at any given time. Throughout this chapter, we assume that the number of active users ($K_a$) is fixed and known at the receiver. The papers~\cite{Durisi2021, Durisi2022} extend these results to the scenario where the number of active users is random and unknown. This extension will be discussed briefly in Chapter~\ref{chap4}. It is worth noting that this extension is a major aspect of any practical scheme, since the receiver lacks a priori knowledge of the number of active users. Instead, it must estimate this quantity from the received signal. All users utilize the same message set $[M]$ and aim to transmit $k = \log_2 M$ bits. Each user selects a message to transmit uniformly and independently of the other users. Communication occurs in a frame-synchronized manner, with each frame consisting of $n$ channel uses.

The main features of this problem formulation are as follows:

\begin{enumerate}
\item \emph{Large number of active users with short data packets.} Unlike the classical \ac{mac} problem formulation, where $\Ka$ is fixed while $n, k \to \infty$, we consider a scenario in which $k$ is fixed while $\Ka$ and $n$ are large.

\item \emph{Users share the same encoder.} Given the vast number of devices, assigning different encoders to each user would result in prohibitive receiver complexity. The receiver would not know which decoder to use and, in the worst case, would need to try all of them. Therefore, a promising strategy is to employ a common encoder $f: [M] \to \mathbb{R}^n$ for all users, leading to a random-access scenario. A natural power constraint,
\[
\left\|f(W)\right\|^2_2 \leq nP, \quad W \in [M],
\]
is also imposed.  

Since users are indistinguishable, the receiver is required only to recover the transmitted messages without identifying their senders. In other words, decoding is performed up to permutation.\footnote{Note that we only consider channels that are permutation-invariant, allowing this type of decoding.} Such schemes are referred to as \acf{ura}, since the source of the message is irrelevant.  

Users may include their identity as part of the $k$-bit payload, but this is not mandatory. For example, a fire sensor only needs to transmit the coordinates of a fire without an explicit identifier.

\item \emph{User-centric probability of error or \ac{pupe}.\footnote{To justify the choice of \ac{pupe}, see Proposition~\ref{prop:p_joint}, which shows that achieving a small joint probability of error is infeasible.}} Let
$$
\dset{T} = \{W_1, W_2, \ldots, W_{K_a}\}
$$
denote the set of transmitted messages. The decoder outputs a set ${\dset{R}} \subseteq [M]$, and we measure the probability of error as follows:
\begin{equation}
\label{eq:ch3:p_missed}
P_e = \mathbb{E}\left[ \frac{| \dset{T} \backslash \dset{R}|}{\Ka} \right] = \frac{1}{K_a} \sum\limits_{i=1}^{K_a} \prob{W_i \not\in {\dset{R}}}.
\end{equation}
Tracking $P_e$ is sufficient when the output list size is fixed (e.g., the usual assumption for bound derivations is $|{\dset{R}}| = K_a$). However, when $|{\dset{R}}|$ is a random variable (as in practical schemes), we also define the \emph{false alarm rate (FAR)}:
\begin{equation}
\label{eq:ch3:p_false}
P_f = \mathbb{E}\left[ \frac{| \dset{R} \backslash \dset{T}|}{|\dset{R}|} \right].
\end{equation}

\item \emph{Energy efficiency.} Since the devices are autonomous and battery-powered, the goal is to minimize the energy per bit while ensuring a maximum tolerable \ac{pupe} of $P_e \leq \varepsilon$:
\begin{equation}\label{eq:ch3_energy_efficiency}
\min\limits_{\text{s. t. } P_e \leq \varepsilon}{\frac{E_b}{N_0}}, \quad \text{where} \quad \frac{E_b}{N_0} = \frac{Pn}{2k}.
\end{equation}
\end{enumerate}

In this monograph, we focus on the Gaussian channel, as it provides key insights into the \ac{ura} problem without unnecessary complications. The fading model and its main implications will be discussed briefly in Chapter~\ref{chap6}. 

Let $\dvec{x}^{(i)} = f(W_i)$, for $i \in [\Ka]$. The output of the Gaussian channel, $\rvec{y} \in \mathbb{R}^n$, is given by:
\begin{equation}
\label{eq:original_cs}
\rvec{y}=\sum_{i=1}^{K_a} \dvec{x}^{(i)} +\rvec{z},
\end{equation}
where $\rvec{z} \sim \mathcal{N}\br{\bzero,\bI_{n}}$ is the \ac{awgn}.

\begin{remark}
In \ac{ura} analysis, we focus on the non-asymptotic regime, as users utilize the same codebook of fixed size, and we cannot increase the number of users indefinitely (since $|\dset{T}| \leq M$). We provide the asymptotic analysis of the bounds described in Chapter~\ref{chap4} in Appendix~\ref{app:asymptotics}. For this purpose, we shift to a different codebook scenario (see Section~\ref{sec:many_mac2}).
\end{remark}

\begin{remark}
Note that the factor $2$ in \eqref{eq:ch3_energy_efficiency} corresponds to a real-valued \ac{awgn} channel. In the case of a complex-valued channel (Chapter~\ref{chap6}), this factor should be omitted.
\end{remark}

\begin{remark}
We note that \ac{ura} does not introduce a new way of accessing the channel. For example, the ALOHA system presented in Section~\ref{ch2:sec_par_active} is a good example of a \ac{ura} scheme. Instead, \ac{ura} provides a new information-theoretic model that incorporates both medium access protocols and physical layer effects, focusing on energy per bit rather than the more classical throughput and rate under vanishing error. This model applies to already existing schemes, such as ALOHA.
\end{remark}

\begin{remark}
In this manuscript, we assume perfect synchronicity, as proposed in the original work~\cite{polyanskiy2017perspective}. We consider that synchronization is supported by additional mechanisms such as beaconing and proper frame structuring. In the case of perfect synchronicity, the analysis becomes clearer and more straightforward.

At the same time, we note that keeping autonomous low-power devices synchronized is a challenging problem. Asynchronous \ac{ura} has been addressed in various research papers~\cite{Amalladinne2019, fengler2023advantagesasynchronyunsourcedmac, Chen2023Async, Decurninge2022Async}; however, even the fundamental limits for this setup remain unknown. We mention this as an interesting open problem in Chapter~\ref{chap7}.
\end{remark}

\section{\acs{ura} as an \acl{asr} problem}\label{sec:ura_as_cs}

In some situations, it is beneficial to represent sets\footnote{or multi-sets if the same message was chosen by multiple users.} of messages as corresponding indicator vectors. Consider a multi-set:
$$
\dset{T} = \{W_1, W_2, \ldots, W_{K_a}\}, \quad W_i \in [M], \quad i \in \Ka.
$$
We can represent this set as a vector:
\[
\dvec{u} = [u_1, u_2, \ldots, u_M]^T,
\]
where $u_i$ denotes the number of times message $i$ was chosen or the multiplicity of $i$ in $\dset{T}$.

\begin{example}
Let $M = 5$, $\Ka = 8$, and
$$
\dset{T} = \{1, 2, 2, 3, 3, 5, 5, 5\}.
$$
The corresponding vector $\dvec{u}$ is given by:
\[
\dvec{u} = [1, 2, 2, 0, 3]^T.
\]
\end{example}

It is worth mentioning that the \ac{ura} problem formulation allows for a standard \ac{cs} interpretation~\cite{Donoho2006CS, CandesTao2006}. Recall that $M$ is the number of messages and $n$ is the blocklength. We can represent the channel output as
\begin{equation}\label{eq:ura_as_cs}
\rvec{y} = \dmat{X} \dvec{u} + \rvec{z},
\end{equation}
where:
\begin{itemize}
\item $\mathbf{y}$ is the channel output, representing a noisy mixture of users' codewords,
\item $\mathbf{X} = [ \dvec{x}_1, \ldots, \dvec{x}_M] = [x_{i,j}]$ is a codebook of size $n \times M$ with $M = 2^k$, where $\dvec{x}_W = f(W)$ for $W \in [M]$,
\item $\mathbf{u}$ represents the activity vector, and $\rvec{z}$ is the \ac{awgn} vector. Note that $|\mathrm{supp}(\dvec{u})| \leq \Ka$ and $\sum\nolimits_{i=1}^M u_i = \Ka$.
\end{itemize}

The matrix $\dmat{X}$ can be viewed as a sensing matrix, and $\dvec{u}$ is a sparse vector (with sparsity defined by the number of active users) to be reconstructed. More precisely, our problem is the \acf{asr} problem~\cite{Reeves2013LB}, as we only need to determine the support of $\dvec{u}$ up to some distortion defined by the desired \ac{pupe}:
\[
P_e = \mathbb{E}\brs{\frac{| \mathrm{supp}({\dvec{u}}) \backslash \mathrm{supp}(\hat{\dvec{u}})|}{\Ka}} \leq \varepsilon,
\]
where $\hat{\dvec{u}}$ is the recovered activity vector or the indicator vector of the set $\dset{R}$.

\begin{figure}[t]
\centering
\includegraphics{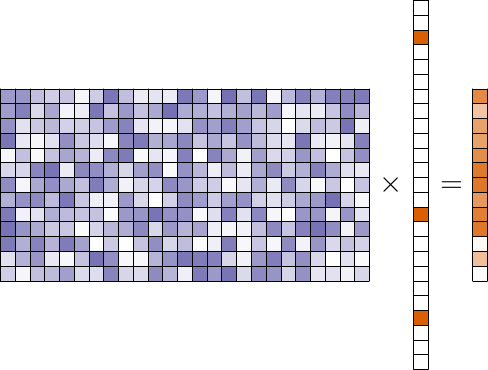}
\includegraphics{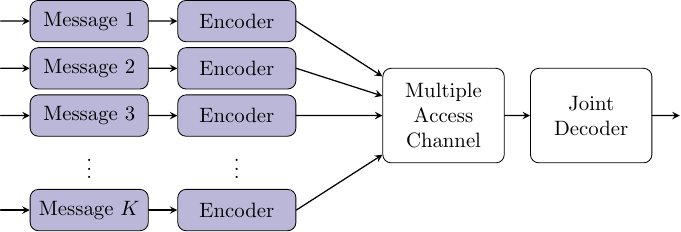}
\caption{Equivalence of the \ac{cs} problem (top) and the same-codebook multiple-access problem (bottom). The \ac{cs} problem is given by~\eqref{eq:ura_as_cs}. Orange squares represent the nonzero elements of the users' activity vector. This figure is reproduced from~\cite{Poly2021tut}}.\label{fig_cs_illustration}
\end{figure}

\begin{remark}\label{remark:msg_unique}
The number of unique messages $|\dset{T}|$ can be less than $K_a$. Assuming $W_i \sim \mathrm{Unif}([M])$, $i \in [\Ka]$, we can estimate the probability of this event as follows:
\[
p' = \prob{\left|\dset{T}\right| < K_a} = 1 - \prod\limits_{i=0}^{K_a-1} \left( 1 - \frac{i}{M} \right) \leq \frac{\binom{K_a}{2}}{M}.
\]
In this monograph, we consider $\log_2 M \approx 100$ bits and $1 \leq \Ka \leq 500$. For these parameters, $p'$ is negligible, allowing us to disregard the case where $|\dset{T}| < \Ka$ and thus work with $\dvec{u}, \hat{\dvec{u}} \in \{0,1\}^M$, where $|\mathrm{supp}({\dvec{u}})|= |\mathrm{supp}(\hat{\dvec{u}})| = \Ka$.
\end{remark}

We establish a schematic correspondence between the \ac{ura} and \ac{asr} problems in~\Fig{fig_cs_illustration}. The sensing matrix is represented by a codebook shared among all users. This matrix is multiplied by the activity vector $\dvec{u}$ (with nonzero elements marked by orange squares). The resulting vector, after being corrupted by noise, produces the channel output.

\begin{remark}
A key observation in related research is that the dimensionality of this problem is enormous---on the order of $2^{100}$ or greater. This characteristic makes the application of standard \ac{cs} algorithms infeasible.
\end{remark}

\section{Converse bounds}

In this section, we describe two converse bounds for the \ac{ura} problem. The first is a single-user converse, with the only difference being that it accounts for the output list of messages. The second utilizes a capacity versus rate-distortion argument.

\begin{theorem}[Single-user converse~\cite{Poly2021tut}]
\label{the:converse_su}
Consider $K_a$ users transmitting messages $W_i \in [M]$, $i \in [K_a]$, via a Gaussian channel within $n$ channel uses, under a maximum tolerable \ac{pupe} $\varepsilon$. Given that the decoder’s output list size is $\ell_0$, then
\begin{equation}\label{eq:converse_ura}
nP \geq \left( Q^{-1}(\ell_0/M) + Q^{-1}(\varepsilon)\right)^2, \quad Q(x) = \int\limits_{x}^{+\infty} \frac{1}{\sqrt{2 \pi} } e^{-y^2/2} dy,
\end{equation}
where $Q^{-1}\br{\cdot}$ is the inverse of the $Q$-function defined above.
\end{theorem}

\begin{proof}
We provide a sketch of the proof:
\begin{itemize}
\item In \cite{ppv2011}, it was shown that any single-user channel code over a Gaussian channel with parameters $(n,M,P)$ and error probability $P_e \triangleq \Pr[W \ne \hat{W}] \leq \varepsilon$ must satisfy
\[
nP \geq \left( Q^{-1}\br{1/M} + Q^{-1}(\varepsilon)\right)^2.
\]
\item This argument extends naturally to show that list-$\ell_0$ decodable codes (where $P_e \triangleq \Pr[W \not\in \dset{R}] \leq \varepsilon$ and $|\dset{R}| = \ell_0$) must satisfy~\eqref{eq:converse_ura}.
\item \emph{Symmetry argument}: It suffices to obtain a lower bound on the probability that a particular user's message is not in the decoded list.
\item \emph{Genie argument}: Assume the decoder has access to the interference from all other users. The equivalent channel is given by:
\[
\rvec{y}' = \dvec{x}^{(1)} + \rvec{z},
\]
to which the bound above applies.
\end{itemize}
\end{proof}

\begin{theorem}[Multi-user converse~\cite{polyanskiy2017perspective}]
\label{the:converse_ura}
Consider $K_a$ users transmitting messages $W_i \in [M]$, $i \in [K_a]$, via a Gaussian channel within $n$ channel uses, under a maximum tolerable \ac{pupe} $\varepsilon$. Given that the decoder’s output list size is $\ell_0$, then
$$
\log_2{M} \leq 
\frac{1}{\left(1 - \varepsilon\right)}\left(
\frac{n}{K_a} C(\Ka P) + h\left(\varepsilon\right) + \left(1 - \varepsilon\right)\log_2{\ell_0}\right),
$$
where $C(x) = \nicefrac{1}{2} \log_2 (1 + x)$.
\end{theorem}

\begin{proof}
Consider the $i$-th user. Let $\varepsilon_i = \prob{\rscalar{W}_i \not\in \dset{R}}$. Using Fano’s list inequality, we obtain:
\begin{equation}\label{eq:list_fano}
H\left(\rscalar{W}_i | \rvec{y} \right) \leq h(\varepsilon_i) + \left(1 - \varepsilon_i \right) \log_2{\ell_0} + \varepsilon_i \log_2{\left(M - \ell_0\right)}.
\end{equation}

Summing inequalities~\eqref{eq:list_fano} over all $i \in [K_a]$ and noting that 
$$
H\left(\rscalar{W}_i | \rvec{y} \right) = H(\rscalar{W}_i) - I(\rscalar{W}_i; \rvec{y}),
$$
we obtain:
\begin{flalign*}
\sum\limits_{i=1}^{K_a}I\left(\rscalar{W}_i;\rvec{y} \right) &\geq
\sum\limits_{i=1}^{K_a} H(\rscalar{W}_i) \\
& - \sum\limits_{i=1}^{K_a} h\left(\varepsilon_i\right) - \left(1 - \sum\limits_{i=1}^{K_a} \varepsilon_i \right) \log_2{\ell_0} \\
& - \sum\limits_{i=1}^{K_a} \varepsilon_i \log_2{\left(M - \ell_0\right)}.
\end{flalign*}

Note that:
\begin{flalign*}
H(\rscalar{W}_i) &= \log_2M, \quad \text{for} \ i \in [K_a], \\
\sum\limits_{i=1}^{K_a} \varepsilon_i &= K_a \varepsilon, \\
\sum\limits_{i=1}^{K_a} h\left(\varepsilon_i\right) &\geq K_a h(\varepsilon).
\end{flalign*}

Using the Markov chain:
\[
(\rscalar{W}_1, \ldots, \rscalar{W}_{\Ka}) \to (\rvec{x}^{(1)}, \ldots, \rvec{x}^{(\Ka)}) \to \rvec{y} \to \dset{R},
\]
we derive:
\begin{flalign*}
\sum\limits_{i=1}^{\Ka}I\left(\rscalar{W}_i;\rvec{y} \right) &\leq I(\rscalar{W}_1, \rscalar{W}_2, \ldots, \rscalar{W}_{\Ka}; \rvec{y}) &\text{(independent messages)} \\
& \leq I(\rvec{x}^{(1)}, \rvec{x}^{(2)}, \ldots, \rvec{x}^{(\Ka)}; \rvec{y}) &\text{(data processing ineq.)}\\
& \leq \sum\limits_{j=1}^n I(\rscalar{x}^{(1)}_j, \rscalar{x}^{(2)}_j, \ldots, \rscalar{x}^{(\Ka)}_j; \rscalar{y}_j) &\text{(memoryless channel)}\\
& \leq n C(\Ka P).
\end{flalign*}
\end{proof}

\begin{remark}
A very similar problem was considered in \cite{Reeves2013LB} and presented in \cite{Poly2021tut}. The main difference is that \cite{Reeves2013LB, Poly2021tut} consider $\dset{T}$ chosen uniformly from $\binom{M}{K_a}$ variants (or equivalently, the indicator vector $\dvec{u} \in \{0,1\}^M$, where $|\mathrm{supp}(\dvec{u})| = K_a$). The converse presented above remains valid even when messages are repeated. For the regime of interest, the two converses coincide.
\end{remark}

We refer the reader to Section~\ref{chap4:num_res}, where different achievability bounds are compared, highlighting the gap between achievability and converse bounds presented in Theorems~\ref{the:converse_su} and~\ref{the:converse_ura}. The results are shown in~\Fig{chap4:num_res}. Note that the converse is depicted as a single line, representing the maximum of the bounds from Theorems~\ref{the:converse_su} and~\ref{the:converse_ura}.

\section{Parameters of interest}\label{sec:params}
We need to fix the parameters $(n, k, \Ka)$. Consider a typical LoRa network:
\begin{itemize}
    \item Payload size: $k \approx 100$ bits.
    \item A message transmitted using LoRa with a spreading factor of 11 occupies approximately $k \cdot\nicefrac{2^{11}}{11}$ complex channel uses, resulting in $n \approx 3\times 10^4$ real channel uses.
\end{itemize}
Following~\cite{polyanskiy2017perspective}, we use the following parameters:
\begin{itemize}
    \item Frame length: $n = 3\times 10^4$ (real channel uses),
    \item User payload: $k = 100$ bits,
    \item Number of active users: $K_a = 1, \dots, 500$,
    \item Target \ac{pupe}: $\varepsilon \in \{0.001, 0.05, 0.1\}$.
\end{itemize}
The goal is to find the minimal $E_b/N_0$ in~\eqref{eq:ch3_energy_efficiency} such that $P_e \leq \varepsilon$.

We also note that users share a large number of channel uses compared to the number of information bits, leading to an individual coding rate of approximately $1/300$.

\chapter{\acl{gura}: achievability bounds}\label{chap4}

In this chapter, we consider achievability bounds for the \acf{gura}. Specifically, we are interested in the tradeoff between energy efficiency ($E_b/N_0$) and the number of active users ($\Ka$). The term ``achievability'' implies that we do not account for decoding complexity and instead consider the potential capabilities under an infeasible decoding algorithm.

We begin the chapter by defining the \ac{ml} decoding rule and establishing a general setup for evaluating achievability bounds, as presented in Section~\ref{ch5:ml_rule_statement}. Within this setup, we assume that the number of active users is known at the receiver.

In Section~\ref{ch5:sec_tricks}, we introduce Gallager's $\rho$-trick and Fano's trick -- key components of the subsequent analysis. Next, in Section~\ref{ch4:sec_ach_bounds}, we derive three achievability bounds. The first bound, presented by Y. Polyanskiy~\cite{polyanskiy2017perspective}, is based on Gallager's $\rho$-trick (see Theorem~\ref{th:polyanskiy_bound}). This bound has become a foundational result in \ac{ura} research. The second bound is based on Fano's trick (see Theorem~\ref{th:gauss_codebook}). We find that this bound provides a similar energy efficiency estimate to the first one. However, a straightforward modification of this bound enables us to estimate energy efficiency under binary codebooks (see Theorem~\ref{th:binary_codebook}). Notably, we observe that switching from Gaussian codebooks to binary does not lead to a loss in achievable energy efficiency estimates.

Next, we describe an achievability bound that combines both tricks. Originally designed to evaluate the performance of \ac{sparcs}, this bound is presented in Theorem~\ref{th:comb_tricks}. However, these results are also applicable to the \ac{ura} problem.

In Section~\ref{ch4:further_comments}, we discuss additional research on achievable energy efficiency analysis. We present results for a random number of active users, where energy efficiency depends on the probability of error and the \acf{far}, as discussed in Section~\ref{ch4:sec:random_ka}. Finally, in Section~\ref{ch4:sec:gordon_lemma}, we explore the application of Gordon's lemma to the \ac{ura} problem. We conclude this chapter with a numerical analysis of the described bounds, as presented in Section~\ref{chap4:num_res}.

\section{\Acl{ml} decoding rule}\label{ch5:ml_rule_statement}

Recall that $\dset{T} = \{W_1, W_2, \ldots, W_{K_a}\}$ is the set of transmitted messages, and ${\dset{R}} \subseteq [M]$, with $|{\dset{R}}| = K_a$, is the set of received messages. The decoder's task is to recover the set of transmitted messages up to permutation. Note that
\[
\Pb{\bigcup\nolimits_{i \ne j} \{ W_i = W_j \}} \leq \frac{\binom{K_a}{2}}{M},
\]
which is negligible for the system parameters of interest (see Remark~\ref{remark:msg_unique}). Thus, we do not consider multi-sets in this section.

Next, we introduce helpful notation. For $\dset{I} \subseteq [M]$, we define
$$
\dvec{s}_{\dset{I}} = \sum\limits_{W \in \dset{I}}f\br{W} = \sum\limits_{W \in \dset{I}} \dvec{x}_W.
$$
We now formulate the \ac{ml} decoding rule:
\begin{equation}
\label{eq:ura_awgn_decoding_rule}
\dset{R} = \argmin_{\dset{R}' \subseteq [M], \left| \dset{R}'\right| = K_a} \norms{\rvec{y} - \dvec{s}_{\dset{R}'}}.
\end{equation}

\begin{remark}
We note that the decoder mentioned above is optimal for the joint error probability rather than for \ac{pupe}. A \ac{pupe}-optimal decoder should compute
\[
\Pr\left[ W \in \dset{T} \: \middle| \: \dvec{y}\right] \propto \sum\limits_{\dset{T}': W \in \dset{T}'} p\left(\dvec{y} \: \middle| \: \dvec{s}_{\dset{\dset{T}'}}\right) \Pr\left[\dset{T}'\right]
\]
for all $W \in [M]$, and then output the top-$\Ka$ messages. The choice of the joint decoder was motivated by its suitability for analysis.
\end{remark}

\subsection{Probabilistic method, \texorpdfstring{$t$}{t}-error events and \texorpdfstring{$P_e$}{Pe} estimate}\label{sec:pt_naive}

It is not possible to analyze $P_e$ for a specific codebook. Therefore, we utilize the so-called probabilistic method~\cite{prob_method} or random coding~\cite{Shannon_MTC_1948}. Specifically, we consider the ensemble $\mathcal{G}$ of codebooks $\rmat{X}$ specified by the probability distribution  $P_{\rmat{X}}$\footnote{To utilize the symmetry property, we consider only ensembles with i.i.d. codewords.} and compute
\[
{P}_e = \mathbb{E}_{\rmat{X}} [P_e(\rmat{X})].
\]

We can then assert that there exists a codebook $\dmat{X}^* \in \Gens$ such that $P_e(\dmat{X}^*) \leq {P}_e$. Due to symmetry when averaging over codebooks, we assume w.l.o.g. that $\dset{T} = [K_a]$. Let $\dset{C} = \dset{T} \bigcap {\dset{R}}$, $\dset{M} = \dset{T} \backslash {\dset{R}}$ and $\dset{F} = {\dset{R}} \backslash \dset{T}$ be the sets of correctly received, missed, and falsely detected messages, respectively (see~Fig.~\ref{fig:hat_s_structure}).

\begin{figure}[t!]
\centering
\includegraphics{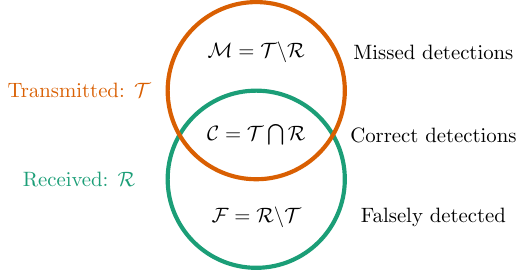}
\caption{Correctly detected, missed, and falsely detected codewords. The set of transmitted messages has size $\left|\dset{T}\right| = K_a$. Throughout this chapter, we assume $\left|\dset{R}\right| = K_a$ for all bounds discussed, which implies that the receiver knows the number of transmitted messages. In this case, we assume that $\left|\dset{M}\right|=\left|\dset{F}\right|$.}
\label{fig:hat_s_structure}
\end{figure}

Following the approach of \cite{polyanskiy2017perspective}, we consider $t$-error events $\dset{E}_t = \{|\dset{M}| = t\}$ and use the following estimate:
\begin{equation}\label{eq:chapter4:pupe_definition}
P_e \leq \sum_{t = 1}^{K_a} \frac{t}{K_a} \Pb{\dset{E}_t} + p_0,
\end{equation}
where
\begin{equation}\label{eq:chapter4:p0}
p_0 = \prob{\left[ \bigcup\limits_{i \ne j} \{ W_i = W_j \} \right] \cup \left[ \bigcup\limits_{i \in [K_a]} \{ \norm{f(W_i)}^2 > Pn \} \right]}
\end{equation}
is the probability that different users select the same codewords from the codebook or that the signal power exceeds $Pn$. Thus, in what follows, we do not consider colliding messages. For ensembles of Gaussian codebooks considered in Section~\ref{ch4:sec_ach_bounds}, the probability of power violation, $p_0$, is bounded by~\eqref{eq:p0_bound}. For binary codebooks considered in Section~\ref{ch4:sec:binary_codebook}, we have $p_0=0$. Recall that the number of active users is assumed to be known at the receiver.

Let us examine the error event in more detail. In the case of \ac{ml} decoding, a decoding error occurs if there exists a set of messages ${\dset{R}}$ such that the sum of the corresponding codewords is closer to the received sequence $\rvec{y}$ than the sum of the codewords corresponding to the transmitted messages $\dset{T}$:
$$
\norms{\rvec{y} - \rvec{s}_{\dset{T}}} \geq \norms{\rvec{y} - \rvec{s}_{{\dset{R}}}}.
$$

Taking into account that $\rvec{y} = \rvec{s}_{\dset{T}} + \rbz$, $\rvec{s}_{\dset{T}} = \rbcm + \rbcc$ and $\rvec{s}_{\dset{R}} = \rbcf + \rbcc$, we obtain the following decoding error event:
\begin{equation}\label{eq:pol_bound_error_event}
\dset{E}(\dset{M}, \dset{F}) = \brc{\norms{\rbz} - \norms{\rbcm - \rbcf + \rbz} \geq 0}.
\end{equation}

Note that
\[
\dset{E}_t = \bigcup_{\dset{M} \subseteq \dset{T}, \dset{F} \subseteq [M] \backslash [K_a]} \dset{E}(\dset{M}, \dset{F}).
\]

Let us begin with the naive approach and apply the union bound in a straightforward manner. We have
\begin{eqnarray*}
\Pb{\dset{E}_t} &=& \Pb{ \bigcup_{\dset{M}, \dset{F}} \dset{E}(\dset{M}, \dset{F}) } = \mathbb{E}_{\rmat{X}} \left[ \Pb{ \bigcup_{\dset{M}, \dset{F}} \dset{E}(\dset{M}, \dset{F}) \: \middle| \: \rmat{X} } \right] \\
&\leq& \binom{K_a}{t} \binom{M-K_a}{t} \mathbb{E}_{\rmat{X}} \left[ \Pb{ \dset{E}(\dset{M}', \dset{F}') \: \middle| \: \rmat{X} } \right],
\end{eqnarray*}
where $\dset{M}' = [t]$, $\dset{F}' = [K_a + t]\backslash[K_a]$. The last step follows from the union bound and the symmetry property (as we are averaging over the codebook).

We note that the second binomial coefficient is extremely large (recall that $M$ is on the order of $2^{100}$). Thus, the naive approach significantly overestimates the desired probability. In what follows, we describe several approaches to tighten the union bound.

\section{Useful tricks to tighten the union bound}\label{ch5:sec_tricks}

In this section, we describe all the tools needed for the proofs. We begin with approaches to improve the union bound, as proposed by Gallager\footnote{also known as Gallager's $\rho$-trick}~\cite{gallager_trick} and Fano~\cite{fano1961transmission}.

\subsection{Gallager's trick}\label{sec:gallager_trick}

Gallager's $\rho$-trick is a method to improve the union bound, which can be formulated as follows. Let $\dset{E}_{i}$, for $i \in [m]$ be the events. Then, for any $0 \leq \rho \leq 1$ we have
\[
\prob{\bigcup_{i=1}^m \dset{E}_{i}} \leq \br{\sum\limits_{i=1}^m\prob{\dset{E}_{i}}}^{\rho}.
\]

It is reasonable to apply this method at intermediate steps in estimating conditional probabilities. Specifically, one needs to find a suitable random variable $V$ such that
\[
\Pr\left[\dset{E}_i \: \middle| \: V \right] \leq \exp\left[-F(V)\right]
\]
for some function $F(V)$. 

Then the $\rho$-trick leads to the following estimate:
\[
\prob{\bigcup_{i=1}^m \dset{E}_{i}} \leq m^\rho \mathbb{E}_V\left[\exp\left[- \rho F(V)\right]\right].
\]

\subsection{Fano's trick}\label{sec:fano_trick}

Let us explain the idea using a simple single-user case. Let $\mathcal{C} = \{\dvec{x}_1, \ldots, \dvec{x}_M\}$ and
\[
\rvec{y} = \dvec{x}_1 + \rvec{z},
\]
where $\dvec{x}_1 \in \mathcal{C}$ and $\rvec{z}$ is the noise vector.

Consider the \ac{ml} or minimum Euclidean distance decoding. The error event is $\mathcal{E} = \bigcup\nolimits_{m=2}^M \mathcal{E}_m$, where $\mathcal{E}_m = \{\norms{\rvec{y} - \dvec{x}_1} \geq \norms{\rvec{y} - \dvec{x}_m}\}$.

Let us consider why the union bound overestimates the probability. This occurs when several $\dvec{x}_m$ are closer to $\rvec{y}$ than $\dvec{x}_1$. Clearly, this situation becomes more likely when the norm of the noise is large.

Fano's method can be described as follows. Let us choose a so-called ``good'' region $\dset{B}$. We proceed as follows:
\begin{equation}\label{eq:fano_trick_def}
\Pb{\dset{E}} \leq \Pb{\dset{E} \bigcap \dset{B}} + \Pb{\dset{B}^c},
\end{equation}
and then apply the union bound for the first term only.

In the single-user example, a reasonable choice is $\dset{B} = \{\norms{z} \leq \beta n\}$ for some $\beta > 0$. In the general case, the choice of $\dset{B}$ poses a significant challenge.

\section{Achievability bounds} \label{ch4:sec_ach_bounds}

In this section, we present three achievability bounds. All of these bounds consider a Gaussian codebook and a paradigm of averaging over multiple codebooks generated from the \emph{ensemble}. To proceed, let us begin with the definition of the ensemble.

\begin{define}\label{def:gauss_codebook}
Let $\Gens$ be the ensemble of Gaussian codebooks of size $n \times M$, where each element is sampled i.i.d. from $\mathcal{N}\br{0,P}$.
\end{define}

The key difference between the approaches presented below lies in how the union bound is tightened. We utilize Gallager's $\rho$-trick~\cite{polyanskiy2017perspective}, Fano's trick~\cite{Glebov2023GMAC}, and their combination presented in~\cite{BarronSparcs2019}. We note that the problem and approaches considered in this section share many similarities with the optimal decoder analysis of \ac{sparcs}~\cite[Chapter 2]{BarronSparcs2019}.

Fix $P' < P$ and consider $\Gensp$. The following estimate is valid for $p_0$ (see~\eqref{eq:chapter4:p0}) 
\begin{equation}\label{eq:p0_bound}
p_0 \leq \overline{p}_0 \triangleq \frac{\binom{K_a}{2}}{M}+K_a\Pb{\frac{1}{n}\sum_{i=1}^{n}\xi_i^2> \frac{P}{P'}}, \:\: \xi_i \overset{i.i.d}{\sim} \mathcal{N}(0,1).    
\end{equation}

\subsection{The bound based on Gallager's \texorpdfstring{$\rho$-}{}trick}

We start with the bound that utilizes the Gallager's $\rho$-trick and the chain rule:
\[
\Pr\nolimits_{\rscalar{x}, \rscalar{y}, \rscalar{z}}
\left[\mathcal{E}(\rscalar{x}, \rscalar{y}, \rscalar{z})\right] =
\mathbb{E}_{\rscalar{z}}\left[
\mathbb{E}_{\left.\rscalar{y} \: \middle| \: \rscalar{z}\right.} \left[ \Pr\nolimits_{\left.\rscalar{x} \: \middle| \: \rscalar{y}, \rscalar{z}\right.}\left[\mathcal{E}\left(\rscalar{x}, \rscalar{y}, \rscalar{z}\right) \: \middle| \: \rscalar{y}, \rscalar{z}\right] \right]
\right]. 
\]

In what follows, we work with independent random vectors, so there is no need to deal with expectations over conditional distributions.

\begin{theorem}[Y.~Polyanskiy, \cite{polyanskiy2017perspective}]\label{th:polyanskiy_bound}
Fix $P' < P$. There exists a codebook $\dmat{X}^* \in \Gensp$ satisfying the power constraint $P'$ and providing $P_e$ using eq.~\eqref{eq:chapter4:pupe_definition}, where $p_{0}$ is bounded by~\eqref{eq:p0_bound}, and $\Pb{\dset{E}_t} \leq p_t$, where
\[
p_t = \inf\limits_{0 \leq \rho_{1},\rho_{2} \leq 1} \exp{\brs{-n\br{-\rho_{1}\rho_{2}R_{1} - \rho_{1}R_{2} + E_{0}\br{\rho_{1}, \rho_{2}} }}},
\]
\[
R_1 = \frac{1}{n}\log \binom{M - \Ka}{t}, \quad 
R_2 = \frac{1}{n}\log \binom{\Ka}{t},
\]
\[
E_0\br{\rho_{1}, \rho_{2}} = \frac{1}{2}\br{\rho_1 a + \log(1-2b\rho_1)},
\]
\[
a = \rho_{2} \log (1+2P't\lambda) + \log (1+2P't\mu), \quad
b = \rho_{2} \lambda - \frac{\mu}{1+2P't\mu},
\]
\[
\mu = \frac{\rho_{2} \lambda}{1+2P't\lambda}, \quad 
\lambda = \frac{P't -1 + \sqrt{D}}{4 (1+ \rho_{1} \rho_{2}) P' t}, \quad
D = (P't-1)^2 + 4P't \frac{1 + \rho_{1} \rho_{2}}{1 + \rho_{2}}.
\]
\end{theorem}
\begin{proof}
If we look at \eqref{eq:chapter4:pupe_definition}, it is clear that we need to deal with $\Pr[\mathcal{E}_t]$, where the randomness is induced by $\rvec{z}$ and the random matrix $\rmat{X}$. Our goal is to show that $\Pr[\mathcal{E}_t] \leq p_t$, where $p_t$ is given in the theorem statement.

Let
\[
\dset{E}(\dset{M}) = \bigcup_{\dset{F} \subseteq [M] \backslash [K_a]} \dset{E}(\dset{M}, \dset{F}).
\]
Let $0 \leq \rho_{1}, \rho_{2} \leq 1$, and we can represent $\Pr[\mathcal{E}_t]$ as follows:
\begin{flalign}
\Pb{\mathcal{E}_t} &= \mathbb{E}_{\rvec{z}}\sbr{\Pr[\mathcal{E}_t \: \middle| \: \rvec{z}]} & \nonumber\\
&= \mathbb{E}_{\rbz} \sbr{\Pb{ \bigcup\nolimits_{\dset{M} \subseteq [\Ka]} \dset{E}(\dset{M}) \: \middle| \: \dbz }} & \nonumber\\
&\leq \mathbb{E}_{\rbz} \sbr{ \br{\Pb{ \bigcup\nolimits_{\dset{M} \subseteq [\Ka]} \dset{E}(\dset{M}) \: \middle| \: \dbz } }^{\rho_1} } & \text{($\rho$-trick)} & \nonumber \\
&\leq \mathbb{E}_{\rbz} \sbr{ \br{ \sum\nolimits_{\dset{M} \subseteq [\Ka]} \Pb{ \dset{E}(\dset{M}) \: \middle| \: \dbz } }^{\rho_1} } & \text{(union bound)} \nonumber\\
&= \mathbb{E}_{\rbz} \sbr{ \br{ \sum\nolimits_{\dset{M} \subseteq [\Ka]} \mathbb{E}_{\rbcm} \sbr{ \Pb{ \dset{E}(\dset{M}) \: \middle| \: \dbz, \dbcm } }}^{\rho_1} } & \nonumber\\
&= \mathbb{E}_{\rbz} \sbr{ \br{ \binom{\Ka}{t} \mathbb{E}_{\rbcmp} \sbr{ \Pb{ \dset{E}(\dset{M}') \: \middle| \: \dbz, \dbcmp } }}^{\rho_1} }, \label{eq:pol_et} &
\end{flalign}
where the last transition is valid due to symmetry. W.l.o.g., $\dset{M}' = [t]$ in the last expression and in what follows.

Now consider
$$
\Pb{\dset{E}(\dset{M}') \: \middle| \: \dbz, \dbcmp} = \Pb{ \bigcup_{\dset{F} \subseteq [M] \backslash [K_a]} \dset{E}(\dset{M}, \dset{F}) \: \middle| \: \dbz, \dbcmp}
$$
\begin{flalign}
&\leq \br{ \Pb{ \bigcup_{\dset{F} \subseteq [M] \backslash [K_a]} \dset{E}(\dset{M}', \dset{F}) \: \middle| \: \dbz, \dbcmp } }^{\rho_2} & \text{($\rho$-trick)} \nonumber \\
&\leq \br{ \sum\nolimits_{\dset{F} \subseteq [M] \backslash [K_a]} \Pb{ \dset{E}(\dset{M}', \dset{F}) \: \middle| \: \dbz, \dbcmp } }^{\rho_2} & \text{(union bound)} \nonumber \\
&\leq \br{ \binom{M-\Ka}{t} \Pb{ \dset{E}(\dset{M}', \dset{F}') \: \middle| \: \dbz, \dbcmp } }^{\rho_2}, \label{eq:pol_eM}&
\end{flalign}
where the last transition is again due to symmetry. W.l.o.g., $\dset{F}' = [\Ka+t]\backslash [\Ka]$ in the last expression and in what follows.

We start with $\Pr\left[\dset{E}(\dset{M}', \dset{F}') \: \middle| \: \dbcm, \dbz\right]$. In accordance with~\eqref{eq:pol_bound_error_event}, we have 
\[
\Pr\left[\dset{E}(\dset{M}', \dset{F}') \: \middle| \: \dbz, \dbcmp\right] = \Pr \left[\norms{\dbz} - \norms{\dbcmp - \rbcfp + \dbz} \geq 0 \: \middle| \: \dbz, \dbcmp\right].
\]

Let $\lambda > 0$. Applying the Chernoff bound (Lemma~\ref{lemma:chernoff}), we obtain
\[
\Pr\left[\dset{E}(\dset{M}', \dset{F}') \: \middle| \: \dbz, \dbcmp\right] \leq \mathbb{E}_{\rbcfp} \sbr{\exp\left[ \lambda \left( \norms{\dbz} - \norms{\dbcmp - \rbcfp + \dbz} \right) \right]}.
\]
and the vectors $\dbz$ and $\dbcmp$ are fixed. Using Corollary~\ref{cor:quad_form_rv_corollary}, we get
\[
\Pr\left[\dset{E}(\dset{M}', \dset{F}') \: \middle| \: \dbz, \dbcmp\right] \leq \exp\left[ E_1(\dbcmp, \dbz) \right],
\]
where
\[
E_1(\dbcmp, \dbz) = \lambda \left( \norm{\dbz}^2 - \frac{\norm{\dbcmp + \dbz}^2}{1 + 2tP' \lambda} \right).
\]
Substituting this result into~\eqref{eq:pol_eM}, we obtain
\[
\Pr\left[\dset{E}(\dset{M}') \: \middle| \: \dbz, \dbcmp\right] \leq \binom{M-K_a}{t}^{\rho_2} \exp\left[\rho_2 E_1(\dbcmp, \dbz) \right].
\]
Taking here expectation\footnote{We again apply Corollary~\ref{cor:quad_form_rv_corollary}.} over $\rbcmp$ (see~\eqref{eq:pol_et}), we get
\[
\Pr\left[\dset{E}(\dset{M}') \: \middle| \: \dbz\right] \leq \binom{M-K_a}{t}^{\rho_{2}} \exp\left[b \norm{\dbz}^2 - na \right].
\]
Substituting this result into~\eqref{eq:pol_et}, we finally obtain
\[
\Pr\left[\dset{E}_t \: \middle| \: \dbz\right] \leq \binom{K_a}{t}^{\rho_1} \binom{M-K_a}{t}^{\rho_1 \rho_2} \exp\left[\rho_1 (b \norm{\dbz}^2 - na) \right].
\]
Taking expectation over $\rbz$, we obtain the theorem statement. 
\end{proof}

\begin{remark}
One may ask whether the provided sequence of expectations is optimal. In the chain rule, we can use any sequence, for example, by first computing the expectation over $\rbz$. Let us provide some intuition. There are $\binom{M-\Ka}{t}$ possible choices for $\dset{F}$, and this binomial coefficient is extremely large. Therefore, it is reasonable to start with $\rbcf$, as done in the proof, since Gallager's trick will hopefully significantly reduce this large coefficient.
\end{remark}

\subsection{The bound based on Fano's trick}

Let us fix $0 \leq \alpha < 1$, $\beta \geq 0$. We introduce the following ``good'' region:
\begin{equation}\label{eq:ball_prj}
\dset{B} = \bigcap\limits_{\dset{M} \subseteq [K_a]} \dset{B}(\dset{M}),
\quad
\dset{B}(\dset{M}) = \brc{\alpha\norms{\rbcm + \rbz} + \beta n > \norms{\rbz}}.
\end{equation}

\begin{figure}[t]
\begin{center}
\includegraphics{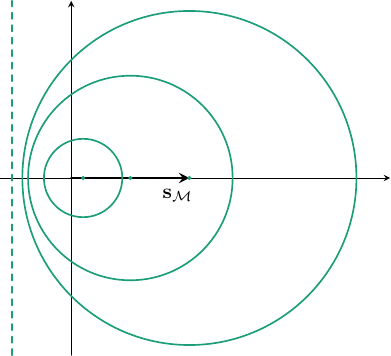}
\end{center}
\caption{Illustration of the ball presented in~\eqref{eq:ball_prj} for $\beta = 0$ and different values of $\alpha$ such that $\nicefrac{\alpha}{1 - \alpha} \in\left\{0.1, 0.5, 1.0\right\}$. Each ball is represented by a green line, and a larger radius corresponds to a larger value of $\alpha$. The dashed line corresponds to the extreme case $\alpha=1$. For $\beta > 0$, the radius of each ball increases, and the dashed line shifts to the left. Note that the zero vector is always inside the ball for any $0 < \alpha \leq 1$  and $\beta \geq 0$. Hence, the intersection of these balls is never empty.\label{ch4:fig_ball}}
\end{figure}

\begin{remark}[Visualization of the ``good'' region]
Let us consider a geometric interpretation of the region~\eqref{eq:ball_prj}. Suppose the vector $\rbcm$ is fixed (since we consider a union over all $\binom{\Ka}{t}$ missed codewords in the subsequent theorem proof). Thus, given $\rbcm$, the condition~\eqref{eq:ball_prj} for a single $\dset{B}(\dset{M})$ can be rewritten as follows (see~\Fig{ch4:fig_ball}).

For $0 < \alpha < 1$, we have:
$$
\left\{\left\|\rbz - \mathbf{v}_0\right\|^2 < r^2\right\}, \quad \mathbf{v}_0 = \frac{\alpha}{1 - \alpha}\rbcm, \quad r^2 = \frac{\alpha\left\|\rbcm\right\|^2 + \beta n\br{1 - \alpha}}{\br{1 - \alpha}^2},
$$
which defines a ball.

Note that $\alpha$ controls both the center $\mathbf{v}_0$ and the radius $r$ of the sphere, while $\beta$ affects only the radius. Additionally, each ball contains the zero point; hence, their intersection is non-empty.

For $\alpha = 1$, we have:
$$
\left\{\rbcm^T\cdot \rbz > -\frac{\left\|\rbcm\right\|^2 +\beta n}{2}\right\}, \quad \alpha = 1,
$$
which defines a hyperplane.
\end{remark}

\begin{remark}
Let us provide some comments on the choice of the ``good'' region:
\begin{itemize}
\item Adding $\rbcc$ to the region description does not lead to a performance improvement, since correctly received codewords do not appear in~\eqref{eq:pol_bound_error_event}. However, the use of $\rbcc$ may be beneficial in the case of fading, as shown in~\eqref{eq:ball}.
\item The region description does not include $\rbcf$ to avoid a large binomial coefficient $\binom{M-K_a}{t}$. Indeed, if $\rbcf$ were included in~\eqref{eq:ball_prj}, then calculating $\Pb{\dset{B}^c}$ in eq.~\eqref{eq:fano_trick_def} would involve taking the complement of an intersection of a large number of events for $\dset{F} \subseteq [M] \setminus [\Ka]$. After applying a union bound, this large binomial coefficient would inevitably appear.
\end{itemize}
\end{remark}

\begin{theorem}[Achievability based on Fano's trick~\cite{Glebov2023GMAC}] \label{th:gauss_codebook}
Fix $P' < P$. There exists a codebook $\dmat{X}^* \in \Gensp$ satisfying the power constraint $P'$ and providing $P_e$ using eq.~\eqref{eq:chapter4:pupe_definition}, where $p_{0}$ is bounded by~\eqref{eq:p0_bound}, and $\Pb{\dset{E}_t} \leq p_{t}$, where
\[
p_{t} = \inf\limits_{\alpha, \beta \geq 0}\br{p_{1,t} + p_{2,t}},
\]
\[
p_{1,t} = \inf\limits_{u,v > 0, \lambda_{\bA} > 0}\exp\brs{-n \br{-R_{1} - R_{2} + \frac{1}{2}\log\det\br{\bI_{3} - 2\bA} - v\beta}},
\]
\[
p_{2,t} = \inf\limits_{\delta > 0, \lambda_{\bB} > 0}\exp\brs{-n\br{-R_{2} + \frac{1}{2}\log\det\br{\bI_{2} - 2\delta\bB} + \delta\beta}},
\]
\[
R_1 = \frac{1}{n}\log \binom{M - \Ka}{t}, \quad 
R_2 = \frac{1}{n}\log \binom{\Ka}{t},
\]
$$
\bA = \br{\begin{array}{rrr}
\br{\alpha - 1}v & \br{\alpha v - u}\sqrt{P't} & u\sqrt{P't} \\
\br{\alpha v - u}\sqrt{P't} & \br{\alpha v - u} P't & uP't \\
u\sqrt{P't} & uP't & -uP't \\
\end{array}},
$$
$$
\bB = \br{\begin{array}{rr}
1 - \alpha & -\alpha\sqrt{P't} \\
-\alpha\sqrt{P't} & -\alpha P't \\
\end{array}}.
$$
By $\lambda_{\bA}$ and $\lambda_{\bB}$, we mean the minimum eigenvalues of $\bI_{3} - 2\bA$ and $\bI_{2} - 2\delta\bB$ accordingly.
\end{theorem}

\begin{proof}
We start by applying Fano's trick~\eqref{eq:fano_trick_def} for $\dset{E}_t$, where $\dset{B}$ is given by~\eqref{eq:ball_prj}.

Let us start with $\Pb{\dset{E}_t \bigcap \dset{B}}$. We have
\begin{flalign}
&\Pb{\dset{E}_t \bigcap \dset{B}} \leq \Pb{ \bigcup_{\dset{M}, \dset{F}} \dset{E}(\dset{M}, \dset{F}) \bigcap \dset{B}(\dset{M})} \nonumber\\
&\leq \binom{K_a}{t} \binom{M-K_a}{t} \Pr\sbr{ \dset{E}(\dset{M}', \dset{F}') \bigcap \dset{B}(\dset{M}')}, \nonumber
\end{flalign}
where $\dset{M}' = [t]$, $\dset{F}' = [K_a + t]\backslash[\Ka]$. The last transition is due to the union bound and symmetry.

Consider $\Pr \sbr{ \dset{E}(\dset{M}', \dset{F}') \bigcap \dset{B}(\dset{M}')}$. Recall that
\[
\dset{E}(\dset{M}', \dset{F}') = \brc{\norms{\rbz} - \norms{\rbcmp - \rbcfp + \rbz} \geq 0}
\]
and
\[
\dset{B}(\dset{M}') = \brc{\alpha\norms{\rbcmp + \rbz} + \beta n > \norms{\rbz}}.
\]

Let us introduce the vector
$$
\be^{T} = \br{\rbz^{T},\, \rbcmp^{T}/\sqrt{P't},\, \rbcfp^{T}/\sqrt{P't}}, \quad \be \sim \mathcal{N}(\bzero, \bI_{3n}).
$$
Using this notation, we can rewrite the events $\dset{E}(\dset{M}', \dset{F}')$ and $\dset{B}(\dset{M}')$ as follows:
\[
\dset{E}(\dset{M}', \dset{F}') = \brc{\be^{T}\bA_{e}\be \geq 0}
\]
and
\[
\dset{B}(\dset{M}') = \brc{\be^{T}\bA_{r}\be + \beta n > 0},
\]
where
$$
\bA_{e} = \br{\begin{array}{rrr}
 \bzero & -\sqrt{P't}\bI_{n} & \sqrt{P't}\bI_{n} \\
 -\sqrt{P't}\bI_{n} & -P't\bI_{n} & P't\bI_{n} \\
 \sqrt{P't}\bI_{n} & P't\bI_{n} & -P't\bI_{n}
\end{array}},
$$
$$
\bA_{r} = \br{\begin{array}{rrr}
 \br{\alpha - 1}\bI_{n} & \alpha\sqrt{P't}\bI_{n} & \bzero \\
 \alpha\sqrt{P't}\bI_{n} & \alpha P't\bI_{n} & \bzero \\
 \bzero & \bzero & \bzero 
\end{array}}.
$$

Now, we apply the Chernoff bound (Lemma~\ref{lemma:chernoff}):
$$
\Pr \sbr{ \dset{E}(\dset{M}', \dset{F}') \bigcap \dset{B}(\dset{M}')} \leq \inf\limits_{u,v > 0}\mathbb{E}_{\be}\exp\brs{\be^{T} \bA_{n} \be + v\beta n },
$$
where $\bA_{n} = u\bA_{e} + v\bA_{r}$.

Using Lemma~\ref{lem:quad_form_rv} (with $\dvec{b}$ equal to the zero vector), we obtain
$$
\Pr \sbr{ \dset{E}(\dset{M}', \dset{F}') \bigcap \dset{B}(\dset{M}')} \leq \inf\limits_{u,v > 0}\exp\brs{-\frac{1}{2}\log\det\br{\bI_{3n} - 2\bA_{n}} + v\beta n},
$$
when $\lambda_{\min}\br{\bI_{3n} - 2\bA_{n}} > 0$.

Let us note that the matrix $\bA_{n}$ can be represented as $\bI_{n} \otimes \bA$ by reordering its columns and rows. So, we can represent $\br{\bI_{3n} - 2\bA_{n}}$ as $\bI_{n} \otimes \br{\bI_{3} - 2\bA}$, and it is clear that
$$
-\nicefrac{1}{2}\log\det\br{\bI_{3n} - 2\bA_{n}} = -\nicefrac{n}{2}\log\det\br{\bI_{3} - 2\bA}.
$$
Thus, $\Pb{\dset{E}_t \bigcap \dset{B}} \leq p_{1,t}$.

Similarly, we deal with 
\begin{eqnarray*}
\Pb{ \mathcal{B}^c} \leq \binom{K_a}{t} \Pr \left[ \mathcal{B}^c(\dset{M}') \right] \leq p_{2,t}. 
\end{eqnarray*}
This concludes the proof.
\end{proof}

\subsection{Combination of the tricks}

One may notice that Gallager's and Fano's tricks can be applied jointly. In this section, we present one possible variant, which is a straightforward modification of the bound for \ac{sparcs} provided in~\cite{BarronSparcs2019}.

\begin{theorem}[Achievability based on the combination of Fano's and Gallager's tricks~\cite{BarronSparcs2019}]
\label{th:comb_tricks}
Fix $P' < P$. There exists a codebook $\dmat{X}^* \in \Gensp$ satisfying the power constraint $P'$ and providing $P_e$ using eq.~\eqref{eq:chapter4:pupe_definition}, where $p_{0}$ is bounded by~\eqref{eq:p0_bound}, and $\Pb{\dset{E}_t} \leq p_{t}$, where
\[
p_{t} = \inf\limits_{\alpha, \tau \geq 0}\br{p_{1,t} + p_{2,t}},
\]
\[
p_{1,t} = \inf\limits_{0 \leq \gamma \leq 1} \exp\brs{-n \br{-\gamma R_{1} - R_{2} + E_{1}\br{\gamma}}}, 
\]
\[
p_{2,t} = \inf\limits_{\delta > 0} \exp\brs{ -n E_{2}\br{\delta} },
\]
\[
R_1 = \frac{1}{n}\log \binom{M - \Ka}{t}, \quad 
R_2 = \frac{1}{n}\log \binom{\Ka}{t},
\]
\[
E_{1}\br{\gamma} = \frac{1}{2}\br{\gamma \log(1+P't) + \log\br{1 - \gamma^2 (1-\rho_e^2)} -\gamma \tau},
\]
\[
E_{2}\br{\delta} = \frac{1}{2} \br{\log\left(1 - \delta^2 (1-\rho_r^2) \right) + \tau \delta},
\]
\[
\rho_e^2 = \frac{(1 + \alpha P' t)^2}{(1 + \alpha^2 \Ka P') (1 + t P')}, \quad
\rho_r^2 = \frac{1}{1+\alpha^2 \Ka P'}.
\]
\end{theorem}

\begin{proof}
Let $\alpha, \beta > 0$. In what follows, we demonstrate how to choose these constants. We start by applying Fano's trick~\eqref{eq:fano_trick_def} for $\dset{E}_t$, where
\begin{flalign*}
\dset{B} &= \left\{ \beta {{\norms{\rvec{y} - (1-\alpha)\rbct}}} - \norms{\rvec{y} - \rbct} + \tau n \geq 0 \right\} \\
& = \left\{ \beta{{\norms{\rvec{z} + \alpha \rbct}}} - \norms{\rvec{z}} + \tau n \geq 0 \right\}, 
\end{flalign*}
and
\[
\beta = \beta' / (1 + \alpha^2 \Ka P). 
\]
Similar to the proofs of Theorems~\ref{th:gauss_codebook} and~\ref{th:polyanskiy_bound}, we can write ($0 \leq \gamma \leq 1$)
\begin{flalign}
&\Pb{\dset{E}_t \bigcap \dset{B}} \leq \Pb{ \bigcup_{\dset{M}, \dset{F}} \dset{E}(\dset{M}, \dset{F}) \bigcap \dset{B}(\dset{M})} \nonumber\\
&\leq \binom{K_a}{t} \mathbb{E}_{\rbz, \rbct} \sbr{ \left( \binom{M-K_a}{t} \Pr\sbr{ \dset{E}(\dset{M}', \dset{F}') \bigcap \dset{B} \: \middle| \: \dbct, \dbz } \right)^{\gamma}}. \nonumber\\
&= \binom{K_a}{t} \binom{M-K_a}{t}^{\gamma} \mathbb{E}_{\rbz, \rbct} \sbr{ \left(\Pr\sbr{ \dset{E}(\dset{M}', \dset{F}') \bigcap \dset{B} \: \middle| \: \dbct, \dbz } \right)^{\gamma}}. \nonumber
\end{flalign}

Recall that 
\[
\dset{E}(\dset{M}', \dset{F}') = \brc{\norms{\rbz} - \norms{\rbcmp - \rbcfp + \rbz} \geq 0}.
\]
Proceeding further,
\begin{flalign*}
&\Pr\sbr{ \dset{E}(\dset{M}', \dset{F}') \bigcap \dset{B} \: \middle| \: \dbct, \dbz } \\
&\leq \Pr\sbr{\beta {{\norms{\dvec{z} + \alpha \dbct}}} - \norms{\dbcmp - \rbcfp + \dbz} + \tau n \geq 0 \: \middle| \: \dbct, \dbz },
\end{flalign*}
where the last transition follows from the fact that
\[
\Pr[\xi_1 \geq 0, \xi_2 \geq 0] \leq \Pr[\xi_1 + \xi_2 \geq 0].
\]

Applying the Chernoff bound with $\lambda = \frac{1}{2}$\footnote{One can optimize over $\lambda$ to improve the bound, but we follow the original approach.}, we obtain
\begin{flalign*}
&\Pr\sbr{ \dset{E}(\dset{M}', \dset{F}') \bigcap \dset{B} \: \middle| \: \dbct, \dbz } \\
&\leq \exp\sbr{ \frac{1}{2} \left( \beta{{\norms{\dvec{z} + \alpha \dbct}}} - \frac{\norms{\dbcm + \dbz}}{1+P't} \right) + \frac{\tau n}{2} -\frac{n}{2} \log(1+P't) }.
\end{flalign*}

Raising the estimate to the power $\gamma$ and proceeding with the expectation over $\rbz$ and $\rbct$, we introduce the notation
\[
A = \exp\sbr{\frac{\gamma \tau n}{2} -\frac{\gamma n}{2} \log(1+P't)},
\]
and obtain
\begin{flalign*}
&\mathbb{E}_{\rbz, \rbct} \sbr{ \Pr\sbr{ \dset{E}(\dset{M}', \dset{F}') \bigcap \dset{B} \: \middle| \: \dbct, \dbz }} \\
&\leq A \mathbb{E}_{\rbz, \rbct} \sbr{\exp\sbr{ \frac{\gamma}{2} \left( \beta{{\norms{\rvec{z} + \alpha \rbct}}} - \frac{\norms{\dbcm + \dbz}}{1+P't} \right)}} \\
&\leq A \mathbb{E}_{\be_1, \be_2} \sbr{\exp\sbr{ \frac{\gamma \beta'}{2}{{\norms{\be_1}}} - \frac{\gamma}{2} \norms{\be_2} }} \\
& = A \exp\sbr{ -\frac{n}{2} \log\left(1 - \gamma \beta' + \gamma - \gamma^2 \beta' (1-\rho_e^2) \right) },
\end{flalign*}
where the last transition follows from Corollary~\ref{cor:quad_form_corr_rv_corollary}, and $\rho_e$ is defined as
\[
\cov(\be_1, \be_2) = \mathbb{E}[\be^T_1 \be_2] = \rho_e \dmat{I}_n,
\]
where
\begin{equation} \label{eq:pe_alpha}
\rho_e^2 = \frac{(1 + \alpha P' t)^2}{(1 + \alpha^2 \Ka P') (1 + t P')}.
\end{equation}

Now consider
\begin{flalign*}
\Pr[\dset{B}^c] &= \Pr\sbr{ \beta{{\norms{\rvec{z} + \alpha \rbct}}} - \norms{\rvec{z}} + \tau n < 0} \\
&\leq \mathbb{E}_{\rbz, \rbct} \sbr{ \exp\sbr{ \frac{\delta}{2} \norms{\rvec{z}} - \frac{\delta}{2}\beta{{\norms{\rvec{z} - \alpha \rbct}}} - \frac{\delta \tau}{2}n}} \\
&= \exp\sbr{- \frac{\delta\tau n}{2}} \mathbb{E}_{\be_1, \be_2} \sbr{ \exp\sbr{ \frac{\delta}{2} \norms{\rvec{\be_1}} - \frac{\delta}{2}\beta'{{\norms{\be_2}}}}} \\
&= \exp\sbr{- \frac{\delta\tau n}{2}} \exp\sbr{ -\frac{n}{2} \log\left(1 - \delta + \delta \beta' - \delta^2 \beta' (1-\rho_r^2) \right) },
\end{flalign*}
as
\[
\rho_r^2 = \frac{1}{1+\alpha^2 \Ka P'}.
\]

For simplicity, the authors of~\cite{BarronSparcs2019} chose $\beta' = 1$. We follow the same choice but note that optimizing this parameter may further improve the bound.

Thus, collecting all terms together and setting $\beta' = 1$, we obtain the theorem statement.

The final region is given by
\begin{flalign*}
\dset{B} &= \left\{ \frac{{\norms{\rvec{y} - (1-\alpha)\rbct}}}{1 + \alpha^2 \Ka P} - \norms{\rvec{y} - \rbct} + \tau n \geq 0 \right\} \\
& = \left\{ \frac{{\norms{\rvec{z} + \alpha \rbct}}}{1 + \alpha^2 \Ka P} - \norms{\rvec{z}} + \tau n \geq 0 \right\}.
\end{flalign*}
\end{proof}

\begin{remark}
The authors of~\cite{BarronSparcs2019} show that $\alpha = t/\Ka$ maximizes the expression~\eqref{eq:pe_alpha}, resulting in
\[
\rho_e^2 = \frac{1 + \frac{t^2}{\Ka}P'}{1+t P'}, \quad
\rho_r^2 = \frac{1}{1+\frac{t^{2}}{\Ka} P'},
\]
which are used in the original theorem. However, as numerical results (see Section~\ref{chap4:num_res}) show, further optimization of $\alpha$ allows for better performance, particularly in the regime where the number of active users $\Ka$ is small.
\end{remark}

Throughout Theorems~\ref{th:polyanskiy_bound}, \ref{th:gauss_codebook}, and \ref{th:comb_tricks}, we considered different approaches to estimate the achievable energy efficiency for the \ac{ura} setup. As we further demonstrate through numerical analysis in Section~\ref{chap4:num_res}, all these bounds yield similar energy efficiency estimates, with the gap being no greater than $0.25$ dB. However, the results presented in Theorem~\ref{th:gauss_codebook} can similarly be generalized to a binary codebook, as we demonstrate later in Theorem~\ref{th:binary_codebook}. 

Let us also recall the converse bounds previously defined in Theorems~\ref{the:converse_su} and~\ref{the:converse_ura}. As we further demonstrate, there is a gap between achievability and converse, with the difference being up to $1.2$ dB. We refer the reader to Section~\ref{chap4:num_res} for a detailed discussion and comparison of the presented bounds.

\section{The case of binary codebooks}\label{ch4:sec:binary_codebook}

We note that the bounds from the previous section assume the use of Gaussian signaling (or codebook), which is not ideal for practical applications. Indeed, a Gaussian codebook may be suboptimal, especially in the case of short blocks, where different codewords may have significantly different energy levels. To overcome this issue, codewords can be sampled from a \emph{power shell} or from a uniform distribution on the multidimensional sphere (see Appendix~\ref{a3:cwbook:sphere}), ensuring equal power for all codewords. Moreover, Gaussian codebooks may have limited practical applications, as storing such a codebook and generating the transmitted signal can be computationally complex.

In this section, we address this problem and consider binary signaling (or \ac{bpsk} modulation). We aim to answer a natural question: does binary signaling restrict energy efficiency? In other words, we seek to establish the fundamental limits for binary-input \ac{gmac}. This question was previously studied in \cite{Glebov2023GMAC}, and we follow the description presented in that paper.

We start with the definition of the binary codebook ensemble.
\begin{define}
Let $\Bens$ be the ensemble of binary codebooks of size $n \times M$, where each element is sampled i.i.d. from $\brc{-\sqrt{P}, \sqrt{P}}$ with probability $1/2$.
\end{define}

\begin{remark}
Note that there is no need to check the power constraint for this ensemble as it is fulfilled by design.
\end{remark}

Now we are ready to formulate and prove the main result.

\begin{theorem}[Achievability for binary codebook based on Fano's trick~\cite{Glebov2023GMAC}]
\label{th:binary_codebook}
There exists a codebook $\dmat{X}^* \in \Bens$ providing $P_e$ using eq.~\eqref{eq:chapter4:pupe_definition}, where $p_{0}=0$, and $\Pb{\dset{E}_t} \leq p_{1,t} + p_{2,t}$, where
\begin{flalign*}
&p_{1,t} = \binom{M-K_{a}}{t}\binom{K_{a}}{t} \inf\limits_{u,v > 0}\left[ \exp\brs{-n\zeta\br{\alpha, \beta, v}} \times \begin{array}{cc}
& \\
& \\
&
\end{array} \right. \nonumber \\
&\times \left. \br{\sum\limits_{m=-t}^{t}\sum\limits_{f=-t}^{t} \rho_{m}\rho_{f}\exp{\brs{P\phi\br{\alpha, u, v, m, f}}}}^{n} \right], \nonumber
\end{flalign*}
where
$$
\zeta\br{\cdot} = \frac{1}{2}\log\br{1 - 2v\br{\alpha - 1}} - \beta v,
$$
$$
\phi\br{\cdot} = \br{2\frac{\br{\alpha v m - u\br{m-f}}^2}{1 - 2v\br{\alpha - 1}} + \alpha v m^{2} - u\br{m-f}^2},
$$
and
\begin{flalign*}
& p_{2,t} = \binom{K_{a}}{t} \inf\limits_{\delta > 0}\left[ \exp\brs{-n\xi\br{\alpha, \beta, \delta}} \br{\sum\limits_{m = -t}^{t} \rho_{m}\exp\brs{P\psi\br{\alpha, \delta, m}}}^{n} \right], \nonumber
\end{flalign*}
where
$$
\xi\br{\cdot} = \frac{1}{2}\log\br{1 - 2\delta\br{1 - \alpha}} + \delta\beta, \quad 
\psi\br{\cdot} = \delta\alpha\frac{2\delta - 1}{1 - 2\delta\br{1 - \alpha}}m^{2},
$$
$$
\rho_{i} =
\begin{cases}
\binom{t}{\frac{1}{2}\br{i+t}}2^{-t}, & i \in \brc{2j-t,\,\forall j: 0 \leq j \leq t} \\
0, & \text{otherwise}
\end{cases}
$$
with the following conditions holding:
\[
1 - 2v\br{\alpha - 1} > 0, \quad  1 - 2 \delta \br{1 - \alpha} > 0.
\]
\end{theorem}

\begin{proof}
We utilize the approach from Theorem~\ref{th:gauss_codebook} with exactly the same good region~\eqref{eq:ball_prj} in Fano's trick~\eqref{eq:fano_trick_def} for $\dset{E}_t$.

We now show how to modify the proof for the case of a binary codebook. Let us start with $\Pb{\dset{E}_t \bigcap \dset{B}}$. We have
\begin{flalign}
&\Pb{\dset{E}_t \bigcap \dset{B}} \leq \binom{K_a}{t} \binom{M-K_a}{t} \Pr\sbr{ \dset{E}(\dset{M}', \dset{F}') \bigcap \dset{B}(\dset{M}')}, \nonumber
\end{flalign}
where $\dset{M}' = [t]$, $\dset{F}' = [K_a + t]\backslash[\Ka]$, and the events are defined as follows:
\begin{flalign*}
\dset{E}(\dset{M}', \dset{F}') &= \brc{\norms{\rbz} - \norms{\rbcmp - \rbcfp + \rbz} \geq 0}, \\
\dset{B}(\dset{M}')            &= \brc{\alpha\norms{\rbcmp + \rbz} + \beta n > \norms{\rbz}}.
\end{flalign*}

The major difference from the proof of Theorem~\ref{th:gauss_codebook} is that $\rbcmp$ and $\rbcfp$ are not Gaussian. Expanding the norms, we rewrite the events as
\begin{flalign*}
\dset{E}(\dset{M}', \dset{F}') &= \brc{-2\br{\rbcmp - \rbcfp}^{T}\rbz - \norms{\rbcmp - \rbcfp} > 0}, \\
\dset{B}(\dset{M}')            &= \brc{\br{\alpha - 1}\norms{\rbz} + 2\alpha \rbcmp^T\rbz + \alpha\norms{\rbcmp} + \beta n > 0}.
\end{flalign*}

Thus, we can estimate the probability $\Pr\sbr{ \dset{E}(\dset{M}', \dset{F}') \bigcap \dset{B}(\dset{M}')}$ using the Chernoff bound for the joint events in the following way:
\begin{flalign*}
\Pr\sbr{ \dset{E}(\dset{M}', \dset{F}') \bigcap \dset{B}(\dset{M}')} &\leq \inf\limits_{u,v > 0}\mathbb{E}_{\rbz,\rbcmp,\rbcfp}\exp\left[\br{\alpha - 1}v\norms{\rbz} \right. \nonumber \\
&\left. + 2\br{\alpha v \rbcmp^T - u\br{\rbcmp -\rbcfp}^T}\rbz \right. \nonumber \\ 
&\left. + \alpha v \norms{\rbcmp} - u\norms{\rbcmp - \rbcfp} + \beta v n \right].
\end{flalign*}

First, let us take the expectation over $\rbz$ while treating $\rbcmp$ and $\rbcfp$ as fixed. To do this, we apply Lemma~\ref{lem:quad_form_rv} and obtain
\begin{flalign}
\Pr\sbr{ \dset{E}(\dset{M}', \dset{F}') \bigcap \dset{B}(\dset{M}')} &\leq \inf\limits_{u,v > 0}\mathbb{E}_{\rbcmp, \rbcfp}\exp\left[-\frac{n}{2}\log\br{1 - 2\br{\alpha - 1}v} \right. \nonumber \\ 
&\left. + 2\frac{\norms{\alpha v \rbcmp - u\br{\rbcmp - \rbcfp}}}{1 - 2\br{\alpha - 1}v} \right. \nonumber \\
&\left. + \alpha v \norms{\rbcmp} - u\norms{\rbcmp - \rbcfp} + \beta v n \right]. \label{eq:bin_error_prob}
\end{flalign}

To take the expectation over $\rbcmp$ and $\rbcfp$, let us consider 
$$
\rbcmp = \br{s_{\dset{M}',1}, s_{\dset{M}',2}, \dots, s_{\dset{M}',n}}
$$
(for $\rbcfp$, the arguments are similar), which is the sum of $t$ codewords, each satisfying the power constraint $P$ for every coordinate. Then, we observe that
$$
s_{\dset{M}',l} \in \brc{\br{2j-t}\sqrt{P}, \forall j: 0 \leq j \leq t}, \quad \forall l \in \brs{n}.
$$
Changing the sign of $j$ terms modifies the sum by $2j\sqrt{P}$ in absolute value. For instance, if all $t$ terms initially have the value $-\sqrt{P}$ and we randomly select $j$ terms to change to $\sqrt{P}$, then the sum will be $ \br{2j - t}\sqrt{P}$.

The number of ways to obtain the sum $s_{\dset{M}',l} = \br{2j - t}\sqrt{P}$, $\forall j: 0 \leq j \leq t$ is given by the binomial coefficient $\binom{t}{j}$. Since each of the $t$ terms can independently take one of two values, the total number of possible combinations is $2^{t}$. Thus, the probability distribution of $s_{\dset{M}',l}$ is:
$$
\rho_{i} = \prob{s_{\dset{M}',l} = i\sqrt{P}} = \frac{1}{2^{t}}\binom{t}{\frac{i+t}{2}}
$$
for $i \in \brc{2j-t, \forall j: 0 \leq j \leq t}$ and zero otherwise. 

Now, we note that
$$
\mathbb{E}\brs{\exp\brs{\norms{\rbcmp}}} = \sum\limits_{\br{i_{1}, i_{2}, \dots, i_{n}}}\prod\limits_{l=1}^{n}\rho_{i_{l}}\exp\brs{i^{2}_{l}P} = \br{\sum\limits_{i=-t}^{t}\rho_{i}\exp\brs{i^{2}P}}^{n}.
$$

Finally, applying this approach to~\eqref{eq:bin_error_prob}, we obtain $\Pb{\dset{E}_t \bigcap \dset{B}} \leq p_{1,t}$. Similarly, we have $\Pb{\dset{B}^c} \leq p_{2,t}$.
\end{proof}

In Section~\ref{chap4:num_res}, we provide a numerical comparison of this bound and find that it yields the same energy efficiency estimates as the bound for the Gaussian codebook presented in Theorem~\ref{th:gauss_codebook}.

\section{Further comments and discussion}\label{ch4:further_comments}

\subsection{Random number of active users}\label{ch4:sec:random_ka}

All the bounds considered in this chapter assume that $\Ka$ is fixed and known to the receiver. The authors of~\cite{Durisi2022} pose a very important and timely question: How do the fundamental limits change when $\Ka$ is \emph{random} and \emph{unknown} to the receiver? For this case, they derive the \acf{rcb}, which is based on Gallager's trick. Let us discuss the major differences in comparison to~\cite{polyanskiy2017perspective}.

We start with the decoding rule. Recall the rule from~\cite{polyanskiy2017perspective}:
\[
\dset{R} = \argmin_{\dset{R}' \subseteq [M], \left| \dset{R}'\right| = K_a} \norms{\rvec{y} - \dvec{s}_{\dset{R}'}},
\]
which is reasonable since the receiver knows~$\Ka$.

The most straightforward way to generalize this rule is to consider:
\begin{flalign*}
\dset{R} &= \argmax_{\dset{R}' \subseteq [M]} \brc{\exp \sbr{-\frac{\norms{\rvec{y} - \dvec{s}_{\dset{R}'}}}{2}} \Pr[\mathcal{R}']}\\
& = \argmin_{\dset{R}' \subseteq [M]} \brc{\norms{\rvec{y} - \dvec{s}_{\dset{R}'}} - 2 \log \Pr[\mathcal{R}']},
\end{flalign*}
where $\Pr[\mathcal{R}']$ is the prior distribution of $\rscalar{K}_a$.

The rule above is difficult to analyze because it requires considering all possible $\dset{R}'$, which makes the union bound greatly overestimate the desired probability. To address this, the authors of~\cite{Durisi2022} propose a two-stage approach. First, they estimate $\rscalar{K}_a$ using the following rule:
\[
\hat{K}_a = \argmax_{K \in [K_l, K_u]} m(\dvec{y}, K), 
\]
where $m(\dvec{y}, K)$ is a suitably chosen metric, and $K_l$ and $K_u$ are appropriately selected lower and upper bounds for $\rscalar{K}_a$.

Then, given $\hat{K}_a$, the receiver outputs:
\[
\dset{R} = \argmin_{\dset{R}' \subseteq [M], \underline{K} \leq |\dset{R}'| \leq \overline{K}} \norms{\rvec{y} - \dvec{s}_{\dset{R}'}},
\]
where $\underline{K} = \max\{K_l, \hat{K}_a -r\}$, $\overline{K} = \min\{K_u, \hat{K}_a + r\}$, and $r$ is a nonnegative integer referred to as the \emph{decoding radius}.

We do not present the final theorem here but note that the techniques are very similar to those in \cite{polyanskiy2017perspective}. However, in this case, one must consider both $P_e$ and $P_f$ since the output list size is a random variable. We discuss the bound for random  $\Ka$ when presenting the numerical results in Section~\ref{chap4:num_res}.

\subsection{Gordon's lemma and achievability improvements}\label{ch4:sec:gordon_lemma}

Let us once again look at Fano's approach and the good region from Theorem~\ref{th:gauss_codebook}. Recall that we decided not to use $\rbcf$ in the region description to avoid a large binomial coefficient $\binom{M-K_a}{t}$. Are there any other ways to calculate the probability $\Pb{\dset{B}^c}$ without using the union bound and the large binomial coefficient? In this section, we describe the approach from \cite{ZPT-isit19}, which relies on Gordon's inequality \cite{gordon1988} for the expectation of the minimum of the Gaussian process. As shown in~\cite{Glebov2023GMAC}, the techniques from \cite{ZPT-isit19} can be reformulated as an instance of Fano's method.

As an example, let us consider the following simplified region:
\[
\dset{B} = \{ \norm{\rbcf} \geq \sqrt{\beta n} \}.
\]

Recall that
\[
\Pb{\dset{E}_t} \leq \Pb{\dset{E}_t \bigcap \dset{B}} + \Pb{\dset{B}^c},
\]
and for the first term ($\Pb{\dset{E}_t \bigcap \dset{B}}$), we can proceed exactly as in the proof of Theorem~\ref{th:gauss_codebook}. We omit this and only consider $\Pb{\dset{B}^c}$.

We have
\begin{eqnarray*}
\Pb{\dset{B}^c} &=& \Pr\brs{\bigcup_{\dset{F} \subset [M]\backslash[K_a]} \{ \norm{\rbcf} < \sqrt{\beta n} \}} \\
&=& \Pr\brs{\min_{\dset{F} \subset [M]\backslash[K_a]} \norm{\rbcf} < \sqrt{\beta n}}.
\end{eqnarray*}

The latter term can be upper-bounded using concentration properties and Gordon's inequality. Let us mention the main ingredients.

\begin{definition}[Gaussian width]
Given a closed set $\dset{S} \subseteq \mathbb{R}^{d}$, its Gaussian width $w(\dset{S})$ is defined as
\[
w(\dset{S}) = \mathbb{E}\brs{\max_{\dvec{x} \in \dset{S}} \{\rvec{g}^T \dvec{x}\}},
\]
where $\rvec{g} \sim \mathcal{N}(0, \dmat{I}_d)$.
\end{definition}

\begin{theorem}[Gordon's theorem, \cite{gordon1988}]
Let $\rmat{G} \in \mathbb{R}^{k \times d}$ be a random matrix with i.i.d. entries $\rscalar{g}_{i,j} \sim \mathcal{N}(0,1), i \in [k], j \in [d]$, and $\dset{S} \subseteq \mathbb{S}^{d-1}$ be a closed subset of the unit sphere in $d$ dimensions. Then
\begin{flalign*}
\mathbb{E}\brs{\max_{\dvec{x} \in \dset{S}} \norm{\rmat{G}\dvec{x}}} &\leq a_k + w(\dset{S}), \\
\mathbb{E}\brs{\min_{\dvec{x} \in \dset{S}} \norm{\rmat{G}\dvec{x}}} &\geq a_k - w(\dset{S}),
\end{flalign*}
where $a_k = \mathbb{E}[\norm{\rvec{g}}]$, $\rvec{g} \sim \mathcal{N}(0, \dmat{I}_k)$, and $w(\dset{S})$ is the Gaussian width of the set $\dset{S}$. Recall that $a_k = \sqrt{2} \Gamma((k+1)/2)/\Gamma(k/2)$ satisfies
\[
\frac{k}{\sqrt{k+1}} \leq a_k \leq \sqrt{k}.
\]
\end{theorem}

\begin{lemma}[Gaussian concentration, see e.g. {\cite[Corollary 3.3]{Chandrasekaran}}]
Let $f: \mathbb{R}^d \rightarrow \mathbb{R}$ be an $\ell$-Lipschitz function and $\rvec{g} \sim \mathcal{N}(0, \dmat{I}_d)$. Then,
\begin{flalign*}
\Pr\brs{f(\rvec{g}) \geq \mathbb{E}_{\rvec{g}}[f] + t} &\geq 1 - e^{-\frac{t^2}{2\ell^2}}, \\
\Pr\brs{f(\rvec{g}) \leq \mathbb{E}_{\rvec{g}}[f] - t} &\leq e^{-\frac{t^2}{2\ell^2}}.
\end{flalign*}
\end{lemma}

\begin{lemma}[Size of the maximum]
Let $\rvec{g} = [\rscalar{g}_1, \ldots, \rscalar{g}_d] \sim \mathcal{N}(0, \dmat{\Sigma})$ and $\mathbb{E}[\norms{\rscalar{g}_i}] = 1$ for $i \in [d]$. Irrespective of the covariance structure, we have the following bound:
\[
\mathbb{E}\brs{\max_{i \in [d]}\rscalar{g}_i} \leq \sqrt{2 \log d}.
\]
\end{lemma}

\begin{proof}
\begin{flalign*}
\mathbb{E}\brs{\max_{i \in [d]}\rscalar{g}_i} &= \frac{1}{\beta} \mathbb{E}\brs{\log \exp[\beta \max_i \rscalar{g}_i]} \\
& \leq \frac{1}{\beta} \mathbb{E}\brs{\log \sum_{i \in [d]} \exp[\beta \rscalar{g}_i]} \\
& \leq \frac{1}{\beta} \log \sum_{i \in [d]} \mathbb{E}\brs{\exp[\beta \rscalar{g}_i]} = \frac{\beta}{2} + \frac{\log n}{\beta}.
\end{flalign*}

Taking $\beta = \log n$ yields the lemma statement.
\end{proof}

Let us now return to our problem. Recall that we consider the case of a Gaussian codebook $\rmat{X}$. Let $\rmat{G} = \rmat{X}_{[M]\backslash[\Ka]}$. The matrix $\rmat{G}$ is of size $n \times (M-K_a)$, with $g_{i,j} \overset{\text{i.i.d.}}{\sim} \mathcal{N}(0,P')$.
\begin{flalign*}
&\Pr\brs{\min_{\dset{F} \subset [M]\backslash[K_a]} \norm{\rbcf} < \sqrt{\beta n}} \\
&= \Pr\brs{\min_{\dvec{x} \in \{0,1\}^{M-K_a}, \norm{x}_0 = t} \norm{\mathbf{G} \dvec{x}} < \sqrt{\beta n}} \\
&\leq \exp\left( -\frac{c^2}{2 \ell^2}\right),
\end{flalign*}
where $(x)_+ = \max\{x,0\}$,
\[
c = \left( \mu - \sqrt{\beta n} \right)_{+},
\]
and
$$
\mu = \mathbb{E}\brs{\min_{\dvec{x} \in \{0,1\}^{M-K_a}, \norm{\dvec{x}}_0 = t} \norm{\rmat{G} \dvec{x}}}  \geq \frac{n\sqrt{P't}}{\sqrt{n+1}} - \sqrt{2P't \ln \binom{M-K_a}{t}}.
$$

Note that the function $F(\dmat{G}) = \min_{\dvec{x} \in \{0,1\}^{M-K_a}, \norm{x}_0 = t} \norm{\dmat{G} \dvec{x}}$ is $\ell$-Lipschitz with $\ell = \sqrt{t}$. Indeed,
\begin{flalign*}
\left| F(\rmat{G}_1) - F(\rmat{G}_2) \right| &\leq \max_{\dvec{x} \in \{0,1\}^{M-K_a}, \norm{\dvec{x}}_0 = t} \left| \norm{\rmat{G}_1 \dvec{x}} - \norm{\rmat{G}_2 \dvec{x}} \right| \\
& \leq \max_{\dvec{x} \in \{0,1\}^{M-K_a}, \norm{\dvec{x}}_0 = t} \brs{\norm{(\rmat{G}_1 - \rmat{G}_2) \dvec{x}}} \\
& \leq \sqrt{t}\norm{(\rmat{G}_1 - \rmat{G}_2)}_{2,2} \leq \sqrt{t} \norm{(\rmat{G}_1 - \rmat{G}_2)}_F,
\end{flalign*}
where
\[
\norm{\dmat{A}}_{\alpha, \beta}
\triangleq \sup\limits_{\dvec{x}: \: \norm{\dvec{x}}_\alpha \leq 1}{\norm{\dmat{A}\dvec{x}}_{\beta}}.
\]

We use the result of this lemma to analyze the asymptotic regime presented in Appendix~\ref{app:asymptotics}. This lemma proves to be highly efficient, and the resulting achievability bound coincides with the converse bound. However, we do not evaluate this lemma in the \ac{fbl} regime, as it yields energy efficiency estimates several decibels higher than those presented below.

\section{Numerical evaluation}\label{chap4:num_res}

\begin{figure}[t]
\centering
\includegraphics{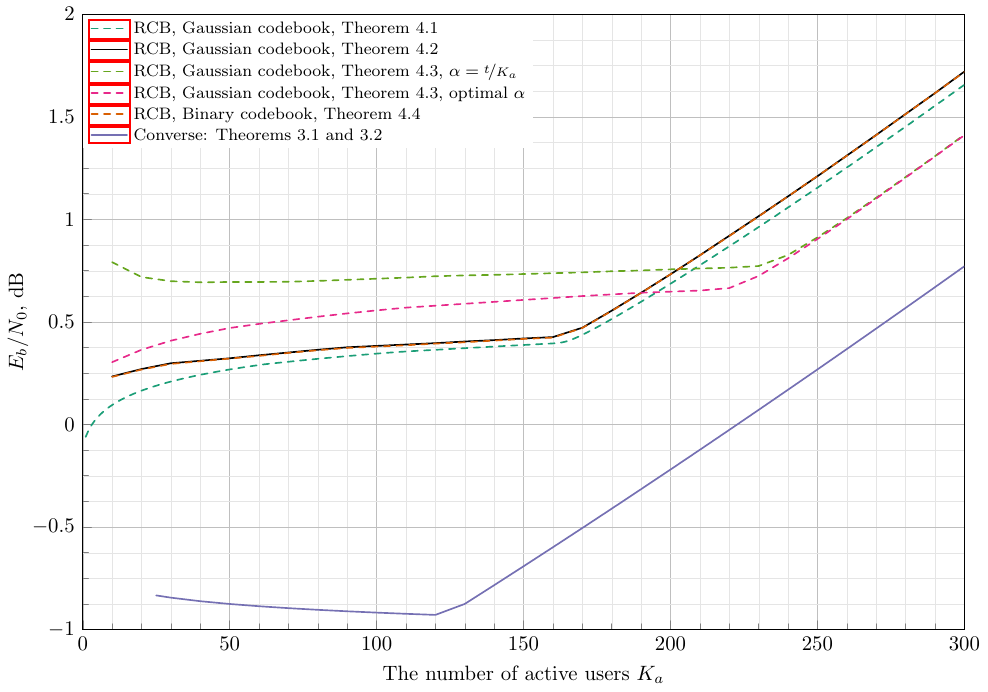}
\caption{\ac{awgn} \ac{mac} achievability bounds are presented for $k=100$ information bits, a frame length of $n=30000$, and an error probability constraint of $P_e<0.05$. The converse bound is depicted as a single line, representing the maximum of the results from Theorems~\ref{the:converse_su} and~\ref{the:converse_ura}.}
\label{fig:ebno_ka}
\end{figure}

\begin{figure}[!ht]
\centering
\includegraphics{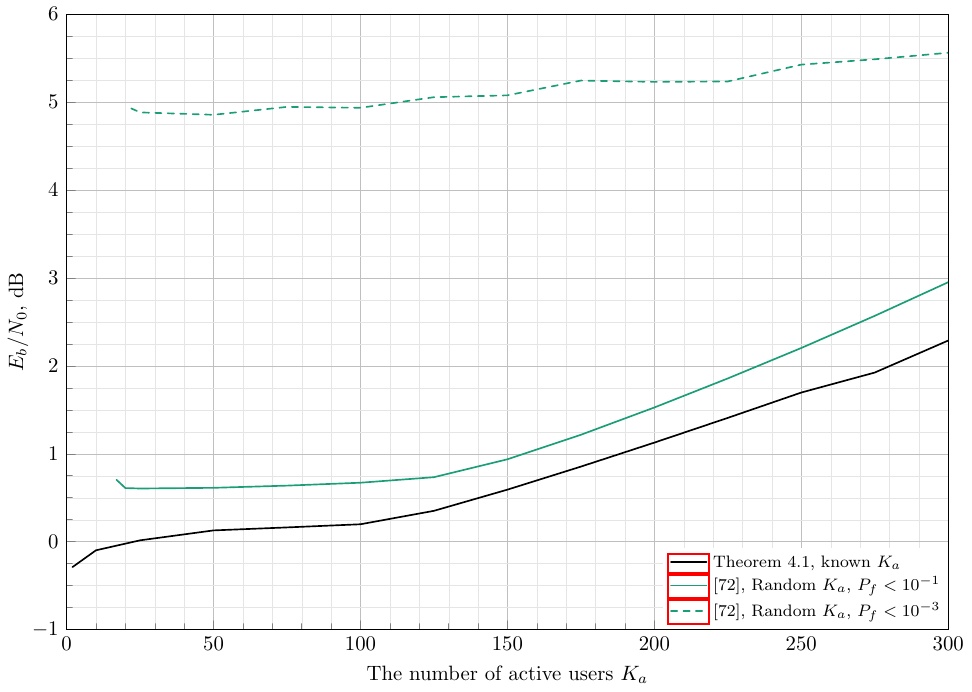}
\caption{\ac{awgn} \ac{mac} achievability bounds are presented for $k=128$ information bits, a frame length of $n=19200$, an error probability constraint of $P_{e}<0.1$, and a random number of active users. Results for different values of $P_f$ are shown. A stricter \ac{far} requirement leads to a higher required $E_b/N_0$.}
\label{fig:ebno_ka_awgn_random}
\end{figure}

Let us consider a setup with a frame length of $n=3\times 10^4$ real channel uses and $k=100$ information bits transmitted by each active user. This setup was proposed in~\cite{polyanskiy2017perspective} and~\cite{ordentlich2017low}. 

We begin with a known number of active users. We select $P_e < 0.05$ and analyze the minimum energy per bit ($E_b/N_0$) as a function of the number of active users, $K_a$. The \ac{rcb} (Theorem~\ref{th:polyanskiy_bound}, \cite{polyanskiy2017perspective}) is represented in~\Fig{fig:ebno_ka} by a dashed green line. A surprising observation is that the energy per bit remains nearly constant---almost as if only a single user were transmitting---until reaching a critical value of $\Ka$ (for the parameters from Theorem~\ref{th:polyanskiy_bound}, this value is approximately $\Ka \approx 150$).

Thus, the performance is \emph{noise-limited}. The key takeaway is that additional users can be accommodated without increasing energy requirements until a critical threshold is reached. The primary reason for this effect is that the dominant probability, $\Pb{\dset{E}_t}$, in~\eqref{eq:chapter4:pupe_definition} corresponds to $t = 1$ when $K_a$ is small. However, as $K_a$ increases, the last term, $t = K_a$, becomes dominant due to the large combinatorial factor $\binom{M-K_a}{t}$, which results in the \emph{interference-limited} regime.

Next, let us consider the \ac{rcb} based on Fano's trick (Theorem~\ref{th:gauss_codebook}). The resulting energy efficiency is presented by a solid black line. As can be seen in Fig.~\ref{fig:ebno_ka}, the results obtained for the Gaussian (Theorem~\ref{th:gauss_codebook}) and binary (Theorem~\ref{th:binary_codebook}) codebooks are almost identical and very close to the original achievability bound from Theorem~\ref{th:polyanskiy_bound} (e.g., for 250 active users, the values are 1.210 dB, 1.211 dB, and 1.154 dB, respectively).

The results for the \ac{rcb}, based on a combination of tricks (Theorem~\ref{th:comb_tricks}), are shown by dashed green and magenta lines for $\alpha = t/\Ka$ and the optimized $\alpha$, respectively. We highlight both curves to improve the original bound from Theorem~\ref{th:polyanskiy_bound} for large values of $\Ka$.

Next, let us consider a random number of active users~\cite{Durisi2022}. For this bound, we evaluate the previously defined energy efficiency as a function of the average number of active users. We note that the number of active users follows a Poisson distribution. 

The main challenge here is the presence of two probabilities: $P_e$~\eqref{eq:ch3:p_missed} and $P_f$~\eqref{eq:ch3:p_false}. The results are presented for $n=19200$, $k=100$, and $P_e < 0.1$, as shown in~\Fig{fig:ebno_ka_awgn_random}. When the condition $P_f < 0.1$ is imposed, the bound for a random $\Ka$ (depicted by a solid green line) demonstrates energy efficiency very close to that of the bound for a known $\Ka$~\cite{polyanskiy2017perspective} (depicted by a black line). Surprisingly, imposing stricter limits on $P_f$ (e.g., $10^{-3}$ instead of $10^{-1}$) results in an additional gap of approximately 2 dB.

The achievability bound for $P_f < 10^{-3}$ is depicted by a dashed green line. Furthermore, the transition between noise-limited and interference-limited regimes (caused by the behavior of different terms in~\eqref{eq:chapter4:pupe_definition}) becomes indistinguishable.

\begin{remark}
A stricter \ac{far} requirement leads to a higher required $E_b/N_0$. In the extreme case where $P_f \to 1$, we obtain $E_b/N_0 \to 0$ (or  $-\infty$ dB). Indeed, in this scenario, the receiver can take the entire codebook as the list of decoded messages, leading to $P_e = 0$ and $P_f = 1$.
\end{remark}

\chapter{\Acl{gura}: low-complexity schemes} \label{chap5}

In this chapter, we address the problem of constructing low-complexity coding schemes for the \acf{gura}. The main challenges we aim to overcome are as follows:

\begin{itemize}
\item Ultra-low rate: if we look at Section~\ref{sec:params}, we observe that the coding rate is $\approx 1/300$ bits per channel use for typical values of parameters of interest.
\item Large \emph{collision order}\footnote{In the \ac{gura} setup, a \emph{collision} is understood as multiple simultaneous transmissions in the frame or slot, depending on the context, and the \emph{order of collision} equals the number of simultaneous transmissions.} in a frame: $\Ka$ can be as high as $500$.
\item Same codebook.
\end{itemize}

The existing \ac{mac} coding schemes are efficient in two opposite situations:
\begin{enumerate}
\item Large dimension, or equivalently, codebook size, and a small collision order, typically not exceeding $2$--$4$.
\item Large collision order and a small codebook size, typically on the order of $2^{15}$.
\end{enumerate}
The latter case corresponds to the classical \ac{cs} problem with standard decoding algorithms such as LASSO, \ac{nnls}, or \ac{omp}; see the description in Appendix~\ref{a:cs_methods}. Thus, the task of \ac{ura} is to support a huge number of extremely low-bit-rate user streams, which appears to be much more complex than supporting a small number of users, each with a high bit rate. Finally, the same encoder assumption introduces additional challenges. Indeed, all low-complexity schemes described in Chapter~\ref{chap2} utilize single-user decoders (explicitly or implicitly). The same codebook and equal powers eliminate user diversity and thus do not allow relying on single-user decoders.

As an example, we consider the case of a two-user \ac{mac}, where users utilize \ac{ldpc} codes. When \ac{ldpc} codes are different\footnote{Different codes can be created from one \ac{ldpc} code by means of different interleavers, or one can consider two different cosets of the same \ac{ldpc} code.}, the performance of the system is good and even close to the point on the dominant facet when $n \to \infty$~\cite{LDPC_2user}. At the same time, when the same code is used, the system performs extremely poorly. We return to this case at the end of the chapter to provide some recent results and formulate open problems.

Before we proceed with the detailed description, we formulate two main ideas behind all existing low-complexity schemes for the \ac{ura} channel: (a) sparsifying the collisions and (b) splitting the high-dimensional problem into a number of low-dimensional ones. If we return to the \ac{cs} interpretation (see eq.~\eqref{eq:ura_as_cs}), the task can be formulated as follows: find the sensing matrix $\dmat{X}$ such that low-complexity algorithms do exist. See the examples of sensing matrices of \ac{ura} low-complexity schemes in Appendix~\ref{a:cs_matrices}.

The question of low-complexity \ac{ura} schemes for the \ac{gmac} is widely addressed in the literature. Below, we present the main techniques:
\begin{itemize}
\item Schemes based on the $T$-fold coded \ac{sa} protocol or ALOHA with multi-packet reception (Section~\ref{ch5:sec:coded_sa}).
\item Schemes based on \ac{idma} techniques~\cite{IDMA2006} were proposed in~\cite{AmalladinneJoint2019} (Section~\ref{sec:ch5:idma}).
\item \Ac{ccs} approach~\cite{codedCS2020} and modifications that include $t$-error correcting list-recoverable codes~\cite{Frolov2022TCOM}, \ac{sparcs} combined with the \ac{amp} algorithm~\cite{fengler2020sparcs}, iterative approaches~\cite{AmalladinneAMP2021}, where the inner \ac{amp} decoder and the outer tree decoder are allowed to exchange soft information within a joint message-passing algorithm (Section~\ref{chap5:ccs_description}).
\item Random spreading and correlation-based energy detector with polar codes and equal power levels~\cite{AmalladinnePolar2020}, polar codes and different power levels~\cite{Duman2021}, \ac{ldpc} codes and \ac{sic} within the decoding process~\cite{PradhanLDPCSC2021} (Section~\ref{sec:ch5:randon_spreading}).
\end{itemize}

The presentation of low-complexity schemes is followed by a discussion on the applicability of the same linear codes to the \ac{gura} problem in Section~\ref{sec:same_lin}. Finally, a numerical comparison of the presented schemes is provided in Section~\ref{ch5:sec_numerical}. The main results are summarized in~\Fig{fig:ch5:mac_numerical}.

\section{\texorpdfstring{$T$}{T}-fold coded \acl{sa}}\label{ch5:sec:coded_sa}

In this section, we focus on the \ac{irsa} protocol presented in Section~\ref{ch2:coded_aloha} and analyze the achievable energy efficiency under \ac{pupe} constraints~\eqref{eq:ch3_energy_efficiency}, instead of throughput. This protocol assumes that users send multiple replicas of the same packet (repetition), forming a subclass of coded \ac{sa} protocols. A key aspect of the \ac{ura} problem is the large number of devices, which leads to a high collision probability, making collision resolution necessary.

$T$-fold \ac{irsa} (or \ac{irsa} with multi-packet reception) enables the resolution of collisions up to order $T$. Recall that in the classical collision model described in Section~\ref{sec:ch2_saloha}, messages from a collided slot cannot be extracted. However, in an \ac{awgn} channel, the successful reception of multiple transmissions---up to $T$---may be possible with a suitable slot decoding algorithm. The details of this algorithm will be outlined throughout this section, depending on the particular scenario. The value of $T$ is determined by computational complexity limitations at the receiver.

We begin with the \ac{irsa} protocol description in Section~\ref{ch5:sec_irsa_basic}, followed by the \ac{de} analysis in Section~\ref{sec:ch5_irsa_de}. Next, we present a $T$-fold \ac{irsa} achievability bound formulated in Theorem~\ref{th:irsa_t_bound}. Some \ac{irsa} modifications adopted for practical schemes are introduced in Section~\ref{ch5:sec_irsa_sf}, followed by an analysis of practical collision-resolution methods. These methods, based on \ac{ldpc} codes~\cite{vem2019user} and polar codes~\cite{Marshakov2019Polar}, as well as the compute-and-forward strategy~\cite{Nazer2011CoF}, which was applied to the \ac{ura} setup in~\cite{ordentlich2017low}, are presented in Section~\ref{sec:ch5:irsa_practical_awgn}.

\subsection{Basic \texorpdfstring{$T$}{T}-fold \ac{irsa} scheme}\label{ch5:sec_irsa_basic}

In this section, we describe the basic $T$-fold \ac{irsa} protocol used in~\cite{vem2019user}. The main difference from previous works~\cite{liva2011graph, MPR_schlegel, MPR_stefanovich} is that we do not add pointers to the locations of the replicas. Instead, we determine the slot count and slot indices based on the message to be transmitted.

Assume the frame is split into $L$ slots of size $n' = n/L$ channel uses. We now describe the main features of the transmission process:
\begin{itemize}
\item The user chooses a message $W\sim \mathrm{Unif}\br{[M]}$, encodes it, and obtains a codeword $f(W) \in \mathbb{R}^{n'}$.
\item Users repeat their codewords in multiple slots. To choose the slots, all users share \emph{the same} slot-selection function $g: [M] \to \left\{0,1\right\}^L$, i.e., the codeword $f(W)$ is transmitted in the slots with indices from $\mathrm{supp}\br{g\br{W}}$.
\item The function $g$ also determines the number of replicas $D: [M] \to [L]$, where $D(W) = \mathrm{wt}\br{g\br{W}}$.
\item The uniform message distribution induces the replica count distribution, i.e., the distribution of the random variable $D(W)$, which is the same for all users.
\end{itemize}

\begin{remark}
We note that the functions $g(\cdot)$ and $D(\cdot)$ are deterministic functions of the message. This fact is crucial because the receiver does not know the positions of multiple replicas until the message is decoded in at least one slot. Once the message is decoded, the receiver can determine from which slots the replicas of the decoded message should be removed and then perform the \ac{sic} step.
\end{remark}

\begin{remark}\label{remark:slot_choice_uniform}
In our analysis (see Section~\ref{sec:ch5_irsa_de}), we assume that the slot-selection function $g\br{\cdot}$ is designed such that, for a given number of replicas $d = D\br{W}$, the slot indices are selected uniformly from the $\binom{L}{d}$ possible choices. In practical schemes, we construct $g\br{\cdot}$ to approximate this assumption as closely as possible. One possible approach is to use the message $W$ as a seed for a pseudo-random number generator. Once the message is recovered, the receiver sets the same seed and determines the replica locations.
\end{remark}

\paragraph{Bipartite graph representation.}

\begin{figure}[t]
\centering
\includegraphics{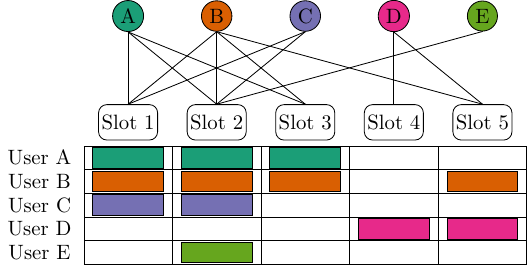}
\caption{Example of a bipartite graph (top) corresponding to the case of $5$ active users transmitting in $5$ slots (bottom). User-nodes are labeled A--E and have different colors. Slot-nodes are marked by numbers. The messages transmitted in different slots are represented by colored rectangles, with each color matching a different user.\label{fig:chap5:aloha_tx}}
\end{figure}

The transmission and decoding processes can be described using a bipartite graph (see~\Fig{fig:chap5:aloha_tx}), known as a Tanner graph~\cite{tanner1981recursive}. The vertex set of this graph consists of the set of user-nodes $\dset{U}= \left\{u_1, u_2, \ldots, u_{\Ka}\right\}$, which represents the users, and the set of slot-nodes $\dset{S} = \left\{s_1, s_2, \ldots, s_L\right\}$, which corresponds to the signals received in the slots. A user-node $u_i$ and a slot-node $s_j$ are connected by an edge if and only if the $i$-th user transmitted a packet in the $j$-th slot.

Following~\cite{RichardsonUrbanke2008}, we represent user-node degree distributions from both the node and edge perspectives using polynomials (or generating functions). Specifically, we introduce the polynomials:
\begin{equation}\label{eq:ch5:irsa_poly}
\Lambda\br{x} = \sum\limits_{i=1}^L\Lambda_i x^i, \quad \text{and} \quad \lambda\br{x} = \frac{\Lambda'\br{x}}{\Lambda'\br{1}}=\sum\limits_{i=1}^L \lambda_i x^{i-1},
\end{equation}
where $\Lambda_i$ denotes the fraction of user-nodes with degree $i$, and $\lambda_i$ denotes the fraction of edges incident to user-nodes of degree $i$. Here, $\Lambda'\br{x}$ represents the derivative with respect to the formal variable $x$. Note that
$$
\Lambda'\br{1} = \sum\limits_{i=1}^L i \Lambda_i
$$
is the average degree of user-nodes.

\begin{example}
Consider the graph in~\Fig{fig:chap5:aloha_tx}. We have:
\[
\Lambda\br{x} = \frac{1}{5}x + \frac{2}{5} x^2 + \frac{1}{5} x^3 + \frac{1}{5} x^4,
\]
and
\[
\lambda\br{x} = \frac{1}{\frac{1}{5} + \frac{4}{5} + \frac{3}{5} + \frac{4}{5}} \left( \frac{1}{5} + \frac{4}{5} x + \frac{3}{5} x^2 + \frac{4}{5} x^3 \right) =  \frac{1}{12} + \frac{1}{3} x + \frac{1}{4} x^2 + \frac{1}{3} x^3.
\]
\end{example}

Similar to~\eqref{eq:ch5:irsa_poly}, we define the slot-node degree distributions from the node perspective $R\br{x}$ and the edge perspective $\rho\br{x}$ as follows:
$$
R\br{x} = \sum\limits_{i=0}^{\Ka}R_i x^i, \quad \text{and} \quad \rho\br{x} = \sum\limits_{i=1}^{\Ka} \rho_i x^{i-1}.
$$

Next, we compute the probability that a user selects any given slot (denoted by event $\mathcal{E}$). Since the process does not depend on specific user or slot indices, we omit them for simplicity. Let the user transmit $d$ replicas. Following Remark~\ref{remark:slot_choice_uniform} (uniform slot selection), we obtain:
\[
\Pr[\mathcal{E} | D = d] = \frac{\binom{L-1}{d-1}}{\binom{L}{d}} = \frac{d}{L}, 
\]
which implies
\[
\Pr[\mathcal{E}] = \mathbb{E}_D \left[ \Pr[\mathcal{E} | D] \right] = \frac{\Lambda'\br{1}}{L} = \frac{G \Lambda'\br{1}}{\Ka},
\]
where $G = \Ka/L$ and $G \Lambda'\br{1}$ represents the average degree of slot-nodes.

Thus, the slot-node distribution (from the node perspective) follows a binomial distribution:
\[
\text{Bino}\left( \Ka, \frac{G \Lambda'\br{1}}{\Ka} \right).
\]
That is, defining $\gamma = G \Lambda'\br{1}$, we obtain:
\begin{flalign}
R\br{x} &= \sum_{i=0}^{\Ka} \binom{\Ka}{i} \left(\frac{\gamma}{\Ka}\right)^i \left(1-\frac{\gamma}{\Ka}\right)^{\Ka-i} x^i \nonumber\\
&= \left( 1 - \frac{\gamma}{\Ka} \right)^{\Ka} \left( 1 + \frac{\frac{\gamma}{\Ka} x}{1 - \frac{\gamma}{\Ka}}\right)^{\Ka}. \label{eq:bino_distr}
\end{flalign}

\paragraph{Decoding.}
Now, let us describe a general approach to the decoding algorithm, which is based on two main procedures:
\begin{enumerate}
\item The slot decoding algorithm, which outputs the list of decoded\footnote{Note that this is an estimate of the transmitted list, so some messages can be missing and some false messages may be added by the decoder. The issue of false messages should be addressed by the decoder to avoid error propagation within the SIC process. One possible solution is to add control information (e.g. add the \ac{crc}) to the packet.} messages.
\item The \ac{sic} step, which relies on the function $g\br{\cdot}$ and subtracts replicas (if present) of successfully decoded messages.
\end{enumerate}

Clearly, after a single \ac{sic} step, the collision order will decrease in some slots. Hence, these steps may be iteratively repeated in different combinations with respect to different slots, depending on the particular decoding algorithm. We refer the reader to Section~\ref{sec:ch5:irsa_practical_awgn} for specific practical slot decoding algorithms.

The example below considers a slot decoding sequence for the messages presented in~\Fig{fig:chap5:aloha_tx}.

\begin{figure}[t]
\centering
\includegraphics{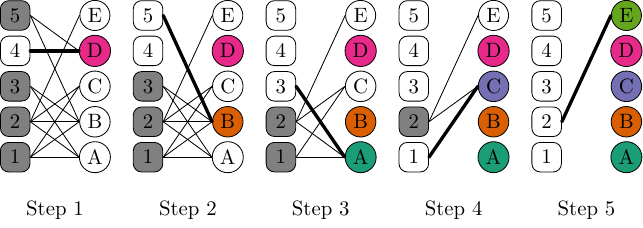}
\caption{Sequence of slot decoding followed by \ac{sic} steps applied to the case presented in~\Fig{fig:chap5:aloha_tx} for $T=1$. The bold line corresponds to the edge incident to the slot that can be decoded at the corresponding step. At the subsequent \ac{sic} step, all edges incident to the same user-node as the bold one will be removed. Resolved slots are marked with a white background, unresolved slots are marked with a gray background. Resolved users are marked with a color that matches~\Fig{fig:chap5:aloha_tx}.\label{fig:chap5:aloha_rx}}
\end{figure}

\begin{example}
Consider the example with $\Ka = 5$ active users and $L = 5$ slots, as presented in~\Fig{fig:chap5:aloha_tx}. Assume the first user (A) sends a codeword in slots 1, 2, and 3; the second user (B) in slots 1, 2, 3, and 5; the third user (C) in slots 1 and 2; the fourth user (D) in slots 4 and 5; and the fifth user (E) in slot 2.

As a result, there are 3 simultaneous transmissions in slot 1, 4 transmissions in slot 2, collisions of order 2 in slots 3 and 5, and the only collision-free slot is slot 4.

Consider a decoder with $T=1$ (see~\Fig{fig:chap5:aloha_rx}). The decoder searches for a collision-free slot, then decodes the message transmitted in this slot and subtracts this message from all slots where it was replicated. There is a natural question: how does the decoder select the collision-free slot? The strategies may be different, but one of them is to try all slots in parallel with a decoder designed for $T=1$. Another reasonable strategy is to select the slot with the minimal received energy.

The \ac{sic} step begins after the successful decoding of slot 4 at step 1 (with all remaining slots marked gray, indicating slots with collisions) and decodes the message of user D. The corresponding edge in the bipartite graph is marked with a bold line. After the \ac{sic} step, all edges incident to the successfully decoded user (D, in this case) are removed, as all transmitted replicas are subtracted.

At the second decoding step, slot 5 becomes collision-free, resulting in the successful decoding of user B. The third step resolves slot 3, allowing user A to be decoded. After the \ac{sic} step, slot 1 becomes collision-free. In the final step, user E, who transmitted in slot 2, can be decoded.

For a decoder with $T=2$, the decoding algorithm may behave differently. For instance, slots 3, 4, and 5 have a collision order not higher than two. Hence, these slots may be decoded first in any order. After the corresponding \ac{sic} steps, the collision order of slot 1 will be reduced to one (collision-free transmission), and slot 2 will have a collision order of two. Moreover, decoding slot 2 will be sufficient to extract all remaining messages in the frame.
\end{example}

\subsection{Achievability bound for \texorpdfstring{$T$}{T}-fold \ac{irsa} scheme}

\subsubsection{Slot decoding} \label{sec:ch5_irsa_de_slot}

Let us first note that, in order to obtain an achievability bound for the whole scheme, we do not restrict the complexity of slot decoding and instead use the random coding bounds described in Chapter~\ref{chap4}. 

Consider a particular slot---w.l.o.g., let it be the first slot. Let $\mathbf{y}_1$ be the received signal of $n'$ channel uses. We assume that $r$ users transmit in the first slot. Recall that $k$ denotes the number of information bits sent by each user and that $P$ is the transmit power. The random coding bound from Theorem~\ref{th:polyanskiy_bound} states that, in the random Gaussian ensemble $\Gensr$ (see Definition~\ref{def:gauss_codebook}), there exists a codebook $\dmat{X}^*$ such that
\begin{flalign*}
&\frac{1}{r} \sum\limits_{i=1}^{r} \Pr\left(W_i \not\in \dset{R}\br{\mathbf{{y}}_1} \: \middle| \: \dmat{X}^*\right) = \Pr\left(W_1 \not\in \dset{R}\left(\mathbf{{y}}_1\right) \: \middle| \: \dmat{X}^*\right)\\
&\leq p\br{n', k, P, r} \triangleq \mathbb{E}_{\dmat{X}}\brs{\Pr\left(W_1 \not\in \dset{R}\br{\mathbf{{y}}_1} \: \middle| \: \dmat{X}, r\right)}.
\end{flalign*}

The first equality holds due to the symmetry of users (it is sufficient to consider the probability that a particular user's message is not in the decoded list). The last expectation is taken over the Gaussian ensemble. 

Here, we emphasize \emph{the main problem}---the bound assumes that the number of active users ($r$) is known. However, due to the randomness of $T$-fold \ac{irsa}, the number of users transmitting in a particular slot is a random variable. 

Another issue is that the codebook $\dmat{X}^*$ is constructed for a specific number of users, but we need a codebook that can resolve collisions for any order $r \in [T]$.  

To address these issues, the authors of~\cite{Glebov2019} propose shifting to blind decoding. Let
$$
\dset{T} = \{W_1, \ldots, W_r\} \subset [M], \quad \text{where} \quad |\dset{T}| = r,
$$
denote the set of transmitted messages. The blind \ac{ml} decoding rule is given by:
\begin{equation}
\label{eq:ura_awgn_blind_decoding_rule}
\dset{R} = \argmin_{\dset{R}' \subseteq [M], \left| \dset{R}'\right| \leq T} \norms{\rvec{y}_1 - \dvec{s}_{\dset{R}'}}.
\end{equation}
Recall that $T$ is the maximum collision order that can be resolved.

\begin{theorem}\label{theor_av}
Fix $P'$ and $T$. Then, the average (over the random Gaussian ensemble) per-user error probabilities can be calculated as follows for $r \in [T]$:
\[
\mathbb{E}_{\dmat{X}} \brs{ \Pr\left(W_1 \not\in \dset{R}\br{\mathbf{{y}}_1} \: \middle| \: \dmat{X}, r\right) } \leq \sum\limits_{t=1}^{r} \frac{t}{r} p_t\br{r, T},
\]
where
\begin{equation}\label{eq:pt_irsa_bound}
p_t(r, T)  = \sum\limits_{\hat{t} = 0}^{T-r+t}\exp\brs{-n E\br{t, \hat{t}}},
\end{equation}
\[
E(t, \hat{t}) = \max\limits_{0\leq \rho, \rho_1 \leq 1, \lambda > 0}\brs{ -\rho_{1}\rho_{2}R_{1} - \rho_{1}R_{2} + E_{0}\br{\rho_1, \rho_2}},
\]
\[
R_1 = \frac{1}{n} \log \binom{M-r}{\hat{t}}, \quad R_2 = \frac{1}{n} \log \binom{r}{t},
\]
\[
E_{0}\br{\rho_1, \rho_2} =  \frac{1}{2}\br{\rho_1 a + \log(1-2b\rho_1)},
\]
\[
a = \rho_{2} \log \br{1+2P'\hat{t}\lambda} + \log \br{1+2P't\mu},
\]
\[
b = \rho_{2} \lambda - \frac{\mu}{1+2P't\mu}, \quad  \mu = \frac{\rho_{2} \lambda}{1+2P'\hat{t}\lambda}.
\]
\end{theorem}
\begin{proof}
The proof of this theorem follows the same structure as that of Theorem~\ref{th:polyanskiy_bound}, with the main difference being in the summation limits of~\eqref{eq:pt_irsa_bound}.
\end{proof}

We now need to choose a codebook that is effective for all collision orders up to $T$. Let us formally state this requirement.

\begin{statement}\label{stat:code}
Let us choose positive values $\alpha_i$, for $i \in [T]$, such that $\sum\nolimits_{i=1}^T \alpha_i < 1$. In a random Gaussian ensemble, there exists a codebook $\tilde{\dmat{X}}$ such that the following inequalities hold for all $r \in [T]$:
\begin{flalign*}
\Pr\left(W_1 \not\in \dset{R}(\mathbf{{y}}_1) \: \middle| \: \tilde{\dmat{X}}, r\right) &\leq \tilde{p}\br{n, k, P, T, r} \\
& \triangleq \frac{1}{\alpha_i}\mathbb{E}_{\dmat{X}} \brs{ \Pr\left(W_1 \not\in \dset{R}(\mathbf{{y}}_1) \: \middle| \: \dmat{X}, r\right) }.
\end{flalign*}

\end{statement}
\begin{proof}
We estimate the probability of a bad code---one for which the inequalities do not hold for at least one $r$. Applying Markov’s inequality and a union bound, we see that this probability is upper-bounded by
$$
\sum\limits_{i=1}^T \alpha_i < 1.
$$
\end{proof}

Finally, we refer the reader to Section~\ref{ch4:sec:random_ka}, where the problem of a random number of active users has been addressed. The results presented in Section~\ref{ch4:sec:random_ka} were developed later and can serve as an alternative approach to the $T$-fold \ac{irsa} achievability bound.

\subsubsection{Performance analysis via \acl{de}} \label{sec:ch5_irsa_de}

\begin{figure}[t]
\centering
\includegraphics{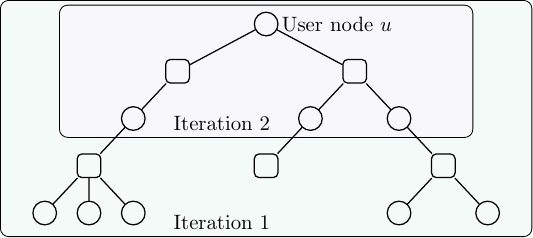}
\caption{Computational graph: Neighborhood of user node $u$ for $2$ iterations.}
\label{fig:irsa_de_tree}
\end{figure}

\begin{figure}[t]
\centering
\includegraphics{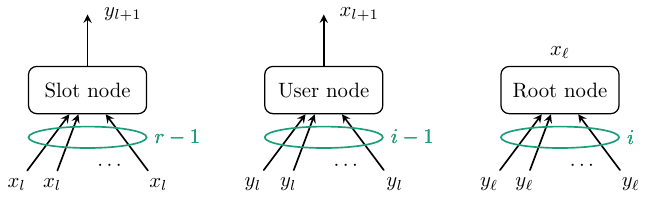}
\caption{\Ac{irsa} \ac{de} rules. Left: Messages to the slot node at intermediate iteration $l$, center: Messages to the user node at intermediate iteration $l$, right: Messages to the user node at the final iteration $\ell$.}
\label{fig:irsa_de_rules}
\end{figure}

The iterative \ac{sic} procedure can be analyzed by means of the \ac{de} method proposed in~\cite{liva2011graph, Pfister2012} for \ac{irsa}. This approach is, in fact, equivalent to the \ac{de} analysis for \ac{ldpc} codes~\cite{RichardsonUrbanke2008} over the erasure channel. We consider message passing decoding, where slot-nodes and user-nodes send outgoing messages along each edge. Each message can take two possible values: either the recovered packet or an erasure\footnote{In what follows, we use the term ``erased message''.} if the packet cannot be recovered.

We consider the ensemble of Tanner graphs $\mathcal{G}(\Ka, L, \lambda\br{x}, \rho\br{x})$ corresponding to the multiple-access scheme with $\Ka$ users, $L$ slots, and the degree distributions $\lambda\br{x}$ and $\rho\br{x}$. We are interested in the decoding performance averaged over this ensemble in the limit as $\Ka, L \to \infty$. Note that in the limit $\Ka \to \infty$, the distribution~(\ref{eq:bino_distr}) becomes a Poisson distribution, i.e.
\[
R\br{x} = \rho\br{x} = e^{-\gamma (1-x)} = \sum\limits_{r=1}^{\infty} \frac{{\gamma}^{r-1}}{(r-1)!} x^{r-1} = \sum\limits_{r=1}^{\infty} \rho_r x^{r-1},
\]
where $\gamma = G \Lambda'\br{1}$ should remain constant.

Let us choose a slot-node $u$. Message passing takes place on the local neighborhood of this node. If we unroll the Tanner graph with the root at user-node $u$, we obtain the computation graph depicted in~\Fig{fig:irsa_de_tree}. In what follows, we use the fact that any neighborhood of constant depth in $\mathcal{G}(\Ka, L, \lambda\br{x}, \rho\br{x})$ is tree-like with probability one. We refer the reader to~\cite{RichardsonUrbanke2008} for a detailed explanation.

We follow the exposition from~\cite{Glebov2019}. Before we proceed with the \ac{de} rules, we list the main assumptions:
\begin{enumerate}
\item A code $\tilde{\dmat{X}}$ constructed in accordance with Statement~\ref{stat:code} is used in the system.
\item If the collision order is greater than $T$, no message can be recovered from the slot. If the collision order is $t < T$, each message is recovered with probability $\tilde{p}\br{n', k, P, T, t}$.
\item The genie reveals information on false packets, i.e., such packets are excluded from the \ac{sic} process.
\end{enumerate}

\begin{remark}
Note that the genie assumption prevents the described bound from being a true achievability bound. This bound is only an optimistic estimate of the performance achievable within a $T$-fold \ac{irsa} scheme.     
\end{remark}

Now let us write the \ac{de} rules. By $x_l$ and $y_l$, we denote the probabilities that an outgoing message (see~\Fig{fig:irsa_de_rules}) from the user-node and the slot node, respectively, are erased during the $l$-th iteration.

\begin{statement}
The DE rules are described in the following form:
\begin{eqnarray*}
y_{l+1} &=& 1 - \rho\br{1-x_l} \sum\limits_{t=0}^{T-1} \br{ 1 - \tilde p\br{n', k, P, T, t+1}} \times \frac{\br{G L'\br{1} x_l}^t}{t!}, \\
x_{l} &=& \lambda\br{y_l}, \:\: 1 \leq l < \ell, \\
x_{\ell} &=& \Lambda\br{y_{\ell}},
\end{eqnarray*}
where the function $\tilde p\br{n', k, P, T, r}$ is defined in Statement~\ref{stat:code}, and the initial condition is $x_0 = 1$, which means that the user messages are erased at the beginning, and we observe only the noisy signal sums in the slots.
\end{statement}

\begin{proof}
Consider the $l$-th iteration. Let us start with the \emph{slot node}. We want to calculate the erasure probability of the outgoing message $y_{l+1}$ based on incoming messages that are erased with probabilities $x_l$. Let us start with a slot node of degree $r$. We can recover the outgoing message (or equivalently, the packet) if there are no more than $T-1$\footnote{In accordance with the \ac{de} rules, the outgoing message is assumed to be erased, and thus the total number of erased messages is no more than $T$.} erased messages out of $r-1$ incoming messages. Finally, the outgoing message is erased with probability
\[
p_r = 1 - \sum\limits_{t=0}^{\min\br{r,T}-1} \br{1 - \tilde p\br{n', k, P, T, t+1}} \binom{r-1}{t} x_l^t \br{1-x_l}^{r-1-t}.
\]

Note that $\rho_r$ is the probability that the outgoing edge is connected to a slot node of degree $r$ (see~\cite{RichardsonUrbanke2008}), thus the total probability of the erased outgoing message can be calculated as
\[
y_{l+1} = \sum\limits_{r = 1}^{\infty} \rho_r p_r.
\]
By changing the summation order, we obtain the result from the proposed statement.

Now proceed with the \emph{user-node}. Clearly, the outgoing message is erased if all incoming messages are erased. Thus, for the node of degree $i$, the outgoing message is erased with probability
\[
q_i = y_l^{i-1},
\]
and the total probability of the erased outgoing message can be calculated as
\[
x_{l} = \sum\limits_{i = 1}^{\infty} \lambda_i q_i = \lambda\br{y_l}.
\]

The root node should be considered separately. The only change is in the total probability calculation. Here, we need to use the probabilities for the user-node (rather than the edge) to be connected to a slot node of degree $i$. Thus,
\[
x_{\ell} = \sum\limits_{i = 1}^{\infty} \Lambda_i q_i = \Lambda\br{y_{\ell}}.
\]
\end{proof}

\begin{theorem}\label{th:irsa_t_bound}
Let us fix $P'$, $n'$, $\Lambda\br{x}$, and $\ell$. Then, the $T$-fold \ac{irsa} scheme with the code $\tilde{\dmat{X}}$ achieves the \ac{pupe} (as $\Ka, L \to \infty$)
\[
P_e \leq x_{\ell} + \prob{\bigcup\limits_{i \ne j} \brc{ W_i = W_j }}.
\]
\end{theorem}

\begin{remark}
We note that, in the case of \ac{ldpc} codes over the erasure channel \ac{de}, one wants to find the so-called threshold (for the erasure probability $\varepsilon$), i.e.,
\[
\varepsilon_{\text{thr}} = \sup\brc{ \varepsilon \in [0,1]:  \lim\limits_{\ell \to \infty} x_\ell = 0 }.
\]

In our case, the graph is infinite, but the slot is finite, and it is not possible to achieve $\lim\nolimits_{\ell \to \infty} x_\ell = 0$ due to the presence of the $\tilde p\br{n', k, P, T, r}$ function.
\end{remark}

\begin{remark}
Clearly, the average transmitted energy can be calculated as $n' P' \Lambda'\br{1}$. Thus, the energy per information bit is given by the following expression:
\begin{equation}\label{eq:ebno_tfold_irsa}
\frac{E_b}{N_0} = \frac{n' P' \Lambda'\br{1}}{2k}.
\end{equation}
\end{remark}

Our main goal is to minimize the required energy per bit $E_b/N_0$ given that $P_e \leq \varepsilon$. In Section~\ref{sec:ch5:num_results_irsa}, we utilize the described \ac{de} method to choose the system parameters: $n'$ and $\Lambda\br{x}$.

\subsection{\texorpdfstring{$T$}{T}-fold \acs{irsa} modifications for the \acs{gura}}\label{ch5:sec_irsa_sf}

Recall the main problem: we need the same codebook scheme, while existing linear codes (such as \ac{ldpc} and polar codes) perform extremely poorly in this regime. We refer the reader to Section~\ref{sec:same_lin} for a detailed discussion of this issue. Hence, we require a nonlinear same-codebook approach. A straightforward method is to use preambles (as illustrated in~\Fig{fig:irsa-s} and~\ref{fig:irsa-f}) that define a transformation over the codewords of a linear code. Preamble detection can be considered as a \ac{cs} problem. Once a preamble is detected, the decoder can recover this transformation and decode the linear code.

In what follows, we refer to the \ac{irsa} presented in Section~\ref{ch5:sec_irsa_basic} as \acf{irsab} and describe two strategies proposed in~\cite{vem2019user, AmalladinneJoint2019}, namely \ac{irsas} and \ac{irsaf}.

\subsubsection{\Acl{irsas}}\label{par:irsa_s_scheme}

\begin{figure}[t]
\centering
\includegraphics{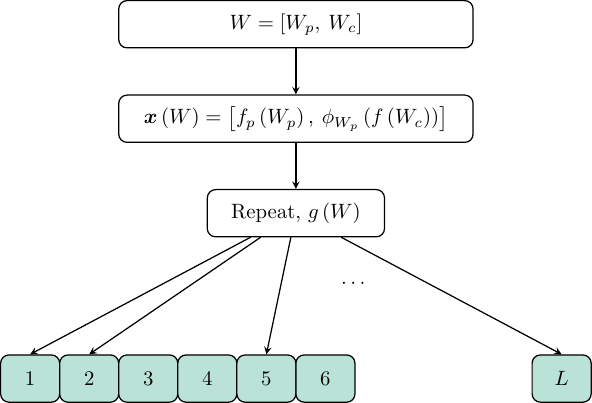}
\caption{\Ac{irsas} encoding process. The transmitted message $W$ is split into two parts: $W_p$, which defines the preamble $f_p\br{W_p}$ and a transform $\phi_{W_p} \br{\cdot}$, and $W_c$, which is encoded using the same linear code with the encoding function $f\br{\cdot}$, followed by a transform $\phi_{W_p}\br{\cdot}$ (e.g., interleaving). The number of replicas and slot indices are determined by a function $g\br{W}$. Slots are represented by a row of green rectangles. Note that the entire message, including the preamble, is transmitted within each chosen slot. For clarity, we omit the user message index.}
\label{fig:irsa-s}
\end{figure}

In this variant, the base protocol remains unchanged, i.e., the replica count and slot selection are determined based on the message. The main modification lies in the word to be transmitted in the slot.

Consider the $i$-th user. We split the message into two parts: $W_i = \brs{W_{i,p}, \: W_{i,c}}$, where $W_{i,p} \in [M_{p}]$, $W_{i,c} \in [M_{c}]$, and $M = M_p M_c$. Analogously, the slot is divided into two parts of lengths $n_{p}$ and $n_{c}$, such that $n' = n_{p} + n_{c}$. Now, let us describe the word to be transmitted:

\begin{itemize}
\item Within the first $n_{p}$ symbols, the user sends the preamble $f_p(W_{i,p})$, where $f_p: [M_p] \to \mathbb{R}^{n_{p}}$ is the encoder of a well-designed \ac{cs} code. Note that the size of this code should not be large ($M_p \leq 2^{15}$), as we apply standard \ac{cs} decoding algorithms such as LASSO, \ac{nnls}, or \ac{omp}. We discuss these codes and their decoding algorithms in Appendix~\ref{a:cs_methods}.
\item Within the last $n_{c}$ symbols, we transmit the user codeword $f(W_{i,c})$ transformed by $\phi_{W_{i,p}}$, where $f: [M_c] \to \mathbb{R}^{n_{c}}$ is the user code encoder\footnote{As previously, the encoder is the same for all users.}, and the family of functions $\phi_{W}: \mathbb{R}^{n_{c}} \to \mathbb{R}^{n_{c}}$, with $W \in \brs{M_p}$, is introduced to ``create'' different codes. In what follows, the function $\phi$ is implemented either as a permutation or as an addition with a shift vector, both determined by $W_p$.
\end{itemize}

Thus, the word transmitted in the slot is given by:
\[
\dvec{x}(W_i) = [f_p(W_{i,p}), \: \phi_{W_{i,p}}(f(W_{i,c}))].
\]

Now, we describe the decoding process. Analogous to the \ac{irsab}, the graph is not known on the receiver side and is reconstructed during the decoding process. The slot decoding process then proceeds as follows:

\begin{enumerate}
\item Select a slot (preferably the one with the smallest estimated collision order, e.g., based on energy estimation).
\item Decode the preamble part using a \ac{cs} algorithm. The output of this algorithm includes $\hat{r}$ -- an estimate of the collision order -- the set of preamble parts $\{\hat{W}_{i,p}\}$, and the set of functions $\{\phi_{\hat{W}_{i, p}}\}$, where $i \in [\hat{r}]$.
\item Decode the user messages and obtain $\{\hat{W}_{i, c}\}$. We discuss specific decoding algorithms for \ac{ldpc} and polar codes in Section~\ref{sec:ch5:irsa_practical_awgn}.
\item Perform the \ac{sic} step. As a result of steps 1--3, we obtain the set of message estimates $\{\hat{W}_i\}$, where $i \in [\hat{r}]$. Use the function $g(\cdot)$ to remove replicas.
\item Remove the slot from the slot list. If the list of slots is non-empty, return to step 1.
\end{enumerate}

\begin{remark}
Note that since we use a different codebook scenario in step 3, we do not need to solve the message assembling problem. For each $\hat{W}_p$, we either reconstruct $\hat{W}_c$ or encounter a decoding failure.
\end{remark}

\subsubsection{\Acl{irsaf}}

\begin{figure}[t]
\centering
\includegraphics{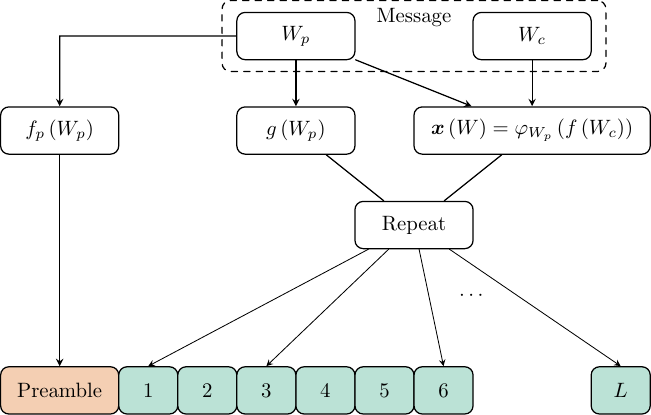}
\caption{\Acs{irsaf} encoding process. The transmitted message $W$ is split into two parts: $W_p$, which defines the preamble $f_p\br{W_p}$ and a transform $\phi_{W_p} \br{\cdot}$, and $W_c$, which is encoded using the same linear code with the encoding function $f\br{\cdot}$, followed by a transform $\phi_{W_p}\br{\cdot}$. The number of replicas and slot indices are determined by a function $g\br{W_p}$. All preambles are sent in a dedicated slot, depicted in orange, while $\varphi_{W_p}\br{f\br{W_c}}$ is sent in the chosen slots, marked by green rectangles. For clarity, we omit the user message index.}
\label{fig:irsa-f}
\end{figure}

In the previous variant, slot-wise preambles were used, which simplify the requirements for the \ac{cs} code and decoder. Indeed, one needs to solve the following \ac{cs} problem:
\begin{equation} \label{eq:ch5:cs_preamble}
\dvec{y}_p = \dmat{A}_p \dvec{u}_p + \dvec{z}_p,
\end{equation}
where $\dvec{u}_p = (u_{p,1}, \ldots, u_{p,M_p})$ and the sparsity $\mathrm{wt}\br{\dvec{u}_p} = r \leq T$. Thus, the length $n_p$ can be small. At the same time, we add preambles to all the slots, and thus the actual overhead is $L n_p$ channel uses. An illustration of the corresponding sensing matrix is given in Appendix~\ref{a:cs_matrices}.

Another preamble placement was proposed in~\cite{AmalladinneJoint2019}. Again, consider the $i$-th user and split the message into two parts: $W_i = [W_{i,p}, \: W_{i,c}]$, where $W_{i,p} \in [M_{p}]$, $W_{i,c} \in [M_{c}]$, and $M = M_p M_c$. However, the main differences are as follows:
\begin{itemize}
\item The codeword placement and replica count depend only on $W_{i,p}$.  
\item The preamble of length $n_p$ is placed before the frame, so the total length $n$ is divided between the preamble ($n_p$) and $L$ slots of length $(n-n_p)/L$.
\end{itemize}

Let us consider the decoding process.
\begin{enumerate}
\item Decode the preamble part using a \ac{cs} decoder, i.e., solve problem~\eqref{eq:ch5:cs_preamble}, where
$$
\dvec{u}_p = (u_{p,1}, \ldots, u_{p,M_p})
$$
and the sparsity $\mathrm{wt}(\dvec{u}_p) = K_a$. This problem is more challenging compared to the slot-wise \ac{cs} problem, and thus a stronger \ac{cs} code is required.
\item Reconstruct the graph and $\{\phi_{\hat{W}_{i, p}}\}$ for $i \in [K_a]$. Note that this is the most significant difference compared to the \ac{irsas} case. If all preambles are recovered, the graph and code transformations are known.
\item Apply the \ac{irsab} decoder with the known graph. In this case, we can select slots with the minimal user population.
\end{enumerate}

\begin{remark}
We note that both protocols (\ac{irsas} and \ac{irsaf}) assume an additional degree of freedom: namely, the preamble and codeword energies can be different. These energies should be chosen to optimize system performance. The only requirement is the total energy constraint:
\begin{flalign*}
\text{\acs{irsas}:} \quad &(P_p n_p + P_c n_c) L'\br{1} \leq Pn,& \quad (n_c = n' - n_p = n/L - n_p) \\ 
\text{\acs{irsaf}:} \quad &P_p n_p + P_c n_c L'\br{1} \leq Pn,& \quad (n_c = n' = (n-n_p)/L)
\end{flalign*}
\end{remark}

\subsection{Collision resolution methods}\label{sec:ch5:irsa_practical_awgn}

In this section we consider collision resolution methods, i.e. we focus on one slot of length $n'$ real channel uses and collision resolution capability $2 \leq T \leq 10$. In what follows we assume $r \leq T$ users that transmitted the codewords $\dvec{x}^{\br{1}}, \dvec{x}^{(2)}, \ldots, \dvec{x}^{\br{r}}$ have collided in the slot of interest
\begin{equation}\label{eq:ch5:slot_received}
\dvec{y} = \sum\limits_{i=1}^r \dvec{x}^{\br{i}} + \dvec{z}, \quad\dvec{z}\sim\mathcal{N}\br{0, \dmat{I}_{n'}}.
\end{equation}

\subsubsection{Compute and forward, \acs{irsab}}

We begin with the first scheme proposed for \ac{gura} in~\cite{ordentlich2017low}. Although this scheme exhibits poor energy efficiency\footnote{We note that the authors of~\cite{ordentlich2017low} consider pure $T$-fold ALOHA rather than $T$-fold \ac{irsa}, which may be the root cause of its poor performance.} compared to fundamental limits and other \ac{ura} schemes, it incorporates inspiring ideas. The scheme relies on the compute-and-forward strategy~\cite{Nazer2011CoF}, which aims to recover integer combinations of packets rather than the packets themselves.

To illustrate this idea, consider the following example. Suppose we have a two-user \ac{gmac}, where the users wish to transmit messages $W_1$ and $W_2$ by sending codewords $f(W_1)$ and $f(W_2)$. Recall that $f: [M] \to \mathbb{R}^{n'}$ and that the energy constraints are given by $\left\|f(W_1)\right\|^2 \leq Pn'$ and $\left\|f(W_2) \right\|^2 \leq Pn'$. The channel output is given by
\[
\rvec{y} = f(W_1) + f(W_2) + \rvec{z}, \quad \rvec{z} \sim \mathcal{N}\br{0, \dmat{I}_{n'}}.
\]

The standard low-complexity decoding strategy utilizes a \ac{tin} step combined with an \ac{sic} step. This can be done in different ways, e.g., in a hard manner, in a soft manner (soft \ac{sic}), or through message passing between two \ac{tin} decoders. Nevertheless, the first \ac{tin} decoder deals with effective noise of energy $(P+1)n'$, which is complicated.

At the same time, if 
\begin{equation}\label{eq:lattice} 
f(W_1) + f(W_2) = f\br{W_1 \boxplus W_2},
\end{equation}
then the receiver can decode $W_1 \boxplus W_2$ as in the single-user case, i.e., without interference. In general, the operation $\boxplus$ can be any arbitrary group operation, i.e., we only need $(\mathcal{M}, \boxplus)$ to be a group, where $\mathcal{M}$ is the closure of the set $\{W_i\}_{i=1}^{M}$ under the group operation $\boxplus$.

In what follows, we use a particular $\boxplus$ operation. Namely, the messages are represented as binary vectors of length $k$, and the $\boxplus$ operation is an element-wise real addition of such vectors. Clearly, we have $\mathcal{M} = \mathbb{Z}^k$. Assume~\eqref{eq:lattice} holds, then the set $\Lambda = \{f(\dvec{u})\}_{\dvec{u} \in \mathbb{Z}^k}$ forms a group under real addition, and $f$ is an isomorphism\footnote{We need to extend the domain of $f$ from $[M]$ to $\mathbb{Z}^k$ based on~\eqref{eq:lattice}.} between $(\mathbb{Z}^k, +)$ and $(\Lambda, +)$. Thus, we come to \emph{lattices} and \emph{lattice codes}~\cite{conway2013sphere}.

Let us introduce necessary definitions.

\begin{definition}
An $n$-dimensional lattice, $\Lambda$, is a discrete additive subgroup of $\mathbb{R}^n$. A lattice can always be described using a generator matrix $\dmat{B} \in \mathbb{R}^{n \times n}$:
\[
\Lambda = \{ \lambda = \dmat{B} \dvec{u}, \:\: \dvec{u} \in \mathbb{Z}^n \}.
\]
\end{definition}

\begin{definition}
The Voronoi region $\dset{V}(\lambda)$ of a lattice point $\lambda$ is the set of all points in $\mathbb{R}^n$ that are closer to $\lambda$ than to any other lattice point.
\end{definition}

\begin{definition}
Let $\dset{R} \subseteq \mathbb{R}^n$. In what follows, we refer to $\dset{R}$ as a shaping\footnote{We do not specify any particular requirements for $\dset{R}$ in this definition. Clearly, it is reasonable to choose $\dset{R}$ from convex bounded sets. From energy considerations, the closer $\dset{R}$ is to the ball centered at $\dvec{0}$, the better.} region. A lattice code $\mathcal{C}$ is a finite subset of a lattice $\Lambda$:
\[
\mathcal{C} = \Lambda \bigcap \dset{R}.
\]

The code consists of a finite subset of $M$ points from the original lattice:
\[
\mathcal{C} = \{ \lambda_1, \lambda_2, \ldots , \lambda_M\},
\]
where $\lambda_i \in \Lambda$, $i \in [M]$.
\end{definition}

Usually, lattice codes are constructed using nested lattices and are called \emph{nested lattice codes}.

\begin{definition}
A lattice $\Lambda_c$ is said to be nested in a lattice $\Lambda_f$ if $\Lambda_c \subseteq \Lambda_f$. We will sometimes refer to $\Lambda_c$ as the coarse lattice and $\Lambda_f$ as the fine lattice.
\end{definition}

\begin{definition}
Let $\Lambda_f$ be a lattice such that $\Lambda_c \subseteq \Lambda_f$, and let $\dset{V}_c(0)$ be the 0-centered Voronoi region for $\Lambda_c$. Then,
\[
\mathcal{C} = \Lambda_f \bigcap \dset{V}_c(0)
\]
is a nested lattice code.
\end{definition}

Note that a nested lattice code $\mathcal{C}$ is isomorphic to the quotient group $\Lambda_f / \Lambda_c$, i.e., $\mathcal{C}$ forms a group under addition modulo $\Lambda_c$.

We note that the use of lattice codes for \ac{mac}s is natural, as electromagnetic signals exhibit linear superposition. The illustration of the transmission scheme using nested lattice codes is presented in Fig.~\ref{fig:nested_lattice_mac}.

\begin{figure}[t]
\centering
\includegraphics{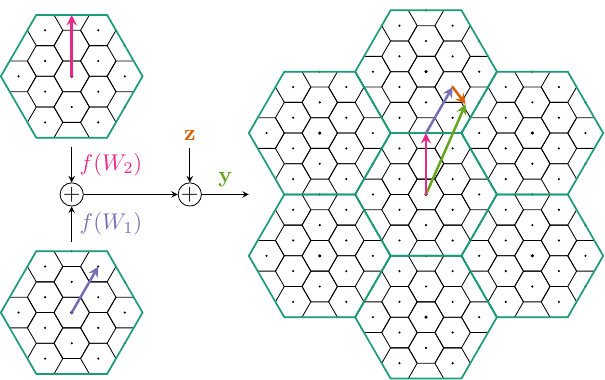}
\caption{Illustration of nested lattice codes for \ac{gmac}.}
\label{fig:nested_lattice_mac}
\end{figure}

The aim of lattice coding is to denoise the sum of messages. Note that only a single-user decoder is required to accomplish this task, which is a clear advantage of this approach. Assume successful denoising has been carried out. The second phase is to recover the transmitted messages from the denoised sum. For this purpose, one may use an outer code for the \emph{noiseless} adder channel. The only requirement for such a code is unique decodability: all the sums of up to $T$ codewords are different. See Section~\ref{sec:ch2_zero} for details.

Now, let us describe the scheme from~\cite{ordentlich2017low}.

\paragraph{Encoding.} We start with the two-phase encoding process:
\begin{enumerate}
\item Encode the message $W$ for the adder channel:
\[
\dvec{x}_{\text{ad}} = f_{\text{ad}}(W),
\]
where $f_{\text{ad}}: [M] \to \mathbb{F}_p^{k'}$.
\item Encode $\dvec{x}_{\text{ad}}$ with the lattice code:
\[
\dvec{x} = f_{\text{lc}}(\dvec{x}_{\text{ad}}),
\]
where $f_{\text{lc}}: \mathbb{F}_p^{k'} \to \mathbb{R}^{n'}$.
\end{enumerate}

Let us describe the main components.

\emph{Outer code}. The technique proposed in~\cite{BarDavid1993} is as follows. Let $\dmat{H} = [\dvec{h}_1, \ldots, \dvec{h}_M]$ be the parity-check matrix of a $T$-error-correcting code over $\mathbb{F}_p$. Set $\mathcal{C}_\text{ad} = \{\dvec{h}_1, \ldots, \dvec{h}_M\}$. This construction ensures the desired property: all sums of up to $T$ codewords (or columns) are distinct. For practical construction, one can use a \ac{bch} code over $\mathbb{F}_p$ with low-complexity decoding algorithms (e.g., the Berlekamp-Massey algorithm).

\emph{Inner code}. For practical implementation, we need a high-dimensional lattice with a low-complexity decoding algorithm. The authors of~\cite{ordentlich2017low} proposed utilizing the so-called Construction A~\cite{conway2013sphere}, i.e.,
\[
\Lambda = \mathcal{C}_\text{lc} + p \mathbb{Z}^n,
\]
where $p$ is a prime number and $\mathcal{C}_\text{lc}$ is a $p$-ary error-correcting code with a low-complexity decoding algorithm.

\begin{remark}
Note that for Construction A, we have $\Lambda_f = \Lambda$ and $\Lambda_c = p \mathbb{Z}^n$. Hence, modulo-$\Lambda_c$ reduction is straightforward: it is just a component-wise modulo-$p$ operation, with the result in the interval $[0, p)$.
\end{remark}

Now, consider the received signal~\eqref{eq:ch5:slot_received} and recall that
\[
\sum\limits_{i=1}^r \dvec{x}^{\br{i}} = \lambda \in \Lambda,
\]
i.e., this sum is a lattice point.

\paragraph{Decoding.}

The paper~\cite{ordentlich2017low} does not specify the lattice code explicitly. Instead, the authors suggest computing
\[
\dvec{y}' = \dvec{y} \mod p = (\lambda + \dvec{z}) \mod p = \dvec{c} + (\dvec{z} \mod p),
\]
where $\dvec{c} \in \mathcal{C}_\text{lc}$, and propose using the \ac{fbl} achievability bound~\cite{ppv2011} to characterize the capabilities of $\mathcal{C}_\text{lc}$.

Here, we describe the algorithm for some low-complexity soft-input code $\mathcal{C}_\text{lc}$ over $\mathbb{F}_p$ (e.g., \ac{ldpc} or a polar code can be used).

\begin{enumerate}
\item Let $\dvec{y} = [y_1, \ldots, y_{n'}]$. The input of the $\mathcal{C}_\text{lc}$ decoder is the a priori distribution vector
\[
\dmat{D} = \left[\dvec{d}_1, \dvec{d}_2, \ldots, \dvec{d}_{n'}\right],
\]
where
\[
\dvec{d}_i = [\Pr(c_i = 0 | y_i), \ldots, \Pr(c_i = p-1 | y_i)], \quad i \in [n'].
\]

First, we calculate $\dmat{D}$ as follows:
\[
{d}_{i, m} = \Pr(c_i = m | y_i) = \frac{1}{\sqrt{2 \pi}} \sum\limits_{\substack{j = -\infty,\\ (j = m \mod p)}}^{+\infty} \exp[-(y_i - j)^2/2].
\]

Note that in practical implementation, it is reasonable to sum over $j \in [y_i - \Delta, y_i + \Delta]$ for some fixed $\Delta$.

\item Decode $\mathcal{C}_\text{lc}$. We obtain
\[
\dvec{y}_\text{ac} = \sum\limits_{i=1}^r \dvec{x}^{\br{i}}_\text{ac}.
\]

\item Decode the outer code for the noiseless adder channel. Taking into account that the code consists of the columns of the \ac{bch} parity-check matrix, $\dvec{y}_\text{ac}$ is nothing more than a syndrome. Indeed, we have
\[
\dvec{y}_\text{ac} = \dmat{H} \dvec{s},
\]
where $\dvec{s}$ is the indicator vector of transmitted messages, i.e., $\{W_1, \ldots, W_r\} = \mathrm{supp}(\dvec{s})$.

Thus, we arrive at the classical \ac{bch} code decoding problem. Note that $r \leq T$ and our code is a $T$-error-correcting code.
\end{enumerate}

\begin{remark}
The paper~\cite{ordentlich2017low} focuses on the case $p=2$, as binary codes are well-developed and significantly outperform non-binary ones in terms of soft decoding complexity. In what follows, we consider the case $p=2$, in which case $\dmat{D}$ can be replaced with the input \acf{llr} vector.
\end{remark}

\begin{remark}
One may argue that \ac{bch} decoding has high complexity since the parity-check matrix $\dmat{H}$ has size $n' \times M$ (recall that $M = 2^k$ and $k = 100$ bits is of interest). Classical \ac{bch} decoding involves calculating the syndrome, which has linear complexity in $M$ (thus, exponential complexity in $k$). However, there is no need to calculate the syndrome in our case, as this operation is performed in the channel, and we receive it directly. It can be shown (see~\cite{ordentlich2017low}) that the complexity of \ac{bch} decoding is $O(k T^2 \log^2(T) \log\log(T))$ operations in $\mathbb{F}_{2^k}$.
\end{remark}

\begin{remark}
In the scheme described above, users transmit symbols from the set $\{0, \sqrt{2P}\}$ (or equivalently, $\{0, 1\}$ under the power constraint $P$). In a practical implementation, it is more reasonable to use the set $\{-\sqrt{P}, \sqrt{P}\}$, which corresponds to a lattice coset. See~\cite{ordentlich2017low} for further details.
\end{remark}

\subsubsection{\Ac{irsas} with \ac{ldpc} codes}\label{sec:dec_ldpc}

In this section, we consider the \ac{irsas} scheme (presented in Section~\ref{par:irsa_s_scheme}), which is based on \ac{ldpc} codes~\cite{vem2019user}. We employ a joint iterative decoding algorithm with low computational complexity~\cite{Boutros2002, LDPC_2user}.  

Recalling the \ac{irsas} scheme, the slot is divided into two parts, such that $n' = n'_p + n'_c$ channel uses. In the first part, we utilize \ac{bch} subcodes, as described in Appendix~\ref{a3:cwbook:sub_bch}, as a suitable candidate for the \ac{cs} codebook. The preamble is recovered using the methods detailed in Appendix~\ref{a:cs_methods}.  

Now, let us focus on the joint decoding of the \ac{ldpc} codes positioned in the last $n_{c}'$ channel uses. We define different codebooks, \( \mathcal{C}_1, \ldots, \mathcal{C}_{r} \), which are derived from a common codebook $\mathcal{C}$ through random permutations based on $\phi_{W_p}$.  

\begin{figure}[t]
\centering
\includegraphics{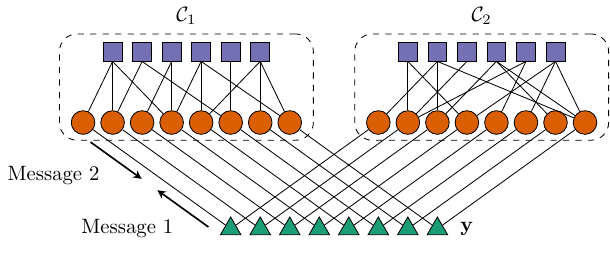}
\caption{The joint decoding algorithm and factor graph for the case $r=2$ are illustrated. Orange circles and blue squares represent the variable and check nodes of the \ac{ldpc} codes $\mathcal{C}_1$ and $\mathcal{C_2}$, respectively. The messages exchanged between these nodes follow the standard decoding procedure of \ac{ldpc} codes. Additionally, there are functional nodes, depicted as green triangles, that represent the received signal $\dvec{y}$.}\label{fig:irsa_ldpc_joint}
\label{fig:fg}
\end{figure}

The received signal in the slot is given by~\eqref{eq:ch5:slot_received}, where \( \dvec{x}^{\br{i}} = \tau\br{\dvec{c}^{\br{i}}} \) is a \ac{bpsk}-modulated codeword of each active user. To recover users' codewords, we employ a decoder based on the iterative \ac{bp} algorithm, which is a specific case of the broader class of \ac{mpa}s. The decoding procedure can be represented as a graph (factor graph)~\cite{KFL}, shown in~\Fig{fig:fg}.  

There are three types of nodes in the graph. The \ac{ldpc} codes of the users are represented by Tanner graphs with variable nodes (orange) and check nodes (blue). For the variable and check nodes in the subgraphs of different users, the message exchange rule follows the same procedure as in the single-user case (see~\cite{RichardsonUrbanke2008}).

Additionally, there is a third type of node in the graph -- functional nodes (marked by green triangles). These nodes correspond to the elements of the received signal \( \rvec{y} \). To fully specify the decoding algorithm, we need to describe the messages exchanged between the variable nodes of each user's code and the functional nodes (Message type 1), as well as the messages from the functional nodes to the variable nodes of the users' codes (Message type 2).

\paragraph{Message type 1 (from functional nodes to \ac{ldpc} codes).}
\newcommand{\msgUF}[1]{
\mu_{#1 \to \mathcal{F}}
}
\newcommand{\msgFU}[1]{
\mu_{\mathcal{F} \to #1}
}
Following the principles of message-passing algorithms~\cite{RichardsonUrbanke2008}, the update rule to compute the message \( \mu \) sent to the \( j \)-th variable node of the \( i \)-th user (\( j \in [n_{c}'] \), \( i \in [r] \)) from the functional node \( \mathcal{F}_{j} \) is given as follows.

First, let us fix the index of the functional node \( j \) and omit it for clarity (as the \( j \)-th functional node may only be connected to the \( j \)-th variable node of each user's \ac{ldpc} code; see~\Fig{fig:irsa_ldpc_joint}). Let \( y = \dvec{y}_j \) be the signal received at the \( j \)-th position. Let \( \dvec{b} \in \brc{0, 1}^r \) represent a binary vector, where \( \dvec{b}_i \) corresponds to the bit value of the \( i \)-th user's codeword \( \dvec{c}^{\br{i}} \). Also, introduce a function:

\begin{equation}\label{eq:gmac_joint:gauss_pdf}
\eta\br{\dvec{b}} = p\left(y \middle|  \dvec{b} \right) = \exp\brc{-\br{
y - \sum\limits_{i = 1}^{r} \tau\br{\mathbf{b}_i}
}^2},
\end{equation}
which represents the p.d.f. of the received signal conditioned on \( \dvec{b} \), up to some constant multiplier that can be ignored when considering likelihood ratios. Note that \( \dvec{b} \) completely defines the noiseless received signal.

Let \( \msgUF{i} \) be the message sent from the \( i \)-th user's variable node to the functional node. Let \( l_i \) represent the \ac{llr} value given by the \( i \)-th user's \ac{ldpc} decoder. Then, define:
\begin{equation}\label{eq:ch5:llr2soft}
\msgUF{i}\br{0} \triangleq \Pb{\dvec{b}_i = 0} = \frac{e^{l_i}}{e^{l_i} + 1}, \quad \msgUF{i}\br{1} \triangleq 1 - \msgUF{i}\br{0},
\end{equation}
where \( l_i \) corresponds to message type 2. The values of \( l_i \) are initialized to zero and then updated by the iterative decoding algorithm.

Similarly, \( \msgFU{i} \) is the message sent from the functional node to the \( i \)-th user's variable node, and \( \msgFU{i}\br{0} \) is given by:
\begin{equation}\label{eq:func_node_dec}
\msgFU{i}\br{\beta}=
\sum\limits_{\dvec{b} \in \brc{0, 1}^r, \dvec{b}_i = \beta}\brs{\eta\br{\dvec{b}}\times\prod\limits_{t \neq i}\msgUF{t}\br{\dvec{b}_t}},
\quad \beta \in \brc{0, 1},
\end{equation}
where \( \eta\br{\dvec{b}} \) is given by~\eqref{eq:gmac_joint:gauss_pdf}, and \( \msgUF{t}\br{\cdot} \) is given by~\eqref{eq:ch5:llr2soft}. The summation above has \( 2^{r - 1} \) terms, where \( r \) is the collision order. Consequently, the number of computations required to obtain the messages from the functional node \( \mathcal{F}_{i} \) grows exponentially with the number of users.

Now, let us express the messages in terms of the \ac{llr}s \( l_{\mathcal{F} \to i} \):
\begin{equation}\label{eq:func_node_llr}
l_{\mathcal{F} \to i} = \log\br{\frac{\msgFU{i}\br{0}}{\msgFU{i}\br{1}}},
\end{equation}
where \( \msgFU{i}\br{\beta} \) is given by~\eqref{eq:func_node_dec}. The value \( l_{\mathcal{F} \to i} \) is then passed as the input \ac{llr} to the single-user \ac{ldpc} decoder.

Note that an alternative approach to functional node processing is given by \ac{soic}, as proposed in~\cite{Wang2019soic}. Soft interference represents the expected value of the codeword symbol \( \mathbb{E} \dvec{c}^{\br{i}} \), derived from~\eqref{eq:ch5:llr2soft} as follows (considering \ac{bpsk} modulation):
\[
\mathbb{E} \dvec{x}^{\br{i}} = \sqrt{P}\br{\msgUF{i}\br{0} - \msgUF{i}\br{1}}.
\]

Then, the \ac{llr} specified in~\eqref{eq:func_node_llr} is calculated as the \ac{llr} for \ac{bpsk} modulation with \ac{soic}:
\[
l_{\mathcal{F} \to i} = 2\sqrt{P}\br{y - \sum\limits_{t=1, t\neq i}^{r}\mathbb{E}\dvec{x}^{(t)}},
\]
where \( y \) is the received signal at the considered position, and \( 2\sqrt{P}y \) corresponds to the \ac{llr} of a \ac{bpsk}-modulated signal with a power constraint \( P \).

\paragraph{Message 2 (decoding the \ac{ldpc} code).}
After decoding the functional nodes and computing the \acs{llr}s for the bits, this information is used to decode the \ac{ldpc} code. The user applies the \ac{mpa} algorithm to the \ac{ldpc} code, which can be either the Sum-Product or the Min-Sum algorithm~\cite{RichardsonUrbanke2008}. The output of this algorithm updates the values $l_i$ presented in~\eqref{eq:ch5:llr2soft}.

Finally, the message exchange schedule can be arbitrary. Functional node decoding (propagation of message type 1) may be followed by several decoding iterations of each single-user \ac{ldpc} code.

To construct the \ac{ldpc} codes, the modified \ac{pexit} method from~\cite{LivaChiani2007} is used. Here, only the significant differences from the single-user case are highlighted. The multi-user Gaussian channel is not symmetric; therefore, the \ac{pexit} method cannot be applied directly to the all-zero codeword.  
To address this issue, we consider \emph{channel adapters} from~\cite{SGMAC}. The resulting channel is symmetric, allowing us to use only the all-zero codeword for code construction. Note that the random adapters were used only for code construction.

\subsubsection{\Ac{irsaf} with polar codes}\label{sec:dec_polar}

Let us consider the \ac{irsaf} scheme with polar codes employed in each slot~\cite{Marshakov2019Polar}. The choice of polar codes is motivated by two reasons. First, polar codes~\cite{Arikan} are a powerful tool for short code lengths under the \ac{scl} decoding algorithm~\cite{TalVardyList2015}. Second, as previously discussed in Section~\ref{sec:ch2_low_complexity}, polar codes exhibit capacity-achieving properties in \ac{mac}. In this section, we describe how polar codes can be used in \ac{gmac} in the \ac{fbl} regime and present a \ac{jsc} decoding algorithm.

Consider Ar{\i}kan's kernel and the corresponding \emph{polar transform} of size \(n' = 2^m\):
\[
G_2 \triangleq \begin{bmatrix} 1 & 0 \\
1 & 1
\end{bmatrix},\quad G_{n'} \triangleq  G_2^{\otimes m},
\]
where \(\otimes\) denotes the Kronecker power.

To construct an \((n', k)\) polar coset code, denote the set of frozen positions by \(\dset{F}\), with \(|\dset{F}| = n' - k\). Let \(\dvec{u}_{\dset{F}}\) denote the projection of the vector \(\dvec{u}\) onto the positions in \(\dset{F}\). Now, we define a \emph{polar coset code} \(\mathcal{C}\) as follows:
\[
\mathcal{C}\br{n', k, \dset{F}, \dvec{f}} = \left\{ \dvec{c} = \dvec{u}^T G_{n'} \:\: \middle| \:\: \dvec{u} \in \{0,1\}^{n'}, \:\: \dvec{u}_{\dset{F}} = \dvec{f} \right\}.
\]

The users utilize \emph{different} polar coset codes \(\mathcal{C}_i({n'}, k, \dset{F}_i, \mathbf{f}_i)\), \(i=1,\ldots, r\). Consider the \(i\)-th user. To send the information word \(\dvec{u}^{\br{i}}\), the user first encodes it using the code \(\mathcal{C}_i({n'}, k, \dset{F}_i, \mathbf{f}_i)\), obtaining a codeword \(\mathbf{c}^{\br{i}}\). This is then followed by \ac{bpsk} modulation, resulting in \(\dvec{x}^{\br{i}} = \tau\br{\mathbf{c}^{\br{i}}}\). We refer the reader to~\cite{Marshakov2019Polar}, where the choice of information and frozen positions, along with frozen bit values, has been carefully optimized. In this section, we discuss how to decode a superposition of polar coset codes in a slot. In what follows, we assume a collision of order \(r\).

To specify the decoding algorithm, let us first consider how the channel output is represented. As in the joint \ac{ldpc} decoding algorithm, we consider a received signal position \(y\), omitting the index \(j\) in \(\dvec{y}_j\). At this position, each of the \(r\) users may transmit either zero or one, resulting in \(2^r\) possible transmitted signal values. These values can be indexed by a binary vector \(\dvec{b} \in \brc{0, 1}^r\), yielding \(2^r\) points. 

For the single-user case, the \ac{llr} represents a sufficient statistic for the received signal. With \(r\) simultaneous transmissions, we instead operate with a p.m.f. specified over \emph{tuples} of bits of length \(r\). The prior p.m.f. is determined by the received signal \(y\) at the given position as follows:
\[
{\mu}\br{\dvec{b}} \propto \eta\br{\dvec{b}},
\]
where \(\eta\br{\cdot}\) is defined in~\eqref{eq:gmac_joint:gauss_pdf}. The p.m.f. specified above can be obtained for each received position. Let us denote the channel output as a matrix \(\dmat{M} = \brs{\mu_1,\ldots,\mu_{n'}}\) of size $2^r\times n'$, where the $j$-th column represents the p.m.f. values $\mu$.

To describe the \ac{jsc} decoding procedure, we begin with a simple yet instructive example for \(n' = 2\) (see~\Fig{fig:PolarRepr}), illustrating the key idea behind the proposed method. We observe that, instead of working with bits from different users and multiple polar codes, we can treat the problem as working with a single polar code over \(\mathbb{Z}_2^r\). In this example, we first decode a bit configuration (tuple) \((u_1, v_1)\) -- the first bits of the users -- followed by a tuple \((u_2, v_2)\) -- the second bits of the users. Recall that the decoder operates with tuple distributions \(\dmat{M}\) rather than probabilities of individual bits.

\begin{figure}[t]
\centering
\includegraphics{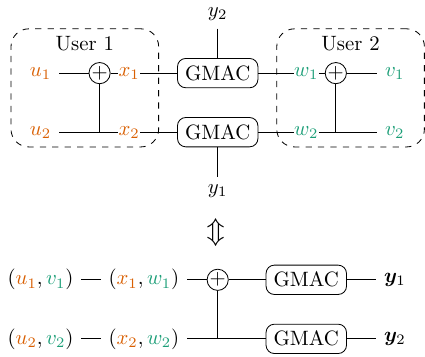}
\caption{Representation of the \ac{gmac} as a polar code over $\mathbb{Z}_2^r$ for $r=2$ and $n'=2$.}
\label{fig:PolarRepr}
\end{figure}

To decode the tuple \((u_1, v_1)\), we first need to calculate the distribution of the sum of two random variables over \(\mathbb{Z}_2^r\). In what follows, we refer to this operation as the \emph{\ac{cnop}}. This can be done via convolution, i.e., \(\hat{\mu}_1 = \mu_1 \ast \mu_2\). Since we work in the Abelian group \(\mathbb{Z}_2^r\), there exists a \ac{ft} \(\Phi\). In what follows, to perform a convolution, we use the \ac{fft}-based technique proposed in~\cite{Declercq} for the case of \ac{ldpc} codes over Abelian groups. Thus, the final rule is as follows:
\begin{equation}\label{eq:jsc_cnop}
\hat{\mu}_1\br{\dvec{b}} \propto \Phi^{-1}\br{\Phi\br{\mu_1} \odot \Phi\br{\mu_2}}\br{\dvec{b}} \quad \forall \dvec{b}\in \mathbb{Z}_2^r,
\end{equation}
where \(\odot\) denotes element-wise multiplication.

After calculating the p.m.f. \(\hat{\mu}_1\), we make a hard decision \(\hat{\dvec{b}}_1\), taking into account the values of frozen bits at this position. Once \(\hat{\dvec{b}}_1\) is found, we proceed with the \emph{\ac{vnop}}. The rule is given by:
\begin{equation}\label{eq:jsc_vnop}
\hat{\mu}_2\br{\dvec{b}} \propto \mu_1\br{\dvec{b} + \hat{\dvec{b}}_1} \mu_2\br{\dvec{b}} \quad \forall \dvec{b}\in \mathbb{Z}_2^r.
\end{equation}
Finally, the hard decision \(\hat{\dvec{b}}_2\) is made using \(\hat{\mu}_2\).

Now let us proceed to the case $n' > 2$. Similar to the single-user case, we consider a recursive decoding algorithm that relies on a polar code representation based on the $\br{U, U+V}$ construction. Let us define a polar code split procedure. Given a polar code $\mathcal{C}$, we specify two subcodes $\mathcal{C}_U$ and $\mathcal{C}_V$ as follows:
\begin{equation}\label{eq:polar_uupv_split}
\mathcal{C}\br{n, k, \dset{F}, \dvec{f}} \overset{\text{split}}{\rightarrow} \mathcal{C}_U\br{\nicefrac{n}{2}, k_U, \dset{F}_U, \dvec{f}_U}, \quad \mathcal{C}_V\br{\nicefrac{n}{2}, k_V, \dset{F}_V, \dvec{f}_V},
\end{equation}
where $\dset{F}_U \subseteq \brs{\nicefrac{n}{2}}$, $\dset{F}_V \subseteq \brc{\nicefrac{n}{2} + 1, \ldots, n}$, $\dset{F}_U\cup \dset{F}_V = \dset{F}$, $\dvec{f} = \br{\dvec{f}_U \: \dvec{f}_V}$, and $k = k_U + k_V$.

Now, let us specify the recursive decoding algorithm. We denote a matrix $\dmat{B} = \brs{\dvec{b}_1, \ldots, \dvec{b}_{n'}}$ of size $r\times n'$ specifying the estimated bit tuples. The \ac{jsc} is presented in Algorithm~\ref{alg:dec_jsc}, where \ac{cnop} and \ac{vnop} operations are applied per column with respect to the input arguments and are denoted as functions $\text{\acs{cnop}}\br{\mu_1, \mu_2}$ and $\text{\acs{vnop}}\br{\mu_1, \mu_2, \dvec{b}}$, respectively. Finally, $\oplus$ denotes an element-wise modulo-two sum.
\begin{algorithm}[t]
\caption{\Acf{jsc} decoding algorithm}\label{alg:dec_jsc}
\begin{algorithmic}[1]
\Function{\acs{jsc}}{$\mathcal{C}\br{n, k, \dset{F}, \dvec{f}}$, $\dmat{M}$} 
\If{$n =1$}
\Comment{Make a hard decision}
\State $\hat{\dvec{b}} \gets \text{hard}\br{\dset{F}, \dvec{f}, \dmat{M}}$
\State \Return $\hat{\dvec{b}}$, $\dmat{M}$
\Else
\Comment{Proceed recursively}
\State $\mathcal{C_U}, \: \mathcal{C}_V \gets \mathcal{C}$ \Comment Perform a split using~\eqref{eq:polar_uupv_split}
\State $\brs{\dmat{M}_U, \dmat{M}_V} \gets \dmat{M}$ \Comment Split into halves
\State $\hat{\dmat{M}}_U \gets \text{\acs{cnop}}\br{\dmat{M}_U, \dmat{M}_V}$ \Comment Per-column \ac{cnop}~\eqref{eq:jsc_cnop}
\State $\hat{\dmat{B}}_U, \: \tilde{\dmat{M}}_U \gets \text{\acs{jsc}}\br{\mathcal{C}_U, \hat{\dmat{M}}_U}$ \Comment Decode $\mathcal{C}_U$ recursively
\State $\hat{\dmat{M}}_V \gets \text{\acs{vnop}}\br{\dmat{M}_U, \dmat{M}_V, \hat{\dmat{B}}_U}$ \Comment Per-column \ac{vnop}~\eqref{eq:jsc_vnop}
\State $\dmat{B}_V, \: \tilde{\dmat{M}}_V \gets \text{\acs{jsc}}\br{\mathcal{C}_V, \hat{\dmat{M}}_V}$ \Comment Decode $\mathcal{C}_V$ recursively
\State $\hat{\dmat{B}} \gets \brs{\hat{\dmat{B}}_U\oplus \hat{\dmat{B}}_V, \hat{\dmat{B}}_V}$ \Comment Element-wise XOR
\State \Return $\hat{\dmat{B}}$, $\brs{\tilde{\dmat{M}}_U, \tilde{\dmat{M}}_V}$ \Comment All tuples of length $n$
\EndIf
\EndFunction
\end{algorithmic}
\end{algorithm}
Note that the structure of Algorithm~\ref{alg:dec_jsc} is exactly the same as that of the single-user polar code \ac{sc} decoding algorithm, where single bits are replaced by bit tuples, and \ac{llr}s are replaced by probability mass functions (p.m.f.) $\mu$.

\begin{remark}
It is worth noting that we can easily improve the decoding procedure by using the \ac{scl} decoding method~\cite{TalVardyList2015}. In other words, we take into account not only the most probable path (which, in our case, consists of tuples) but also $\ell$ different paths with the highest metric. To derive the path metric update, recall the hard decision step from Algorithm~\ref{alg:dec_jsc}:
$$
\hat{\dvec{b}} \gets \text{hard}\br{\dset{F}, \dvec{f}, \dmat{M}},
$$
which corresponds to the case where $n=1$. Hence, the matrix $\dmat{M}$ represents a single p.m.f. $\mu$. Thus, the path metric is updated by multiplying it with $\mu\br{\hat{\dvec{b}}}$.
\end{remark}

\section{Sparse \acl{idma}}\label{sec:ch5:idma}

In this section, we describe the scheme proposed in~\cite{AmalladinneJoint2019}. This scheme is very similar to \ac{irsas} with \ac{ldpc} codes, as presented above. Specifically, it implements the idea of sparsifying collisions, inherent in $T$-fold \ac{irsa}, while also utilizing users' \ac{ldpc} codes. However, the major difference lies in the randomization of the locations of the \ac{ldpc} codeword symbols. The scheme assumes random placement rather than slotted transmission.

Slotted transmission is well-suited for the basic \ac{irsa} protocol, as the receiver does not know the transmission graph and reconstructs it on the fly. Knowing the slot locations enables the receiver to decode. Now, consider the \ac{irsaf} protocol. The \ac{cs} part (or preambles) allows the receiver to first reconstruct the graph and then perform decoding. In this case, the main drawbacks of the scheme from~\cite{vem2019user} are as follows:
\begin{enumerate}
\item Slotted transmission depends on the existence of high-performance short-length codes. As the number of active users increases, the scheme ensures that the slot length decreases. Short \ac{ldpc} codes are known to perform poorly, while polar codes are good but lack flexibility in length adaptation. Thus, designing effective multi-user codes for short block lengths remains an open problem.
\item Once the message is decoded, its copies are peeled from other slots using \ac{sic}. Consequently, the decoding and peeling operations are separate. Moreover, peeling is performed in a hard decision manner. However, since the transmission graph is known, a joint message-passing decoder can be applied.
\item Slotted transmission is merely a specific interleaving strategy. The class of all possible permutations is much broader and potentially contains more effective solutions.
\end{enumerate}

The scheme from~\cite{AmalladinneJoint2019} addresses these issues. The authors demonstrate that a single sparse joint Tanner graph spanning all transmissions can significantly improve performance compared to the schemes in~\cite{vem2019user, Marshakov2019Polar}.

\begin{remark}
The scheme in~\cite{AmalladinneJoint2019} can be interpreted as a sparse version of \ac{idma}~\cite{IDMA2006}, adapted to the uncoordinated and unsourced \ac{mac} by incorporating an additional \ac{cs} component. The sparsity ensures acceptable computational complexity when decoding over the joint graph. We also note that the idea of using sparse access sequences appeared in~\cite{LDS}, although that work focused on short repetition codes.
\end{remark}

Now, let us describe the transmission process. Consider the $i$-th user and split the message into two parts: $W_i = [W_{i,p}, \: W_{i,c}]$, where $W_{i,p} \in [M_{p}]$, $W_{i,c} \in [M_{c}]$, and $M = M_p M_c$. Similar to the \ac{irsaf} protocol: (a) the replica count depends only on $W_{i,p}$, and (b) the preamble of length $n'_p$ is placed before the frame, so the total length $n$ is allocated to the preamble ($n_p$) and the remaining portion of length $n-n_p$ where users transmit their codewords.

In what follows, we consider the case $\log_2 M_p = k_p$ and $\log_2 M_c = k_c$, where $k_p, k_c \in \mathbb{N}$, with $k_p + k_c = k$. 
\begin{enumerate}
\item The preamble is chosen using the preamble encoder $f_p: [M_p] \to \mathbb{R}^{n_p}$. Recall that $\|f_p(W_p)\|^2_2 \leq P_p n_p$ for all $W_p \in [M_p]$.
\item All users utilize the same binary $[n_c, k_c]$ \ac{ldpc} code. The remaining $k_c$ message bits are encoded using this code, and the $i$-th user obtains the binary codeword $\mathbf{c}_{i}$. The user then performs \ac{bpsk} modulation, equivalently expressed as
\[
\tilde{\mathbf{x}}^{\br{i}} = \tau(\mathbf{c}^{\br{i}}), \quad \tau(\mathbf{c}^{\br{i}}) = (\tau(c^{\br{i}}_1), \ldots, \tau(c^{\br{i}}_{n_c})).
\]
\item Let $L = (n-n_p)/n_c$, which represents the maximum repetition count\footnote{Recall that this scheme does not use slots.}. Applying the function $\ell: [M_p] \to [L]$, we determine the replica count $\ell(W_{i,p})$. Once again, we emphasize that the replica count is a deterministic function of $W_{i,p}$ and remains independent of $W_{i,c}$. At this step, we construct the word
\[
\mathbf{x}^{\br{i}} = [\underbrace{\tilde{\mathbf{x}}^{\br{i}}, \: \tilde{\mathbf{x}}^{\br{i}}, \ldots, \: \tilde{\mathbf{x}}^{\br{i}}}_{\ell(W_{i,p})}, \: \mathbf{0}],
\]
where the codeword $\tilde{\mathbf{x}}^{\br{i}}$ is repeated $\ell(W_{i,p})$ times and then appended with $(n-n_p) - \ell(W_{i,p}) n_c$ zeros. Thus, $\mathbf{x}^{\br{i}}$ has length $n-n_p$.
\item Finally, we select the permutation $\pi_{W_{i,p}}$ based on the $W_{i,p}$ message part and transmit the word
$$
\left[f_p(W_{i,p}), \:\: \pi_{W_{i,p}}(\mathbf{x}^{\br{i}})\right]
$$
over the channel.
\end{enumerate}

Now, consider the decoding process:
\begin{enumerate}
\item Decode the preamble part using a \ac{cs} decoder by solving the problem~\eqref{eq:ch5:cs_preamble}, where $\mathbf{u}_p = (u_{p,1}, \ldots, u_{p,M_p})$ and the sparsity $\mathrm{wt}(\mathbf{u}_p) = K_a$.
\item Reconstruct the graph and extract $\{\pi_{\hat{W}_{i, p}}\}$ for $i \in [K_a]$. If all preambles are successfully recovered, the graph and code transformations are known.
\item Apply a joint message-passing decoder. The rules remain exactly the same as in Section~\ref{sec:dec_ldpc}, with the only difference being the graph structure.
\end{enumerate}

\begin{remark}
Similar to the \ac{irsaf} protocol, we have
\[
P_p n_p + P_c n_c L'\br{1} \leq Pn.
\]
\end{remark}

Additionally, we highlight an approach to the \ac{ura} problem that employs a tensor-based modulation method~\cite{Guillaud2021}. Given a large coherence block, this method is also applicable to the block fading model and scenarios where the receiver is equipped with multiple antennas.

\section{\Acl{ccs}}\label{chap5:ccs_description}

In this section, we describe an inspirational scheme from~\cite{amalladinne2018coupled, codedCS2020}. A similar approach has already been used in \ac{cs} and group testing literature~\cite{Cormode2006, Hung2012, Gilbert2007, Indyk2008}. However,~\cite{codedCS2020} presents the first application of this approach to the \ac{ura} problem. Recall that the \ac{ura} problem is a \ac{cs} problem of immense dimension. The scheme from~\cite{codedCS2020} utilizes the divide-and-conquer strategy, i.e., it splits the task into subtasks of smaller dimensions, solves the corresponding \ac{cs} problems, and then assembles the results using an outer tree code. We note that a similar code construction, namely a convolutional code, was used in~\cite{ZigJ}, but for a different single-user channel model (jamming channel or J-channel). The main drawback of the tree code-based scheme proposed in~\cite{codedCS2020} is its inability to deal with errors, i.e., the codeword is not recovered if at least one of its fragments is lost. This drawback was addressed in~\cite{Frolov2022TCOM}. Our description is based on the latter paper.

\begin{figure}[!t]
\centering
\includegraphics{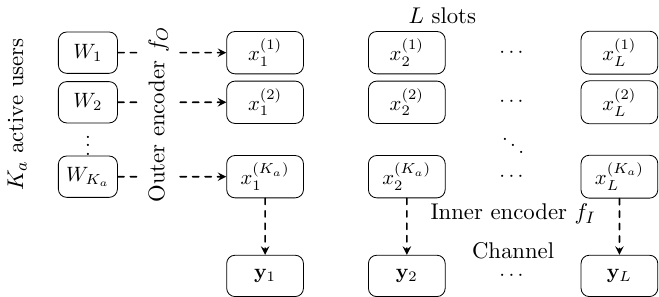}
\caption{\ac{ccs} scheme. First, the outer code of length $L$ over the alphabet $[Q]$, with the encoder $f_O$, is applied. Then, the inner encoder $f_I$ is applied to each symbol of the outer code. Finally, each inner codeword is transmitted through the Gaussian channel~\eqref{eq:block_fading}.}
\label{fig:ccs_scheme}
\end{figure}

Let us briefly describe a \ac{ccs} scheme from~\cite{codedCS2020}. The transmission scheme is shown in~\Fig{fig:ccs_scheme}. The idea is to apply a divide-and-conquer strategy implemented via concatenated coding. Recall that a frame of $n$ real channel uses is divided into $L$ slots of length $n'$. Consider the $i$-th user aiming to transmit a message $W_i \in [M]$. We use an outer code of length $L$ over the alphabet $[Q]$ with the encoder $f_O: [M] \to [Q]^L$. First, we obtain a codeword $\ind{\dvec{x}}{i}{} = (\ind{x}{i}{1}, \ind{x}{i}{2}, \ldots, \ind{x}{i}{L}) = f_O(W_i)$, where $\ind{\dvec{x}}{i}{} \in [Q]^L$. Then, the inner encoder $f_I: [Q] \to \mathbb{C}^{n'}$ is applied to each symbol $\ind{x}{i}{l}$, $l\in[L]$. The resulting codewords of the inner code are transmitted in the corresponding slots. Note that, in contrast to ALOHA protocol-based schemes, the user's transmission occupies the entire frame ($L$ slots).

For the $l$-th slot, we have 
\begin{equation}\label{eq:block_fading}
\rvec{y}_l = \sum\limits_{i=1}^{K_a} f_I\left(\ind{x}{i}{l}\right) + \rvec{z}_l, \quad l \in [L],
\end{equation}
where $\rvec{z}_l \sim \mathcal{N}(\bm{0},\dmat{I}_{n'})$.

To recover the transmitted codewords, we first solve a \ac{cs} problem for each slot. Note that the dimensions of these problems are much smaller compared to the dimension of the original problem (see~\eqref{eq:original_cs}). Thus, one can use standard \ac{cs} algorithms. For further details, see Section~\ref{sec:ch5:num_results_ccs}.

\begin{figure}[t]
\centering
\includegraphics{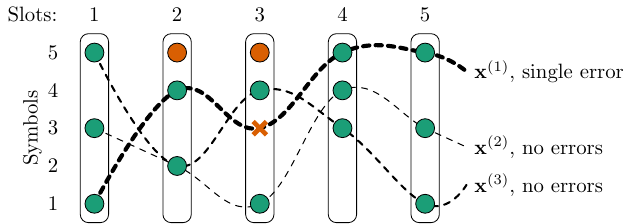}
\caption{Outer code and the list-recovery problem (Example~\ref{ex:frame_decoded}, case B). Rectangles represent slots. Green circles correspond to correctly detected symbols, orange circles correspond to falsely detected symbols, and orange crosses correspond to missed symbols. Dashed lines represent outer code codewords. For example, a codeword $\dvec{x}^{\br{1}}$ is covered by the received lists in all but one position.}
\label{fig:t_error_code}
\end{figure}

In the following example, we illustrate the transmission and the output of the inner decoder.

\begin{example}\label{ex:frame_decoded}
Let $K_a = 3$, $Q=L=5$ and assume the codewords 
\begin{eqnarray*}
\ind{\bx}{1}{} &=& (1, \ 4,\ 3,\ 5,\ 5),\\ 
\ind{\bx}{2}{} &=& (3, \ 2,\ 1,\ 4,\ 3),\\
\ind{\bx}{3}{} &=& (5, \ 2,\ 4,\ 3,\ 1).
\end{eqnarray*}
are transmitted. The channel output for each slot is described by~\eqref{eq:block_fading}. For example, for the third slot, we have the following channel output realization:
\[
\dvec{y}_3 = f_I\left( 3 \right) + f_I\left( 1 \right) + f_I\left( 4 \right) + \dvec{z}_3.
\]

Next, we solve a \ac{cs} problem for each slot and obtain $\dset{Y}_l \subseteq [Q]$, $l \in [L]$ which represents the estimates of the sets of $Q$-ary symbols transmitted in each slot. Let us consider two cases:

\textbf{Case A\textup:} There are no errors in the received lists. Thus,
\begin{flalign*}
\dset{Y}_1 = & \{1,3,5\}, \\
\dset{Y}_2 = & \{2,4\},   \\
\dset{Y}_3 = & \{1,3,4\}, \\
\dset{Y}_4 = & \{3,4,5\}, \\
\dset{Y}_5 = & \{1,3,5\}.
\end{flalign*}
We note that the decoder returns a set, and thus we have a list of size $2$ for the second slot. There is no information about the symbols' multiplicities.

\textbf{Case B\textup:} There are errors in the received lists. Let 
\begin{flalign*}
\dset{Y}_1 = & \{1,3,5\}, \\
\dset{Y}_2 = & \{2,4,5\}, \\
\dset{Y}_3 = & \{1,4,5\}, \\
\dset{Y}_4 = & \{3,4,5\}, \\
\dset{Y}_5 = & \{1,3,5\}.
\end{flalign*}
For the second slot, we observe a falsely detected symbol, $5$. For the third slot, the symbol $3$ is lost, and the symbol $5$ is falsely detected.
\end{example}

The task of the outer code is to assemble the original codewords from the received lists. The lists may contain errors (missed and falsely detected symbols). This problem is known as the list-recovery problem~\cite{Guru}. We illustrate this process for Example~\ref{ex:frame_decoded}, case B, in~\Fig{fig:t_error_code}.

\subsection{Channel for the outer code}\label{sec:outer_code_channel}

Let us start with the case with no errors in the output lists. Clearly, the resulting channel $\mathfrak{W}_A$ is the channel without intensity information (A-channel) from~\cite{Wolf1981}, which is also called a hyperchannel in the literature (see~\cite{Hyper}). Let the symbols $\rscalar{x}^{\br{1}}, \rscalar{x}^{(2)}, \ldots, \rscalar{x}^{(K_a)} \in [Q]$ be transmitted. Then, the channel output is 
$$
\dset{Y}^{(A)} = \bigcup\limits_{i=1}^{K_a} \rscalar{x}^{\br{i}}.
$$
The capacity of the A-channel is derived in~\cite{BasPin00}. If we consider indicator vectors of sets, this channel can be represented as a vector OR-channel.

Our channel $\mathfrak{W}$ is a concatenation of the A-channel $\mathfrak{W}_A$ with the channel $\mathfrak{W}_E$ that introduces errors. Let us consider the channel $\mathfrak{W}_E$ in more detail. The input of the channel is the set of transmitted symbols $\dset{Y}^{(A)} \subseteq [Q]$. Let $\dset{Y} \subseteq [Q]$ be the output set. We introduce the missed detection probability $p_m$ as the probability that the transmitted symbol is not in the received list $\dset{Y}$. The false alarm probability $p_f$ is the probability that some symbol that is not transmitted appears in the received list, i.e., for $x \in [Q]$,
\begin{flalign*}
\pM &= \prob{x \not\in \dset{Y} | x \in \dset{Y}^{(A)}},\\
\pF &= \prob{x \in \dset{Y} | x \not\in \dset{Y}^{(A)}},
\end{flalign*}
where the channel operates independently on the elements.

\begin{example}
Let us consider Example~\ref{ex:frame_decoded} and the second slot with $\dset{Y}_2^{(A)} = \{2,4\}$. Each element $\{2,4\}$ can be missed independently with probability $p_m$, and each element $\{1,3,5\}$ can appear in the received list $\dset{Y}_2$ independently with probability $p_f$.
\end{example}

The capacity of this channel is derived in Appendix~\ref{sec:outer_channel_capacity}.

\subsection{\Acl{rcb}}\label{sec:random_coding_bound}

In this section, we present a \ac{rcb} for the outer code. In what follows, we consider the \emph{ensemble} of codes.

\begin{define}
Let $\mathcal{E}_1(M, L)$ be the ensemble of codebooks of size $M \times L$, where each element is sampled i.i.d. from $\mathrm{Unif}\br{[Q]}$.
\end{define}

Now, let us describe the decoding algorithm. Let
$$
\vdset{Y} = (\dset{Y}_1, \ldots, \dset{Y}_L), \quad \dset{Y}_l \subseteq [Q],
$$
i.e., the received sequence of $L$ sets, each set being an output of the channel $\mathfrak{W}$.

We require the decoder\footnote{Here and in what follows, we consider the outer code only and thus omit the word ``outer''.} to output all the messages $W$, such that
\begin{equation}\label{eq:dec_cond}
d(\vdset{Y}, \dvec{x}) \leq t,
\end{equation}
where $ \dvec{x}= f_O(W)$, $d(\vdset{Y}, \dvec{x}) = | \{ l: x_l \not\in \dset{Y}_l \}|$.

\begin{theorem}[\Acl{ccs} \acl{rcb}~\cite{Frolov2022TCOM}]\label{thm:rcb}
There exists a code $\code \in \mathcal{E}_1(M, L)$, such that
\begin{equation}\label{eq:rcb_pupe}
P_e = \sum\limits_{i=t+1}^{L} \binom{L}{i} \pM^i (1-\pM)^{L-i},
\end{equation}
and
\begin{equation}\label{eq:rcb_far}
P_f \leq \sum\limits_{r=1}^{K_a} \left[\nu_r (M-r) \sum\limits_{i=0}^{t} \binom{L}{i} (1-\mu_r)^i \mu_r^{L-i} \right],
\end{equation}
where $\mu_r$ is given by~\eqref{eq:py} and
$$
\nu_r = \frac{M!}{(M-r)!M^{K_{a}}} \stirling{K_{a}}{r},
$$
where $\stirling{K_{a}}{r}$ is the Stirling number of the second kind.
\end{theorem}

\begin{proof}
Let us start with the \ac{far}. Recall that $\dset{T}$ is the set of transmitted messages and ${\dset{R}} \subseteq [M]$ is a received list. Let us introduce the events
\[
E_r = \{ |\dset{T}| = r \}, \:\: r \in [K_a].
\]
Note that the probability $\prob{E_r}$ of obtaining $r$ unique original elements when $K_{a}$ items are taken with replacement from a sample of size $M$ is equal to $\nu_r$ (see, e.g.,~\cite{knuthConcrete2nd}).

We have
\[
P_f = \sum\limits_{r=1}^{K_a} \nu_r \prob{{\dset{R}} \backslash {\dset{T}} \ne \emptyset \: | \: E_r}.
\]

Let us proceed with $\prob{{\dset{R}} \backslash {\dset{T}} \ne \emptyset \: | \: E_r}$. W.l.o.g., let us assume that $\dset{T} = [r]$ and calculate the probability that some another message satisfies condition~\eqref{eq:dec_cond}. Let us fix the message $\hat{W} \in [M] \setminus [r]$, and let $\hat{\dvec{x}} = f_O(\hat{W})$. Recall (see Section~\ref{sec:outer_code_channel}), 
\[
\prob{\hat{x}_l \in \dset{Y}_l \: | \: E_r} = \mu_r,
\]
thus the probability of accepting the message $\hat{W}$ is equal to
$$
\prob{\sum\limits_{l=1}^L \xi_l \leq t},
$$
where $\xi_l \overset{i.i.d.}{\sim} \mathrm{Bern}(1-\mu_r)$. Applying the union bound, we obtain $P_f$ from the theorem statement.

Finally, note that as we utilize a single-user receiver, $P_e$ is just the probability that more than $t$ errors in the transmitted codeword have occurred.
\end{proof}

In what follows, we are interested in codebooks of size $M \approx 2^{100}$ and utilize the following upper bound for $P_f$.

\begin{corollary}
The following inequality holds:
\[
P_f \leq (M-K_a) \sum\limits_{i=0}^{t} \binom{L}{i} (1-\mu_{K_a})^i \mu_{K_a}^{L-i} + p',
\]
where
\[
p' = \prob{\left|\dset{S}\right| < K_a} = 1 - \prod\limits_{i=0}^{K_a-1} \left( 1 - \frac{i}{M} \right) \leq \frac{\binom{K_a}{2}}{M}.
\]
\end{corollary}

\subsection{\texorpdfstring{$t$}{t}-tree code-based practical scheme} \label{sec:treecode_practical}

Let us describe and analyze a practical code construction, which is a modification of the tree code. The code structure is the same as in~\cite{codedCS2020}, with a crucial difference: a decoder capable of correcting up to $t$ errors.

\paragraph{Code construction.}

Let us represent the user message $W$ as a binary $k$-bit vector $\dvec{u}$ and split it into chunks $\dvec{u}^T = (\dvec{u}^T_1, \ldots, \dvec{u}^T_L)$, such that $\dvec{u}_l$ is of length $b_l$ bits, $l\in [L]$, and $\sum\limits_{l=1}^{L} b_l = k$.

Recall that $\beg{\dvec{u}}{l} = (\dvec{u}_1, \ldots, \dvec{u}_l)$, and let $B_l = \sum\nolimits_{l'=1}^l b_{l'}$. To construct the outer code, we choose the following encoding function $f_O$.
\begin{equation}\label{tree:def_nl}
x_l = f_{O, l}(\beg{\dvec{u}}{l}), \:\: l\in[L],
\end{equation}
where $f_{O, l}: \{0,1\}^{B_l} \to [Q]$. The main idea of the proposed code construction is that the symbol $x_l$ for the $l$-th slot depends only on the message chunks $\dvec{u}_1, \ldots, \dvec{u}_l$. This property simplifies the decoding process (see Section~\ref{tree:dec}).

Linear codes are preferred for practical schemes. Thus, we construct the functions $f_{O,l}(\cdot)$, $l\in[L]$ as follows. Let $c \in \mathbb{N}$, $Q=2^c$. Let us fix a bijective mapping $\phi: \{0,1\}^{c} \to [Q]$. The major part of our construction is a binary linear code with a block upper-triangular generator matrix.
\begin{equation}\label{eq:G}
\dmat{G}= 
\begin{pmatrix}
\dmat{G}_{1,1} & \dmat{G}_{1,2} & \dmat{G}_{1,3} & \dots & \dmat{G}_{1,L}\\
\dmat{0} & \dmat{G}_{2,2} & \dmat{G}_{2,3} & \dots & \dmat{G}_{2,L} \\
\dmat{0} & \dmat{0} & \dmat{G}_{3,3} & \dots & \dmat{G}_{3,L}\\
\vdots & \vdots & \vdots & \dots & \vdots \\
\dmat{0} & \dmat{0} & \dmat{0} & \dots & \dmat{G}_{L,L}\\
\end{pmatrix},
\end{equation}
where $\dmat{G}_{l',l}$, $l', l \in[L]$, is a binary matrix of size $b_{l'} \times c$. 

The codeword $\dvec{x} = (x_1, \ldots, x_L) \in [Q]^L$ is obtained as follows. We start with the binary vector $\dvec{c} = \dvec{u} \dmat{G}$ and then obtain a codeword $\dvec{x}$ by splitting $\dvec{c}$ into chunks of length $c$ and applying the mapping $\phi$, i.e.,
$$
x_l = \phi\left( \sum\limits_{l'=1}^{l} \dvec{u}^T_{l'} \dmat{G}_{l',l} \right), \:\: l\in[L].
$$

\paragraph{Decoding.}\label{tree:dec}

Recall that the goal of the decoder is to recover all the messages $\dvec{u}$, such that
\[
d(\vdset{Y}, \dvec{x} ) \leq t,
\]
where $\dvec{x} = f_O(\dvec{u})$.

We note that $\dvec{u}_{[l]}$ uniquely defines $\dvec{x}_{[l]}$ for each $l \in [L]$. In what follows, we write $\dvec{x}_{[l]} = f_O(\dvec{u}_{[l]})$. This fact allows us to utilize a low-complexity decoding algorithm, which decodes the blocks $\dvec{u}_l$ sequentially.

Let
\[
\dset{V}_l = \left\{ \dvec{v}_l \in \{0,1\}^{B_l} : d(\vdset{Y}_{[l]}, f_O(\dvec{v}_l)) \leq t \right\}, \:\: l \in [L],
\]
which represents the list of messages at each decoding step.

Let us consider the $l$-th decoding step. For each of the elements $\dvec{v}_{l-1} \in \dset{V}_{l-1}$, we consider all possible continuations $\dvec{u}_l \in \{0,1\}^{b_l}$. We include the path $\dvec{v}_l = (\dvec{v}_{l-1}, \dvec{u}_l)$ into the set $\dset{V}_{l}$ if $d(\vdset{Y}_{[l]}, f_O(\dvec{v}_l)) \leq t$. See Algorithm~\ref{alg:dec} for the full description.

\begin{algorithm}
\caption{$t$-tree linear code decoding algorithm}\label{alg:dec}
\begin{algorithmic}[1]
\INPUT $\vdset{Y}$
\OUTPUT $\dset{V}_L$ \Comment Decoded messages
\State $\dset{V}_0 \gets \emptyset$
\For {$l \in [L]$} 
\State $\dset{V}_l \gets \emptyset$ 
\For {$\dvec{v}_{l-1} \in \dset{V}_{l-1}$} \Comment For each element from the list
\For {$\dvec{u}_l \in \{0,1\}^{b_l}$} \Comment For each next block
\State $\dvec{v}_l \gets (\dvec{v}_{l-1}, \dvec{u}_l)$
\If { $d(\vdset{Y}_{[l]}, f_O(\dvec{v}_l)) \leq t$ }
\State $\dset{V}_l \gets \dset{V}_l \bigcup \dvec{v}_l$
\EndIf 
\EndFor
\EndFor
\EndFor
\State \Return $\dset{V}_L$
\end{algorithmic}
\end{algorithm}

\begin{remark}
As we see, Algorithm~\ref{alg:dec} is guaranteed to recover the transmitted message in case no more than $t$ errors have occurred. 
\end{remark}

\begin{remark}
Note that the decoding complexity depends on $|\dset{V}_l|$, $l \in [L]$. Indeed, at the $l$-th decoding step we perform $|\dset{V}_{l-1}| 2^{b_l}$ encodings. Taking into account that the encoding requires $B_l \times c$ additions in the binary field, we obtain the total decoding complexity of $c \times \left(\sum\nolimits_{l=1}^L |\dset{V}_{l-1}| 2^{b_l} B_l\right)$ binary additions. 

See Appendix~\ref{appendix:t_tree_linear_paths} for a detailed analysis of the average number of decoding paths.
\end{remark}

\subsection{\Acl{rs} code-based practical scheme} \label{sec:reed_solomon_practical}

\ac{rs} codes combined with Guruswami--Sudan decoding algorithm are known to solve the list-recovery problem. In this section, we describe an \ac{rs}-code-based practical scheme from~\cite{Frolov2022TCOM}. We start with the necessary definitions. 

Recall that $\mathbb{F}_Q$ is the field with $Q$ elements. Let $\mathbb{F}_Q[X]$ and $\mathbb{F}_Q[X,Y]$ denote the rings of univariate and bivariate polynomials over $\mathbb{F}_Q$. Let $\beta_1, \beta_2, \ldots, \beta_L \in \mathbb{F}_Q$, with $\beta_i \ne \beta_j$ for $i \ne j$. We define an $[n_O = L, k_O]$ \ac{rs} code $\mathcal{C}$ as follows:
$$
\mathcal{C}_{RS} \triangleq \left\{ ( f(\beta_1), f(\beta_2), \ldots, f(\beta_L) ) : f\br{x} \in \mathbb{F}_Q[X], \deg{f\br{x}} < k_O \right\}.
$$

Let us consider the bivariate polynomial
$$
g(x, y) = \sum\limits_{i=0}^\infty \sum\limits_{j=0}^\infty g_{i,j} x^i y^j \in \mathbb{F}_Q[X,Y].
$$
We define the $(w_1, w_2)$-weighted degree of $g(x, y)$ as follows.
\[
\deg_{w_1, w_2} g(x, y) = \max\limits_{g_{i,j} \ne 0} \deg_{w_1, w_2} (x^i y^j),
\]
where the $(w_1, w_2)$-weighted degree of the monomial $x^i y^j$ equals to $w_1 i + w_2 j$. Note that the $(1, 1)$-weighted degree is simply the degree of a bivariate polynomial.

In what follows, we search for polynomials that pass through some point with multiplicity $m$. The bivariate polynomial $g(x,y)$ passes through the point $(\beta, \alpha) \in \mathbb{F}_Q^2$ if $g(\beta, \alpha) = 0$. Let us introduce the concept of root multiplicity. In the real field, the root multiplicity can be defined using partial derivatives, but for the finite field it is more challenging due to finite characteristic. Following~\cite{GS}, we say that the polynomial $g(x,y)$ passes through the point $(\beta, \alpha)$ with multiplicity $m$ if the shifted polynomial $g(x + \beta, y + \alpha)$ contains a monomial of degree\footnote{Here, we refer to the $(1, 1)$-weighted degree.} $m$ and does not contain a monomial of smaller degree.

\paragraph{Naive approach.}\label{sec:naive}

Let us apply the Guruswami--Sudan list recovery algorithm to our problem in a straightforward manner. We briefly explain the idea and refer the reader to~\cite{PPM3} for details. Let us enumerate the elements of the field $\mathbb{F}_Q$ in some order, i.e., $\mathbb{F}_Q = \{\alpha_1, \alpha_2, \ldots, \alpha_Q\}$.

\emph{Encoding.} Let $k_O = k/\log_2Q$. Consider the user message $W$ as a vector$\dvec{f} = (f_0, \ldots, f_{k_O-1})$, where $f_i \in \mathbb{F}_Q$, and introduce the information polynomial $f\br{x} = \sum\nolimits_{i=0}^{k_O-1} f_i x^i$. The user codeword is
$$
\dvec{x} = \br{f(\beta_1), f(\beta_2), \ldots, f(\beta_L)} \in \mathcal{C}_{RS}.
$$

\emph{Decoding.} Let us recall that the channel output is $\vdset{Y} = (\dset{Y}_1, \ldots, \dset{Y}_L)$, where $\dset{Y}_l \subseteq \mathbb{F}_Q$. Let us introduce the following set of points:
\[
\dset{P} = \left\{ (\beta_l, \alpha_i): \alpha_i \in \dset{Y}_l, l \in [L], i \in [Q] \right\}.
\]

We assign a multiplicity $m_{i,l} \in \mathbb{N}$ to each point $(\beta_l, \alpha_i) \in \dset{P}$. We note that multiplicities affect the decoding algorithm only. The larger $m_{i,l}$ are, the better the decoding performance is. At the same time, the decoding complexity grows with $m_{i,l}$. We provide all the details below.

To perform the decoding, we apply the Guruswami--Sudan algorithm. The input to this algorithm is the set of points $\dset{P}$ with the corresponding multiplicities $m_{i,l}$. The full description is given by Algorithm~\ref{alg:GS}.

\begin{algorithm}
\caption{Guruswami--Sudan algorithm}\label{alg:GS}
\begin{algorithmic}[1]
\INPUT $L$, $k_O$, set of points $\dset{P}$ with multiplicities $m_{i,l}$
\OUTPUT List of polynomials $f\br{x}$

\State \emph{Interpolation}. Find a bivariate polynomial $g(x,y)$ of minimal $(1,k_O-1)$-weighted degree that passes through each point $(\beta_l, \alpha_i)$, $i\in[Q]$, $l \in [L]$, with multiplicity $m_{i,l}$. 
\State \emph{Factorization}. Find all the factors of $g(x,y)$ of the form $y - f\br{x}$, where $\deg{f\br{x}} < k_O$.
\end{algorithmic}
\end{algorithm}

\emph{Analysis.} It is convenient to represent the set $\dset{P}$ with corresponding multiplicities as a matrix $\dmat{M} = \left(m'_{i,l}\right)$ of size $Q\times L$ as follows: $m'_{i,l} = m_{i,l}$ if $(\alpha_i, \beta_l) \in \dset{P}$ and $m'_{i,l} = 0$ otherwise.

Let us define the complexity of matrix $\dmat{M}$ as
\[
C(\mathbf{M}) = \frac{1}{2} \sum\limits_{i=1}^Q \sum\limits_{l=1}^L m'_{i,l} (m'_{i,l}+1). 
\]

Algorithm~\ref{alg:GS} is known (see~\cite{GS, KV}) to include $f\br{x}$ in the output list if
\begin{equation}\label{eq:rec_cond}
(\mathbf{M}, \mathbf{X}) \geq \sqrt{2(k_O-1)C(\mathbf{M})},
\end{equation}
where $(\cdot , \cdot)$
denotes the dot product of matrices, and the matrix $\mathbf{X} = \left(x_{i,l}\right)$ corresponds to the codeword $\mathbf{x} = \left(f\left(\beta_1\right), \ldots, f\left(\beta_L\right)\right)$, i.e., $x_{i,l} = 1$ if $\alpha_i = f(\beta_l)$. This matrix has only one unit in each column. 

The following upper bound holds for the list size:
\[
L(\mathbf{M}) \leq \sqrt{\frac{2 C(\mathbf{M})}{k_O-1}}.
\]

Let us investigate the recovery condition~\eqref{eq:rec_cond} in more detail for the case when all the points from $\dset{P}$ have the same multiplicity equal to $m$. Assume also that we have lists of size $K_a$ in each position. Then we have
\[
m (L - d(\vdset{Y}, \dvec{x})) = (\mathbf{M}, \mathbf{X}) \geq \sqrt{(k_O-1) m(m+1)K_a L},
\]
and thus we can recover the codeword if the number of errors $t$ satisfies the inequality
\[
t = d(\vdset{Y}, \dvec{x}) \leq L \left( 1 - \sqrt{\frac{(k_O-1)}{L} K_a \frac{(m+1)}{m}} \right).
\]

As we can see, to have error-correcting capabilities, we must require the outer code rate $R_O \triangleq k_O / L \leq 1 / K_a$, which is infeasible even for a moderate number of users. Thus, our next goal is to reduce the order of collisions.

\paragraph{Modified scheme.}\label{sec:rs_modified}

\begin{figure}[t]
\centering
\includegraphics{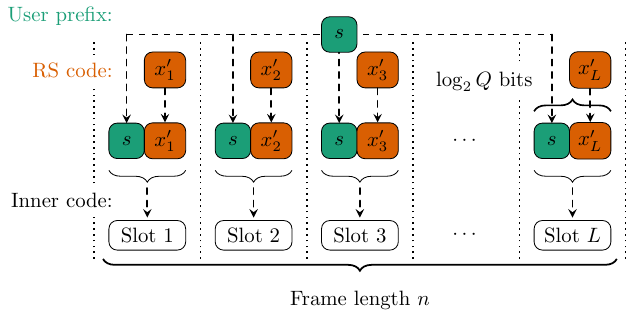}
\caption{Frame structure with slots consisting of \ac{rs} code symbols $x'_l$, where $l\in [L]$, and the prefix $s$.}
\label{fig:rs_scheme}
\end{figure}

Let $Q = q_p \times q$ and $b_p = \log_2{\left(q_p\right)}$. We propose to use $b_p$ bit prefixes $s_i$, $i \in [K_a]$ to group $K_a$ users in $q_p$ groups, ensuring an acceptable level of collisions within each group.

\emph{Encoding.} Consider the user aiming to transmit a $k$ bit message $W$. First $b_{C}$ \ac{crc} bits are calculated and added to the message to suppress the \ac{far} below the required threshold. Each user encodes $k + b_C$ bits by using an $[L,k_O]$ \ac{rs} code $\dset{C}'_{RS}$ over the alphabet $\mathbb{F}_q$ and obtains the codeword ${\dvec{x}}' \in \dset{C}'_{RS}$. The user then generates the prefix $s \sim \mathrm{Unif}([q_p])$. Consider the bijective mapping $\phi': [q_p]\times[q] \to [Q]$, which is common to all users. The resulting codeword $\dvec{x} \in [Q]^L$ is calculated as follows:
\[
x_l = \phi'(s, x'_l), \:\: l \in [L].
\]
The encoding process is illustrated in~\Fig{fig:rs_scheme}. 

\emph{Decoding.} Prefixes split the set $[Q]$ into non-intersecting ranges $\dset{I}_i$, where $|\dset{I}_i| = q$, $i \in [q_p]$. We handle the ranges $\dset{I}_i$ separately. Note that by applying $\phi'^{-1}$ and removing the prefix, we reduce the problem to the one described in Section~\ref{sec:naive}, but with alphabet $\mathbb{F}_q$ and a smaller collision order.

\subsection{Further modifications}

The paper~\cite{fengler2020sparcs} continues this line of work and utilizes \ac{sparcs} combined with the \ac{amp} algorithm. A further modification of the \ac{ccs} approach is presented in~\cite{AmalladinneAMP2021}, where the inner \ac{amp} decoder and the outer tree decoder are allowed to exchange soft information within a joint \ac{mpa}. In~\cite{EbertCouplingCS2021}, the authors reduce complexity by utilizing the same sensing matrix as in~\cite{codedCS2020}. The main difference lies in the use of a joint \ac{mpa} (as in~\cite{AmalladinneAMP2021}) and the assumption that inner codewords are coupled via an outer code. Finally, the paper~\cite{Amalladinne2022} proposes a coded demixing approach to enable the joint recovery of very high-dimensional signals that are sparse with respect to separate bases.

We also note the papers~\cite{Fengler2022_list_prune, Ebert2023_lloop}, which further investigate coding for the A-channel. In particular, the paper~\cite{Fengler2022_list_prune} proposes an interesting list-pruning algorithm that may be extremely helpful for tree-code decoding.

The state-of-the-art \ac{cs}-based solution proposed in~\cite{Truhachev2024} suggests eliminating the outer code entirely. The main idea is to use data from previous slots to customize the sensing matrix for subsequent slots. This scheme demonstrates the best energy efficiency among all \ac{cs}-based approaches.

\section{Random spreading} \label{sec:ch5:randon_spreading}

Previously, we discussed \ac{irsa} and sparse \ac{idma} schemes. Simply put, such schemes rely on a \ac{tdma} strategy to sparsify collisions. From the \ac{mac} literature, we know of one more technique---\ac{dsss}, which serves as the basis for the \ac{cdma} method. This section is devoted to applying this strategy to the \ac{ura} scenario. In what follows, we describe in detail the scheme from~\cite{AmalladinnePolar2020} and briefly discuss the improvements introduced in~\cite{Duman2021, PradhanLDPCSC2021}.

Let us start with an informal description of the scheme. The payload corresponding to each active user is split into two parts. The first component acts as a preamble, selecting a spreading or signature sequence from a codebook of sequences with good correlation properties. The remaining bits are encoded using the user's code. This codeword is then spread using the signature sequence chosen by the first part of the message.

The decoding procedure is iterative, with each iteration consisting of four main steps:
\begin{itemize}
\item[(a)] an energy detector that utilizes the correlation properties of the signature sequences to identify the list of spreading sequences used,
\item[(b)] a \ac{mmse} estimator that produces \ac{llr} estimates based on the spreading sequences,
\item[(c)] a single-user decoding operation, and
\item[(d)] a \ac{sic} step, resulting in a residual channel output.
\end{itemize}

We note that the basic scheme relies on single-user decoding, which drastically reduces its complexity.

The codebook of spreading sequences plays a crucial role in determining the scheme’s error performance. These sequences should exhibit good correlation properties or, equivalently, form a good \ac{cs} code (see Appendix~\ref{a:cs_methods} for details on how to construct such codes). When the size of this codebook is large, active users are less likely to experience collisions in the space of signature sequences, a scenario that favors the use of a single-user decoder. However, the mutual coherence of such a codebook would be very high, adversely affecting the performance of the polar decoder. Conversely, smaller codebooks possess better correlation properties but result in a higher probability of collisions between the signature sequences chosen by active users. This, in turn, renders single-user decoding ineffective for users whose signature sequences collide.

Now, we proceed with a detailed description, starting with the encoding process.

\subsection{Encoder}\label{sec:encoder}

\begin{figure}
\centering
\includegraphics{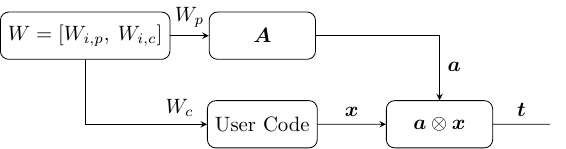}
\caption{The message is split into two parts, $W_p$ and $W_c$. The component $W_p$ is responsible for selecting a spreading sequence, while $W_c$ represents the message encoded by a user code. Spreading sequences are chosen from the matrix $\dmat{A}$, and the encoded vector $\dvec{t}$ is given by~\eqref{eq:ch5_spread_rkon}.}
\label{fig:encoder}
\end{figure}

Consider the $i$-th user. Similar to previous schemes (see Section~\ref{ch5:sec_irsa_sf}), we split the message into two parts:
$$
W_i = [W_{i,p}, \: W_{i,c}],
$$
where $W_{i,p} \in [M_{p}]$, $W_{i,c} \in [M_{c}]$, and $M = M_p M_c$. Recall that the message parts $W_{i,p}$ and $W_{i,c}$ consist of $k_p$ and $k_c$ bits, respectively, i.e., $M_p = 2^{k_p}$ and $M_c = 2^{k_c}$. 

The message part $W_{i,p}$ encodes the spreading sequence $f_p(W_{i,p})$, where $f_p: [M_p] \to \mathbb{R}^{n_{p}}$ is the encoder of some good \ac{cs} code. Let us introduce the codebook (or sensing matrix) of \ac{cs} code
\[
\dmat{A} = \left[\dvec{a}_1, \ldots, \dvec{a}_{M_p}\right],
\]
where $\dvec{a}_W = f_p(W)$ for $W \in [M_p]$. The encoding process consists of selecting the column indexed by $W_{i,p}$ from the matrix $\dmat{A}$. In other words, the message part $W_{i,p}$ determines the choice of the spreading sequence. Throughout this work, we denote the spreading sequence chosen by the $i$-th user as
$$
\dvec{a}^{\br{i}} \triangleq \dvec{a}_{W_{i, p}} = f_p\br{W_{i, p}}.
$$

The message part $W_{i,c}$ is encoded into the user's codeword $\dvec{x}^{\br{i}} = f_c(W_{i,c})$, where $f_c: [M_c] \to \mathbb{R}^{n_{c}}$. In the numerical analysis presented in Section~\ref{sec:ch5:num_results_rspread}, we consider a \ac{bpsk}-modulated binary linear $[n_c, k_c]$-code $\mathcal{C}$, i.e.,
$$
\dvec{x}^{\br{i}} = \tau\br{f_c\br{\dvec{u}_c}},$$
where $\dvec{u}_c \in \{0,1\}^{k_c}$ is the information vector corresponding to $W_c$. Throughout this discussion, $f_c\br{\dvec{u}_c}$ and $f_c\br{W_c}$ are used interchangeably.

Finally, the modulated codeword $\dvec{x}^{\br{i}}$ is spread using the chosen spreading sequence $\dvec{a}^{\br{i}}$ to produce the transmitted signal
\begin{equation}\label{eq:ch5_spread_rkon}
\dvec{t}^{\br{i}} = \dvec{x}^{\br{i}} \otimes \dvec{a}^{\br{i}},
\end{equation}
where $\otimes$ denotes the Kronecker product operation. A graphical representation of the encoding process is shown in~\Fig{fig:encoder}.

\begin{remark}[Signal representation in matrix form]
Let us transition to the matrix representation. We have
\begin{equation}\label{eq:ch5:spread_matrix}
\dvec{T}^{\br{i}} = \dvec{a}^{\br{i}} \br{\dvec{x}^{\br{i}}}^T,
\end{equation}
where $\dvec{t}^{\br{i}}$ can be obtained by reshaping this matrix, i.e., by reading it column by column.
\end{remark}

Taking into account that $\Ka$ users are transmitting simultaneously, we observe the following noisy signal mixture at the channel output
\begin{equation}\label{eq:ch5:spread_reception}
\rmat{Y} = \sum\limits_{i=1}^{\Ka} \dmat{T}^{\br{i}} + \rmat{Z},
\end{equation}
where $\rmat{Z} \sim \mathcal{N}(\dvec{0}, \dmat{I}_{n_p \times n_c})$. 

\begin{remark}
Let $\dset{P} = \brc{W_{1,p}, \ldots, W_{\Ka,p}}$. Consider the general case where some preambles may coincide, i.e., $|\dset{P}| = r \leq \Ka$ and $\dset{P} = \{j_1, j_2, \ldots, j_r\}$, where $j_1 < j_2 < \ldots < j_r$ are unique elements. Then, we have
\begin{equation}\label{eq:mmse}
\rmat{Y} = \dmat{A}_\dset{P} \tilde{\dmat{X}} + \rmat{Z}, 
\end{equation}
where $\dmat{A}_\dset{P}$ is formed from $\dmat{A}$ by selecting columns indexed by $\dset{P}$, and
\[
\tilde{\dmat{X}} = \begin{bmatrix} 
\tilde{\dvec{x}}_1^T\\
\tilde{\dvec{x}}_2^T\\
\vdots\\
\tilde{\dvec{x}}_r^T
\end{bmatrix},
\]
and $\tilde{\dvec{x}}_l$ is the sum of the codewords of users that chose the preamble $j_l$, i.e.,
\[
\tilde{\dvec{x}}_l = \sum\limits_{i: W_{i,p} = j_l} \dvec{x}^{\br{i}}, \quad l \in [r].
\]
\end{remark}

\subsection{Decoder}

The iterative decoder consists of several components. First, an energy detector decodes the preamble parts of the transmitted messages, which are then mapped to columns of $\dmat{A}$ to identify the set of spreading sequences selected by active users. Since signature sequences are chosen solely based on preambles, two or more users select the same sequence whenever they share a common preamble.

Based on the energy corresponding to a detected sequence, the energy detector also provides an estimate of the number of active users sharing that sequence. The \ac{mmse} estimator is then employed to produce soft estimates of the signals corresponding to the detected sequences. These soft estimates are passed to the users' decoders, and the recovered codewords are subtracted from the received signal in the spirit of \ac{sic}. The residual signal is then fed back into the energy detector for the next decoding iteration.

This iterative decoding process continues until all transmitted messages are recovered or the number of decoded messages no longer improves between consecutive iterations. Let us examine this process in more detail.

\subsubsection{Detection of spreading sequences}

At this stage, we need to determine the set of spreading sequences that were transmitted. Recall that we are utilizing a good \ac{cs} codebook; thus, the natural approach is to examine the correlations between the channel output and all possible spreading sequences (i.e., the columns of the sensing matrix $\dmat{A}$).

One challenge in implementing this algorithm arises from~\eqref{eq:ch5_spread_rkon}, as $\dvec{x}^{\br{i}}$ is unknown at this stage of the decoding process. The optimal strategy, in terms of performance, would be to consider all possible codewords $\dvec{x}^{\br{i}}$, but this approach is computationally infeasible.

To address this, the authors of~\cite{AmalladinnePolar2020} proposed a suboptimal strategy. Assume we are considering a spreading sequence $\dvec{a}$\footnote{$\dvec{a}$ is an arbitrary column of the matrix $\dmat{A}$; we omit the subscript as nothing depends on it.}. We split the channel output into $n' = n_c/g$ groups of length $n_p g$ as follows:
\begin{equation*}
\dvec{y}^T = \left[\dvec{y}_1^T, \dvec{y}_2^T, \ldots, \dvec{y}_{n'}^T\right]^T.
\end{equation*}

Now, to calculate the \emph{score} of $\dvec{a}$ we proceed as follows:
\[
S(\dvec{a}) = \sum\limits_{i=1}^{n'} S_i(\dvec{a}), \quad \text{where} \quad S_i(\dvec{a}) = \max\limits_{\dvec{b} \in \{\pm 1\}^g} \br{\dvec{b} \otimes \dvec{a}}^T \dvec{y}_i. 
\]
The group size is carefully chosen to balance error performance and computational complexity. The columns of $\dmat{A}$ are then sorted in descending order of their scores, and the detector outputs the top $\Ka+\Delta$ sequences, where $\Delta$ is a fixed small non-negative integer. We denote by $\mathcal{D}$ the set of indices of these sequences. Notably, $\mathcal{D}$ serves as an estimate of the set of message parts $\dset{P}$.

\subsubsection{Demodulation and decoding}

In this section, we describe the steps involved in decoding the messages corresponding to the detected sequences.

Recall eq.~\eqref{eq:mmse} and note that we know $\dset{D}$, the estimate of the set of spreading sequences $\dset{P}$. We have
\begin{equation}\label{ed:demod}
\rmat{Y} = \dmat{A}_\dset{D} \tilde{\dmat{X}} + \rmat{Z}',
\end{equation}
where $\rmat{Z}'$ consists of the additive Gaussian noise $\rmat{Z}$ and the interference from users with undetected spreading sequences. In what follows, we work under the assumption that $\dset{P} \subseteq \dset{D}$, implying that $\rmat{Z}' = \rmat{Z}$.

First, we apply the \ac{lmmse} estimator (see details in Appendix~\ref{a:mmse}) to obtain
\[
\hat{{\dmat{X}}} = \dmat{A}^{T}_{\dset{D}} \dmat{R}^{-1} \rmat{Y},
\quad
\dmat{R}=(\dmat{A}_{\dset{D}}\dmat{A}_{\dset{D}}^{T}+\dmat{I}_{{n_p}}).
\]

The \ac{mse} of this estimator is well approximated by
\begin{equation}
\mathbf{\Sigma}= \dmat{I} -\mathbf{A}^T_\mathcal{D}\mathbf{R}^{-1}\mathbf{A}_\mathcal{D},
\end{equation}
where
$$
\mathbf{\Sigma} \approx \operatorname{diag} \left( \sigma^2_{\mathrm{mse}}\br{1}, \sigma^2_{\mathrm{mse}}(2), \ldots, \sigma^2_{\mathrm{mse}}(\mathcal{D}) \right).
$$
In this context, $\sigma^2_{\mathrm{mse}}(j)$ is the \ac{mse} in the estimate of the signal corresponding to sequence~$a_j$ for $j \in \mathcal{D}$.

Assume the case when there are no repeating sequences, i.e. $|\dset{P}| = \Ka$. The next step is to calculate input \ac{llr} sequences for the users' decoders. Assuming the \ac{mse} is Gaussian, we proceed as follows:
\[
\dvec{l}_i = \frac{2}{\sigma^2_i} \hat{\dvec{x}}_i, \quad i \in [\Ka], 
\]
where $\hat{\dvec{x}}_i$ is the $i$-th row of the matrix $\hat{{\dmat{X}}}$.

Then, the calculated \ac{llr}s are passed to users' codes. In~\cite{AmalladinnePolar2020}, the authors utilize single-user polar codes. The decoders can converge to a correct codeword, output a failure, or even produce a wrong decision. For the latter case, we assume the presence of a technique to discard such decisions (e.g., \ac{crc}). We denote by $\dset{D}^* \subseteq \dset{D}$ the collection of indices such that users' codes are decoded successfully.

The \ac{sic} removes the contributions from all the successfully decoded codewords $\hat{\dmat{X}}_{\dset{D}^*}$ from the received signal to compute the residual
\begin{equation*}
\dmat{Y} -\dmat{A}_{\dset{D}^*}\hat{\dmat{X}}_{\dset{D}^*}.
\end{equation*}
This residual is passed for the next decoding iteration. This process continues until all the transmitted messages are recovered successfully or there is no improvement between two consecutive rounds of the iterative process.

\subsection{Discussion and further improvements}\label{sec:ch5_random_spreading_practical}

Let us describe possible improvements suggested in the literature.

\paragraph{Joint decoding.} The scheme in~\cite{AmalladinnePolar2020} relies on single-user codes. Thus, if several users choose the same spreading sequence, the corresponding messages will be missed with a relatively high probability. A straightforward solution is to enlarge the set of spreading sequences, but in this case, we lose the good correlation properties. A better solution is to adopt joint decoding (in the spirit of Section~\ref{sec:dec_ldpc} with \ac{ldpc} or polar codes) to avoid the need for large-sized codebooks for the spreading sequences. The idea is as follows: for those active users whose signature sequences collide, we leverage a joint decoder that does not treat the heavy interference from colliding users as noise. We emphasize that, although the envisioned scheme utilizes joint decoding, it is computationally less burdensome than~\cite{Marshakov2019Polar} because joint decoding is performed only for colliding users.

\paragraph{Power diversity.} It is well-known that power diversity improves the \ac{sic} process. The idea proposed in~\cite{Duman2021} is to divide the columns of $\dmat{A}$ into $m$ groups and assign different power levels to each. There are $l_k$ columns with power level $P_k$ in the $k$-th group, with $k \in [m]$. The power levels for each group must be chosen in such a way that the average power constraint is satisfied. The proposed approach divides the active users into different groups. The encoding procedure in this scheme remains unchanged, while the decoding differs in the preamble detection part: the energy detector is replaced with a covariance-based detector. Namely, we form
\[
\Sigma = \dmat{A}^T_N \dmat{Y} \dmat{Y}^T \dmat{A}_N,
\]
where $\dmat{A}_N$ is obtained by scaling the columns of the codebook to $1$ (after removing signatures of correctly decoded users) to prevent a higher detection probability for the sequences with greater power levels. The covariance-based detector outputs $\Ka + \Delta$ signatures corresponding to the largest diagonal elements of $\Sigma$.

Simulations show that the proposed approach outperforms the existing methods, especially when the number of active users is large.

\paragraph{\Ac{soic}.} The authors of~\cite{PradhanLDPCSC2021} noted that to apply the \ac{sic} procedure, we should decode the user code. At the same time, the information from non-converged decoders is not utilized, even though it may help the overall iterative decoding process. The paper~\cite{PradhanLDPCSC2021} proposes replacing polar codes with \ac{ldpc} codes (since \ac{ldpc} codes are much better in terms of soft output) and utilizing \ac{soic}~\cite{softSIC_wang_poor, softSIC_lampe, softSIC_elgamal}. We also mention the paper~\cite{softSIC_caire}, which describes \ac{de} analysis of iterative \ac{soic} multi-user detection/decoding.

Assume the user code provides us with the extrinsic information vector $\dvec{\gamma} = [\gamma_1, \gamma_2, \ldots, \gamma_{n_c}]$. We calculate the expected values of the modulated symbols and remove them from the channel output, i.e.
\[
\overline{\dvec{x}} = [\mathbb{E}[\rscalar{x}_1 | \gamma_1], \ldots, \mathbb{E}[\rscalar{x}_{n_c} | \gamma_{n_c}]] = [\tanh(\gamma_1/2), \ldots, \tanh(\gamma_{n_c}/2)],
\]
and then subtract them:
\[
\mathbf{Y} -\dmat{A}_{\dset{D}} \overline{{\dmat{X}}}_{\dset{D}}.
\]
This scheme is the state-of-the-art scheme for the \ac{gura}. Notably, the proposed scheme outperforms the \ac{rcb} derived in Theorem~\ref{th:polyanskiy_bound} when the active population is less than $75$ for parameters taken from~\cite{ordentlich2017low}. Furthermore, the proposed scheme outperforms the schemes with joint decoding~\cite{AmalladinnePolar2020} and power diversity~\cite{Duman2021} when the active population exceeds $200$ users. However, in a private discussion, the authors of this approach informed us that they were unable to validate the resulting energy efficiency curve (which will be discussed in Section~\ref{ch5:sec_numerical}).

\paragraph{Hadamard-based spreading with convolutional codes.} Finally, let us highlight an elegant scheme that considers spreading differently. All previously discussed spreading techniques are based on a Kronecker product~\eqref{eq:ch5_spread_rkon}. The authors of~\cite{Duman2021conv} propose a scheme based on the \emph{Hadamard} product, where the encoded codeword spread by the signature will not increase in length. The authors also use a slotted system, where users choose the transmission slot at random. To detect different signatures, dedicated pilot bits are inserted. The choice of the signature is also based on splitting the message into two parts, with a convolutional code considered as the user code. The final decision makes it straightforward to apply the Viterbi algorithm to a joint decoding problem. Despite its simplicity, the proposed solution offers quite good performance.

\section{Comments on same linear codes}\label{sec:same_lin}

As we know (Chapter~\ref{chap4}), there exist codes that can operate in the \ac{gura}. At the same time, the same \ac{bpsk}-modulated linear codes (e.g., \ac{ldpc} or polar codes) perform poorly in the \ac{gura}. However, these codes could still be strong candidates for practical schemes, as binary linear codes are well-studied and have low-complexity soft decoding algorithms.

Thus, as explained throughout this chapter, all low-complexity schemes rely on the following approach:
\begin{itemize}
\item[a)] generate different codes from a single linear code (e.g., through permutations or shifts);
\item[b)] append the transmitted codewords with a preamble that helps the receiver identify the correct permutation or shift.
\end{itemize}

This strategy is indeed effective, but the use of a preamble introduces additional overhead. This naturally leads to the question:
\emph{Are there linear codes that can operate in the \ac{gura}?}

\begin{figure}[t]
\centering
\includegraphics{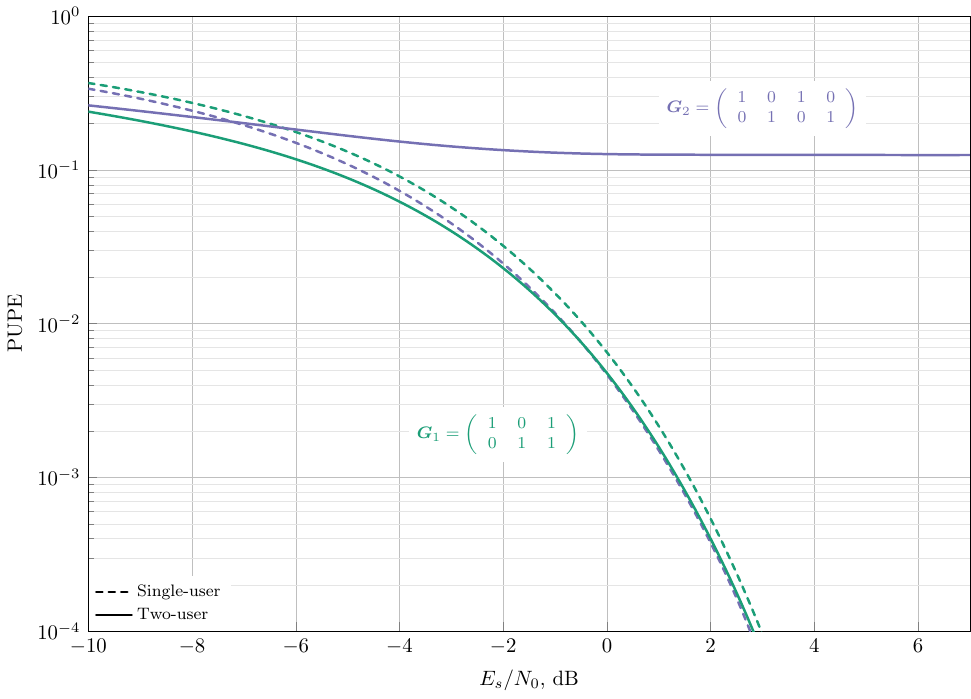}
\caption{\ac{pupe} vs. \ac{snr} for two binary linear block codes over a \ac{bac} with \ac{awgn}. A toy example demonstrates two different linear codes with generator matrices $\dmat{G}_1$ and $\dmat{G}_2$, followed by \ac{bpsk} modulation. These codes exhibit similar performance in a single-user regime but significantly different performance in a multi-user regime. This result is presented in~\cite[Figure 1]{2userBAC}.}
\label{fig:2bac}
\end{figure}

To better understand the problem, let us examine~\Fig{fig:2bac}, which is also presented in~\cite[Figure 1]{2userBAC}. This figure presents the results for the two-user \ac{gura}, namely,
\[
\rvec{y} = \dvec{x}^{\br{1}} + \dvec{x}^{(2)} + \rvec{z},
\]
where $\dvec{x}^{\br{1}}, \dvec{x}^{(2)} \in {\pm \sqrt{P}}^n$ and $\rvec{z} \sim \mathcal{N}(0, I_n)$.

Two linear codes are compared: $\mathcal{C}_1$, with generator matrix $\dmat{G}_1$, which is a $[3,2]$ single-parity-check (SPC) code, and $\mathcal{C}_2$, a $[4,2]$ code with generator matrix $\dmat{G}_2$, which is the direct sum of two repetition codes of length 2. The block error probability of these codes in the absence of interference (i.e., single-user transmission) is depicted as a reference with dashed lines.

For the multi-user case, the decoder outputs the unordered set of codewords ${\hat{\dvec{x}}^{\br{1}}, \hat{\dvec{x}}^{(2)}}$ that maximizes $p(\dvec{y} | \dvec{x}^{\br{1}} + \dvec{x}^{(2)})$. We observe that $\mathcal{C}_2$ performs better in the single-user setting but exhibits a high error floor in the multi-user scenario. Hence, the root cause of the error floor is multi-user interference rather than \ac{awgn}.

Thus, let us focus on a simplified scenario of \ac{ura} in a two-user noiseless \ac{bac}. Despite its simplicity, this setting proves to be particularly rich from a coding theory perspective.

Here, we briefly discuss the results presented in~\cite{2userBAC}, starting with the decoding algorithm. Since we are not considering noise, there is no need to work with the alphabet $\{\pm \sqrt{P}\}$; instead, we use the binary alphabet $\{0,1\}$. We now note that \ac{ml} decoding over the \ac{bac} reduces to an erasure decoding problem. Assume we receive $\dvec{y} \in \{0,1,2\}^n$. Considering the $j$-th element of $\dvec{y}$, we have three possible cases:
\begin{enumerate}
\item ${y}_j = 0 \implies {x}^{\br{1}}_j = 0 \:\: \text{and} \:\: {x}^{(2)}_j = 0$; 
\item ${y}_j = 2 \implies {x}^{\br{1}}_j = 1 \:\: \text{and} \:\: {x}^{(2)}_j = 1$;
\item ${y}_j = 1 \implies \left({x}^{\br{1}}_j, {x}^{(2)}_j\right) \in \{ (0, 1), \:\: (1, 0) \}$;
\end{enumerate}
To handle the decoding process, we introduce an auxiliary vector $\tilde{\dvec{y}}$, defined as follows:
\begin{enumerate}
\item $\tilde{y}_j = 0$ if $y_j=0$,
\item $\tilde{y}_j = 1$ if $y_j=2$,
\item $\tilde{y}_j = *$ if $y_j=1$.
\end{enumerate}
Notably, uncertainty arises only in case (3). Consequently, the decoder begins by constructing the list $\mathcal{L}$ of all codewords $\dvec{x}$ that are compatible with the channel output. Specifically, defining the erasure positions as $\mathcal{E} = \{j : y_j = 1\}$ and the non-erasure positions as $\mathcal{E}^c = [n] \backslash \mathcal{E}$, we obtain
\[
\mathcal{L}' = \{\dvec{x} : \dvec{x}_{\mathcal{E}^c} = \tilde{\dvec{y}}_{\mathcal{E}^c}\}.
\]

The final step involves searching for all unordered codeword pairs $\{\dvec{x}_1, \dvec{x}_2\} \in \mathcal{L}' \times \mathcal{L}'$ whose integer sum matches $\dvec{y}$. If a unique solution exists, the decoder outputs this pair; otherwise, if multiple solutions exist, the decoder selects one at random.

For the case of linear codes, the first step is straightforward. Let  $\dmat{H}$ be the parity-check matrix of the users' code. Then, $\mathcal{L}'$ is the set of solutions to the following system of linear equations:
\begin{equation}\label{eq:2bac}
\dmat{H}_{\mathcal{E}} \dvec{x}^T_{\mathcal{E}} = \dmat{H}_{\mathcal{E}^c} \dvec{x}^T_{\mathcal{E}^c}. 
\end{equation}

We note that this system always has at least two solutions, implying that $\rank(\dmat{H}_{\mathcal{E}}) < |\mathcal{E}|$. One possible approach to solving this system is the \emph{pivoting} strategy: setting an arbitrary element (pivot) of $\dvec{x}_{\mathcal{E}}$ to a value in $\{0,1\}$ and then solving the resulting system. If $\rank(\dmat{H}_{\mathcal{E}}) = |\mathcal{E}| -s$, then there are $2^{s-1}$ solutions. The authors of~\cite{2userBAC} make the following observations:
\begin{enumerate}
\item Any solution of~\eqref{eq:2bac} participates in exactly one valid pair.
\item If $\dvec{x}^{\br{1}} \not\in \mathcal{L}$, then $\dvec{x}^{(2)}$ is also not in $\mathcal{L}$.
\item The probability of decoding error is given by
$$
P_e = 1 - \sum\limits_{s=1}^{|\mathcal{E}|} 2^{-(s-1)} \Pr[\rank(\dmat{H}_{\mathcal{E}}) = |\mathcal{E}| - s].
$$
\item The number of valid codeword pairs computed by the decoder depends on the transmitted pair only through the support set of their difference.
\end{enumerate}

Let us formulate an auxiliary lemma from~\cite{2userBAC}:
\begin{lemma}[All-zero codeword]\label{lemma:all_zero}
Consider two users transmitting over the \ac{bac} using an $[n,k]$ binary linear block code $\mathcal{C}$, and denote the transmitted codewords as $\dvec{x}^{\br{1}}$ and $\dvec{x}^{(2)}$. Then,
\[
\Pr\left[ \rvec{x}^{\br{1}} \not\in \mathcal{L} | \rvec{x}^{\br{1}} = \dvec{c} \right] = \Pr\left[ \rvec{x}^{\br{1}} \not\in \mathcal{L} | \rvec{x}^{\br{1}} = \dvec{0} \right], \quad \forall \dvec{c} \in \mathcal{C}. 
\]
\end{lemma}

Note that the statement remains valid if we exchange $\dvec{x}^{\br{1}}$ and $\dvec{x}^{(2)}$. Lemma~\ref{lemma:all_zero} implies that
$$
P_e = \Pr\left[ \rvec{x}^{\br{1}} \not\in \mathcal{L} | \rvec{x}^{\br{1}} = \dvec{0} \right],
$$
i.e., we can evaluate the \ac{pupe} of a binary linear block code under \ac{ml} decoding by the special case where one of the two users transmits the all-zero codeword.

Now, consider the case where $\rvec{x}^{\br{1}} = \dvec{0}$. In this case, we have $\mathcal{E} = \supp(\dvec{x}^{(2)})$, and~\eqref{eq:2bac} has exactly two solutions if and only if $\dvec{x}^{(2)}$ is \emph{minimal}.

\begin{definition}
The codeword $\dvec{c} \in \mathcal{C}$ is called \emph{minimal} if it does not cover any other codeword, i.e., there does not exist $\dvec{c}' \in \mathcal{C}$ such that $\supp(\dvec{c}') \subseteq \supp(\dvec{c})$.
\end{definition}

Let $\mathcal{M}(\mathcal{C})$ denote the set of minimal codewords in $\mathcal{C}$. We can formulate the following statement:
\[
\Pr[\rank(\dmat{H}_{\mathcal{E}}) < |\mathcal{E}| - 1] = \Pr[\dvec{x}^{(2)} \not\in \mathcal{M}(\mathcal{C})].
\]
As a result, minimal codes achieve zero $P_e$ over a two-user \ac{bac} with the same-codebook constraint.

Returning to~\Fig{fig:2bac}, we observe that the $[3, 2]$ \ac{spc} code is a minimal code, meaning it consists only of minimal codewords. At the same time, the $[4, 2]$ code includes a non-minimal all-one codeword, which is the root cause of the error floor.

\begin{theorem}[Converse for two-user \ac{bac}]\label{th:bac_converse}
Consider two users transmitting over the \ac{bac} with a $[n,k]$ binary linear block code. Then, we have
\[
P_e \geq \frac{1}{2}\left( 1 - \frac{n}{2k-2}\right).
\]
Thus, the maximal achievable symmetric rate is $R \leq 1/2$.
\end{theorem}
\begin{proof}
Let $\dvec{x}^{\br{1}} = \dvec{0}$. We have
\begin{flalign*}
P_e &\geq \frac{1}{2}\Pr\left[\dvec{x}^{(2)} \not\in \mathcal{M}(\mathcal{C})\right]\\
&\geq \frac{1}{2}\Pr\left[\wt{\dvec{x}^{(2)}} > n-k+1\right]\\
&= \frac{1}{2}\left(1 - \Pr\left[n-\wt{\dvec{x}^{(2)}} \geq k-1\right]\right)\\
&\geq \frac{1}{2}\left(1 - \frac{\mathbb{E}\left[n-\wt{\dvec{x}^{(2)}}\right]}{k-1}\right) \quad \text{(Markov's inequality)}\\
&= \frac{1}{2}\left( 1 - \frac{n}{2k-2}\right).
\end{flalign*}
\end{proof}

\begin{remark}
The following results show that transmission over the two-user \ac{bac} with linear block codes is fundamentally limited to a symmetric rate no larger than $\nicefrac{1}{2}$, in contrast to the maximum symmetric rate of $\nicefrac{3}{4}$ achievable by nonlinear block codes over the two-user \ac{bac}.
\end{remark}

\begin{remark}
Note that the result stated in Theorem~\ref{th:bac_converse} extends to cosets of the linear code.
\end{remark}

The papers~\cite{2userBAC, 2userBAC_ldpc} also address the question of whether it is possible to achieve $R=\nicefrac{1}{2}$ via iterative decoding. It appears that for \ac{ldpc} codes, the choice of the pivot\footnote{A good pivot is a variable node associated with a received "erasure" symbol for which revealing its value allows recovering both transmitted messages, up to a vanishingly small fraction of residual erasures, by simple \ac{bch} decoding.} is of critical importance. If the pivot is selected randomly, then any irregular \ac{ldpc} code ensemble has a non-vanishing error probability. If the pivot is selected optimally, then it is possible to attain vanishing \ac{pupe}. The paper~\cite{2userBAC_ldpc} continues the investigation. In particular, the authors analyze the expected fraction of good pivots, compute optimized degree distributions, and provide a simple multi-edge type \ac{ldpc} code construction that can provably achieve the two-user unsourced \ac{bac} limit for linear codes.

In this section, we address the question of constructing the same codebook codes for the \ac{ura}. All the schemes described in this chapter utilize preambles chosen from a small CS codebook. The preambles can be used in two different ways:
\begin{itemize}
\item Same codes that are constructed as a concatenation of preambles and different linear codes (the interleaver or shift is chosen based on the preamble).
\item Same codes that are simply a concatenation of preambles.
\end{itemize}

The main disadvantages are as follows: (a) preambles introduce additional overhead; (b) we cannot significantly increase the number of preambles ($2^{15}$ is an upper bound due to complexity considerations). 

Thus, an important question arises: what are alternative approaches to constructing low-complexity same codes? This section considers a relatively simple scenario of a two-user \ac{bac}, but even in this case, it becomes evident that users cannot effectively utilize well-developed linear codes and their cosets. Non-linear codes are necessary to achieve the capacity ($C = \nicefrac{3}{4}$) of the two-user unsourced \ac{bac}. 

The design of such codes -- where the main challenge is developing low-complexity decoders -- remains an open and challenging problem. We can reformulate the problem as follows: we need a large CS codebook with a polynomial-time recovery algorithm. Some progress in this direction has been made in~\cite{CHIRRUP}.

\section{Numerical comparisons}\label{ch5:sec_numerical}
\begin{figure}
\centering
\includegraphics{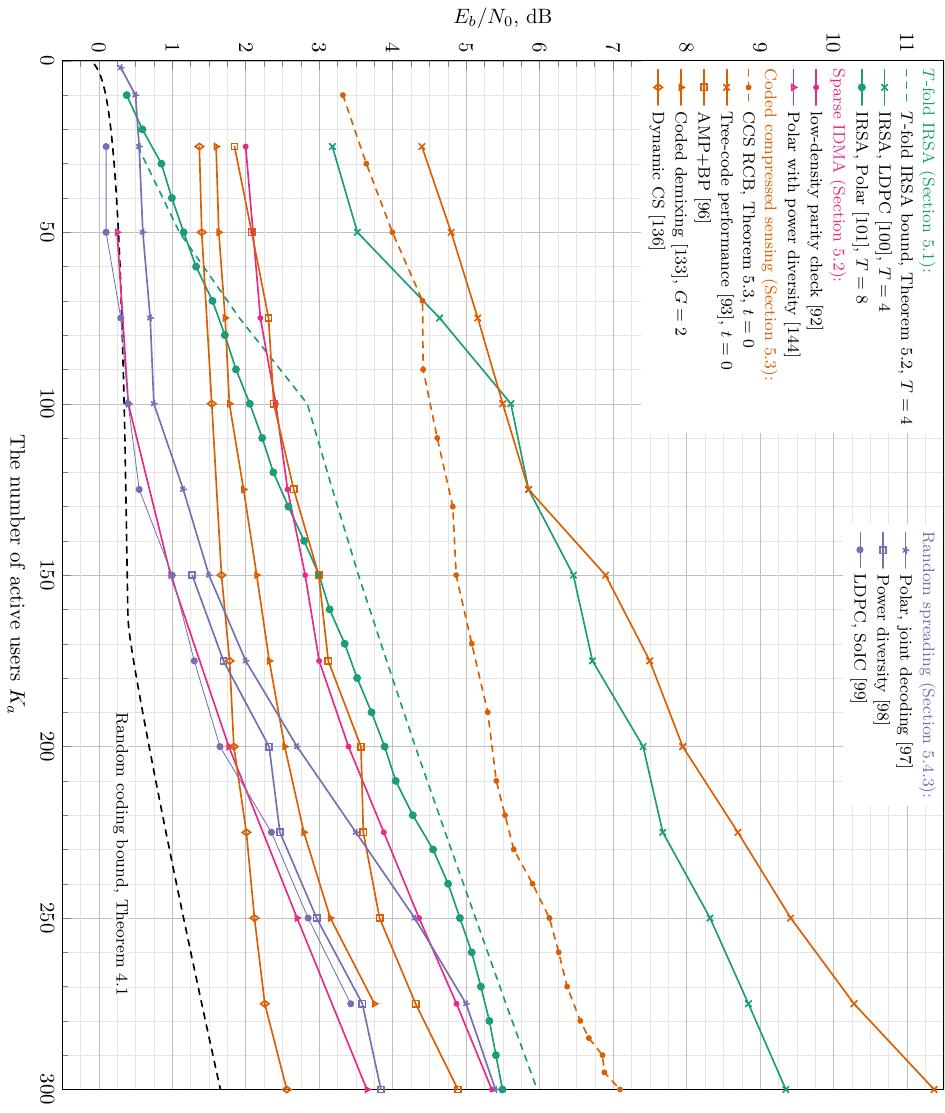}
% \nocite{vem2019user}\nocite{Marshakov2019Polar}\nocite{AmalladinneJoint2019}\nocite{Duman2024ODMA}\nocite{codedCS2020}\nocite{AmalladinneAMP2021}\nocite{Amalladinne2022}\nocite{Truhachev2024}\nocite{AmalladinnePolar2020}\nocite{Duman2021}\nocite{PradhanLDPCSC2021}
\caption{The minimum $E_b/N_0$ required to achieve an error rate below 5\% as a function of the active user count. Theoretical bounds are represented by dashed lines, while practical schemes are shown with solid lines. Different colors correspond to various families of algorithms described in Sections~\ref{ch5:sec:coded_sa}, \ref{sec:ch5:idma}, \ref{chap5:ccs_description}, and \ref{sec:ch5:randon_spreading}. In a private discussion, the authors of~\cite{PradhanLDPCSC2021} informed us that they could not validate the resulting curve. Hence, we mark it with a thin line.}\label{fig:ch5:mac_numerical}
\end{figure}

In this section, we present the results of numerical experiments that estimate the energy efficiency of the solutions described above. All previously introduced schemes assume a frame-synchronous transmission system. Energy efficiency is defined as the minimum energy spent per transmitted information bit, $E_b/N_0$, under the maximum tolerable \ac{pupe}, as defined in~\eqref{eq:ch3_energy_efficiency}.

Starting from~\cite{polyanskiy2017perspective, ordentlich2017low}, a commonly used scenario for \ac{awgn} \ac{ura} assumes a frame length of $n=3\times 10^4$ real channel uses, $k=100$ information bits, and a maximum tolerable \ac{pupe} of $\varepsilon = 0.05$.

A numerical comparison is presented in~\Fig{fig:ch5:mac_numerical} alongside state-of-the-art results. This figure illustrates the energy efficiency $E_b/N_0$ of different schemes proposed in the literature as a function of the number of active users, $K_a$. The latter term represents the number of messages sent within a single frame. This figure omits the first low-complexity scheme proposed in~\cite{ordentlich2017low} because it demonstrated relatively low energy efficiency.

We focus on four previously mentioned groups of algorithms. The first group is based on the $T$-fold \ac{irsa} (see Section~\ref{sec:ch5:num_results_irsa}). The second group follows a sparse \ac{idma} approach (Section~\ref{sec:ch5:num_results_idma}). The third group is based on a \ac{ccs} approach (Section~\ref{sec:ch5:num_results_ccs}). Finally, we consider a family of random spreading algorithms presented in Section~\ref{sec:ch5:num_results_rspread}. Different families of algorithms are represented by different colors. As a reference, we have also included Polyanskiy's achievability bound (presented in Theorem~\ref{th:polyanskiy_bound}) and the \ac{irsa} achievability bound, formulated in Theorem~\ref{th:irsa_t_bound}. Theoretical bounds are marked by dashed lines.

\subsection{\texorpdfstring{$T$}{T}-fold \acl{irsa}-based schemes}\label{sec:ch5:num_results_irsa}

\subsubsection{\texorpdfstring{$T$}{T}-fold \acl{irsa} achievability bound}

Let us start with the achievability bound for \ac{irsa} and numerically evaluate Theorem~\ref{th:irsa_t_bound}. The main objective is to select the polynomial $\Lambda$ (see~\eqref{eq:ch5:irsa_poly}) and determine the number of slots, or equivalently, the slot length $n'$:
$$
\left\{\Lambda\left(x\right), n'\right\} = \argmin_{\Lambda\left(x\right), n'}\left(\frac{E_b}{N_0}: P_e \leq \varepsilon \right),
$$
where $E_b/N_0$ is defined in~\eqref{eq:ebno_tfold_irsa}. The optimization procedure was performed separately for each value of $K_a$. The resulting energy efficiency is represented by a green dashed line in~\Fig{fig:ch5:mac_numerical} and corresponds to $T=4$. The number of slots varies with $K_a$, as shown in~Appendix~\ref{a:irsa} for $T=1, 2, 4$. 

Note that transitioning from a random coding approach to the $T$-fold \ac{irsa} results in a performance loss. However, all schemes based on the $T$-fold ALOHA with the same value of $T$ can be considered as a reference.

\subsubsection{\texorpdfstring{$T$}{T}-fold \ac{irsa} with \ac{ldpc} codes and a joint decoding algorithm}

Consider the joint decoding of \ac{ldpc} codes, as presented in Section~\ref{sec:dec_ldpc}. The authors of this scheme~\cite{vem2019user} performed an optimization search for system parameters as follows. Recall that this scheme relies on the \ac{irsas} approach (see~\Fig{fig:irsa-s}), where $k_p$ out of $k$ information bits are allocated to the preamble (required for proper permutation to enhance joint decoding), while the remaining $k_c = k - k_p$ information bits are encoded using an \ac{ldpc} code for each sequence transmitted within a slot. 

The authors found that a sensing matrix based on a \ac{bch} code exhibited better decoding performance compared to randomly generated sensing matrices. Additionally, they optimized the \ac{ldpc} code length by applying the \ac{rcb} from Theorem~\ref{th:polyanskiy_bound} to estimate the slot-level decoding error. Given the optimal slot length, a regular $(3,6)$ \ac{ldpc} code was selected. The resulting parameters of the scheme are: $k_p = 9$, $k_c = 91$, a preamble length of $n_p = 63$, and an \ac{ldpc} code length of $n_c = 394$, followed by \ac{bpsk} modulation. 

The search for the polynomial $\Lambda\br{x}$ in~\eqref{eq:ch5:irsa_poly} was restricted to a single parameter $\beta$, with $\Lambda\br{x} = \beta x + (1 - \beta)x^2$. The authors also limited the maximum number of replicas to two. The reasoning for this is provided in Appendix~\ref{a:irsa}: higher values of $T$ result in optimal polynomials in~\eqref{eq:ch5:irsa_poly} having fewer transmission attempts (i.e., lower-degree polynomials). Recall that the authors considered $T=4$, and the polynomials presented in Table~\ref{DE_T4} are of degree two.

The energy efficiency of this scheme is depicted in~\Fig{fig:ch5:mac_numerical} by a green line with $\times$-shaped markers and corresponds to a maximum of $T=4$ simultaneous transmissions per slot. Note that there is a gap between this scheme and the achievability bound presented in Theorem~\ref{th:irsa_t_bound}. Nevertheless, it shows a significant improvement over~\cite{ordentlich2017low}, exceeding $10$ dB for $K_a = 300$.

\subsubsection{\texorpdfstring{$T$}{T}-fold \ac{irsa} with polar codes and a joint decoding algorithm}

The next improvement was made by replacing \ac{ldpc} codes with polar codes~\cite{Marshakov2019Polar}, as presented in Section~\ref{sec:dec_polar}. Moreover, switching to polar codes resulted in higher values of $T$ that a joint decoder can successfully resolve. Inspired by~\cite{AmalladinneJoint2019}, the \ac{irsaf} (see~\Fig{fig:irsa-f}) was utilized. This preamble has a length of $n_p = 2000$ and carries $k_p = 15$ information bits, which are dedicated to recovering the transmission graph and frozen-bit patterns for each polar code. The users employ $(512, 85)$ polar coset codes. 

The slot length is fixed for all $K_a$ values, as it is challenging to assign an arbitrary length to the polar code. To choose the frozen positions, we utilized the proposed design procedure. We selected a common set of frozen tuples for all users, and the frozen bit values were chosen at random for each user, since using identical frozen values leads to poor performance. The list size is $\ell = 64$. At the same time, $(512, 85)$ polar coset codes with the \ac{jsc} decoder can recover collisions of order up to $T = 8$.

The resulting energy efficiency is presented in~\Fig{fig:ch5:mac_numerical} by a green line with circular markers. This scheme outperforms the \ac{ldpc}-based scheme presented above and is very close to the \ac{irsa} achievability bound~\cite{Glebov2019} ($T = 4$). Moreover, this scheme closely matches the sophisticated scheme from Section~\ref{sec:ch5:idma}.

\subsection{Sparse \acl{idma}-based schemes}\label{sec:ch5:num_results_idma}

The energy efficiency of the sparse \ac{idma} scheme described in Section~\ref{sec:ch5:idma} is presented in~\Fig{fig:ch5:mac_numerical} by a magenta line with circular markers. The authors chose $k_p = 15$ bits for the preamble in the \ac{irsaf} scheme, and the remaining $k_c = 85$ bits were encoded using the \ac{ldpc} code. The structure of the \ac{ldpc} code was carefully optimized for each number of active users $K_a$. The choice of the coding rate is:
\begin{itemize}
    \item $0.125$ for $K_a = 25, \ldots, 125$,
    \item $0.25$ for $K_a = 150, \ldots, 200$,
    \item $0.4$ for $K_a = 225, \ldots, 300$.
\end{itemize}
Each codeword appeared in the channel twice.

The length of the preamble is $n_p = 2000$. The \ac{cs} design matrix uses $\nicefrac{n_p}{2}$ randomly chosen rows from the \ac{dft} matrix. Then, the imaginary and real parts of this matrix are stacked together to form a real-valued sensing matrix. The resulting matrix is a good choice for \ac{cs}-based preamble detection because it satisfies the \ac{rip} with high probability~\nocite{Duman2024ODMA}\cite{Haviv2016RIP}.

The resulting energy efficiency shows better performance compared to the \ac{irsa}-based scheme. The plot presented in~\Fig{fig:ch5:mac_numerical} corresponds to the case without \ac{sic}. The use of \ac{sic} provides an improvement of about $0.5$ dB for small values of $K_a \lesssim 100$ and approximately $0.2$ dB for larger $K_a$ values.

This scheme was further improved in~\cite{Duman2024ODMA} and is known as \acf{odma}. Compared to the previously described scheme, there are three main differences. The first is that the \ac{irsa}-like scheme has been replaced by a single transmission without replicas. The second is the replacement of \ac{ldpc} codes with polar codes. The third difference is power diversity: instead of transmitting with a fixed energy, a random energy is chosen by each user such that the average power constraint is satisfied.

The positions where \ac{bpsk}-modulated symbols of the polar-encoded message are transmitted are derived from the preamble. The choice of the preamble is given by the first $k_p$ bits of the information message. Preamble detection is solved by a \ac{cs} decoding algorithm.
The authors have chosen the following parameters:
\begin{itemize}
    \item Preamble codebook size $2^{k_p}$, where $k_p = 12$ for $K_a = 50$, $k_p = 13$ for $K_a \leq 150$, and $k_p = 14$ otherwise.
    \item Polar codes of length $n_c = 512$ bits for $K_a \leq 200$ and $n_c = 256$ otherwise.
\end{itemize}

The authors considered the \ac{crc} of size $16$ bits and the \ac{scl} decoding algorithm with a list size of 128. The performance of this scheme is represented by a magenta line with triangular markers.

\subsection{\Acl{ccs}-based schemes}\label{sec:ch5:num_results_ccs}

\subsubsection{Numerical evaluation of Theorem~\ref{thm:rcb}}

Let us start the numerical analysis of the \ac{ccs}-based schemes with the \ac{rcb} from Theorem~\ref{thm:rcb} for $t=0$, presented by a dashed orange line. The number of slots in a frame was optimized for each number of active users, and an \ac{omp}~\cite{Cai2011} decoding algorithm was utilized. The inner codebook is Gaussian and has a size of $2^{15}$ for each slot. The resulting energy efficiency is presented in~\Fig{fig:ch5:mac_numerical} by a dashed orange line.

\subsubsection{\texorpdfstring{$t$}{t}-tree code numerical evaluation}

Next, let us consider the performance of the tree-code scheme presented in Section~\ref{sec:treecode_practical} with $t=0$. Recall that the case $t=0$ was first considered in~\cite{codedCS2020}, and we follow these results in our presentation. The inner codebook size is $2^{15}$ for $K_a \geq 150$, and $2^{14}$ for $K_a < 150$. The inner codebook is a subcode of the $(2047, 23)$ \ac{bch} code. The slot length equals $2047$. The inner code was decoded by the \ac{nnls} algorithm. The number of slots is fixed and equals $11$. Hence, the frame length equals $22517$. The authors suggested transmitting $k=75$ information bits, which results in a similar code rate of $100/30000$ for fair comparison.

The main problem was the allocation of information bits to each slot. When decoding each subsequent slot, all codeword candidates must be kept, and this number may dramatically grow~\cite{Frolov2022TCOM}. Moreover, the final performance of the system may drastically degrade in the case of inappropriate information bit allocation. The information bit allocation has been carefully optimized, and the resulting performance is shown in~\Fig{fig:ch5:mac_numerical} by an orange line with $\times$-shaped markers. The distance between the \ac{rcb} and the scheme is about $2$ dB at $K_a = 150$, and it gradually increases.

\subsubsection{Further improvements}

More sophisticated schemes that include interaction between inner and outer codes show significantly better energy efficiency results. The results from~\cite{AmalladinneAMP2021} are shown by an orange line with square markers, and the results from~\cite{Amalladinne2022} are shown by triangle markers.

It is worth mentioning the last scheme, which does not use an outer code~\cite{Truhachev2024}. The outer code requires additional redundancy. The proposed solution eliminates this redundancy and makes the scheme the best-performing one when $K_a > 200$.

\subsection{Random spreading-based solutions}\label{sec:ch5:num_results_rspread}

In this section, we provide the numerical results of the practical schemes described in Section~\ref{sec:ch5_random_spreading_practical}, namely joint decoding of polar codes, power diversity, and \ac{soic}.

\subsubsection{Joint decoding of polar codes}

Let us start with polar codes and spreading~\cite{AmalladinnePolar2020}. As in the case of \ac{irsaf}, the authors propose to select the spreading sequence based on the first $k_p$ information bits. The length of the spreading sequence $n_s$ generates a trade-off between the interference mitigation capability and the single-user code rate, because the frame length $n = n_p \times n_c$, and $n_c$ is the code length. The optimal parameters for each number of active users are presented in Table~\ref{tab:rs_polar:code_parameters}. As soon as the \ac{crc}-aided polar codes have been selected, the \ac{crc} size is also specified. The resulting energy efficiency is presented in~\Fig{fig:ch5:mac_numerical} by a blue line with star-shaped markers.

\begin{table}
\caption{This is a summary of the encoding parameters used in~\cite{AmalladinnePolar2020} as functions of the number of active users.}
\label{tab:rs_polar:code_parameters}
\begin{center}
\begin{tabular}{ |c|c|c|c|c|c|c|c| } 
\hline
{$\Ka$} &$k_p$ &$k_c = k - k_p$ & $n_p$ & $n_c$ & {List size} & {\ac{crc} size}\\
\hline
$25-75$ & $9$ & $91$ & $29$ & $1024$ & $32$ & $16$ \\
\hline
$100-175$ & $8$ & $92$ & $29$ & $1024$ & $32$ & $16$ \\
\hline 
$200-225$ & $9$ & $91$ & $59$ & $512$ & $32$ & $12$ \\
\hline
$250-300$ & $10$ & $90$ & $117$ & $256$ & $32$ & $10$ \\
\hline
\end{tabular}
\end{center}
\end{table}

For a fixed value of $\Ka$, the quantities $n_c$ and $n_p$ are optimized empirically to minimize the energy per bit required to achieve a target probability of error. The spreading sequences were generated from a Gaussian codebook.

\subsubsection{Power diversity}

The next scheme~\cite{Duman2021} considers power diversity for spreading sequences that improves single preamble detection and subsequent \ac{sic}. The authors of~\cite{Duman2021} did not change the spreading sequence and code lengths ($n_c = 256$, $n_p = 117$), but varied the number of bits used to derive the spreading sequence: $k_p = 14$ for $150 \leq \Ka < 200$, $k_p = 15$ for $200 \leq \Ka < 250$, and $k_p = 16$ for $\Ka \geq 250$. This choice of parameters results in a small probability of preamble collision. The decoding list size was substantially larger ($512$) compared to the previous solution ($32$). The spreading sequences were also generated from a Gaussian codebook, but all spreading sequences were split into several groups having different powers. The number of groups and power levels were optimized to improve the signal-to-interference-plus-noise ratio at each decoding step. The resulting energy efficiency is presented in~\Fig{fig:ch5:mac_numerical} by a blue line with square markers. The curve starts from $K_a = 150$ because smaller $K_a$ values have only a single power splitting group.

\subsubsection{\Acl{soic}}

Finally, let us consider random spreading with a \ac{soic}~\cite{PradhanLDPCSC2021}. The spreading sequences are generated using $k_p = 12$ bits from the information message, and the length of the spreading sequence is $n_p = 86$. The \ac{ldpc} code rate is $1/4$. The structure of \ac{ldpc} codes was optimized using a \ac{de} on protographs~\cite{LivaChiani2007}. The performance of this scheme is presented by a blue line with circular markers in~\Fig{fig:ch5:mac_numerical}.

\chapter{Fading channels}\label{chap6}

More realistic channel models, such as the Rayleigh fading channel, have also been considered in the literature. While the \ac{awgn} channel may be a suitable model for satellite communications when multipath propagation is limited, the Rayleigh fading channel is more appropriate for cellular communications, as it assumes a rich scattering environment.

To operate in a fading channel, \ac{csi} must be acquired. However, estimating the channel for a large number of devices transmitting short packets is prohibitively difficult. Indeed, the best way to estimate the channel is to include a preamble or reference signals in the transmitted message. However, in the case of uncoordinated transmission, we must ensure that these reference signals remain nearly orthogonal. Unfortunately, orthogonality can be compromised by both short transmitted messages and a large number of simultaneously active users. Consequently, the \ac{ura} model assumes the absence of the \ac{csi} at both the transmitters and the receiver. The works~\cite{Kowshik2021TIT, Kowshik2020TCOM, Truhachev2022, Frolov2022TCOM} describe fundamental limits and practical schemes for single-antenna quasi-static Rayleigh fading channels. The most notable observation is that fading increases the required energy per bit by approximately $10$ dB under the same order of \ac{pupe} requirements (see~\Fig{fig:ch5:mac_numerical} and~\Fig{fig:chap6:ebno_ka_single_antenna}).

A.~Fengler and G.~Caire were the first to demonstrate that employing multiple antennas at the \ac{bs} can significantly improve energy efficiency without increasing transmitter complexity~\cite{Fengler2021}. Recall that the transmitters are autonomous, battery-powered devices and should be as simple as possible. The fundamental limits for \ac{ura} with a \ac{mimo} receiver were established in \cite{Poor2022ML, noPoorUnsourcedML2023, MIMO2023ours}. Of particular importance is the scaling law derived in \cite{Fengler2022} and refined in \cite{Poor2022ML}: with $n$ channel uses and a sufficiently large number of BS antennas ($\Ka/L = o(1)$), up to $\Ka = O(n^2)$ active users can be served among $\Kt$ potential users.

Practical receiver schemes have been extensively studied in the literature. The design of these schemes typically involves pilot-based channel estimation followed by decoding and \ac{sic}~\cite{Fengler2022, Duman2022}, or a \ac{ccs} approach combined with non-Bayesian activity detection~\cite{Fengler2021}, \ac{sic}~\cite{AmalladinneMIMO2021}, or the exploitation of slot-wise correlations~\cite{Shyianov2020}. The scheme proposed in~\cite{FASURA} integrates these approaches and delivers state-of-the-art performance for the \ac{mimo} channel.

This chapter is organized as follows. In Section~\ref{ch6:sec:channel_model}, we specify the channel model -- a quasi-static Rayleigh fading MAC with a \ac{bs} equipped with either a single or multiple receive antennas. To define the achievability bounds, we propose two decoding rules in Section~\ref{ch6:sec:decoding_rules}: \ac{ml} decoding and projection-based decoding. Furthermore, we show that projection-based decoding achieves the $\varepsilon$-capacity. Next, in Section~\ref{ch6:sec:single_antenna}, we consider single-antenna bounds, and in Section~\ref{ch6:sec:multi_antenna}, we describe known results for the multiple-antenna case. Finally, we conclude this chapter by discussing low-complexity schemes presented in Sections~\ref{sec:ch6:numerical_single} and~\ref{sec:ch6:numerical_multi} for single-antenna and multiple-antenna scenarios, respectively.

\section{Channel model}
\subsection{Rayleigh fading model}\label{ch6:sec:channel_model}
The main cause of fading is multipath propagation. The receiver experiences interference of multiple transmitted signal copies. Each copy of the signal can experience random attenuation due to reflections from obstacles and random delays due to different propagation path length. Finally, a moving receiver may experience a Doppler spread, but this is not the case for the massive machine-type communications, where most transmitting devices are stationary or moving relatively slowly.

The Rayleigh fading model corresponds to a so-called rich scattering environment. When the number of scatterers becomes large and each reflection coefficient is random, the central limit theorem can be applied. Moreover, in such a rich scattering environment, there is no line-of-sight component, and the received power distribution becomes uniform over the \ac{aoa}.

As a result, for a single-antenna receiver, the channel coefficient follows a circularly symmetric complex normal distribution, and the magnitude of this channel coefficient becomes Rayleigh-distributed. For a \ac{mimo} receiver, assuming a uniform power distribution over the \ac{aoa} and a regular antenna array, the resulting channel coefficients at each receiving antenna follow the same distribution~\cite{Tse_Viswanath_2005}.

When the transmitting devices are stationary and operate in a slowly-varying environment, the resulting channel model becomes quasi-static. Furthermore, since each device transmits infrequently, the channel coefficient remains constant during a transmission but changes from block to block, with channel realizations being independent across different blocks. Moreover, the channel coefficients (under the aforementioned central limit theorem) follow a circularly symmetric complex normal distribution, meaning that the phase of the signal is uniformly distributed. Additionally, small synchronization errors can also be described by this phenomenon.

The time-domain delay results in a phase shift of the signal in the frequency domain. On the other hand, \ac{ofdm} provides an elegant solution for frequency-domain channel equalization, as long as the length of the cyclic prefix exceeds the maximum relative delay value~\cite{FadingASYNC2019}. However, \ac{ofdm} signals have a large peak-to-average power ratio, which may lead to high energy consumption under linear amplification, potentially impacting battery life. Therefore, we plan to focus on simpler (and constant-envelope) modulations at low rates. In this context, synchronization errors can partially be modeled as phase multipliers in the channel coefficients.

Now, let us formulate the \ac{ura} problem for the fading channel scenario with a receiver equipped with $L$ antennas. We assume that each transmitting device has a single transmit antenna. Let $\dset{T} = {W_1, W_2, \ldots, W_{K_a}}$ denote the set of transmitted messages (using the same codebook with blocklength $n$). The channel model can be described as follows.
\begin{equation}
\label{eq:channel_model}
\rmat{Y} = \dmat{X} {\dmat{\Phi}}(\dset{T}) \rmat{H} + \rmat{Z},
\end{equation}
where $\rmat{Y} = [\dvec{y}_1, \ldots, \dvec{y}_L] \in \mathbb{C}^{n \times L}$ is the channel output matrix, $\dmat{X} = [\dvec{x}_1, \ldots, \dvec{x}_M] \in \mathbb{C}^{n \times M}$ with $\dvec{x}_W = f(W), W \in [M]$, is the common codebook. Here, $\dmat{\Phi}(\dset{T}) \in \left\{0,1\right\}^{M \times K_a}$ represents the activity matrix, where each column contains a single non-zero element at the position corresponding to the user's message. The matrix $\rmat{Z} \in \mathbb{C}^{n \times L}$ denotes the \ac{awgn} matrix, with each element sampled i.i.d. from $\cn{0,1}$. Finally, $\rmat{H} \in \mathbb{C}^{K_a \times L}$ is the matrix of fading coefficients, where each element is sampled i.i.d. from $\cn{0,1}$. Note that the fading coefficients are independent of the codewords and the noise matrix $\rmat{Z}$.

In what follows, we assume that $\rmat{H}$ is unknown to both the transmitters and the receiver (the unknown \ac{csi} scenario). As a result, the receiver observes $L$ linear combinations of the transmitted messages, where the coefficients of these linear combinations are determined by the rows of the matrix $\rmat{H}$.

\subsection{Decoding rules}\label{ch6:sec:decoding_rules}

In the case of the \ac{awgn} channel model, we considered the set of codewords closest to the received signal, as described in~\eqref{eq:ura_awgn_decoding_rule}. However, since the new channel model involves unknown \ac{csi}, the decoding rule must be modified. The first idea was proposed in~\cite{Polyanskiy2014quasi}. In the absence of noise, the signal received by each antenna lies in the subspace spanned by the transmitted codewords, regardless of the fading coefficients. Thus, the most likely subset of transmitted codewords is the one that maximizes the projection of the received signal onto the subspace spanned by these codewords.

As we will see later, the projection method ignores the prior distribution of fading coefficients. Additionally, the projection-based decoder has a limitation: it restricts the maximum number of active users to the blocklength, i.e., $\Ka < n$. For the case of multiple antennas, as shown in~\cite{Poor2022ML}, the maximum number of users that can be decoded may exceed the blocklength. To address this issue, an \ac{ml}-based decoding rule was proposed in~\cite{Poor2022ML}.

\ac{ml} decoding rules~\cite{Poor2022ML, noPoorUnsourcedML2023} estimate the transmitted codeword set by maximizing the conditional probability
$$
\Pb{\rmat{Y} \middle| \dmat{X}_{\dset{R}'}} = \mathbb{E}_{\rmat{H}}\left\{\Pb{\rmat{Y} \middle| \dmat{X}_{\dset{R}'}, \rmat{H}}\right\},
$$
which is the expectation over the channel statistics. Note that the receiver does not know the channel realizations but may have access to the channel statistics in many scenarios. Let us denote the matrix composed of the codewords from $\dset{R}'$ as $\dmat{X}_{\dset{R}'}$. The \ac{ml} decoding rule is given by
\begin{equation}\label{eq:ml_decoder}
\dset{R} = \argmax_{\dset{R}' \subset [M], \left| \dset{R}'\right| = K_a}
\mathbb{E}_{\rmat{H}}\left\{\Pb {\dmat{Y} \middle| \dmat{X}_{\dset{R}'}, \rmat{H}}\right\}.
\end{equation}

The idea of the projection decoder is to take not the expectation, but the maximum of the corresponding conditional probability, as described in~\cite{Kowshik2020TCOM}. The projection-based decoding rule is given by
\begin{equation}\label{eq:proj_decoder}
\dset{R} = \argmax_{\dset{R}' \subseteq [M], \left|\dset{R}'\right| = K_a} \max_{\dmat{H}}\left\{ \Pb{\dmat{Y} \middle| \dmat{X}_{\dset{R}'}, \dmat{H}}\right\}.
\end{equation}

Note that the probability $\Pb{\dmat{Y} \middle| \dmat{X}_{\dset{R}'}, \dmat{H}}$ is governed by Gaussian noise, as the conditioning results in a known linear combination of a candidate set $\dmat{X}_{\dset{R}'}$ of codewords. Assuming Gaussian noise, we can rewrite~\eqref{eq:proj_decoder} as
\begin{eqnarray}
 \label{eq:max_proj_problem}
 \dset{R} &=& \argmin_{\dset{R}' \subseteq [M], |\dset{R}'| = K_a} \min_{\dmat{H}} \norm{\dmat{Y} - \dmat{X}_{\dset{R}'} \dmat{H}}_F^2 \nonumber \\
 &=& \argmin_{\dset{R}' \subseteq [M], |\dset{R}'| = K_a} \left(\norm{\dmat{Y}}^2_F - \norm{\dmat{P}_{\dset{R}'}\dmat{Y}}^2_F \right) \nonumber \\
 &=& \argmax_{\dset{R}' \subseteq [M], |\dset{R}'| = K_a} \norm{\dmat{P}_{\dset{R}'}\dmat{Y}}^2_F,
\end{eqnarray}
where the second equality follows from the fact that the received signal $\dmat{Y}$ is fixed. Here $\norm{\cdot}^2_F$ denotes the Frobenius norm, and $\dmat{P}_{\dset{R}'}$ is the orthogonal projection onto the subspace spanned by the codewords from $\dset{R}'$.

A comparison of~\eqref{eq:ml_decoder} and~\eqref{eq:proj_decoder} clarifies the previous statement that a projection-based decoding rule does not take channel statistics into account.

\section{Single-antenna quasi-static Rayleigh fading \acs{mac}}\label{ch6:sec:single_antenna}

In this section, we present existing bounds for the single-antenna quasi-static Rayleigh fading \ac{mac}. First, we describe the $\varepsilon$-capacity bound, which serves as a suitable capacity measure for quasi-static channels. Next, we describe the Shamai-Bettesh asymptotic bound, which provides insight into the corresponding bounds in the \acl{fbl} regime. Finally, we consider the \ac{rcb} for the \ac{fbl} scenario (Theorem~\ref{theorem:pupe_rayleigh_single}) and the converse bound (Theorem~\ref{th:converse_sa}).

\subsection{Capacity region, all-active users}

In the presence of fading, the capacity of the channel should be defined differently compared to Shannon's formula for the \ac{awgn} channel. Consider a slow fading channel model~\cite{Tse_Viswanath_2005}, or a channel with a large coherence block. Our quasi-static model fits well with the definition of the slow fading model. For any prospective rate of reliable communication, there is a probability that a channel realization will be of such a small magnitude that the user will be unable to communicate reliably over this channel at the selected rate. Thus, in the slow-fading regime, an $\varepsilon$-capacity should be introduced instead.

Similarly, let $C_{\varepsilon, J}$ denote the $\varepsilon$-capacity region, which can be defined for the case of all-active $\Ka$~\cite{Kowshik2021TIT} under \emph{joint} decoding. A rate tuple $\br{R_1, \ldots, R_{K_a}}$ is said to be $\varepsilon$-achievable~\cite{Han1998} if there is a sequence of codes whose rates are asymptotically at least $R_i$ such that the joint error is asymptotically smaller than $\varepsilon$.

$$
C_{\varepsilon, J} = \left\{R = \br{R_1, \ldots, R_{K_a}}: \forall i, R_i \geq 0 \, \text{and} \, P_0\br{R} \leq \varepsilon \right\},
$$
where
$$
P_0\br{R} = \Pb{
\bigcup_{S\subset [K_a]}\left\{
\log\br{1 + P\sum_{i\in S}\left|H_i\right|^2} \leq \sum_{i\in S}R_i
\right\}
}.
$$
\begin{theorem}[Projection decoding achieves $C_{\varepsilon, J}$~\cite{Kowshik2021TIT}]
Let $R\in C_{\varepsilon, J}$. Then $R$ is $\varepsilon$-achievable through a sequence of codes with the decoder being the projection decoder~\eqref{eq:proj_decoder}.
\end{theorem}

Nevertheless, in the \ac{fbl} regime and under the \ac{pupe} metric, the projection decoder may be suboptimal. However, this type of decoder allows the achievability bound to be addressed, as presented in Theorem~\ref{th:mimo_ach_prj}.

\subsection{Shamai-Bettesh capacity bound}\label{ch6:shamai_bettesh}

There is another asymptotic bound ($n \to \infty$) for the \ac{pupe} in the case of symmetric rates and large $\Ka$, the Shamai-Bettesh capacity bound from~\cite{bettesh2000outages}. The idea is as follows: the joint decoder knows the realization of the fading coefficients, and users are ranked according to the strength of their fading coefficients. It first tries to decode all users. If it fails (i.e., if the rate vector is not inside the instantaneous full capacity region), it drops the user with the smallest fading coefficient and tries to decode the remaining $\Ka - 1$ users. The dropped user then contributes to the noise.

This process continues iteratively, and the fraction of users that were not decoded corresponds to the \ac{pupe}. Since the case under discussion involves large $\Ka$, the order statistics of the absolute value of the fading coefficients crystallize (i.e., become almost non-random), and hence analytical expressions can be derived for the outage in terms of spectral efficiency ($kK_a / n$) and total power.

Note that the Shamai-Bettesh bound is only an achievable bound (i.e., it is not guaranteed to be tight) for the capacity under \ac{pupe}. It does not apply to our setting for two reasons. First, it assumes different codebooks for different users, and second, it assumes asymptotically large blocklength.

Note also that the asymptotic regime considered in the Shamai-Bettesh bound is as follows: first, $n$ is taken to $\infty$ (under a fixed $K_a$), and second, $K_a$ is also taken to infinity. However, based on studies of the non-fading channels~\cite{ZPT-isit19}, the correct asymptotic regime is to take both $K_a$ and $n$ to infinity at a fixed ratio~\cite{Kowshik2021TIT}.

\subsection{\ac{fbl} regime, partial activity}

Let us formulate the achievability bound for a single-antenna \ac{ura} under a quasi-static Rayleigh fading channel. We provide a simplified formulation and outline the main steps of the proof, referring the reader to~\cite{Kowshik2020TCOM} for the complete proofs.

As in the case of \ac{awgn}, we consider the event $\Pb{\dset{E}_t}$~\eqref{eq:chapter4:pupe_definition}, which represents exactly $t$ errors occurring. We assume a circularly symmetric Gaussian codebook, similar to Definition~\ref{def:gauss_codebook}.

\begin{define}\label{def:gauss_codebook_cn}
Let $\mathcal{G}_{\mathcal{CN}}\br{M,n,P}$ be the ensemble of Gaussian codebooks of size $n \times M$, where each element is sampled i.i.d. from $\mathcal{CN}\br{0,P}$.
\end{define}
To estimate the \ac{pupe}, we must evaluate $\Pb{\dset{E}_t}$ for $t \in [K_a]$. Similar to the \ac{awgn} case, we denote the probability of power violation and message collision as $p_0$.

\begin{theorem}[Achievability bound for single-antenna case~\cite{Kowshik2020TCOM}]\label{theorem:pupe_rayleigh_single}
Fix $P' < P$. There exists a codebook $\dmat{X}^* \in \mathcal{G}_{\mathcal{CN}}\br{M,n,P'}$ satisfying the power constraint $P'$ and providing $P_e$ using eq.~\eqref{eq:chapter4:pupe_definition}, where $p_{0}$ is bounded by~\eqref{eq:p0_bound}, and $\Pb{\dset{E}_t} \leq p_t$, where
\begin{equation}
\label{eq:pt_fading_siso}
p_t \leq
\inf_{\delta>0}\left(\binom{K_a}{t}e^{-(n-K_a)\delta} + p_{1, t}\right), 
\end{equation}
where
\begin{equation}\label{eq:ch6_prj_best}
p_{1, t} = \Pb{\bigcup_{\substack{\dset{M}\subset \dset{T}\\ \left|\dset{M}\right|=t}} \left\{G_{\dmat{Y},\dset{T},\dset{C},t}\geq V_{n,t} \right\}},
\end{equation}
\begin{equation}
\label{eq:g_projection}
G_{\dmat{Y},\dset{T},\dset{C},t}=\frac{\norm{\dmat{Y}}^2-\norm{\dmat{P}_{\dset{T}}\dmat{Y}}^2}{\norm{\dmat{Y}}^2-\norm{\dmat{P}_{\dset{C}}\dmat{Y}}^2}.
\end{equation}
The Frobenius norm in~\eqref{eq:g_projection} becomes a simple Euclidean norm for the single-antenna case, and
$$
\quad V_{n,t}=e^{-\left(\delta+R_1+s_t\right)}, \quad 
s_t=\frac{\ln\binom{n-K_a+t-1}{t-1}}{n-K_a}, \quad
R_1=\frac{\ln\binom{M-K_a}{t}}{n-K_a}.
$$
\end{theorem}

The original formulation~\cite{Kowshik2020TCOM} of the theorem above assumed that the number of active users may differ from the intended list size at the receiver. This approach was inspired by the Shamai-Bettesh asymptotic bound described above (see Section~\ref{ch6:shamai_bettesh}). This approach may be beneficial in the case of large $\Ka$: it may be better not to consider the weakest users but instead decode the remaining ones with a smaller probability of error. In the formulation above, we omit this assumption for simplicity and consider the case when the decoder tries to extract all active users.

This theorem introduces a projection ratio~\eqref{eq:g_projection}, which depends on the choice of the correctly received codewords, and a union probability is taken in~\eqref{eq:ch6_prj_best}. This union over all choices can be replaced by simply selecting the $t$-\emph{weakest}\footnote{Or codewords that have the smallest received energy} codewords as missed, significantly reducing the computational complexity. To evaluate the achievability bound numerically, the tail probability of the projection ratio can be computed either by sampling or analytically (see the complete proof of this theorem in~\cite{Kowshik2020TCOM}).

The main step of this theorem's proof avoids a large combinatorial term $\binom{M-K_a}{t}$ described in Section~\ref{sec:pt_naive}. This is done by observing that a projection of the fixed vector onto a random subspace generated by Gaussian vectors arises. The norm of this projection follows a beta distribution~\cite{Kowshik2021TIT}.

Let us rewrite the condition that exactly $t$ errors occurred conditioned on the channel coefficients, noise, and a set of transmitted codewords. By applying the union bound, we have
\begin{equation}
\label{eq:P_FT}
\Pb{\left.\dset{E}_t \right| \dset{T}, \rmat{H}, \rmat{Z}} \leq \Pb{
\left.
\bigcup\limits_{\dset{M}}\bigcup\limits_{\dset{F}} F_{\dset{F}, \dset{C}, t} \right| \dset{T}, \dmat{H}, \dmat{Z}
},
\end{equation}
where
$$
F_{\dset{F}, \dset{C}, t} = \left\{
\norm{\dmat{P}_{\dset{C} \cup \dset{F}}\dmat{Y}}^2 > \norm{\dmat{P}_{\dset{C} \cup \dset{M}}\dmat{Y}}^2
\right\}.
$$
Note that the union $\bigcup\limits_{\dset{M}}$ is taken over $\binom{M-K_a}{t}$ variants. Hence, the expectation over $\dset{M}$ should be performed first. To perform this step, let us fix the set of correctly received codewords $\dset{C}$ and show that the projection norm $\norm{\dmat{P}_{\dset{C} \cup \dset{F}}\dmat{Y}}^2$ is distributed as
$$
\norm{\dmat{P}_{\dset{C} \cup \dset{F}}\dmat{Y}}^2 \sim\norm{\dmat{P}_{\dset{C}}\dmat{Y}}^2 + \norm{\dmat{P}^\perp_{\dset{C}}\dmat{Y}}^2\beta\br{n-K_a, t}.
$$

Following the notation from~\cite{MIMO2023ours}, $\dmat{P}_{\dset{C} \bigcup \dset{F}} = \dmat{P}_{\dset{C}} + \dmat{P}_{\dset{C}^\perp, \dset{F}}$ by a Gram-Schmidt procedure, where $\dmat{P}_{ \dset{C}^\perp, \dset{F}}$ is a projection onto $\Span{\dmat{P}^\perp_{\dset{C}} \dmat{X}_{\dset{F}}}$, and $\dmat{P}^\perp_{\dset{C}} = \dmat{I} - \dmat{P}_{\dset{C}}$.

Finally, under the fixed $\dset{C}$, $\norm{\dmat{P}_{\dset{C} \cup \dset{F}}\dmat{Y}}^2$ is a random variable that depends only on falsely detected codewords. Moreover,
\begin{equation}\label{eq:ch6:proj_ratio_beta}
\frac{\norm{\dmat{P}_{\dset{C}^\perp, \dset{F}}\dmat{Y}}^2}{\norm{\dmat{P}^\perp_{\dset{C}}\dmat{Y}}^2}\sim \beta\br{n-K_a, t}
\end{equation}
is beta-distributed~\cite{Nielsen1999} as a projection of a fixed vector onto a random subspace spanned by $\Span{\dmat{P}^\perp_{\dset{C}} \dmat{X}_{\dset{F}}}$ in the dimension $n - K_a + t$ -- the dimensionality of a subspace orthogonal to $\Span{\dmat{X}_\dset{C}}$.

The final building block of the theorem proof is that the cumulative density function of $1 - \beta$ (where $\beta$ is beta-distributed) random variable is bounded by
\begin{equation}
\label{eq:beta_cdf_bound}
F_{1 - \beta}\left( x; n-K_a, t \right) \leq \binom{n-K_a + t - 1}{t - 1}x^{n - Ka},
\end{equation}
and the resulting probability of $t$ errors occurring will be upper-bounded by
$$
\Pb{\left.\dset{E}_t \right| \dset{T}, \rmat{H}, \rmat{Z}} \leq
\mathbb{E}_{\dset{T}, \mathbf{H}, \mathbf{Z}}\left[\min\left\{1, 
\sum_{
\substack{\dset{C}\subset \dset{T}\\ \left|\dset{C}\right| = K_a - t}
}e^{\left(n - K_a\right)\left(s_t + R_1\right)}G_{\dmat{Y},\dset{T},\dset{C},t}^{n - K_a}
\right\}\right].
$$
The parameter $\delta$ in the theorem appears after applying Fano's trick (see Section~\ref{sec:fano_trick}) on the formula above.

To perform numerical evaluation of this bound, the following steps must be performed:
\begin{itemize}
\item Sample $\rmat{H}$, codewords from $\dset{T}$, and $\rmat{Z}$,
\item Evaluate the cumulative distribution function of~\eqref{eq:g_projection}. A simplification by choosing $t$ \emph{weakest} users rather than checking $\binom{K_a}{t}$ variants can be considered,
\item Perform optimization over $\delta$.
\end{itemize}

Now let us describe a simple converse bound based on results from \cite{Polyanskiy2013SIMO} and the meta-converse from \cite{polyanskiy2010channel}.
\begin{theorem}[Multi-user, single-antenna converse~\cite{Kowshik2020TCOM}]
\label{th:converse_sa}
Let 
$$
L_n=n\log(1+PG)+\sum_{i=1}^{n}\left(1-|\sqrt{PG}Z_i-\sqrt{1+PG}|^2\right)
$$
and
$$
S_n=n\log(1+PG)+\sum_{i=1}^n \left(1-\frac{|\sqrt{PG}Z_i-1|^2}{1+PG}\right),
$$
where $G = |H|^2$ and $Z_i\overset{\text{i.i.d.}}{\sim}\mathcal{CN}\left(0,1\right)$. Then, for every $n$ and $0 < \epsilon < 1$, any $(M,n-1,\epsilon)$ code for the quasi-static $K_a$ \ac{mac} satisfies
$$
\log(M)\leq \log(K_a)+\log\frac{1}{\Pb{L_n\geq n\gamma_n}}
$$
where $\gamma_n$ is the solution of 
$$
\Pb{S_n\leq n\gamma_n}=\epsilon.
$$
\end{theorem}

\section{Multi-antenna quasi-static Rayleigh fading \acs{mac}}\label{ch6:sec:multi_antenna}

In this section, we consider the case where a \ac{bs} is equipped with multiple antennas -- the so-called \ac{mimo} scenario. As discussed previously, the projection-based decoder achieves the $\varepsilon$-capacity in the single-antenna case. However, as shown in Section~\ref{ch6:scaling_laws}, the number of successfully detected users can exceed the frame length, i.e., $\Ka$ can be greater than the frame length $n$ in the \ac{mimo} case. Consequently, when $\Ka > n$, the projection-based decoding~\eqref{eq:proj_decoder} will no longer work, as the subspace spanned by the transmitted codewords coincides with the full signal space.

To address this problem, the \ac{ml}-based bound in~\eqref{eq:ml_decoder} was proposed in~\cite{Poor2022ML}. Previously, we considered the single-antenna case, where the approach relied on eliminating the large combinatorial term $\binom{M - \Ka}{t}$. However, in the proposed \ac{ml}-based bound, this simplification does not lead to an efficient solution. Evaluating this bound requires sampling codewords from all possible sets: correctly detected codewords $\dset{C}$, falsely detected codewords $\dset{F}$, and missed codewords $\dset{M}$ (see \Fig{fig:hat_s_structure}).

On the other hand, a similar approach based on projection decoding~\eqref{eq:proj_decoder} offers some improvements when the number of active users $\Ka$ is less than $n$. Nonetheless, deriving a projection-based achievability bound in a manner similar to the single-antenna case still does not yield an efficient solution. This is because the projection ratios~\eqref{eq:ch6:proj_ratio_beta} have an intractable distribution, requiring the handling of angles between random subspaces.

\subsection{Scaling laws}\label{ch6:scaling_laws}

The benefits of multiple antennas were analyzed by Fengler et al.~\cite{Fengler2021}, where the authors examined the scenario of \emph{activity detection}. A set of active users sends their $K_a$ randomly chosen preambles (selected from a total of $K_\text{tot}$ preambles) in a dedicated slot. Since detecting a single transmission in a massive \ac{mimo} regime requires accurate channel estimation, activity detection plays a crucial role in practical implementations. The goal of the receiver is to detect the subset of preambles selected for transmission and then proceed to the user's signal detection step, where the \ac{csi} is obtained from the preambles.

This problem is closely related to the \ac{ura} problem. Given multiple antennas ($L > 1$), the corresponding \ac{ura} problem~\eqref{eq:channel_model} can be formulated as a \ac{mmv} problem in the context of \ac{cs}~\cite{Chen2008}\nocite{FadingISIT2019} (see also Section~\ref{sec:ura_as_cs}).

The scaling law proposed in~\cite{Fengler2021} states that given
$$
\frac{K_a}{L} \approx {o}\br{1},  
$$
the number of detectable users is
$$
\mathcal{O}\br{\frac{n^2}{\log^2\br{\frac{K_{\text{tot}}}{K_a}}}}.
$$
Notably, this number can be significantly higher compared to the projection-based decoding algorithm, where $K_a < n$.

\subsection{Achievability bounds for \ac{mimo} receiver}\label{chap6:achievability}

In this section, we describe the achievability bound for the \ac{mimo} case based on random coding, as introduced in Section~\ref{sec:pt_naive}. Given the scaling law described above, there is a motivation to shift away from the projection-based decoding algorithm, which is limited by $\Ka < n$. As the number of receiver antennas increases, an \ac{ml}-based decoding algorithm is chosen for the bound~\cite{Poor2022ML}.

We consider a Gaussian codebook ensemble (see Definition~\ref{def:gauss_codebook_cn}) and recall the channel model~\eqref{eq:channel_model}. Suppose the~\emph{channel statistics}, i.e., the distribution of channel coefficients -- the elements of $\rmat{H}$ -- are known. Then, conditioned on ${\dmat{X}\dmat{\Phi}}(\dset{T})$, the columns of $\rmat{Y}$, the signal vector received by each antenna, are independent and normally distributed:
$$
\rvec{y}_i \sim \mathcal{CN}\br{\dmat{0}, \dmat{I}_n + \dmat{X} \dmat{\Gamma} \dmat{X}^H}, \quad \dmat{\Gamma} = \dmat{\Phi}(\dset{T})\dmat{\Phi}^H(\dset{T}),
$$
where $\br{\cdot}^H$ denotes the Hermitian transpose. Letting $\dmat{\Sigma} = \dmat{I}_n + \dmat{X}\dmat{\Gamma}\dmat{X}^H$, we obtain the following p.d.f.:
$$
p\br{\rmat{Y}\middle|\dmat{X}\dmat{\Phi}\br{\dset{T}}} = \pi^{-Ln}\left|\dmat{\Sigma}\right|^{-L} \exp\br{-\text{tr}\br{\dmat{\Sigma}^{-1}\rmat{YY}^H}}.
$$
The corresponding decoding function to be optimized is:

\begin{equation}
\label{eq:ml_loss_function}
g\br{\rmat{Y}, \dmat{X} \dmat{\Gamma} \dmat{X}^H} = L \log \left|\dmat{\Sigma}\right| + \text{tr}\br{\dmat{\Sigma}^{-1} \rmat{Y} \rmat{Y}^H}.
\end{equation}

The decoder aims to find the set of active users that minimizes~\eqref{eq:ml_loss_function}. In the subsequent analysis, we define the matrix $\dmat{\Gamma}$ formed by the message set $\dset{S}$ as $\dmat{\Gamma}_\dset{S}$ and denote~\eqref{eq:ml_loss_function} as $g\br{\dmat{\Gamma}_\dset{S}}$. As before, we estimate the probability of the event $\dset{E}_t$, in which exactly $t$ errors occur:
$$
\Pb{\dset{E}_t} = \Pb{
\dset{E}_{t, \dset{M}, \dset{F}}
}, \quad \dset{E}_{t, \dset{M}, \dset{F}} = \bigcup_{\dset{F}} \bigcup_{\dset{M}} \left\{
g\br{
\dmat{\Gamma}_{\dset{C}\cup\dset{F}}
} \leq g\br{
\dmat{\Gamma}_{{\dset{C}\cup\dset{M}}}
}
\right\}.
$$

Clearly, we can use the union bound to upper-bound the right-hand side. However, the union bound is known to overestimate the resulting probability. To tighten the bound, we use Fano's trick, i.e., we introduce the following region:
\begin{equation}
\label{eq:ball}
\dset{B}_{\dset{M}} = \{\rmat{Y}: g\br{\dmat{\Gamma}_{\dset{T}}} \leq \alpha g\br{\dmat{\Gamma}_{\dset{C}}} + \beta nL \}, \quad \dset{B} = \bigcap\limits_{\dset{M}} \dset{B}_{\dset{M}}.
\end{equation}
Using Fano's trick, we upper-bound $\Pb{\dset{E}_t}$ as follows:
\begin{equation}\label{eq:mimo_ml_fano}
\Pb{\dset{E}_t} = \Pb{
\dset{E}_{t, \dset{M}, \dset{F}}
} \leq
\Pb{
\dset{E}_{t, \dset{M}, \dset{F} } \bigcap \dset{B}
} + \Pb{\dset{B}^c}.
\end{equation}

The main idea is as follows. We apply the union bound only when $\bY$ is within an allowed region, while we assume that $\Pr[E_t] = 1$ otherwise. This approach prevents overestimating the probability by avoiding the use of the union bound for $\bY$ outside the good region, where the union bound is effective.

\begin{theorem}[\ac{ml}-based achievability bound~\cite{Poor2022ML} for the same codebook~\cite{noPoorUnsourcedML2023}]\label{th:mimo_ach_ml}
Fix $P' < P$ and consider a receiver equipped with $L$ antennas. There exists a codebook $\dmat{X}^* \in \mathcal{G}_{\mathcal{CN}}\br{M,n,P'}$ satisfying the power constraint $P'$ and providing $P_e$ using eq.~\eqref{eq:chapter4:pupe_definition}, where $p_{0}$ is bounded by~\eqref{eq:p0_bound}, and $\Pb{\dset{E}_t} \leq p_t$, where
\begin{equation}\label{eq:pt_ml_mimo}
p_t \leq \inf_{0\leq \alpha \leq 1, 0 \leq \beta}\left\{q_{1, t}\br{\alpha, \beta} + q_{2, t}\br{\alpha, \beta}\right\},
\end{equation}
where
\begin{equation}\label{eq:mimo_ach_ml_q1t}
q_{1, t}\br{\alpha, \beta} = \binom{K_a}{t}\binom{M - K_a}{t}\mathbb{E}_{\rmat{X}}\left[
\inf_{\begin{array}{c}u\geq 0, r\geq 0, \\ \lambda_{\text{min}}\br{\rmat{B}} > 0\end{array}}\br{e^{Lrn\beta}\cdot e^{L\mu}}
\right],
\end{equation}
$$
\mu = \br{u - r}\logdet{\rmat{F}_{\dset{T}}} - u\logdet{\rmat{F}_{\dset{R}}} + r\alpha\logdet{\rmat{F}_{\dset{C}}} - \logdet{\rmat{B}},
$$
\begin{equation}\label{eq:mimo_ach_ml_q2t}
q_{2,t} = \binom{K_a}{t}\inf_{\delta \geq 0}\mathbb{E}_{\rmat{X}}\left[
\frac{\gamma\br{Lm, c_\delta}}{\Gamma\br{Lm}} + 1 - \frac{\gamma\br{nL, \br{1 + \delta}nL}}{\Gamma\br{Ln}}
\right],
\end{equation}
$$
c_{\delta} = L\frac{n\br{1 + \delta}\br{1 - \alpha}-\alpha\logdet{\rmat{F}_{\dset{C}}} + \logdet{\rmat{F}_{\dset{T}}} - n\beta}{\alpha\prod_{i=1}^m\lambda_i^{1/m}},
$$
$$
\rmat{B} = \br{1 - u + r}\dmat{I}_n + u\rmat{F}_{\dset{R}}^{-1}\rmat{F}_{\dset{T}} - r\alpha\rmat{F}_{\dset{C}}^{-1}\rmat{F}_{\dset{T}},
$$
where $m = \min \left\{n, t\right\}$. The matrix $\rmat{F}_{\dset{S}}$ is the covariance matrix corresponding to the codeword set $\dset{S}$, given by
$$
\rmat{F}_{\dset{S}} = \dmat{I}_n + \rmat{X}\rmat{\Gamma}_{\dset{S}}\rmat{X}^H.
$$
By the notation above,
$$
\rmat{F}_{\dset{S}} = \dmat{I}_n + \rmat{X}_{\dset{S}}\rmat{X}^H_{\dset{S}}.
$$
Finally, $\lambda_1, \ldots, \lambda_m$ are the eigenvalues of the matrix $\rmat{F}_{\dset{C}}^{-1}\rmat{X}\Gamma_{\dset{M}}\rmat{X}^H$ of rank $m$, arranged in decreasing order.
\end{theorem}

In this theorem, the expectation is taken over a codebook ensemble. Thus, w.l.o.g., let us assume that the correctly received, missed, and falsely detected codewords (see Section~\ref{chap3:sys_model} and Fig.~\ref{fig:hat_s_structure}) correspond to the first $K_a + t$ codewords of the codebook $\rmat{X}_{[K_a + t]}$. Hence, we define the sets as follows:
\begin{align*}
\text{Transmimtted codewords }\dset{T}=& [K_a], \\
\text{Received codewords }\dset{R}=& [K_a + t] \backslash [t], \\
\text{Correctly received codewords }\dset{C}=& [K_a] \backslash [t], \\
\text{Missed codewords }\dset{M}=& [t], \\
\text{Falsely detected codewords }\dset{F}=& [K_a + t] \backslash [K_a].
\end{align*}

The idea is to first apply Fano's trick~\eqref{eq:mimo_ml_fano}. Then, for the first probability term, $\Pb{\dset{E}_{t, \dset{M}, \dset{F} } \bigcap \dset{B}}$, a Chernoff bound can be applied (Lemma~\ref{lemma:chernoff}).

The proof ultimately results in an expectation over a codebook ensemble, which can be evaluated numerically. However, the minimum required number of samples remains unknown. Moreover, sampling falsely detected codewords may be challenging due to the large number of possible combinations, given by $\binom{M-K_a}{t}$.

\begin{theorem}[Projection-based achievability bound~\cite{MIMO2023ours} for the same codebook]\label{th:mimo_ach_prj}
Fix $P' < P$ and consider a receiver equipped with $L$ antennas. There exists a codebook $\dmat{X}^* \in \mathcal{G}_{\mathcal{CN}}\br{M,n,P'}$ satisfying the power constraint $P'$ and providing $P_e$ using eq.~\eqref{eq:chapter4:pupe_definition}, where $p_{0}$ is bounded by~\eqref{eq:p0_bound}, and $\Pb{\dset{E}_t} \leq p_t$, where
\begin{equation}
\label{eq:pe_theorem1}
\Pb{\dset{E}_t} \leq \inf_{0 \leq \alpha \leq 1}(p_{1,t} + p_{2,t}),
\end{equation}
where
\begin{equation}
p_{1,t} = \binom{K_a}{t} \binom{M-K_a}{t} \mathbb{E}_{\rmat{X}}\left[ \inf_{u, v > 0, \lambda_{\mathcal{D}} > 0}\determ{\dmat{I} - \rmat{D} \rmat{\Sigma}}^{-L} \right], \label{eq:p1t}
\end{equation}
\begin{equation}
p_{2,t} = \binom{K_a}{t} \mathbb{E}_{\rmat{X}} \left[ \inf_{\delta > 0, \lambda_{\mathcal{B}} > 0}\determ{\dmat{I} - \rmat{B}\rmat{\Sigma}}^{-L} \right], \label{eq:p2t}
\end{equation}
where 
\begin{eqnarray*}
\rmat{\Sigma} &=& \dmat{I}_n + \rmat{X}_{\dset{T}} \rmat{X}^H_{\dset{T}}, \nonumber \\
\rmat{D} &=& u \rmat{P}_{\dset{R}} - u \rmat{P}_{\dset{T}} + \alpha v \rmat{P}_{\dset{C}}^\perp - v \rmat{P}_{\dset{T}}^\perp, \label{eq:D}\\
\rmat{B} &=& - \alpha \delta \rmat{P}_{\dset{C}}^\perp + \delta \rmat{P}_{\dset{T}}^\perp \label{eq:B}.
\end{eqnarray*}
Here, $\lambda_{\mathcal{D}}$ and $\lambda_{\mathcal{B}}$ denote the minimum eigenvalues of $\dmat{I} - \rmat{D} \rmat{\Sigma}$ and $\dmat{I} - \rmat{B} \rmat{\Sigma}$, respectively. The expectations are taken over $\rmat{X}_{\dset{T}\cup\dset{R}}$. The sets $\dset{T}$, $\dset{R}$, and $\dset{C}$ are defined above.
\end{theorem}

The idea behind this theorem is similar to that of the previous one, but the decoding function~\eqref{eq:ml_loss_function} is replaced with the projection rule~\eqref{eq:proj_decoder}. Similar to the \ac{ml}-based achievability, Fano's trick is applied, followed by a Chernoff bound. Finally, the expectations over codewords must be evaluated.

\subsection{Converse bound for \ac{mimo} receiver}

The converse bound described below is presented in~\cite[Theorem 9]{Poor2022ML}.
 
\begin{theorem}[No-\ac{csi} case converse~\cite{Poor2022ML}]
\label{theorem:no_csi_poor_mimo}
Assume that there are $K_a$ active users among $K$ potential users, each equipped with a single antenna, and that the number of BS antennas is $L$.
Each user has an individual codebook of size $M$ and length $n$.
For massive random access in \ac{mimo} quasi-static Rayleigh fading channels with no \ac{csi} and known $K_a$, the minimum energy per bit required to satisfy the \ac{pupe} requirement in~\eqref{eq:ch3:p_missed} can be lower-bounded as
\begin{equation} \label{eq:EbN0_conv_noCSI}
E^{*}_{b,\text{no-\acs{csi}},K_a}(n,M,\epsilon) \geq \inf \frac{nP}{\log_2M}.
\end{equation}
The infimum is taken over all $P>0$ satisfying the following condition:
\begin{equation}\label{P_tot_conv_noCSI}
\br{1 - \epsilon}\log_2M - h_2 \br{\varepsilon} \leq  \frac{LC}{K_a} - \frac{L}{K_a} \mathbb{E}_{\rmat{X}} \left[ \log_2  \left| \dmat{I}_{K_a} + \rmat{X}_{\dset{T}}^{ H}  \rmat{X}_{\dset{T}} \right|  \right],
\end{equation}
\begin{equation}\label{P_tot_conv_noCSI_C}
C = \min \left\{ n \log_2 \left( 1 + K_aP \right),  K_a M \log_2 \left( 1 + \frac{nP}{M} \right) \right\},
\end{equation}
where codewords $\rmat{X}_{\dset{T}}$ are drawn from $\mathcal{G}_{\mathcal{CN}}(M,n,P)$. The right-hand side of~\eqref{P_tot_conv_noCSI} can be further loosened to
$$
\frac{LC}{K_a} - \frac{L}{K_a} \sum\limits_{i=0}^{\tilde{n} - 1}\left( \psi\left(\tilde{n}-i\right) \log_2 e +  \log_2  \left( P + \frac{1}{ \tilde{n} -i } \right) \right),
$$
where $\tilde{n} = \min(K_a, n)$, and $\psi(\cdot)$ denotes Euler's digamma function.
\end{theorem}

The idea behind this theorem is based on Fano's inequality, followed by a mutual information analysis inspired by~\cite{Reeves2013LB}.

\subsection{Numerical comparison of the proposed bounds}

In this section, we perform a numerical analysis of the proposed bounds for the scenario taken from~\cite{noPoorUnsourcedML2023}. The results are presented in~Fig.~\ref{fig:chap6:fbl_mimo_numerical_results}. The \ac{ml} bound shows better performance for $K_a < 400$, but beyond this point, the projection-based bound demonstrates better $E_b/N_0$, with a gap of approximately 1 dB. However, this gap vanishes as $K_a \rightarrow n$.
\begin{figure}
\centering
\includegraphics{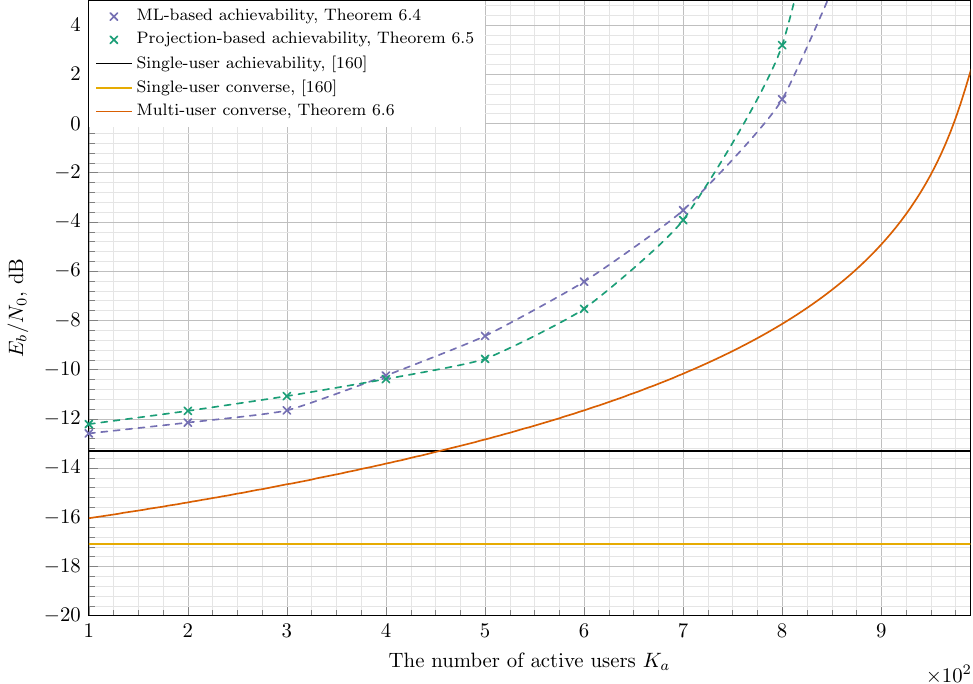}
\caption{\ac{ura} achievability bounds for the no-\ac{csi} setting with a frame length of $n=1000$ channel uses and $k=100$ information bits. The BS is equipped with $L=64$ antennas, and the target \ac{pupe} is $P_e \leq 10^{-3}$. \ac{ml}-based and projection-based achievability bounds are marked by dashed blue and green lines, respectively. As a reference, we consider a single-user converse (marked by a yellow solid line) and a multi-user converse (marked by an orange line). \label{fig:chap6:fbl_mimo_numerical_results}}
\end{figure}

This phenomenon is driven by the following: Given a similar setup, the \ac{pupe} estimate is determined by individual terms corresponding to exactly $t$ errors, as shown in equations~\eqref{eq:pt_ml_mimo} and~\eqref{eq:pe_theorem1}. In both equations, there is a large combinatorial term, $\binom{M - \Ka}{t}$, which leads to increased energy efficiency at high $\Ka$ values when the last term for $t = \Ka$ dominates in~\eqref{eq:chapter4:pupe_definition}. In our case, this term starts dominating earlier in the \ac{ml} bound. However, as $\Ka$ approaches the blocklength $n$, the \ac{ml} bound once again exhibits better energy efficiency, as the projection-based bound is not applicable for $\Ka > n$.

We also found that the projection-based bound provides a tighter estimation of the ``region'' probability compared to the \ac{ml}-based bound. Specifically, with the same parameters $\alpha$ in~\eqref{eq:pe_theorem1} and~\eqref{eq:pt_ml_mimo}, and the same set of samples, the probability~\eqref{eq:mimo_ach_ml_q1t} for the \ac{ml}-based bound is always smaller than its projection-based counterpart~\eqref{eq:p1t}, whereas the probability~\eqref{eq:mimo_ach_ml_q2t} for the \ac{ml}-based bound may be higher than the corresponding projection-based term~\eqref{eq:p2t}.

Finally, let us once again highlight a key drawback of both bounds -- the lack of a rigorous analysis of the required number of samples.

\section{Low-complexity receiver architectures}

Compared to the \ac{awgn} channel, the Rayleigh block-fading channel model assumes that channel coefficients are random~\eqref{eq:channel_model}. Nevertheless, this randomness can be leveraged as follows.

For the single-antenna case, there is a natural diversity in received power. There always exists a user with the maximum received power and, likely, the highest \ac{sinr}. Hence, a straightforward approach is to utilize \ac{tinsic}. Considering the quasi-static fading model, we assume that the channel coherence time is smaller than the frame length, which naturally leads to a slotted system, or \ac{sa}. We assume that the fading coefficient remains constant within a slot but changes randomly and independently from slot to slot.

In Section~\ref{ch5:sec:coded_sa}, we considered \ac{irsa}, where the main idea was to utilize time diversity by transmitting multiple copies of the same message in different slots. After successful decoding, the message can be subtracted from all other slots. Unfortunately, \ac{irsa} cannot be used in a fading scenario because subtraction requires channel estimation, which may have very low precision in the presence of other messages with higher received power. Hence, in the subsequent analysis (see Section~\ref{sec:ch6:numerical_single}), we consider a $T$-fold \ac{sa} instead of \ac{irsa}. The analysis of \ac{sa} is followed by an evaluation of \ac{cs}-based schemes and their results.

In the multiple-antenna case, the \ac{bs} has $L > 1$ receive antennas~\eqref{eq:channel_model}. In this case, random channel vectors (of length $L$) become almost orthogonal, especially when $L$ is large. As a result, the number of simultaneously active users may increase significantly. However, channel estimation becomes more challenging. To obtain accurate channel estimates, preambles are typically embedded within each slot or coherence block. Low-complexity schemes are presented in Section~\ref{sec:ch6:numerical_multi}.

\subsection{Single-antenna case}\label{sec:ch6:numerical_single}

In this section, we consider two main approaches: a $T$-fold \ac{sa} and a \ac{cs}-based approach. In the subsequent numerical analysis for the single-antenna case, we assume $k=100$ information bits, a frame length of $n=3\times 10^4$ complex channel uses, and a maximum tolerable \ac{pupe} value of $\varepsilon = 0.1$.

\subsubsection{\texorpdfstring{$T$}{T}-fold \acl{sa}}\label{chap6:sec_low_complexity:sa}

\begin{figure}[t]
\begin{center}
\includegraphics{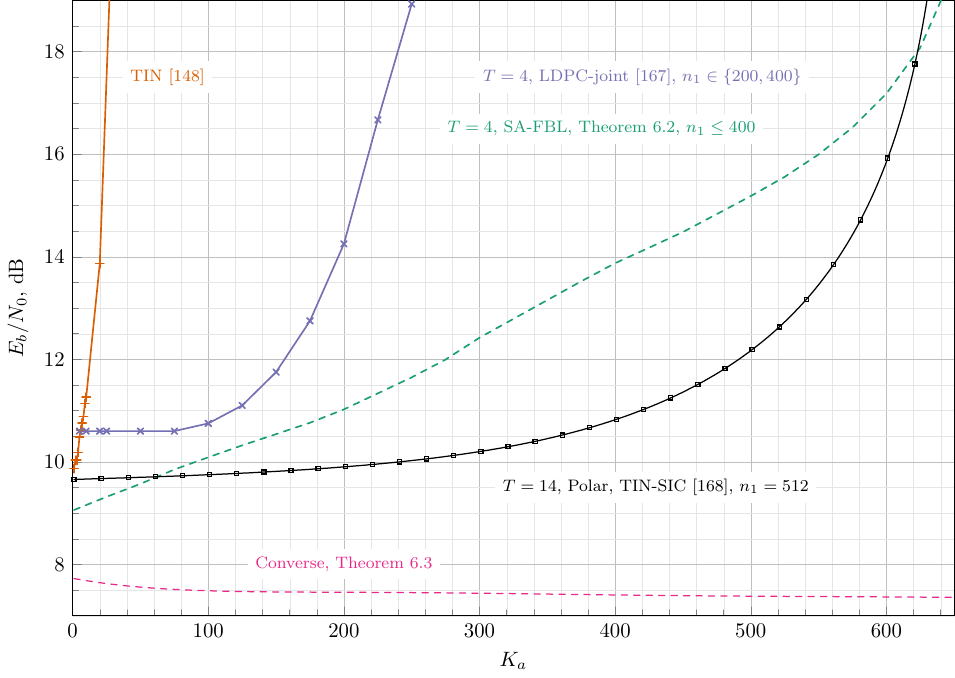}
\caption{Minimum $E_b/N_0$ required to achieve the \ac{pupe} below $10$\% ($\varepsilon=0.1$) as a function of the active user count. Theoretical bounds are represented by dashed lines, while practical schemes are shown with solid lines. The green dashed line corresponds to the achievability bound from Theorem~\ref{theorem:pupe_rayleigh_single}, where the slot length $n'$ is chosen optimally using~\eqref{eq:perf_slot2frame} for $T=4$ and $n' \leq 400$. The dashed magenta line corresponds to the converse bound from Theorem~\ref{th:converse_sa}. For practical schemes, we consider \ac{tin}, $T=4$-fold \ac{sa} with a joint decoder~\cite{Kowshik2020TCOM} ($n' = 200$ and $400$ depending on $\Ka$), and the best-performing \ac{tinsic} scheme based on polar codes~\cite{Frolov2020ISIT}.\label{fig:chap6:ebno_ka_single_antenna}}
\end{center}
\end{figure}

The main results are presented in~\Fig{fig:chap6:ebno_ka_single_antenna}. We begin with a \ac{tin} analysis, represented by an orange line. Similar to the \ac{awgn} case discussed in~Chapter~\ref{label:chap1}, this strategy is highly inefficient.

Next, a joint decoder was introduced in~\cite{FadingASYNC2019, Kowshik2020TCOM, FadingASYNC2019}. The key idea of this approach is to jointly perform decoding and channel estimation to recover a linear combination of \ac{ldpc}-encoded and \ac{bpsk}-modulated messages. The resulting energy efficiency of this scheme for $T=4$ is represented by a blue line. This energy efficiency was achieved by using two slot lengths, depending on the value of $\Ka$. For small $\Ka$, a longer slot of length $n' = 400$ was more efficient, whereas for large $\Ka$, a shorter slot of length $n' = 200$ was preferable, as it reduces the maximum number of simultaneous transmissions in a slot.

Having the slot length $n'$ and the number of simultaneously active users in a slot, $r$, we can evaluate the \ac{pupe} within each slot for a given $T$. Let us denote this \ac{pupe} value as $p_{r, k, n'}$, where $k$ is the number of information bits. Note that $p_{r, k, n'}$ can represent either a theoretical bound or the slot-level performance of a practical scheme. The final goal is to evaluate the \ac{pupe} for the entire frame.

\begin{remark}[Evaluating per-frame \ac{pupe} given per-slot \ac{pupe}]
Given $p_{r, k, n'}$, the per-frame \ac{pupe} $p_e$ can be evaluated as follows. Assuming that each user selects a transmission slot randomly, the \ac{pupe} is given by:
\begin{equation}\label{eq:perf_slot2frame}
p_e = 1 - \sum\limits_{r = 1}^{T}\br{1 - p_{r, k, n'}}\binom{K_a - 1}{r - 1}\br{\frac{1}{L}}^{r-1}\br{1 - \frac{1}{L}}^{K_a - r}.
\end{equation}
\end{remark}

To evaluate the achievable energy efficiency of this scheme, we applied the achievability bound from Theorem~\ref{theorem:pupe_rayleigh_single} to the slot, followed by an assessment of the overall frame performance using eq.~\eqref{eq:perf_slot2frame}. The result is depicted by a dashed green line, representing the minimum $E_b/N_0$ over all possible slot lengths for each value of $\Ka$.

During numerical analysis, we observed that the first codeword successfully decoded by the joint decoder typically has the highest received energy. Furthermore, given a known channel coefficient, this codeword can also be decoded using a \ac{tin} algorithm. This insight led to a modification of the scheme, replacing the joint decoder with a \ac{tinsic} decoder~\cite{Frolov2020ISIT}, where \ac{ldpc} codes were substituted with polar codes due to their superior error-correcting performance for short blocks.

We found that the polar code-based scheme exhibits better energy efficiency (black line). We employed a \ac{crc}-aided list decoding algorithm and did not use preambles for \ac{csi} estimation. Instead, assuming the transmitted signal is \ac{bpsk}-modulated, we leveraged the Gaussian mixture structure of the received signal to retrieve \ac{csi} via a clustering algorithm. Another challenge was the presence of falsely detected codewords. To mitigate this, we increased the \ac{crc} size to 21 bits. The \ac{sic} step was implemented using the \ac{omp} algorithm (see~Appendix~\ref{a3:alg:omp}).

\subsubsection{\Acl{ccs}}\label{chap6:sec_low_complexity:cs}

\begin{figure}[t]
\centering
\includegraphics{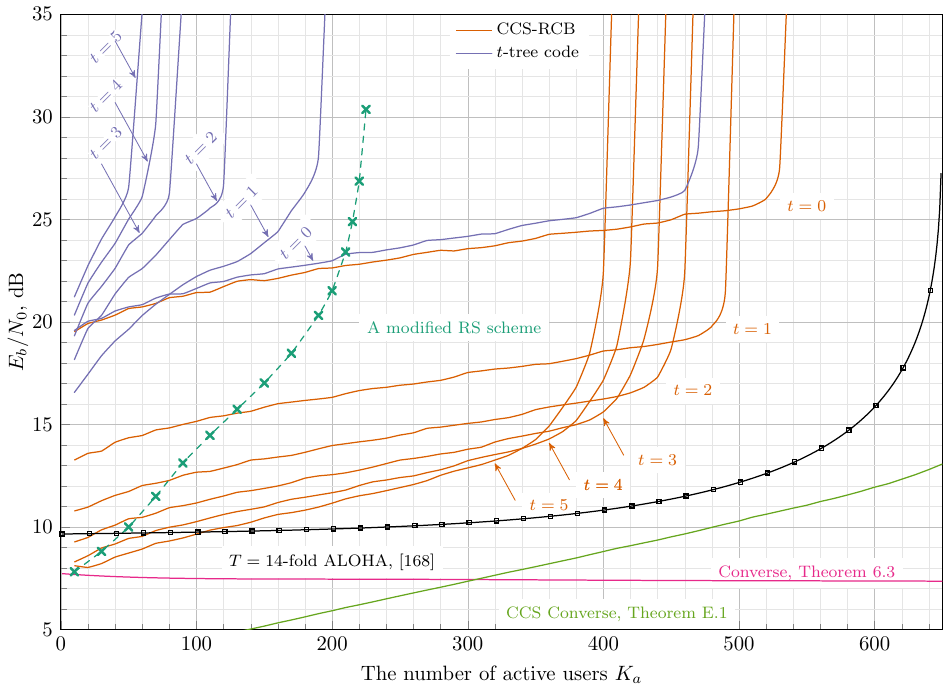}
\caption{Numerical results for the single-antenna quasi-static Rayleigh fading channel. The maximum tolerable \ac{pupe} is $P_e \leq 10^{-1}$, and \ac{far} is $P_f \leq 10^{-3}$. Parameters $K_0$ and $n'$ are chosen to minimize the required $E_b/N_0$. The following curves are presented: \ac{ccs}-\ac{rcb} (see Theorem~\ref{thm:rcb}); $t$-tree code bounds for $t = 0, \ldots, 5$ (see Corollary~\ref{thm:linear_bound}), with the maximum average number of decoding paths equal to $\mathbb{E}[|V_l|] \leq v^\star = 2^{10}$; a modified \ac{rs} scheme from Section~\ref{sec:reed_solomon_practical} (parameters taken from Table~\ref{tab:RA_params}). A $14$-fold \ac{sa} practical scheme from~\cite{Frolov2020ISIT} is added as a reference. Additionally, two converse bounds are provided: one from Theorem~\ref{th:converse_sa} and another for the \ac{ccs} scheme, based on Fano's inequality (Theorem~\ref{prop:converse}).
}
\label{fig:chap6:ebno_ka_cs}
\end{figure}

In this section, we consider the \ac{ccs} scheme presented in Section~\ref{chap5:ccs_description}, now evaluated for a Rayleigh fading channel. Given the channel coherence time assumption, \ac{ccs} also operates under a slotted transmission model. For each slot, we consider an inner codebook of size $Q=2^{15}$. Larger inner codebook sizes improve energy efficiency but require significant computational resources for numerical evaluation.

As previously discussed, we consider a setup with $k=100$ information bits, a frame length of $n=3\times 10^4$, and a slot length of $n' = n / L$, where $L$ is the number of slots. Codewords of the inner code are generated with an i.i.d. uniform distribution over the (complex) power shell. We decode the inner code using \ac{omp}\cite{Cai2011} and its \ac{mmse}-based extension\cite{Sparrer2016}. Since \ac{omp} is a sequential algorithm, the number of output codewords equals the number of decoding steps, denoted as $K_0$.

Numerical results are presented in~\Fig{fig:chap6:ebno_ka_cs}, where we plot energy efficiency as a function of the active user count $\Ka$ for different schemes. Energy efficiency is defined as the minimum $E_b/N_0$ over $K_0$, $L$, and other scheme-specific parameters/constraints, such that the \ac{pupe} satisfies $P_e\leq 0.1$ and the \ac{far} satisfies $P_f\leq 10^{-3}$. These additional parameters/constraints include the maximum average number of decoding paths for the $t$-tree code~\eqref{eq:linear_decoder_constraint} and the user prefix size for the \ac{rs} case.

Let us start with the \ac{ccs}-\ac{rcb} analysis as presented in Theorem~\ref{thm:rcb}. The orange lines in the figure correspond to the \ac{ccs}-\ac{rcb} for $t=0, \ldots, 5$. Recall that the parameter $t$ represents the number of erasures in the outer-code codewords (see eq.\eqref{eq:dec_cond} and the illustrative example presented in\Fig{fig:t_error_code}). For $t=5$, there is a significant improvement in $E_b/N_0$ (more than $10$ dB for $K_a = 50$) compared to the $t=0$ case. When the number of active users is small, the scheme with $t=5$ demonstrates better energy efficiency than a $T$-fold \ac{sa} with polar codes from~\cite{Frolov2020ISIT}, which is known as the best practical solution for the quasi-static fading channel with a single antenna at the receiver.

\begin{table}
\centering
\caption{Optimal slot count $L$ (and outer code rate $R_O = \frac{k}{L\log_2{Q}}$) for \ac{rcb}\label{tab:RCB_slots}}
\begin{tabular}{|r|cccccc|}
\hline
$K_a$ & $t = 0$ & $t = 1$ & $t = 2$ & $t = 3$ & $t = 4$ & $t = 5$\\
\hline
$50$ & $12$ & $14$ & $15$ & $16$ & $18$ & $19$ \\
$100$ & $14$ & $15$ & $17$ & $18$ & $19$ & $20$ \\
$150$ & $15$ & $16$ & $18$ & $19$ & $20$ & $21$ \\
$200$ & $16$ & $17$ & $19$ & $20$ & $21$ & $22$ \\
$250$ & $16$ & $18$ & $19$ & $21$ & $22$ & $23$ \\
$300$ & $17$ & $18$ & $20$ & $21$ & $22$ & $24$ \\
\hline
\end{tabular}
\end{table}

We also observe that higher values of $t$ allow for better energy efficiency but support a smaller maximum number of simultaneously active users. The main reason for this behavior is the \ac{far} constraint. Indeed, as the number of active users grows, the \ac{far} also increases~\eqref{eq:rcb_far}. Moreover, the higher the $t$ values, the faster the \ac{far} grows. Suppressing the \ac{far} requires more slots, which necessitates a more robust (smaller $\pM$) inner code to guarantee low $P_e$~\eqref{eq:rcb_pupe}. Ensuring more robust inner code performance under shorter slot lengths leads to inner code failure at some $\tilde{K}_a$, and this $\tilde{K}_a$ is smaller for higher values of $t$. For $t=5$, $\tilde{K}_a \approx 410$, and for $t=0$, $\tilde{K}_a \approx 540$. The optimal slot count values are presented in Table~\ref{tab:RCB_slots} for different values of $t$ and $\Ka$.

Next, we evaluated the converse bound derived in Theorem~\ref{prop:converse}. We determined the minimum $E_b/N_0$ (over $K_0$ and $L$) such that~\eqref{eq:converse_outer} holds for different values of $K_a$. Additionally, we included the converse bound from Theorem~\ref{th:converse_sa} as a reference. The \ac{ccs} converse (green line in~\Fig{fig:chap6:ebno_ka_cs}) starts below the bound from Theorem\ref{th:converse_sa} (magenta line in~\Fig{fig:chap6:ebno_ka_cs}) and intersects it at $K_a \approx 300$. We consider the overall converse bound to be the region above both bounds.

\begin{table}
\centering
\caption{Optimal slot count $L$ (and outer code rate $R_O = \frac{k}{L\log_2{Q}}$) achievability bound for the code from Corollary~\ref{thm:linear_bound}. The maximum average number of paths $\mathbb{E}[|V_l|] \leq v^\star = 2^{10}$\label{tab:LIN_slots}.}
\begin{tabular}{|r|cccc|}
\hline
$K_a$: & $50$ & $100$ & $150$ & $200$ \\
\hline
$t=0$ & $13 \ (0.513)$ & $14 \ (0.476)$ & $15 \ (0.444)$ & $16 \ (0.417)$ \\
$t=1$ & $27 \ (0.247)$ & $34 \ (0.196)$ & $37 \ (0.180)$ & -- \\
$t=2$ & $46 \ (0.145)$ & $59 \ (0.145)$ & -- & -- \\
$t=3$ & $66 \ (0.101)$ & -- & -- & -- \\
$t=4$ & $85 \ (0.078)$ & -- & -- & -- \\
$t=5$ & $100 \ (0.067)$ & -- & -- & -- \\
\hline
\end{tabular}
\end{table}

\begin{table}
\centering
\caption{Greedy bit allocation results for $v^\star = 2^{10}$\label{tab:LIN_bits}.}
\begin{tabular}{|r|l|}
\hline
\multicolumn{1}{|c}{$K_a$} & \multicolumn{1}{|c|}{Information bits pattern}\\
\hline
\multicolumn{2}{|c|}{$t=0$} \\
\hline
$50$ & $\mathbf{b} = \left[ 15\ 11\ 9\ 8\ 9\ 8\ 9\ 8\ 9\ 8\ 6\ 0\ 0 \right]$ \\
$100$ & $\mathbf{b} = \left[ 15\ 10\ 8\ 8\ 7\ 8\ 8\ 8\ 8\ 8\ 8\ 4\ 0\ 0 \right]$\\ 
$150$ & $\mathbf{b} = \left[ 15\ 9\ 7\ 7\ 8\ 7\ 7\ 8\ 7\ 7\ 8\ 7\ 3\ 0\ 0 \right]$ \\
$200$ & $\mathbf{b} = \left[ 15\ 8\ 7\ 7\ 7\ 6\ 7\ 7\ 7\ 7\ 7\ 7\ 6\ 2\ 0\ 0\right]$ \\
\hline
\multicolumn{2}{|c|}{$t=1$} \\
\hline
$50$ & $\mathbf{b} = [10\ \underbrace{4\ldots 4}_{22}\ 2\ 0\ 0\ 0]$ \\
$100$ & $\mathbf{b} = [10\ \underbrace{3\ldots 3}_{30}\ 0\ 0\ 0]$ \\
$150$ & $\mathbf{b} = [10\ 2\ 2\ 2\ 2\ 2\ 3\ 2\ 3\ 2\ \underbrace{3\ldots 3}_{7}\ 2\ \underbrace{3\ldots 3}_{15}\ 2\ 0\ 0\ 0]$ \\
\hline
\end{tabular}
\end{table}

To evaluate the achievability bound for the $t$-tree code from Corollary~\ref{thm:linear_bound}, the minimum $E_b/N_0$ search procedure must take the maximum average number of decoding paths ($v^\star$) into account as follows: minimize $E_b/N_0$, subject to
\begin{equation}\label{eq:linear_decoder_constraint}
\mathbb{E}[|V_l|] \leq v^\star, \ l\in[L], \ \sum\limits_{l=1}^L b_l = k, \ P_e < 0.1, \ P_f < 10^{-3},
\end{equation}
where $\mathbb{E}[|V_l|]$ is the sum of~\eqref{eq:vlc} and~\eqref{eq:vlf} in accordance with~\eqref{eq:total_paths}. To solve this problem, an optimization over $K_0$ and $n'$ was performed jointly with a greedy information bits allocation, assigning the maximum number of information bits at each subsequent slot $l$ while keeping the constraint $\mathbb{E}[|V_l|] \leq v^\star$. If the total number of allocated bits $\sum\limits_{l=1}^L b_l < k$, we assume $P_e=1$.

The resulting energy efficiency is presented in~\Fig{fig:chap6:ebno_ka_cs} by blue lines for $t=0, \ldots, 5$ and $v^\star = 2^{10}$. For $t=0$, the result is very close to the \ac{ccs}-\ac{rcb}, but the difference becomes dramatic for $t=5$. The constraint $\mathbb{E}[|V_l|] \leq v^\star = 2^{10}$ requires significantly more slots for $t>0$ compared to the \ac{ccs}-\ac{rcb}. The outer coding rates are presented in Table~\ref{tab:LIN_slots}, and the optimal bits allocation is shown in Table~\ref{tab:LIN_bits}. For $t=5$, one needs $L=100$ slots for $K_a = 50$, which is much higher compared to the \ac{ccs}-\ac{rcb} ($L=19$ for $K_a=50$). A larger slot count makes it impossible to construct the $t$-tree code with $t=5$ for $K_a > 60$.

\begin{table}
\centering
\caption{System parameters for the \ac{rs}-based solution\label{tab:RA_params}.}
\begin{tabular}{|r|ccccc|}
\hline
$K_a$ & $K_0$ & $L$ & $b_p$ & $R_{\text{\acs{rs}}}$ & $b_C$\\
\hline
$50$ & $53$ & $53$ & $9$ & $0.3208$ & $19$\\
$100$ & $109$ & $52$ & $9$ & $0.3269$ & $19$\\
$150$ & $167$ & $45$ & $9$ & $0.3778$ & $19$\\
$200$ & $229$ & $41$ & $9$ & $0.4146$ & $19$\\
\hline
\end{tabular}
\end{table}

Finally, we evaluated the \ac{rs}-based practical solution. For the \ac{rs} scheme, we found the minimum $E_b/N_0$ over $K_0$, $L$ and the prefix size $b_p$ (bits). We also adjusted the number of bits $b_C$ to suppress falsely detected messages, but during the simulations, we simply evaluated $b_C$ and corrected the $E_b/N_0$. The resulting energy efficiency of the \ac{rs}-based scheme is presented by green crosses in~\Fig{fig:chap6:ebno_ka_cs}. The \ac{rs} code parameters are presented in Table~\ref{tab:RA_params}, where $R_{\text{\acs{rs}}}$ does not take the \ac{crc} into account (the total coding rate is $R = R_{\text{\acs{rs}}} \cdot R_{\text{\acs{crc}}}$, where $R_{\text{\acs{crc}}} = k/\left(k+b_C\right)$). Increasing $K_a$ requires both a longer prefix and a larger \ac{rs} code length. These two requirements are contradictory because the \ac{rs} code is constructed over the field of size $q = Q / 2^{b_p}$, and $q > L$. On the other hand, increasing the prefix length decreases $q$, making it impossible to increase the number of slots while fitting $k=100$ bits into the frame. The slot count increase also weakens the inner-code performance. We note that the practical \ac{rs}-based scheme operates over a larger outer code length compared to the \ac{ccs}-\ac{rcb}.

\subsection{Multiple-antenna case}\label{sec:ch6:numerical_multi}

\begin{figure}[t]
\begin{center}
\includegraphics{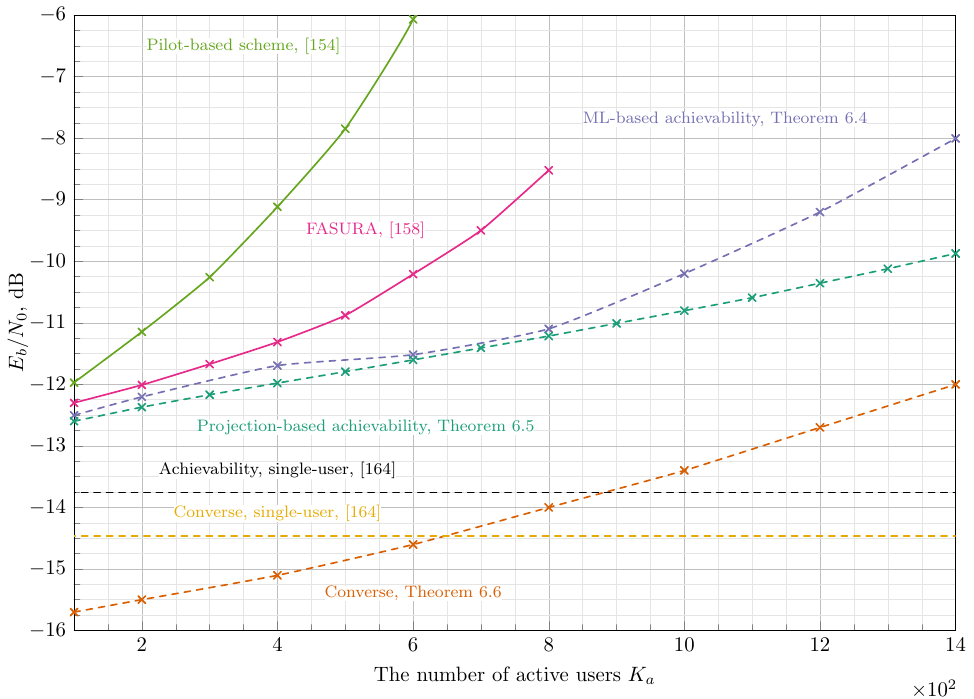}
\caption{Minimum $E_b/N_0$ required to achieve an error rate of $P_e \leq 0.025$ as a function of the active user count. Theoretical bounds are represented by dashed lines, while practical schemes are shown with solid lines. As a reference, single-user achievability and converse bounds are included, along with \ac{ml}-based and projection-based multi-user bounds from Theorems~\ref{th:mimo_ach_ml} and~\ref{th:mimo_ach_prj}. The frame length is $n=3200$, with $k=100$ information bits. The \ac{bs} is equipped with $L=50$ antennas.\label{fig:chap6:ebno_ka_fasura}}
\end{center}
\end{figure}

For the multiple-antenna receiver, we consider a scenario with a frame length of $n = 3200$. The channel coherence time is assumed to be at least as long as the proposed frame length. The \ac{bs} is equipped with $L=50$ antennas. To maintain consistency with results presented in the literature, we also consider an error rate constraint of $P_e \leq 0.025$.

In this setup, employing multiple antennas at the receiver provides two key benefits. First, with $L$ antennas, receive beamforming can be used to enhance the received signal energy. Second, a large number of antennas introduces a proportionally larger number of degrees of freedom in the communication system~\cite{Tse_Viswanath_2005}, allowing for a significantly higher number of simultaneously active users.

The most challenging aspect of the receiver is channel estimation. Since each user has $L$ fading coefficients, the previously described clustering-based approach becomes inefficient. A straightforward method is to allocate a portion of the available resources for channel learning and then use the resulting channel estimate to perform a \ac{tin} step, followed by \ac{sic}. This approach has been implemented in~\cite{Fengler2022}. Pilot detection is closely related to the \ac{cs} problem. In the case of multiple receiver antennas, this problem is known as the \ac{mmv} problem. The energy efficiency of the pilot-based scheme is shown in~\Fig{fig:chap6:ebno_ka_fasura} as a light green line. As a reference, we have re-evaluated the achievability bounds presented in~\Fig{fig:chap6:fbl_mimo_numerical_results}. The current state-of-the-art solution~\cite{FASURA} integrates activity detection, single-user coding, pilot- and temporally decision-aided iterative channel estimation, decoding, \ac{mmse} estimation, and \ac{sic}. The energy efficiency of this scheme is given by a magenta line.

Notably, there is a drastic improvement in energy efficiency compared to the single-antenna case. The resulting energy efficiency improves by approximately $10 \times \log_{10} L$ dB compared to the single-antenna case, due to receive beamforming. Additionally, the number of simultaneously active users increases: up to $\Ka \approx 800$ at a frame length of $n=3200$, compared to $\Ka \approx 600$ at $n=3\times 10^4$.

\chapter{Open problems}\label{chap7}

In this monograph, we have presented a comprehensive overview of the \ac{ura} problem in noisy channels. The \ac{ura} model provides an information-theoretic framework that accounts for \ac{fbl} effects and energy efficiency. This framework enables the derivation of fundamental limits, facilitates the comparison and optimization of existing uncoordinated access schemes, and supports the development of new ones (e.g. coded \ac{cs} schemes). Despite the substantial progress made and the numerous results obtained, several open problems remain:

\begin{itemize}

\item \textit{Are there any better ``good'' regions which improve the achievability bound based on the Fano's trick?} Gordon's inequality seems to be useful as it improves the asymptotic bound~\cite{ZPT-isit19}. Can this method improve finite-length results?

\item \textit{Is it possible to close the gap between achievability and converse bounds in the asymptotic regime (see Appendix~\ref{app:asymptotics})?} It seems counterintuitive that for high user density per channel use, the current achievability bounds coincide with the converse bound, whereas for low user density, we observe a gap of approximately $0.8$~dB.

\item \textit{Is there any polynomial complexity coding scheme with non-linear user codes that outperforms preamble-based approaches?} We note preamble-based approach is one of the methods to create non-linear codes. Let us briefly recall that the same codebook is constructed as follows: there are $L$ linear codes $\mathcal{C}_1, \mathcal{C}_2, \ldots, \mathcal{C}_L$. The user randomly chooses the codebook and sends the codebook index and the codeword to the channel. The clear drawback of this scheme is the dependence on the preamble length on the number of active users and the required $P_e$. 

\item \textit{Are there better codes for the A-channel and A-channel with errors?} The solutions from the literature are based either on the tree-coding or \ac{rs} codes -- the classes of codes suitable for the list recovery problem. In particular, it is interesting whether it is possible to modify classical \ac{ldpc} or polar codes to work in this setup. Note the \ac{scl} decoder seems to be very close to the tree-code decoder, namely it recovers next symbol based on the results for the previous symbols. But there are no results on the use of polar codes for the A-channel.

\item \textit{Are there any other principles to construct sensing matrices with polynomial decoding algorithms?} Standard \ac{cs} decoding algorithms check all correlations which is impossible in high-dimensions. One of the ideas behind this question in the CHIRRUP approach from \cite{CHIRRUP}, which utilizes binary chirps which are simply codewords in the second order \ac{rm} code. Can we go further while preserving polynomial complexity?

\item \textit{Is it possible to obtain closed-form expression of the bounds for the fading \ac{mac}s with \ac{mimo} receiver (and avoid the expectation over the codebook)?} The expected values need to be calculated in a precise manner, which is computationally expensive and probably not feasible.

\item \textit{Are there good low-complexity schemes for the \ac{mimo} \ac{ura} setup which do not perform channel estimation and do not use pilots?}

\item \textit{General fading channels.} One more open problem is the design of \ac{ura} schemes for more general fading channels. Most schemes are designed for the quasi-static Rayleigh fading channel and rely on the quasi-static property. However, quasi-static fading is rarely observed in practice, especially over the relatively long blocklengths considered in \ac{ura}.

\item \textit{Asynchronous \ac{ura}.} All the schemes and bounds presented in this monograph assume perfect per-frame synchronization. At the same time this condition is extremely difficult to fulfill as we consider low-power autonomous devices. The question of asynchronous \ac{mac} was considered in the literature~\cite{asynch}, it is important to propose low-complexity schemes and bounds for this case. We mention several papers that may serve as useful starting points \cite{Async1, Async2, Async3}.

\end{itemize}

\appendix

\chapter{\Acl{mmse} estimation}\label{a:mmse}

Consider the following Markov chain
\[
\rvec{x} \to \rvec{y} \to \hat{\rvec{x}},
\]
where $\hat{\rvec{x}} = \hat{\rvec{x}}(\dvec{y})$ is the estimate of $\rvec{x}$.

The goal is to minimize the \ac{mse}
\[
\mathsf{MSE} = \mathbb{E}[(\hat{\rvec{x}} - {\rvec{x}})^2].
\]
The \ac{mmse} estimator is defined as the estimator that achieves the minimal \ac{mse}:
\[
\hat{\dvec{x}}_*(\dvec{y}) =  \argmin_{\dvec{x}} \mathsf{MSE}.
\]

It is known that, for a given $\dvec{y}$,
\[
\hat{\dvec{x}}_*(\dvec{y}) = \mathbb{E}[\rvec{x} | \dvec{y}].
\]

In many cases, finding an analytical expression for the \ac{mmse} estimator is not feasible. Additionally, the numerical evaluation of the conditional expectation is computationally expensive. Thus, it is reasonable to consider a specific class of estimators -- linear estimators, i.e.,
\[
\hat{\dvec{x}}^{(l)}(\dvec{y}) =  \dmat{W} \dvec{y} + \dvec{b}.
\]

The linear estimator is optimal if it minimizes the \ac{mse} within the class of linear estimators, i.e., solving
\[
(\dmat{R}_*, \dvec{b}_*)  =  \argmin_{\dmat{W}, \dvec{b}} \mathsf{MSE},
\]
where
\[
\hat{\dvec{x}}^{(l)}(\dvec{y}) =  \dmat{W} \dvec{y} + \dvec{b}.
\]

The optimal pair is given by
\begin{flalign}
\dmat{W}_* &= \cov\sbr{\rvec{x}, {\rvec{y}}} (\cov\sbr{\rvec{y},{\rvec{y}}})^{-1}, \\
\dmat{b}_* &= \mathbb{E}\sbr{\rvec{x}} - \dmat{W}_* \mathbb{E}\sbr{\rvec{y}}, 
\end{flalign}
where $\cov\sbr{\rvec{a}, \rvec{b}} = \mathbb{E}\sbr{\left( \rvec{a} - \mathbb{E}\sbr{\rvec{a}} \right) \left( \rvec{b} - \mathbb{E}\sbr{\rvec{b}} \right)^T}$.

In what follows, we refer to the optimal linear estimator as the \ac{lmmse} estimator, given by
\[
\hat{\dvec{x}}^{(l)}_*(\dvec{y}) =  \dmat{W}_* \dvec{y} + \dvec{b}_*.
\]

The \ac{mse} value achievable by the \ac{lmmse} estimator is given by
\[
\mathsf{MSE}^{(l)}_* = \mathrm{tr} \dmat{\Sigma},
\]
where
\[
\dmat{\Sigma} = \cov\sbr{\rvec{x},{\rvec{x}}} - \cov\sbr{\rvec{x}, {\rvec{y}}} (\cov\sbr{\rvec{y},{\rvec{y}}})^{-1} \cov\sbr{\rvec{y}, {\rvec{x}}}.
\]

Now consider a linear observation process
\[
\rvec{y} = \dmat{A} \rvec{x} + \rvec{z},
\]
where $\rvec{x}$ and $\rvec{z}$ are independent, $\rvec{z} \sim \mathcal{N}(\dvec{0}, \dmat{I}_n)$, $\mathbb{E}[\rvec{x}] = \dvec{\mu}$ and $\cov\sbr{\rvec{x},{\rvec{x}}} = \dmat{\Gamma}$. We have
\begin{flalign*}
\cov\sbr{\rvec{x}, {\rvec{y}}} &= \Gamma A^T,\\
\cov\sbr{\rvec{y},{\rvec{y}}} &= (\dmat{A} \dmat{\Gamma} \dmat{A}^{T}+\dmat{I}_{{n}}), 
\end{flalign*}
and thus,
\[
\hat{\dvec{x}}^{(l)}_*(\dvec{y}) =  \dmat{\Gamma} A^T (\dmat{A} \dmat{\Gamma} \dmat{A}^{T}+\dmat{I}_{{n}})^{-1} (\dvec{y} - \dmat{A} \dvec{\mu}) + \dvec{\mu}.
\]

\begin{remark}
If $\rvec{x} \sim \mathcal{N}(\dvec{\mu}, \dmat{\Gamma})$, then the \ac{mmse} estimator coincides with the \ac{lmmse} estimator presented above. 
\end{remark} % MMSE definition
\chapter{Standard \acs{cs} codebooks and decoders}\label{a:cs_methods}

Let us start with a noisy \ac{cs} problem, also known as a sparse linear regression problem. The main task is to recover the sparse vector $\dvec{u}$ given $n$ noisy linear measurements of the form $\dvec{a}^T_i \dvec{u} + \rscalar{z}_i$, $i \in [n]$. That is, we have
\begin{equation}\label{eq:annx:cs_setup}
\rvec{y} = \dmat{A} \dvec{u} + \rvec{z}, \quad \rvec{z}\sim \mathcal{N}\left(\dvec{0}, \dmat{I}_n\right),
\end{equation}
where $\dvec{u} \in \mathbb{R}^M$, and $\norm{\dvec{u}}_0 = t$. Recall that 
$\dmat{A} \in \mathbb{R}^{n \times M}$ is referred to as the \emph{sensing matrix}.

In most cases throughout this monograph, we are primarily interested in recovering only the support of $\dvec{u}$, defined as 
$$
\mathrm{supp}(\dvec{u}) = \{i: u_i \ne 0\}.
$$
Such problems are known as \emph{support recovery problems}.

In this monograph, we encounter two instances of problem~\eqref{eq:annx:cs_setup}, namely:
\begin{itemize}
\item Case A (\ac{gmac}): $\dvec{u} \in \{0,1\}^M$, $\norm{\dvec{u}}_0 = t$;
\item Case B (Fading\footnote{Strictly speaking, we should write $\dvec{u} \in \mathbb{C}^M$, but we focus on the real case in this chapter. Note that the approaches remain largely the same.} \ac{mac}): $\dvec{u} \in \mathbb{R}^M$, $\norm{\dvec{u}}_0 = t$;
\end{itemize}

\begin{remark}
The support recovery problem is not necessarily simpler than the general \ac{cs} problem. The reasoning is as follows: suppose we know the support $\dset{I}$, then we can apply \ac{mmse} or linear \ac{mmse} estimators (see Appendix~\ref{a:mmse}) to obtain an estimate $\hat{\dvec{u}}_\dset{I}$ that minimizes the mean squared error, leading to a good overall estimate $\hat{\dvec{u}}$.
\end{remark}

In this chapter, we describe standard \ac{cs} approaches, focusing on decoding algorithms and codebook design strategies for cases where $M$ is small. In particular, we consider $M\leq 2^{15}$. The elements of these codebooks serve as preambles in practical schemes, as discussed in Section~\ref{chap5}.

\section{Decoding algorithms}\label{a3:alg_awgn}

\subsection{Non-negative least squares}\label{a3:alg:nnls}

Since we are dealing with additive white Gaussian noise (\acs{awgn}) $\rvec{z}$, one possible approach to solving problem~\eqref{eq:annx:cs_setup} is to find a vector $\dvec{u}$ that minimizes the residual norm:
\begin{equation}
\label{eq:annx:min_norm}
\hat{\dvec{u}} = \argmin_{\dvec{u}}\left\|\rvec{y} - \dmat{A} \dvec{u}\right\|^2.
\end{equation}

The problem above is convex and can be solved using standard algorithms, either analytically or via gradient-based methods. However, the solution~\eqref{eq:annx:min_norm} may contain negative elements in $\hat{\dvec{u}}$, which can be undesirable in some scenarios.

To enforce non-negativity, the solution can be obtained through a \acf{nnls} approach, which imposes an additional constraint:
$$
\hat{\dvec{u}} = \argmin_{\dvec{u} \geq \dvec{0}}\left\|\rvec{y} - \dmat{A} \dvec{u}\right\|^2,
$$
where $\dvec{u} \geq \dvec{0}$ means that $u_i \geq 0$ for all $i \in [M]$. Since $\left\|\dvec{x}\right\|^2 = \dvec{x}^T\dvec{x}$, the problem above can be reformulated equivalently as
$$
\hat{\dvec{u}} = \argmin_{\dvec{u} \geq \dvec{0}}\left(\frac{1}{2}\dvec{u}^T\dmat{Q}\dvec{u} + \dvec{c}^T\dvec{u}\right), \quad \dmat{Q} = \dmat{A}^T\dmat{A}, \quad \dvec{c} = -\dmat{A}^T\dvec{y},
$$
which is a convex optimization problem (see~\cite{cvxoptbook}).

\subsection{LASSO: Sparse linear regression}\label{a3:alg:lasso}

\Ac{cs} can be viewed as a linear regression problem. To enforce sparsity in the solution, $\ell_1$-regularization can be applied:
$$
\hat{\dvec{u}} = \argmin_{\dvec{u}}\left\|\rvec{y} - \dmat{A} \dvec{u}\right\|^2 + \lambda \norm{\dvec{u}}_1.
$$
This regularized linear regression approach is known as LASSO~\cite{bishop2007}. The parameter $\lambda$ controls the sparsity of the solution.

\subsection{\Acl{amp}}\label{a3:alg:amp}

\Ac{amp}~\cite{Donoho2010AMP} is a family of low-complexity iterative algorithms used to solve a wide range of problems, including\eqref{eq:annx:cs_setup}.

In this section, we provide a simplified description from~\cite{Truhachev2024}, listed in Algorithm~\ref{alg:amp_cs_awgn}. In the algorithm description, subscripts indicate iterations, $\langle\cdot\rangle$ denotes the average of a vector, $\eta\left(\cdot\right)$ is a nonlinear denoising function (see, e.g.,~\cite{Truhachev2024}), and $\eta'\left(\cdot\right)$ represents the component-wise derivative.

\begin{algorithm}
\caption{\Acf{amp} for~\eqref{eq:annx:cs_setup} problem.}\label{alg:amp_cs_awgn}
\begin{algorithmic}[1]
\State \textbf{Input:} $\dmat{A}$, $\dvec{y}$ \Comment Design matrix, observed vector
\State $\tilde{\dvec{y}}_0 = \dvec{y}$ \Comment Soft residual signal
\State $\hat{\dvec{u}}_0 = 0$ \Comment Soft activity pattern estimate
\For {$t = 1, \ldots, T$} \Comment Run $T$ iterations
\State $\hat{\dvec{u}}_{t + 1} = \eta_t\left(\hat{\dvec{u}}_t + \dmat{A}^T\tilde{\dvec{y}}_t\right)$
\State $\tilde{\dvec{y}}_{t + 1} = \dvec{y} - \dmat{A}\hat{\rvec{u}}_{t + 1} + \frac{M}{n}\tilde{\dvec{y}}\left\langle \eta'\left( \hat{\dvec{u}}_t + \dmat{A}^T\tilde{\dvec{y}}_t \right) \right\rangle$
\EndFor
\State \Return $\hat{\dvec{u}}$
\end{algorithmic}
\end{algorithm}

The main idea of the algorithm is as follows. Consider the prior distribution of the elements of $\dvec{u}$ as
$$
p\left(\dvec{u}\right) = \prod\limits_{i = 1}^{M}p\left(\dvec{u}_i\right) = \prod\limits_{i = 1}^{M}\exp\left\{-\beta \left|\dvec{u}_i\right|\right\}
$$
Given the Gaussian channel, let us consider the posterior distribution $\mu$:
$$
\mu\left(\rvec{u}\right) = \frac{1}{Z}
\prod\limits_{i = 1}^{n}\exp\left\{-\frac{\beta}{2}\left(\dvec{y}_i - \left(\dmat{A}\rvec{u}\right)_i\right)^2\right\}
 \prod\limits_{i = 1}^{M}\exp\left\{-\beta \left|\dvec{u}_i\right|\right\},
$$
where $Z$ is a normalizing constant.

To compute the posterior probabilities, an iterative sum-product algorithm on a dense graph can be used~\cite{Donoho2010AMP}. The idea of \ac{amp} is to apply the large system limit ($M\rightarrow \infty$) and sharpen the prior probability distribution as $\beta \rightarrow \infty$~\cite{Donoho2010AMP}.

Note that~\cite[Section V]{Donoho2010AMP} also describes the iterative algorithm that solves a LASSO problem.

The idea of \ac{amp} for Case B of the problem~\eqref{eq:annx:cs_setup} is similar and is described in~\cite{Rangan2011gamp}.

\subsection{\Acl{omp}}\label{a3:alg:omp}

The idea of \ac{omp} is based on a \ac{sic} approach and the near-orthogonal property of the columns of random sensing matrices. The observed vector $\dvec{y}$ in~\eqref{eq:annx:cs_setup} (case B) is a linear combination of the columns of the sensing matrix $\dmat{A}$. Thus, if these columns are nearly orthogonal, the concept of an energy detector from \ac{cdma} can be applied here. \ac{omp} is an iterative algorithm. At each iteration, it identifies the column of the design matrix that is most correlated with the received signal, adds the corresponding column to the set of \emph{active} columns, and then performs cancellation of the already detected columns. The algorithm is specified in~\cite{Cai2011} and listed in Algorithm~\ref{alg:omp}.

\begin{algorithm}
\caption{\Acl{omp} for~\eqref{eq:annx:cs_setup} problem.}\label{alg:omp}
\begin{algorithmic}[1]
\State \textbf{Input:} $\dmat{A}$, ${\dvec{y}}$, $t$ \Comment Sensing matrix, observed vector, output list size
\State $\tilde{\dvec{y}} = \dvec{y}$ \Comment{Residual signal}
\State $\dset{D} = \varnothing$ \Comment Estimate of $\mathrm{supp}(\dvec{u})$
\While {$\left|\dset{D}\right| < t$} \Comment Run until $t$ candidates detected
\State $i = \argmax_{j \in [M] \backslash \dset{D}}\left| \dvec{y}^T \dvec{a}_i \right|$ \Comment{Find the highest correlation coefficient}
\State ${\dset{D}} = {\dset{D}} \bigcup \{i\}$ \Comment Update the detected columns set
\State $\hat{\dvec{u}}_{\dset{D}} = \dmat{A}_{\dset{D}}^\dagger \dvec{y}$ \Comment Get weights estimate by pseudoinverse
\State $\tilde{\dvec{y}} = \dvec{y} - \dmat{A}_{\dset{D}} \hat{\dvec{u}}_{\dset{D}}$ \Comment update the residual signal
\EndWhile
\State \Return $\hat{\dset{D}}$, $\hat{\dvec{u}}$
\end{algorithmic}
\end{algorithm}

The weight coefficients estimate should be calculated jointly for the entire set of detected columns $\hat{\dset{T}}$, as the columns of the design matrix $\dmat{A}$ may not be orthogonal. When calculating the estimate of weights $\hat{\dvec{h}}$, we treat the set of design matrix columns as another matrix.

Note that the \ac{omp} is also suitable for case A in~\eqref{eq:annx:cs_setup}, where $\hat{\dvec{u}}_i = 1$ for $i \in \dset{D}$, without the need to calculate a pseudoinverse.

\section{Codebook design}\label{a3:cwbook}

All the algorithms described above calculate correlations and thus rely on codebooks with good correlation properties. In this section, we describe several approaches to constructing such codebooks.

Our goal is to design a sensing matrix $\dmat{A}$ such that $|\dvec{a}^T_i \dvec{a}_j|$, for $i, j \in [M]$ and $i \ne j$, is small. Recall also the energy constraint $\norms{\dvec{a}_i} \leq Pn$, for $i \in [M]$. In this section, we focus on $P=1$, i.e., the columns of $\dmat{A}$ have unit norms. The case of arbitrary $P$ is obtained by multiplying the codebooks above by $\sqrt{P}$.

\subsection{Gaussian codebook}\label{a3:cwbook:gauss}

The most popular approach is to utilize the Gaussian codebook, i.e., each element $a_{i,j}$ of $\dmat{A}$ is sampled i.i.d. from $\mathcal{N}\br{0,1-\varepsilon}$ for some $\varepsilon > 0$. We need $\varepsilon$ to satisfy the power constraint. Such codebooks have been proven to be effective for large $n$. Indeed, for any $i,j \in [M]$, $i \ne j$, we have
\[
\frac{1}{n}\dvec{a}^T_i \dvec{a}_j \overset{P}{\to} 0.
\]

To be precise, we can use the following equality:
\[
\dvec{a}^T_i \dvec{a}_j = {\norms{\frac{\dvec{a}_i+\dvec{a}_j}{2}} - \norms{\frac{\dvec{a}_i-\dvec{a}_j}{2}}},
\]
where the vectors $\be_1 = (\dvec{a}_i + \dvec{a}_j)/2$ and $\be_2 = (\dvec{a}_i - \dvec{a}_j)/2$ are independent. We note that $\be_1, \be_2 \sim \mathcal{N}(\dvec{0}, (1-\varepsilon)\dmat{I}_n)$.

Let us consider the variance:
\[
\mathrm{Var}\sbr{\frac{1}{n}\dvec{a}^T_i \dvec{a}_j} = \frac{1}{n^2}\left(\mathrm{Var}\sbr{\norms{\be_1}} - \mathrm{Var}\sbr{\norms{\be_2}}\right) = \frac{4(1-\varepsilon)^2}{n},
\]
where the last equality follows from the fact that the variance of a chi-squared distribution with $n$ degrees of freedom is $2n$.

Thus, the random variable $\frac{1}{n}\dvec{a}^T_i \dvec{a}_j$ is well-concentrated around $0$.

\subsection{Spherical codebook}\label{a3:cwbook:sphere}

To avoid power constraint issues, it is reasonable to use a spherical codebook. Let $\mathbb{S}^{n-1}$ be the unit sphere in $n$ dimensions. The codewords are uniformly and independently sampled from $\mathbb{S}^{n-1}$, i.e.,
\[
\dvec{a}_i \overset{i.i.d}{\sim} \mathrm{Unif}(\mathbb{S}^{n-1}), \quad i \in [M].
\]

\begin{remark}
The standard way to generate such a column is to sample $\tilde{\dvec{a}}_i \sim \mathcal{N}\br{\dvec{0}, \dmat{I}_n}$ and then normalize the result as $\dvec{a}_i = \tilde{\dvec{a}}_i / \norm{\tilde{\dvec{a}}_i}$. It is easy to show that the resulting vector is distributed uniformly on $\mathbb{S}^{n-1}$.
\end{remark}

Note that for the conventional \ac{gmac}\footnote{Different codebook case, all active users, fixed $\Ka$} \cite{Yavas2021}, spherical codebooks have been shown to achieve a better second-order (dispersion) term in the asymptotic expansion of the achievable rate region as $n \to \infty$ compared to Gaussian codebooks. However, it is not clear if this result is valid for the \ac{gura} setup. For the \ac{ura} setup in the case of the \ac{mimo} quasi-static Rayleigh fading channel, the authors of~\cite{Gao2024sphere} demonstrated that the use of spherical codebooks is beneficial in the \ac{fbl} regime.

\subsection{\acs{bch} subcodes}\label{a3:cwbook:sub_bch}

The codebooks described in the previous sections are easy to generate and work well for large $n$. However, when considering the preamble design task, we are required to make the preambles as short as possible to avoid additional overhead. In this regime, Gaussian and spherical codebooks may not perform as well.

To overcome this issue, we describe how to construct a codebook with good correlation properties for existing error-correction codes, particularly from \ac{bch} codes.

Let us start with the $[n, k]$-\ac{bch} code $\mathcal{C}$ that corrects $t$ errors. The columns of $\dmat{A}$ are to be chosen from the set $\dset{X} = {\dvec{x} = \tau(\dvec{c}): \dvec{c} \in \mathcal{C} }$, where $\tau\br{\cdot}$ is a \ac{bpsk} modulation. Let us consider the correlation properties, specifically calculating ($i,j \in [M], i \ne j$):
\begin{flalign*}
\dvec{x}^T_i \dvec{x}_j &= \frac{1}{2} (\norms{\dvec{x}_i} + \norms{\dvec{x}_j} - \norms{\dvec{x}_i - \dvec{x}_j}) = n - 2 d_H(\dvec{x}_i, \dvec{x}_j).
\end{flalign*}
Thus,
\[
n - 2 \mathrm{wt}_{\max}(\mathcal{C}) \leq \dvec{x}^T_i \dvec{x}_j \leq n - 2 \mathrm{wt}_{\min}(\mathcal{C}) = n - 2 (2t+1).
\]
An interesting observation is that we need to control not only the minimal weight (code distance) but also the maximal weight of the code. The ideal case is to use an equidistant code.

\begin{example}
For the $[n=63, k=10, d=27]$-\ac{bch} code, we have
\[
-63 \leq \dvec{x}^T_i \dvec{x}_j \leq 9
\] 
since the code contains the all-one codeword of weight $63$.

It seems to be easy to solve our problem for the \ac{bch} code. We consider only a subcode $\mathcal{C}_0 \subset \mathcal{C}$, where $\dvec{1} \not\in \mathcal{C}_0$. Let us consider the $[n=63, k=9, d=27]$-\ac{bch} subcode. We have
\[
-9 \leq \dvec{x}^T_i \dvec{x}_j \leq 9
\] 
since the maximal weight is equal to $36$.
\end{example}

Let $\mathcal{C}_0$ be the desired subcode with $\dim(\mathcal{C}_0) = k_p$. We construct the sensing matrix $\dmat{A}$ with $\dvec{a}_i = \tau(\dvec{c}_i)$, where $\dvec{c}_i \in \mathcal{C}_0$. The size of $\dmat{A}$ is $n \times 2^{k_p}$.

\subsection{Codebook from the parity-check matrix of a \texorpdfstring{$t$}{t}-error-correcting code}\label{a3:cwbook:bch}

To obtain a good codebook, we need all the sums $\sum\nolimits_{i \in \dset{I}} \dvec{a}_i$, where $\dset{I} \subset [M]$ and $|\dset{I}| = t$, to be distinct and well-separated (to account for noise). Such codes are called \emph{superimposed codes}. In what follows, we present one of the possible approaches to constructing such codes.

Let us utilize the technique proposed in \cite{BarDavid1993}. Let $\dmat{H} = [\dvec{h}_1, \ldots, \dvec{h}_n]$ be the parity-check matrix of a $t$-error-correcting code over $\mathbb{F}_2$. Set
\[
\dvec{a}_i = \tau(\dvec{h}_i), \quad i \in [M].
\]
We have the desired property: all sums of up to $t$ columns of $\dmat{A}$ are distinct.

\begin{example}
Consider a \ac{bch} code of length $n$ that corrects $t$ errors. In this case, we can construct a codebook of size approximately $t \log n \times n$ (since \ac{bch} codes are known to be quasi-perfect).
\end{example}

\subsection{Binary chirps}\label{a3:cwbook:calderbank}

In this section, we describe a codebook based on \ac{rm} codes. Let us introduce the main definitions.

\begin{definition}[Algebraic power]
\[
x^{\sigma} = \left\{ \begin{array}{l} x, \quad \sigma=1 \\ 1, \quad \sigma=0\end{array} \right.
\]
\end{definition}

\begin{definition}[Algebraic normal form or Zhegalkin polynomial]
$$
f(\dvec{x}) = \bigoplus_{\dvec{\sigma} \in \{0,1\}^m} b_{\dvec{\sigma}} \dvec{x}^{\dvec{\sigma}}, \quad \textbf{x}^{\dvec{\sigma}} = x_1^{{\sigma}_1} \cdot \ldots \cdot x_m^{{\sigma}_m}, \quad  \deg( \dvec{x}^{\dvec{\sigma}}) = \left\| \dvec{\sigma} \right\|_0.
$$
\end{definition}
Note that each Boolean function $f$ has a unique algebraic normal form representation.

\begin{definition}
The degree of the function $f(\dvec{x})$ is defined as
\[
\deg(f) = \max\limits_{\dvec{\sigma}: a_{\dvec{\sigma}} = 1} \deg(\dvec{x}^{\dvec{\sigma}}).
\]
\end{definition}

\begin{definition}[\Acf{rm} code]
Let $f: \left\{0, 1\right\}^m \rightarrow \left\{0, 1\right\}$ be the Boolean function with $m$ arguments, and let
$$
\mathrm{Eval}(f) = [f(\dvec{\sigma}_1), \ldots, f(\dvec{\sigma}_{2^m})],
$$
where $\dvec{\sigma}_i \in \{0,1\}^m$, $i \in [2^m]$, then the \ac{rm} code of order $r$ is defined as follows:
$$
RM(m, r) = \left\{ \mathrm{Eval}(f) : \deg f\leq r \right\}.
$$
\end{definition}

The code is linear over $\mathbb{F}_2$ has length $n=2^m$, and dimension $k=\sum_{i=0}^r \binom{m}{i}$.

Now, let us consider the codebook proposed in \cite{CHIRRUP}. Let us start with the first-order \ac{rm} code (or the case when $r=1$). Consider all $f$ such that $\deg(f) = 1$\footnote{The reader may ask why the function of degree $0$ was removed. The reason is exactly the same as in Section~\ref{a3:cwbook:sub_bch} -- the all-one codeword is removed to improve the correlation properties.}. These functions can be expressed as
\[
f(\dvec{x}) = \dvec{b}^T \dvec{x} \quad \mod 2.
\]

Now, consider the modulated codeword. It can be represented as $\mathrm{Eval}\left((-1)^f\right)$, which is equivalent to $\mathrm{Eval}\left((-1)^{\dvec{b}^T \dvec{x}}\right)$\footnote{Note that $\mod 2$ may be omitted.}. The functions
\[
\phi_{\dvec{b}}(\dvec{x}) = (-1)^{\dvec{b}^T \dvec{x}}
\]
are called the Walsh functions, $\phi_{\dvec{b}}: \{0,1\}^m \to \mathbb{R}$. These functions are orthogonal. However, we can construct only $2^m$ such functions.

In \cite{CHIRRUP}, it is suggested to use the second-order \ac{rm} codes. Again, the constant function is removed, and we consider
\[
\phi_{\dvec{b}, \dmat{P}}(\dvec{x}) = (-1)^{\dvec{b}^T \dvec{x} + \frac{1}{2}\dvec{x}^T \dmat{P} \dvec{x}},
\]
where $\dvec{b} \in \brc{0,1}^m$ and the matrix $\dmat{P} \in \brc{0,1}^{m \times m}$ is symmetric and has zeroes on the main diagonal\footnote{We only consider the real codebook here, but note that the authors of~\cite{CHIRRUP} also suggest the complex codebook $\phi_{\dvec{b}, \dmat{P}}\br{\dvec{x}} = i^{2\dvec{b}^T \dvec{x} + \dvec{x}^T \dmat{P} \dvec{x}}$, where $\dmat{P}$ is symmetric, but diagonal elements are allowed to be non-zero.}.

We set
\[
\dvec{a}_{i(\dvec{b}, \dmat{P})} = \mathrm{Eval}(\phi_{\dvec{b}, \dmat{P}}),
\]
where $i\br{\dvec{b}, \dmat{P}}$ is the function that returns the column index given $\dvec{b}$ and $\dmat{P}$. The size of $\dmat{A}$ is $2^m \times 2^{m + \binom{m}{2}}$.

Finally, let us mention the main benefit of this codebook. The standard decoders (e.g., the \ac{omp} algorithm) calculate the correlation between the current measurement residual and every column, selecting the one with the largest absolute value. This is an $O\br{\left|\mathcal{C}\right|}$ procedure and, therefore, computationally intractable for large codebook sizes. In contrast, the reconstruction algorithm proposed in~\cite{CHIRRUP} performs a small number of Walsh-Hadamard transforms to identify the matrix $\dmat{P}$ row by row, and subsequently identifies the vector $\dvec{b}$. The overall complexity of the procedure for finding a codeword with maximum correlation is $O(m^2 2^m)$, which is sub-linear in the size of the codebook.
 % Here, a definition of sensing matrix is given
\chapter{\acs{cs} sensing matrices}\label{a:cs_matrices}

As mentioned in Chapter~\ref{chap3}, the \ac{ura} problem can be considered a \ac{cs} problem. In this section, we illustrate the \ac{cs} sensing matrices corresponding to the \ac{irsa} protocol presented in Section~\ref{ch5:sec_irsa_sf} and the \ac{ccs} protocol presented in Section~\ref{chap5:ccs_description}.

\begin{figure}[t]
\centering
\includegraphics{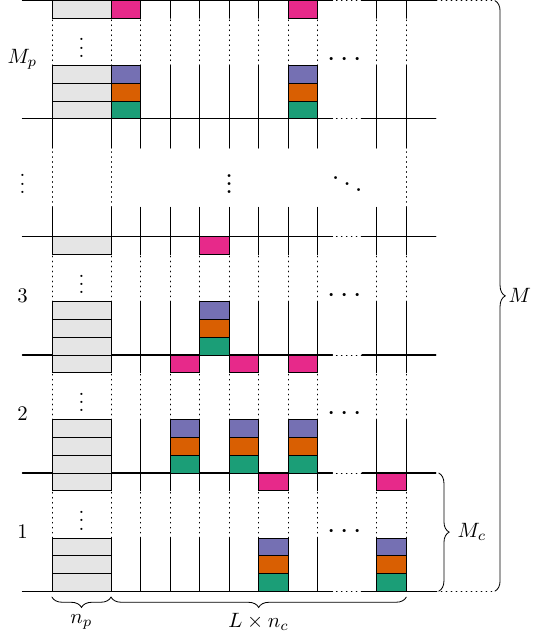}
\caption{\acs{irsaf} Sensing Matrix Example. The preamble length is $n_p$, and the slot length is $n_c$. Different colors represent different codewords of user messages. Each user message appears $M_p$ times in the resulting sensing matrix, with varying replica counts and slot indices.
\label{fig:irsa_f_cs}}
\end{figure}

Let us start with an \ac{irsaf} protocol presented in~\Fig{fig:irsa-f} and consider the corresponding \ac{cs} problem~\eqref{eq:ch5:cs_preamble}. Each user transmits $k=\log_2M$ information bits. The user message $W \in [M]$ is split into $W_p$ and $W_c$. Then, $W_p$ determines a preamble, while the remaining part, $W_c$, is encoded using a code of length $n_c = \br{n - n_p} / L$. The encoded message is further replicated across multiple slots. The replica count and slot indices are determined by both $W_p$ and $W_c$. The resulting sensing matrix is illustrated in~\Fig{fig:irsa_f_cs}.

\begin{figure}[t]
\centering
\includegraphics{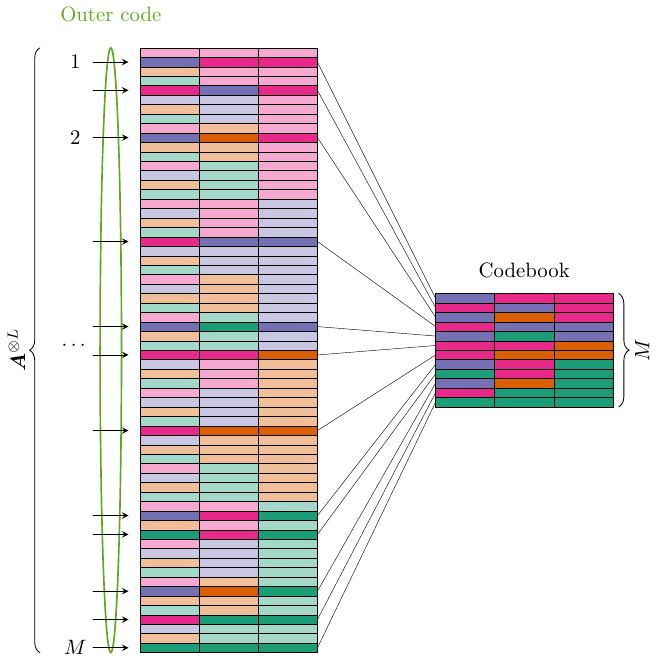}
\caption{\acs{ccs}: sensing matrix. An example for $L=3$ slots and a given outer code. Different colors represent different rows of the inner codebook matrix $\dmat{A}$ of size $n'\times Q$. In this illustrative case, $Q=4$. The matrix construction begins by combining all possible inner-code symbol combinations across $L$ slots, forming the large matrix $\dmat{A}^{\otimes L}$ of size $n'L\times Q^L$ (left). Then, the outer code selects $M$ valid outer codewords from the $Q^L$ variants. These valid codewords are highlighted with saturated colors on the left side of the figure, while the resulting \ac{cs} matrix is shown on the right side.\label{fig:ccs_sensing_matrix}}
\end{figure}

In the case of \ac{ccs}, the sensing matrix construction procedure is presented in \Fig{fig:ccs_sensing_matrix}. Within each slot, transmission is performed using an inner code with a codebook matrix $\dmat{A}$ of size $n'\times Q$. Given that there are $L$ slots in the frame, the sensing matrix construction begins by applying the direct product $L$ times to the matrix $\dmat{A}$, resulting in $\dmat{A}^{\otimes L}$ of size  $n'L\times Q^L$. Finally, only $M$ rows of the resulting matrix are selected -- forming a set of valid codewords for the outer code.
 % Here, some sensing matrices are described, alongwith some decoding algorithms based on MMSE
\chapter{\acs{irsa}: replica count distribution}\label{a:irsa}

Replica count distributions optimized using the density evolution procedure (over $10$ iterations) are shown in Tables \ref{DE_T1}, \ref{DE_T2}, and \ref{DE_T4} for $T=1$, $2$, and $4$, respectively. Recall that the replica count distribution is given by the polynomial $\Lambda\br{x}$~\eqref{eq:ch5:irsa_poly}, where the coefficient of the $i$-th degree term corresponds to the probability that $i$ replicas will be chosen.

Notably, smaller values of $T$ require greater time diversity, which is reflected in the higher-degree terms of the replica count polynomial. On the other hand, the decoder for $T=4$ can successfully decode slots with high collision orders, making additional replicas unnecessary in this case. Moreover, fewer replicas reduce the overall system load and improve the energy efficiency of the transmission scheme.

\begin{table}
\caption {Optimal node degree distribution and $E_b/N_0$ (dB) for $T=1$, $10$ iterations} 
\label{DE_T1}
\centering
\begin{tabular}{|r|llll|r|}
\hline
$\Ka$ & $\Lambda(x)$: & &&& $E_b/N_0$, dB \\
\hline
25   & $0.0928x$ & $+0.9072x^2$ &              &              &  3.71 \\ 
50   &           & $+1.0000x^2$ &              &              &  4.48 \\ 
100  &           & $\textcolor{white}{+}1.0000x^2$ & &        &  5.89 \\ 
150  & & $\textcolor{white}{+}0.6211x^2$ & $+0.3789x^3$ &     &  7.17 \\ 
200  & & $\textcolor{white}{+}0.4781x^2$ & $+0.5219x^3$ &     &  8.31 \\ 
250  & $0.0706x$ & $+0.2011x^2$ & $+0.7283x^3$ &              &  9.42 \\ 
300  & $0.1297x$ &              & $+0.8703x^3$ &              & 10.52 \\ 
350  & $0.1234x$ &              & $+0.8766x^3$ &              & 11.62 \\ 
400  & $0.1184x$ &              & $+0.8816x^3$ &              & 12.76 \\ 
450  & $0.1247x$ &              & $+0.7991x^3$ & $+0.0763x^4$ & 13.90 \\ 
500  & $0.1396x$ &              & $+0.6716x^3$ & $+0.1889x^4$ & 15.05 \\ 
550  & $0.1474x$ &              & $+0.5906x^3$ & $+0.2620x^4$ & 16.20 \\ 
600  & $0.1549x$ &              & $+0.5239x^3$ & $+0.3212x^4$ & 17.37 \\ 
\hline
\end{tabular}
\end{table}

\begin{table}
\vskip 0.3cm
\caption {Optimal node degree distribution and $E_b/N_0$ (dB) for $T=2$, $10$ iterations} 
\label{DE_T2}
\centering
\begin{tabular}{|r|l|r|}
\hline
$\Ka$ & $\Lambda(x)$ & $E_b/N_0$, dB \\
\hline
 25  & $1.0000x$                         &  1.63 \\ 
 50  & $0.4443x + 0.5557x^2$             &  3.09 \\ 
100  & $0.2099x + 0.7901x^2$             &  4.00 \\ 
150  & $0.1680x + 0.8320x^2$             &  4.92 \\ 
200  & $0.1382x + 0.8618x^2$             &  5.88 \\ 
250  & $0.1205x + 0.8795x^2$             &  6.86 \\ 
300  & $0.1008x + 0.8992x^2$             &  7.86 \\ 
350  & $0.0895x + 0.9105x^2$             &  8.87 \\ 
400  & $0.1206x + 0.7609x^2 + 0.1185x^3$ &  9.90 \\ 
450  & $0.1609x + 0.5953x^2 + 0.2438x^3$ & 10.92 \\ 
500  & $0.1867x + 0.4916x^2 + 0.3217x^3$ & 11.94 \\ 
550  & $0.2061x + 0.4160x^2 + 0.3779x^3$ & 12.97 \\ 
600  & $0.2184x + 0.3639x^2 + 0.4177x^3$ & 14.00 \\
\hline
\end{tabular}
\end{table}

\begin{table}
\centering
\caption {Optimal node degree distribution and $E_b/N_0$ (dB) for $T=4$, $10$ iterations} 
\label{DE_T4}
\begin{tabular}{|r|l|r|r|r|}
\hline
$K$ & $\Lambda(x)$ & $E_b/N_0$, dB \\
\hline
 25  & $1.0000x$             &  0.53 \\ 
 50  & $1.0000x$             &  1.10 \\ 
100  & $0.5781x + 0.4219x^2$ &  2.84 \\ 
150  & $0.3987x + 0.6013x^2$ &  3.54 \\ 
200  & $0.3764x + 0.6236x^2$ &  4.31 \\ 
250  & $0.3620x + 0.6380x^2$ &  5.13 \\ 
300  & $0.3538x + 0.6462x^2$ &  5.97 \\ 
350  & $0.3479x + 0.6521x^2$ &  6.83 \\ 
400  & $0.3416x + 0.6584x^2$ &  7.71 \\ 
450  & $0.3349x + 0.6651x^2$ &  8.60 \\ 
500  & $0.3302x + 0.6698x^2$ &  9.52 \\ 
550  & $0.3269x + 0.6731x^2$ & 10.45 \\ 
600  & $0.3225x + 0.6775x^2$ & 11.40 \\
\hline
\end{tabular}
\end{table} % IRSA appendix
\chapter{Coded CS: further results}\label{a:ccs}

\section{Capacity of the channel for the outer code}\label{sec:outer_channel_capacity}

In this section, we estimate the capacity of the channel $\mathfrak{W}$ presented in Section~\ref{sec:outer_code_channel}. Recall~\Fig{fig:t_error_code} for an illustrative example. For simplicity, we will not perform optimization over all independent distributions of $\rscalar{x}^{(i)}$, $i\in[K_a]$. Instead, we consider only the uniform distribution. Thus, we obtain an expression for the achievable rate rather than the true capacity. We have
\begin{equation}\label{eq:cu}
I_u = I(\rscalar{x}^{(1)}, \ldots, \rscalar{x}^{(K_a)}; \dset{Y}) = H(\dset{Y}) - H(\dset{Y} | \rscalar{x}^{(1)}, \ldots, \rscalar{x}^{(K_a)}),
\end{equation}
where $\rscalar{x}^{(i)} \overset{i.i.d.}{\sim} \mathrm{Unif}([Q])$.

Let $\Omega = |\dset{Y}^{(A)}|$, then we have
\begin{flalign*}
H(\dset{Y} | \rscalar{x}^{(1)}, \ldots, \rscalar{x}^{(K_a)}) &= \mathbb{E}_{\Omega} [H(\dset{Y} | \Omega)] \\
&= \mathbb{E}_{\Omega} \left[ \Omega h(\pM) + (Q-\Omega) h(\pF)
\right] \\
&= \mathbb{E}\left[ \Omega \right] h(\pM) + \left(Q-\mathbb{E}\left[ \Omega \right]\right) h(\pF)
\end{flalign*}
$$
= Q \left( 1 - \left(\frac{Q-1}{Q} \right)^{K_a} \right) h(\pM) + Q \left( \frac{Q-1}{Q} \right)^{K_a} h(\pF),
$$
where the second equality follows from the following facts: the channel $\mathfrak{W}_E$ operates on elements independently, $H\left(1_{\{x \in \dset{Y} | x \in \dset{Y}^{(A)} \}}\right) = h(\pM)$ and $H\left(1_{\{x \in \dset{Y} | x \not\in \dset{Y}^{(A)} \}}\right) = h(\pF)$. The fourth equality follows from
\begin{flalign*}
\mathbb{E}\left[ \Omega \right] &= \sum\limits_{x \in [Q]} \mathbb{E} \left[ 1_{\{x \in \dset{Y}^{(A)}\}} \right] = Q \prob{1 \in \dset{Y}^{(A)}} \\
&= Q \left( 1-\left(\frac{Q-1}{Q}\right)^{K_a} \right).
\end{flalign*}

The exact calculation of $H(\dset{Y})$ is more complicated. Let us introduce the notation
\begin{equation}\label{eq:py}
\mu_r = \left( 1-\left(\frac{Q-1}{Q}\right)^{r} \right) \left(1-\pM\right) + \left(\frac{Q-1}{Q}\right)^{r}  \pF
\end{equation}
and use the following estimate:
\[
H(\dset{Y}) \leq \sum\limits_{x \in [Q]} H\left(1_{\{x \in \dset{Y}\}}\right) = Q h\left(\prob{1 \in \dset{Y}}\right) = Q h\left(\mu_{K_a}\right).
\]

We note that the estimates above are quite simple and can be found in~\cite{fengler2020sparcs}. These results are useful for formulating the following theorem.

\begin{theorem}[Outer channel converse, \cite{Frolov2022TCOM}]
\label{prop:converse}
Consider $K_a$ users transmitting messages $W_i \in [M]$, $i \in [K_a]$, over the channel $\mathfrak{W}$ within $L$ channel uses, under a maximum tolerable \ac{pupe} $P_e$. Under the following conditions: (a) the users employ a uniform input distribution, and (b) the decoder output list size is equal to $\ell_0$, the following inequality holds:
\begin{equation}\label{eq:converse_outer}
\log_2{M} \leq 
\frac{1}{\left(1 - P_e\right)}\left(
\frac{L I_u}{K_a} + h\left(P_e\right) + \left(1 - P_e\right)\log_2{\ell_0}\right),
\end{equation}
where $I_u$ is given by~\eqref{eq:cu}.
\end{theorem}
\begin{proof}
Consider the $i$-th user. We have the following Markov chain:
$$
\rscalar{W}_i \to \rvec{x}^{(i)} \to \vdset{Y} \to {\dset{R}},
$$
where $\rvec{x}^{(i)} \in [Q]^L$, $\vdset{Y} = (\dset{Y}_1, \ldots, \dset{Y}_L)$, $\dset{Y}_l \subseteq [Q]$, $l \in [L]$, and $|{\dset{R}}| = \ell_0$. Let $\varepsilon_i = \prob{\rscalar{W}_i \not\in {\dset{R}}}$. Using Fano's list inequality, we have~\eqref{eq:list_fano}, repeated below
$$
H\left(\rscalar{W}_i | \vdset{Y} \right) \leq h(\varepsilon_i) + \left(1 - \varepsilon_i \right) \log_2{\ell_0} + \varepsilon_i  \log_2{\left(M - \ell_0\right)}.
$$

Summing the inequalities~\eqref{eq:list_fano} for $i \in [K_a]$ and using the fact that
$$
H\left(\rscalar{W}_i | \vdset{Y} \right) = H(\rscalar{W}_i) - I(\rscalar{W}_i; \vdset{Y}),
$$
along with the data processing inequality
$$
I\left(\rvec{x}^{(i)};\vdset{Y} \right) \geq I\left(\rscalar{W}_i;\vdset{Y} \right),
$$
we obtain
\begin{flalign*}
\sum\limits_{i=1}^{K_a}I\left(\rvec{x}^{(i)};\vdset{Y} \right) \geq&
\sum\limits_{i=1}^{K_a} H(\rscalar{W}_i) \\
& - \sum\limits_{i=1}^{K_a} h\left(\varepsilon_i\right) - \left(1 - \sum\limits_{i=1}^{K_a} \varepsilon_i \right) \log_2{\ell_0} \\
& - \sum\limits_{i=1}^{K_a} \varepsilon_i \log_2{\left(M - \ell_0\right)}.
\end{flalign*}

Finally, since
$$
H(\rscalar{W}_i) = \log_2M, \ \text{for} \ i \in [K_a],
$$
$$
\sum\limits_{i=1}^{K_a}I\left(\rvec{x}^{(i)};\vdset{Y} \right) \leq L I_u, \quad \sum\limits_{i=1}^{K_a} \varepsilon_i = K_a P_e,
$$
and
$$
\sum\limits_{i=1}^{K_a} h\left(\varepsilon_i\right) \geq K_a h(P_e),
$$
we obtain the result stated in the theorem.

\end{proof}

\section{\texorpdfstring{$t$}{t}-tree code: the average number of paths}\label{appendix:t_tree_linear_paths}

Recall the practical scheme based on linear codes, as presented in Section~\ref{sec:treecode_practical}. In this section, we analyze the \ac{far} and the average decoding complexity. Note that at each decoding step $l \in [L]$, the list of messages $\dset{V}_l$ can be expressed as the following disjoint union:
\begin{equation}\label{eq:total_paths}
\dset{V}_l = \dset{V}_l^{(c)} \bigsqcup \dset{V}_l^{(f)}, \:\: l \in [L],
\end{equation}
where $\dset{V}_l^{(c)}$ represents the beginnings of transmitted messages, and $\dset{V}_l^{(f)}$ corresponds to falsely alarmed messages.

We begin by calculating $\mathbb{E}[|\dset{V}_l^{(c)}|]$ and $\mathbb{E}[|\dset{V}_l^{(f)}|]$ for $l \in [L]$. To do so, we apply random coding to the following ensemble.

\begin{define}
The elements of the ensemble $\mathcal{E}_2\br{b_1, \ldots, b_L, L}$ are obtained by randomly selecting the generator matrix $\dmat{G}$ with the structure defined by~\eqref{eq:G}, i.e., the elements of each matrix $\dmat{G}_{l',l}$, where $1\leq l' \leq l \leq L$, are sampled i.i.d. from the $\mathrm{Bern}(1/2)$ distribution.
\end{define}

\begin{theorem}\label{th:linear_av}
Consider the ensemble $\mathcal{E}_2(b_1, \ldots, b_L, L)$. Let $M_l = 2^{B_l}$. The following bounds hold for $l \in [L]$:

\begin{eqnarray}
\mathbb{E}[|\dset{V}_l^{(c)}|] &=& {v}_l^{(c)} \triangleq M_l \rho_l \lambda_l, \label{eq:vlc} \\
\mathbb{E}[|\dset{V}_l^{(f)}|] &\leq& {v}_l^{(f)} \triangleq M_l \sum\limits_{j=0}^{l-1} \rho_j \lambda_j, \label{eq:vlf}
\end{eqnarray}
where
\begin{flalign}\label{eq:Pl}
&\lambda_j = \left(1 - \frac{1}{M_{j + 1}}\right)^{K_a} - \left(1 - \frac{1}{M_{j}}\right)^{K_a}, \:\: j\in [l-1], \nonumber\\
&\lambda_0 = \left(1 - \frac{1}{M_{1}}\right)^{K_a}, \nonumber\\
&\lambda_l = 1 - \left(1 - \frac{1}{M_{l}}\right)^{K_a},
\end{flalign}
and
\[
\rho_j = \sum\limits_{\begin{array}{c} 0 \leq x \leq j\\ 0 \leq y \leq l-j\\ x+y \leq t\end{array}} \binom{j}{x} \binom{l-j}{y} \pM^x (1-\pM)^{j-x} 
\gamma_1^{y} \gamma_2^{l-j-y}, 
\]
where
\[
\gamma_1 = \left( \frac{K_a}{Q} \right) \pM + \left(1 - \frac{1}{Q} \right)(1-\pF),
\]
and
\[
\gamma_2 =  \left( \frac{K_a}{Q} \right) (1-\pM) + \left(1 - \frac{1}{Q} \right)\pF.
\]
\end{theorem}
\begin{proof}
Assume that the messages
\begin{equation}\label{eq:outer_code_tx_sample}
\ind{\rvec{u}}{1}{}, \ind{\rvec{u}}{2}{}, \ldots, \ind{\rvec{u}}{K_a}{}
\end{equation}
were transmitted. Let $\hat{\dvec{u}}$ be an information word and define $\hat{\dvec{x}}_{[l]} = f_O(\hat{\dvec{u}}_{[l]})$. In what follows, we assume $\hat{\dvec{u}}$ to be fixed, while~\eqref{eq:outer_code_tx_sample} are chosen uniformly at random. 

Introduce the r.v.
$$
\rscalar{D} = d(\vdset{Y}_{[l]}, \hat{\dvec{x}}_{[l]})
$$
and the events
\[
E_j = \bigcap\limits_{i=1}^{K_a} E^c_{i,j}, \quad j \in [l],
\]
where $E_{i,j} = \{ \beg{\hat{\dvec{u}}}{j} = \ind{\rvec{u}}{i}{[j]} \}$ for $i \in [K_a]$. Note that the sequence of events satisfies $E_1 \subseteq E_2 \subseteq \ldots \subseteq E_l$.

Observe that $\beg{\hat{\dvec{u}}}{l}$ is included in $\dset{V}_l^{(c)}$ if it was chosen by at least one user (event $E_l^c$) and $\rscalar{D} \leq t$, and is included in $\dset{V}_l^{(f)}$ if it was not chosen (event $E_l$) and $\rscalar{D} \leq t$. Thus, we have
\begin{eqnarray*}
\mathbb{E}[|\dset{V}_l^{(c)}|] &=& M_l \prob{\rscalar{D} \leq t, E_l^c}, \\
\mathbb{E}[|\dset{V}_l^{(f)}|] &=& M_l \prob{\rscalar{D} \leq t, E_l}.
\end{eqnarray*}

\textit{Calculation of $\mathbb{E}[|\dset{V}_l^{(c)}|]$.} Note that $\prob{E_l^c} = \lambda_l$ since there are $M_l$ possibilities for each $\ind{\rvec{u}}{i}{[l]}$, $i \in [K_a]$, and
\[
\prob{\rscalar{D} \leq t \: | \: E_l^c} = \rho_l = \sum\limits_{i=1}^{t} \binom{l}{i} \pM^i (1-\pM)^{l-i}.
\]
Multiplying by $M_l$, we obtain~\eqref{eq:vlc}.

\textit{Estimation of $\mathbb{E}[|\dset{V}_l^{(f)}|]$.} The main difference compared to the proof of Theorem~\ref{thm:rcb} is as follows. The beginning of the information word $\beg{\hat{\dvec{u}}}{l}$ may coincide with the beginning of one of the transmitted information words. In this case, the beginnings of the codewords will also coincide, which must be accounted for in the analysis. 

Introduce the events $E'_j$ indicating that the longest match length is equal to $j$, i.e.,
\begin{eqnarray}
E'_j = E_{j+1} \setminus E_j, \quad j = 0, \ldots, l-1, \quad \text{where} \quad E_0 \triangleq \emptyset. \label{eq:set_lambda_j}
\end{eqnarray}

We note that $E_l = \bigsqcup_{j=0}^{l-1} E'_j$, and thus
\[
 \prob{\rscalar{D} \leq t, E_l} =  \sum\limits_{j=0}^{l-1} \prob{\rscalar{D} \leq t | E'_j} \cdot \prob{E'_j}.
\]
According to~\eqref{eq:set_lambda_j}, the probability $\prob{E'_j} = \lambda_j$, where $\lambda_j$ is given by~\eqref{eq:Pl}. 

Now consider $\prob{\rscalar{D} \leq t \: | \: E'_j}$. Clearly,
$$
\prob{\hat{x}_{l'} \in \dset{Y}^{(A)}_{l'} \: | \: E'_j} = 1, \quad \text{for} \quad l' \in [j].
$$
Thus, an error can only occur due to missed detection, and the result is determined by i.i.d. r.v.s $\xi_{l'} \sim \mathrm{Bern}(p_m)$.

Consider a slot $l' \in [l] \setminus [j]$. Unlike in Theorem~\ref{thm:rcb}, the symbols $\ind{\rvec{x}}{i}{l'}$ are not independent due to the linear mapping $f_{O,l'}$. Consequently, we apply the following bounds:
\[
\frac{1}{Q} \leq \prob{\hat{x}_{l'} \in \dset{Y}^{(A)}_{l'} | E'_j} \leq \frac{K_a}{Q}, \quad l' \in [l] \setminus [j].
\]
Hence, 
\[
\prob{\hat{x}_{l'} \not\in \dset{Y}_{l'} | E'_j} \leq \gamma_1,
\quad
\prob{\hat{x}_{l'} \in \dset{Y}_{l'} | E'_j} \leq \gamma_2,
\]
for $l' \in [l] \setminus [j]$.

Under the condition $E'_j$, the r.v. $\rscalar{D}$ can be expressed as follows:
\[
\rscalar{D} = \sum\limits_{l' = 1}^j \xi_{l'} + \sum\limits_{l' = j+1}^l \psi_{l'},
\]
where $\xi_{l'} \sim \mathrm{Bern}(p_m)$ and $\psi_{l'} \sim \mathrm{Bern}(\tilde{p})$. The following bounds hold: $\tilde{p} \leq \gamma_1$ and $1-\tilde{p} \leq \gamma_2$. 

Finally, we note that $\prob{\rscalar{D} \leq t \: | \: E'_j} = \rho_j$. Combining the obtained results, we obtain the theorem statement.
\end{proof}

\begin{corollary}\label{thm:linear_bound}
There exists a code $\code \in \mathcal{E}_2(b_1, \ldots, b_L, L)$, such that
\[
P_e = \sum\limits_{i=t+1}^{L} \binom{L}{i} \pM^i (1-\pM)^{L-i} \quad \text{and} \quad P_f \leq {v}^{(f)}_L,
\]
where ${v}^{(f)}_L$ is given by~\eqref{eq:vlf}.
\end{corollary}
\begin{proof}
The proof for $P_e$ is exactly the same as in Theorem~\ref{thm:rcb}. To estimate $P_f$, we apply Markov's inequality and use Theorem~\ref{th:linear_av}. We have
\[
P_f = \prob{|{\dset{R}} \backslash \dset{T}| \geq 1} \leq \mathbb{E}[|{\dset{R}} \backslash \dset{T}|] = \mathbb{E}[|\dset{V}^{(f)}_L|] \leq {v}^{(f)}_L.
\]
\end{proof}

 % Coded CS proofs.
\chapter{Asymptotics}\label{app:asymptotics}

Let us consider asymptotic regime as described in Section~\ref{sec:many_mac2}, where $n \to \infty$ and $K_{a} \to \infty$ such, that $K_a = \mu n$, where $\mu$ represents density of users per channel use, and the number of bits transmitted by each user remains fixed at $k = \log_2M$. As shown, in such a setup, it makes sense to consider only the case where each user has a different codebook of size $M$.

In what follows, we will be interested in the fundamental trade-off between the user density $\mu$ and minimal required energy per bit:
\[
\br{\frac{E_{b}}{N_{0}}}_{\min} = \inf\brc{\frac{nP}{2\log_{2}M}},
\]
where the infimum is taken over all $(n, M, P)$ for which there exist $K_{a} = \mu n$ codebooks, each of size $M$, that are decodable with \ac{pupe} $\leq \epsilon$. 

To formulate the asymptotic bounds, it is convenient to introduce the following notation of the total energy $P_{tot} = \Ka P$, so that 
\[
P_{tot} = 2\frac{E_{b}}{N_{0}}\mu\log_{2}M.
\]
Thus, $P_{tot}$ determines $\frac{E_{b}}{N_{0}}$ and is fixed. Additionally, let $t=\theta \Ka$, then $Pt = P\Ka\theta = \theta P_{tot}$.

\section{Converse bound}
Let us consider the converse bound from \cite{polyanskiy2017perspective}. In this case, we have the best of the following two bounds.

The first one is given by Fano's inequality:
\begin{equation}
\label{eq:asympt_conv_fano}
\br{1-\epsilon} \mu \log_{2} M \leq \frac{1}{2}\log_{2}\br{1 + P_{tot}} + \mu h\br{\epsilon},
\end{equation}
where $h\br{\epsilon} = -\epsilon\log_{2}\br{\epsilon} - \br{1 - \epsilon}\log_{2}\br{1 - \epsilon}$ is the binary entropy function.

The second bound follows from the fact that each user transmits only finitely many bits, $\log_{2}M$ \cite[Theorem 2]{ppv2011}
\begin{equation}
\label{eq:asympt_conv_bit_num}
\log_{2}M \leq -\log_{2} Q\br{\sqrt{\frac{P_{tot}}{\mu}} + Q^{-1}\br{1 - \epsilon}},
\end{equation}
where $Q^{-1}\br{\cdot}$ denotes the inverse of the standard Gaussian tail function $Q\br{\cdot}$. 

For both bounds, we need to determine such minimal $\frac{E_{b}}{N_{0}}$ that corresponding inequality holds and then select the maximal among them.

\section{Achievability bounds}
Now, let us consider the achievability bounds for the asymptotic regime. Recall the \ac{pupe} definition from \eqref{eq:chapter4:pupe_definition}
\[
P_e \leq \sum_{t = 1}^{K_a} \frac{t}{K_a} \Pb{\dset{E}_t} + p_0.
\]
It is clear that in this case, $p_0=0$ due to concentration theorem and different codebooks. 

Let us rewrite the notation of \ac{pupe} in the following form for the fixed $\epsilon > 0$ (omitting $p_0$)
\[
\Tilde{P}_e \leq \sum_{t = \lfloor \epsilon\Ka \rfloor}^{K_a} \frac{t}{K_a} \exp\brs{-nE_{t}}.
\]
and consider the following exponent
\begin{flalign*}
E = & \lim\limits_{n \to \infty} \frac{-\ln \Tilde{P}_{e}}{n} \\
& \geq \lim\limits_{n \to \infty}-\frac{1}{n}\ln \sum\limits_{t=\lfloor{\epsilon\Ka}\rfloor}^{\Ka}\frac{t}{\Ka}\exp\brs{-nE_{t}} \\
& \geq \lim\limits_{n \to \infty}-\frac{1}{n}\ln \sum\limits_{t=\lfloor{\epsilon\Ka}\rfloor}^{\Ka}\exp\brs{-nE_{t}} \\
& \geq \lim\limits_{n \to \infty}-\frac{1}{n}\ln \br{ \Ka \max\limits_{\lfloor \epsilon\Ka \rfloor \leq t \leq \Ka}\brc{\exp\brs{-nE_{t}}}} \\
& = E_{min} = \min\limits_{\epsilon \leq \theta \leq 1}E_{\theta}
\end{flalign*}
where the interval $\epsilon \leq \theta \leq 1$ is justified by the fact that ensuring $E_{min} > 0$ in the asymptotic regime guarantees that \ac{pupe} $\leq \epsilon$.

Thus, for the fixed $\epsilon > 0$, the fundamental tradeoff of $\frac{E_{b}}{N_{0}}$ vs. $\mu$ is defined by the following condition:
$$
E_{min} = \min\limits_{\epsilon \leq \theta \leq 1}E_{\theta} > 0.
$$
In other words, for the fixed $\epsilon$ and $\mu,$ we need to find minimal $\br{\frac{E_{b}}{N_{0}}}_{min}$, such that this condition holds.

To formulate asymptotic achievability bounds, we introduce some helpful notations. Consider the binomial coefficients $R_{1}$ and $R_{2}$, used in the bounds formulation of the Section~\ref{ch4:sec_ach_bounds}, in the asymptotic regime. Since we consider different codebooks, $R_{1} = \frac{1}{n}\log\binom{M-\Ka}{t}$ is replaced by $\frac{1}{n}\log{\br{M-1}^{t}}$. Thus,
\[
\Tilde{R}_{1} = \lim\limits_{n \to \infty}\frac{1}{n}\ln{\br{M-1}^{t}} = \lim\limits_{n \to \infty}\frac{t}{n}\ln{\br{M-1}} = \theta\mu\ln{\br{M-1}}
\]
and
\[
\Tilde{R}_{2} = \lim\limits_{n \to \infty}R_{2} = \lim\limits_{n \to \infty}\frac{1}{n}\ln\binom{\Ka}{t} = \lim\limits_{n \to \infty}\frac{1}{n}\ln\binom{\Ka}{\theta\Ka} = \mu H\br{\theta},
\]
where $H\br{\theta} = -\theta\ln\br{\theta} - \br{1 - \theta}\ln\br{1 - \theta}$ is entropy function. Moreover,
\[
\lim\limits_{n \to \infty}\frac{1}{-n}\ln\br{\exp\brs{-nE_{1}} + \exp\brs{-nE_{2}}} = \min\brc{E_{1}, E_{2}}.
\]

Now, we are ready to formulate the achievability bounds obtained in Section~\ref{ch4:sec_ach_bounds} in the asymptotic regime.

For Theorem~\ref{th:polyanskiy_bound}, we have
\begin{equation}
\label{eq:asympt_polyanskiy_bound}
E_{\theta} = \max\limits_{0 \leq \rho_{1},\rho_{2} \leq 1} \brs{-\rho_{1}\rho_{2}\Tilde{R}_{1} - \rho_{1}\Tilde{R}_{2} + E_{0}\br{\rho_{1}, \rho_{2}} },
\end{equation}
where
\[
E_0\br{\rho_{1}, \rho_{2}} = \frac{1}{2}\br{\rho_1 a + \log(1-2b\rho_1)},
\]
\[
a = \rho_{2} \log (1+2\lambda\theta P_{tot}) + \log (1+2\mu\theta P_{tot}), \quad
b = \rho_{2} \lambda - \frac{\mu}{1+2\mu\theta P_{tot}},
\]
\[
\mu = \frac{\rho_{2} \lambda}{1+2\lambda\theta P_{tot}}, \quad 
\lambda = \frac{\theta P_{tot} -1 + \sqrt{D}}{4 (1+ \rho_{1} \rho_{2}) \theta P_{tot}},
\]
\[
D = (\theta P_{tot}-1)^2 + 4\theta P_{tot} \frac{1 + \rho_{1} \rho_{2}}{1 + \rho_{2}}.
\]

For Theorem~\ref{th:gauss_codebook}, we can write
\begin{equation}
\label{eq:asympt_gauss_codebook}
E_{\theta} = {\max\limits_{\alpha, \beta \geq 0}{\min\br{\Tilde{E}_{1}, \Tilde{E}_{2}}}},
\end{equation}
where
\[
\Tilde{E}_{1} = \max\limits_{u,v > 0, \lambda_{\Tilde{\bA}} > 0}\brc{-\Tilde{R}_{1} - \Tilde{R}_{2} + \frac{1}{2}\log{\det}{\br{\mathbf{I}_{3} - 2\Tilde{\bA}}} - v\beta},
\]
\[
\Tilde{E}_{2} = \max\limits_{\delta > 0, \lambda_{\Tilde{\bB}} > 0}\brc{-\Tilde{R}_{2} + \frac{1}{2}\log{\det{ \br{\mathbf{I}_{2} - 2\delta \Tilde{\bB}} }} + \delta\beta},
\]
\[
\Tilde{\bA} = \br{\begin{array}{rrr}
\br{\alpha - 1}v & \br{\alpha v - u}\sqrt{\theta P_{tot}} & u\sqrt{\theta P_{tot}} \\
\br{\alpha v - u}\sqrt{\theta P_{tot}} & \br{\alpha v - u} \theta P_{tot} & u\theta P_{tot} \\
u\sqrt{\theta P_{tot}} & u\theta P_{tot} & -u\theta P_{tot} \\
\end{array}},
\]
\[
\Tilde{\bB} = \br{\begin{array}{rr}
1 - \alpha & -\alpha\sqrt{\theta P_{tot}} \\
-\alpha\sqrt{\theta P_{tot}} & -\alpha \theta P_{tot} \\
\end{array}},
\]
and $\lambda_{\Tilde{\bA}}$ and $\lambda_{\Tilde{\bB}}$ are the minimum eigenvalues of $\bI_{3} - 2\Tilde{\bA}$ and $\bI_{2} - 2\delta\Tilde{\bB}$ as previous.

Finally, for Theorem~\ref{th:comb_tricks}, we obtain
\begin{equation}
\label{eq:asympt_comb_tricks}
E_{\theta} = \max\limits_{\alpha,\tau > 0}\min\brc{\Tilde{E}_{1}, \Tilde{E}_{2}'},
\end{equation}
\[
\Tilde{E}_{1}=\max\limits_{0 \leq \gamma \leq 1}\brc{-\gamma \Tilde{R}_{1} - \Tilde{R}_{2} + E_{1}\br{\gamma}},
\]
\[
\Tilde{E}_{2} = \max\limits_{\delta > 0}\brc{E_{2}\br{\delta}},
\]
\[
E_{1}\br{\gamma} = \frac{1}{2}\br{\gamma \log(1+\theta P_{tot}) + \log\br{1 - \gamma^2 (1-\rho_e^2)} -\gamma \tau},
\]
\[
E_{2}\br{\delta} = \frac{1}{2} \br{\log\left(1 - \delta^2 (1-\rho_r^2) \right) + \tau \delta},
\]
\[
\rho_e^2 = \frac{(1 + \alpha \theta P_{tot})^2}{(1 + \alpha^2 P_{tot}) (1 + \theta P_{tot})}, \quad
\rho_r^2 = \frac{1}{1+\alpha^2 P_{tot}}.
\]

In \cite{ZPT-isit19}, the best asymptotic achievability bound was obtained by applying Gordon's lemma (see Section~\ref{ch4:sec:gordon_lemma}):
\begin{equation}
\label{eq:asympt_gordon_lemma}
E_{\theta} = -\Tilde{R}_{1} - \Tilde{R}_{2} + \frac{1}{2}\ln{\br{1 + 2\lambda\theta P_{tot}}} + \lambda\frac{\psi\br{\theta, \mu}}{1 + 2\lambda\theta P_{tot}} - \lambda,
\end{equation}
where
\[
\lambda = \frac{\theta P_{tot} + \sqrt{\br{\theta P_{tot}}^{2} + 4\psi\br{\theta,\mu}^{2} - 2}}{4\theta P_{tot}},
\]
\[
\psi\br{\theta,\mu} = \sqrt{1 + \theta P_{tot}} - \gamma\br{\theta}\sqrt{\mu P_{tot}},
\]
\[
\gamma\br{\theta} = \frac{1}{\sqrt{2\pi}}\exp\brs{-\frac{1}{2}\br{Q^{-1}\br{\theta}}^{2}}.
\]

\section{Numerical results}
\begin{figure}
\centering
\includegraphics{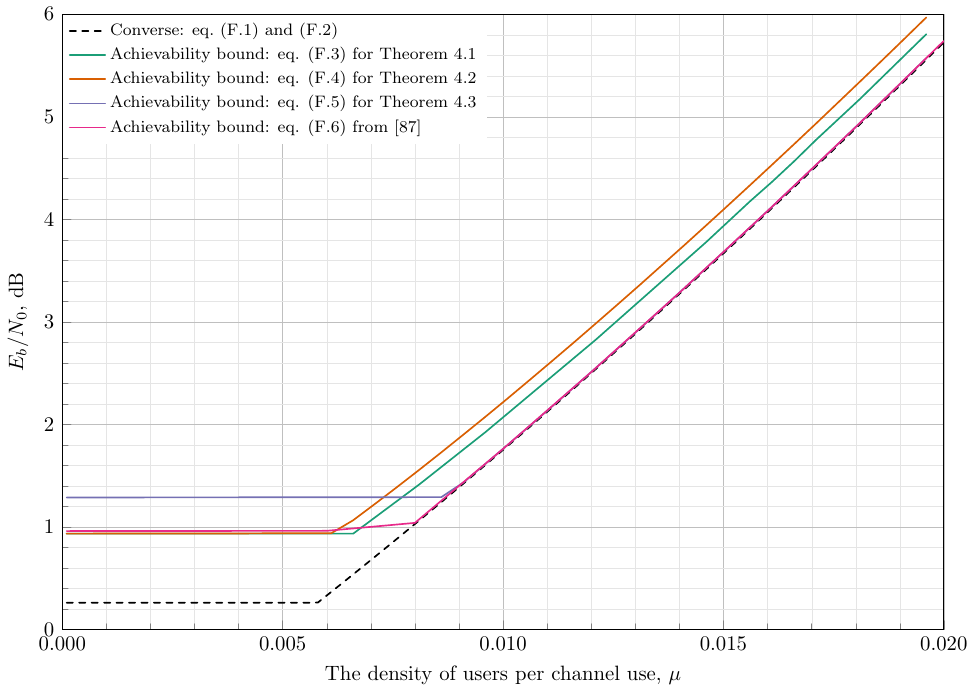}
\caption{Asymptotic bounds for $\epsilon = 10^{-3}$ and $k=100$.}
\label{fig:ebno_mu}
\end{figure}

Let us consider a setup with a \ac{pupe} constraint of $\epsilon = 10^{-3}$ and $k=100$ information bits transmitted by each user. In \Fig{fig:ebno_mu}, we analyze the minimum energy per bit ($E_b/N_0$) as a function of the user density per channel use, $\mu$.

Surprisingly, the energy per bit remains nearly constant until reaching a critical value of $\mu$ (for the converse bound, this value is approximately $\mu \approx 0.006$). The behavior of Theorems~\ref{th:polyanskiy_bound}, \ref{th:gauss_codebook}, and \ref{th:comb_tricks} is almost the same in both the finite-length and asymptotic regimes. Specifically, Theorems~\ref{th:polyanskiy_bound} and \ref{th:gauss_codebook} outperform Theorem~\ref{th:comb_tricks} for small values of $\mu$, but the latter performs better for larger values of $\mu$.

Interestingly, Theorem~\ref{th:gauss_codebook} loses to Theorem~\ref{th:polyanskiy_bound} for large values of $\mu$ in the asymptotic regime, although they almost coincide in the finite-length regime. It is also noteworthy that Theorem~\ref{th:comb_tricks} achieves the converse bound for large values of $\mu$. The best achievability bound is the one from \cite{ZPT-isit19}, although it slightly underperforms compared to Theorems~\ref{th:polyanskiy_bound} and \ref{th:gauss_codebook} at small values of $\mu$.

The most interesting observation is the gap of approximately $0.8$~dB between the converse bound and the achievability bounds for small values of user density, although, the bounds coincide for larger values of user density.
 % Asymptotic analysis (Gordon Lemma)
\chapter{Some useful lemmas}\label{a:expectations}

In this chapter, we provide an overview of some useful lemmas that are used throughout this monograph. We present the Chernoff bound (Lemma~\ref{lemma:chernoff}) and a closed-form solution for the expectation of a quadratic form exponent (Lemma~\ref{lem:quad_form_rv}).

\begin{lemma}[Chernoff bound, \cite{cover2012elements}]\label{lemma:chernoff}
Let $\chi_1$ and $\chi_2$ be arbitrary random variables. Then, for any $u, v > 0$, the following bounds hold:
\[
\prob{\chi_1 \geq 0} \leq \mathbb{E}_{\chi_1} \left[ \exp\left( u \chi_1 \right) \right]
\]
and
\[
\prob{\chi_1 \geq 0, \chi_2 \geq 0} \leq \mathbb{E}_{\chi_1, \chi_2} \left[ \exp\left( u \chi_1 + v \chi_2 \right) \right].
\]
\end{lemma}

\begin{lemma}[Theorem 3.2a.1, \cite{mathai1992quadratic}]\label{lem:quad_form_rv}
Let $\be \sim \mathcal{N}(\bzero, \bI_n)$, $\mathbf{b}$ be an arbitrary vector of length $n$, $\bA$ be a symmetric matrix of size $n \times n$. If $\bI_{n} - 2\bA$ is positive definite, then
\begin{flalign*}
&\mathbb{E}_{\be}\exp\brs{ \be^{T}\bA\be + \mathbf{b}^{T}\be} = \nonumber \\
&= \exp\brs{-\frac{1}{2}\log\det\br{\bI_{n} - 2\bA} + \frac{1}{2}\mathbf{b}^{T}\br{\bI_{n}-2\bA}^{-1}\mathbf{b}}. \nonumber
\end{flalign*}
\end{lemma}

\begin{corollary}\label{cor:quad_form_rv_corollary}
Let $\rvec{z} \sim \mathcal{N}(\bzero, \dmat{I}_n)$, and let $\dvec{b}$ be a fixed vector of length $n$. Then, for any $a > 0$ and $\gamma > -\frac{1}{2a}$, we have
\[
\mathbb{E}_{\rvec{z}}\exp\brs{-\gamma\norm{\sqrt{a}\rvec{z} + \mathbf{b}}^{2}_{2}}=\frac{\exp\brs{-\frac{\gamma\norms{\mathbf{b}}}{1 + 2a\gamma}}}{\br{1 + 2a\gamma}^{\frac{n}{2}}}.
\]
\end{corollary}

\begin{corollary}\label{cor:quad_form_corr_rv_corollary}
Let $\rvec{z} = [\rscalar{z}_1, \rscalar{z}_2] \sim \mathcal{N}(\bzero, \dmat{\Gamma})$ with
\[
\dmat{\Gamma} = \begin{bmatrix}
1 & \rho \\
\rho & 1
\end{bmatrix}.
\]
Then, for any $\lambda_{1}, \lambda_{2}$ satisfying $\lambda_{1}\lambda_{2}\rho^{2} > \br{\lambda_1 - 1}\br{1 + \lambda_2}$, we have
\[
\mathbb{E}_{\rscalar{z}_1,\rscalar{z}_2}\exp\brs{\frac{\lambda_{1}}{2}\rscalar{z}_1^2 - \frac{\lambda_{2}}{2}\rscalar{z}_2^2} = \frac{1}{\br{1 - \lambda_{1} + \lambda_{2} - \lambda_{1}\lambda_{2}\br{1 - \rho^{2}}}^{\frac{1}{2}}}.
\]
\end{corollary}

\begin{lemma}[\cite{mathai1992quadratic}]\label{lemma:qf_inv}
Let $\rvec{x} \sim \mathcal{CN}\left(0, \dmat{\Sigma}\right)$, and let $\dmat{A}$ be a Hermitian matrix. If $\dmat{I} -\dmat{A} \dmat{\Sigma}$ is positive definite, then
\[
\mathbb{E}_{\rvec{x}} \left[ \Exp{\rvec{x}^H \dmat{A} \rvec{x} } \right] = \determ{\dmat{I} -\dmat{A} \dmat{\Sigma}}^{-1}.
\]
\end{lemma}
 % Some useful lemmas placed at the end

\addcontentsline{toc}{chapter}{Bibliography}
\emergencystretch 2em % Fix overfull in bibliography
\bibliographystyle{unsrt}
\bibliography{main}

\end{document}